\documentclass[a4paper,10pt,preprint,3p,number,sort&compress]{elsarticle}
\pdfoutput=1
\usepackage[utf8]{inputenc}
\usepackage[english]{babel}
\usepackage[T1]{fontenc}
\usepackage[no-theorems, no-hyperref]{mplmacros}
\usepackage{hyperref}
\hypersetup{colorlinks=true}
\hypersetup{
	pdfauthor={Martin Plávala},
	pdftitle={General probabilistic theories: An introduction}
}
\usepackage{tikz}
\usetikzlibrary{diatikz}
\usepackage{caption}
\usepackage{subcaption}
\newlength{\rowsep}
\newlength{\originsep}
\newlength{\prepfix}
\makeatletter
\def\ps@pprintTitle{
	\def\@oddfoot{\hfill{\footnotesize\itshape\@date}}
}
\makeatother
\numberwithin{equation}{section}
\theoremstyle{plain}
\newtheorem{proposition}{Proposition}[section]
\newtheorem{theorem}[proposition]{Theorem}
\newtheorem{lemma}[proposition]{Lemma}
\newtheorem{corollary}[proposition]{Corollary}
\theoremstyle{remark}
\newtheorem{example}[proposition]{Example}
\theoremstyle{definition}
\newtheorem{definition}[proposition]{Definition}

\begin{document}

\title{General probabilistic theories: An introduction}

\author{Martin Pl\'{a}vala}
\ead{martin.plavala@uni-siegen.de}
\address{Naturwissenschaftlich-Technische  Fakult\"{a}t, Universit\"{a}t Siegen, 57068 Siegen, Germany}

\begin{abstract}
We introduce the framework of general probabilistic theories (GPTs for short). GPTs are a class of operational theories that generalize both finite-dimensional classical and quantum theory, but they also include other, more exotic theories, such as the boxworld theory containing Popescu-Rohrlich boxes. We provide in-depth explanations of the basic concepts and elements of the framework of GPTs, and we also prove several well-known results. The review is self-contained and it is meant to provide the reader with consistent introduction to GPTs. Our tools mainly include convex geometry, but we also introduce diagrammatic notation and we often express equations via diagrams.
\end{abstract}




\maketitle

{
	\hypersetup{linkcolor=black}
	\tableofcontents
}

\section{Introduction} \label{sec:intro}

%
%
General probabilistic theories (GPTs for short) are a framework developed within the foundations of physics in many different forms and flavors. The main goals of GPTs was to answer the question: what is a physical theory? This question usually appeared in the context of axiomatizations of quantum theory, as many researchers were attempting to derive quantum theory from a set of reasonably motivated axioms. In the current days, the aim of the research is no longer only to search for axiomatizations of quantum theory, the current research in GPTs is oriented towards operational properties of GPTs. It is often investigated what structure is needed to realize certain protocols or constructions known from quantum information theory or classical information theory. For example:
\begin{itemize}
\item The nonlocal features of quantum theory, such as  entanglement, steering and Bell inequality violations we investigated in GPTs \cite{Barrett-GPTinformation,BarnumGaeblerWilce-steering,Banik-steering,KarGhoshChoudharyBanik-uncertainty,Plavala-channels,PlavalaZiman-PRbox,JencovaPlavala-PRbox,BhattacharyaSahaGuhaBanik-nonlocality,CzekajHorodeckiTylec-bipartiteEffects}. One can, for example, construct theories which maximally violate the CHSH inequality and hence provide an implementation of the Popescu-Rohrlich (PR) boxes \cite{PopescuRohrlich-PRbox}. Other research directions include investigating and generalizing steering in GPTs, using the general definition of steering provided in \cite{WisemanDohertyJones-nonlocal}.

\item Uncertainty relations \cite{DahlstenGarnerVedral-uncertainty,KarGhoshChoudharyBanik-uncertainty,SahaOszmaniecCzekajHorodeckiHorodecki-uncertainty,TakakuraMiyadera-uncertainty,TakakuraMyiadera-entropicUncertainty, SunLiangZhouKwekYu=uncertaintyRelations} and incompatibility of measurements and channels \cite{BarnumBarrettLeiferWilce-noBroadcasting,BarnumBarrettLeiferWilce-noBroadcastingPRL,BuschHeinosaariSchultzStevens-compatibility,Plavala-simplex,JencovaPlavala-maxInc,FilippovHeinosaariLeppajarvi-compatibility,Plavala-channels,HeinosaariLeppajarviPlavala-noFreeInformation,Jencova-incomaptibility,Kuramochi-simplex,CzekajSainzSelbyHorodecki-compositeMeasurements,BluhmJencovaNechita-spectrahedra} were investigated in GPTs. Incompatibility of measurements is a generalization of the non-commuting observables in quantum theory to an operational framework. One can then easily extend the definition of incompatibility of measurements to channels. Most of the results show that the existence of certain type of incompatibility or uncertainty is a consequence of some non-classical features of the theory.

\item Noncontextuality \cite{Spekkens-contextuality,Spekkens-toyTheory,ChiribellaYuan-contextuality,SchmidSelbyWolfeKunjwalSpekkens-noncontextuality,SchmidSelbyPuseySpekkens-nonconModels} of GPTs was investigated. Note that some of the works on noncontextuality in operational theories use more general framework, e.g. they do not assume the no-restriction hypothesis.

\item Causal structures \cite{WeilenmannColbeck-causalStructures,ScandoloSalazarKorbiczHorodecki-objectivity}, dynamics \cite{GrossMullerColbeckDahlsten-boxworldDynamics,AlSafiShort-boxworldDynamics,AlSafiRichens-reversibleDynamics,BranfordDahlstenGarner-dynamicsInGPTs,GalleyMasanes-dynamics}, physical properties \cite{GarnerDahlstenNakataMurao-phaseInGPTs,DahlstenGarnerThompsonGuVedral-particleExchange,GalleyGiacominiSelby-gravity}, double-slit and multiple-slit interference \cite{UdudecBarnumEmerson-interference,BarnumLeeScandoloSelby-higherOrderInterference,Kleinmann-multipleSlit,DakicPaterekBrukner-densityCubes,LeeSelby-multipleSlit,LeeSelby-phaseKickback,BarnumMullerUdudec-interferenceDerivatonQT,HorvatDakic-interference} were investigated. Most of the research into physical properties and interference tries to capture some underlying phenomena and investigate them either in terms of black boxes, or in terms of information-theoretic applications, almost always using finite-dimensional effective theories.

\item Computation \cite{BarnumLeeSelby-oraclesComputation,LeeHoban-proofs,LeeSelby-Grover,LeeBarrett-computation,Garner-computation,KrummMuller-computation,BarrettbeaudrapHobanLee-computation}, resource theories \cite{TakagiRegula-resourceTheories,LamiRegulaTakagiFerrari-resourceTheories}, cryptography and other information-theoretic tasks \cite{BarnumWilce-informationPrcessing,BarnumBarrettLeiferWilce-teleportation,MullerUdudec-computation,MullerDahlstenVedral-coinTossing,Chiribella-dilation,BarnumDahlstenLeiferToner-bitCommitment,CzekajHorodeckiHorodeckiHorodecki-informationContent,FilippovHeinosaariLeppajarvi-simulability,BaeKimKwek-discrimination,SelbySikora-money,SikoraSelby-bitCommitment,SikoraSelby-coinFlipping,LamiPalazuelosWinter-dataHiding,YoshidaAraiHayashi-discrimination,BanikSahaGuhaAgrawalBhattachryaRoyMajumdar-informationSymmetry,SahaBhattacharyaGuhaHalderBanik-communication,SahaGuhaBhattacharyaBanik-distributedComputation} were investigated in GPTs. Large part of the protocols used in quantum information theory do not rely on the formalism of Hilbert spaces, but they can be realized using only limited amount of states and measurements. It is then natural to use convex geometry to characterize the required relations between the states and measurements and then one can characterize all theories where certain protocol can be performed.

\item Different notions of entropy \cite{ShortWehner-entropy,KimuraNuidaImai-entropyDistiguishability,BarnumBarrettClarkLeiferSpekkensStepanikWilceWilke-entropy,KimuraIshiguroFukui-entropyHolevo,Takakura-entropy} and foundational aspects of thermodynamics \cite{ChiribellaScandolo-thermodynamics,ChiribellaScandolo-microcanonicalThermodynamics,KrummBarnumBarrettMuller-thermodynamics} were investigated in GPTs. There are several different possible operational constructions that one can use to define entropy in GPTs. Therefore some proofs that are immediate in quantum information theory may, in the framework of GPTs, depend on the chosen definition of entropy and on the properties of the state space.

\item Diagonalization and existence of spectral decompositions were investigated \cite{ChiribellaScandolo-diagonalization,BarnumHilgert-spectral,Gudder-spectralEA,JencovaPlavala-spectralEA}. It seems that some form of spectral decompositions is crucial for singling out quantum and quantum-like theories among other GPTs. Moreover, existence of suitable spectral decompositions would unify the different definitions of entropy in GPTs.

\item The original motivation for GPTs is still an active field of research to this day. There are many papers on what axioms we need to add to GPTs to single-out quantum theory, or how to test whether a theory is quantum \cite{Hardy-derivationQT,ChiribellaDArianoPerinotti-derivationQT,PfisterWehner-discreteGPTs,Kleinmann-emergenceQT,RichensSelbyAlSafi-entanglement,Wilce-derivationQT,MasanesMuller-derivatonQT,LeeSelby-decoherenceToQT,vandeWetering-derivatonQT,vandeWetering-sequential,MazurekPuseyReschSpekkens-experiment,WeilenmannColbeck-selfTesting,GarnerMuller-retractsToQT}. The derivations of quantum theory are usually done in the finite-dimensional framework and they often prove that all GPTs that satisfy certain axioms are connected to Euclidean Jordan algebras. Jordan algebras \cite{Mccrimmon-Jordan} are vector spaces that contain the generalization of the symmetric operator product $\frac{1}{2}(AB + BA)$ and it is known that a class of Jordan algebras, called Euclidean Jordan algebras, contains only quantum and quantum-like theories \cite{JordanNeumannWigner-algebras}.
\end{itemize}

There are also several practical reasons to use GPTs, to name a few: one can often use GPTs to get better understanding of what makes many things in quantum information theory work, or why they give us advantage compared to classical theory. In GPTs, ensembles of objects, conditional probabilities, conditional states and even joint systems of the aforementioned object can be represented by their respective state spaces and so we can treat them as any other state space and we can use known results, instead of having to prove them from scratch, often by mimicking known proofs from other scenarios. Representing all transformations by channels allows us to use the constructions from frameworks based on category theory, such as operational probabilistic theories \cite{ChiribellaDArianoPerinotti-GPTpurification,BisioPerinotti-higherOrder,Perinotti-cellularAutomata} and effectus theory \cite{ChoJacobsWesterbaanWesterbann-effectus}, since one can interpret state spaces as objects and channels as morphisms.

There are several different approaches to GPTs, but all of them are either equivalent or only marginally different. The approach we will use is to start with an abstract definition of state space. We find this approach the easiest to explain and easy to work with, compared to the other options. An equivalent approach would be to start with a table of all possible probabilities that can be generated in an experiment, conditioned on preparation and measurement procedures. One can show that this is the same as starting with an abstract state space, but instead of using vectors we would be describing states in terms of all of the probabilities they can produce. A different approach would be to start with an ordered vector space with order unit, or, equivalently, with a convex effect algebra. We will see that state spaces and effect algebras are dual objects and starting from either one, we can reconstruct the other. Therefore we can freely choose whether we start with state spaces, convex effect algebras, or order unit spaces; we will choose state spaces as our starting point.

Before we proceed further, we must comment on the name of the framework of GPTs. There are different names for the same (or very similar) framework and they usually follow the formula
\begin{equation*}
\begin{gathered}
\text{General} \\
\text{Generalized}
\end{gathered}
\quad
+
\quad
\begin{gathered}
\text{Probability} \\
\text{Probabilistic} \\
\text{Physical}
\end{gathered}
\quad
+
\quad
\begin{gathered}
\text{Theory} \\
\text{Theories}
\end{gathered}
\quad
=
\quad
\text{GPTs}.
\end{equation*}
All possible combinations are frequently and interchangeably used by many authors and the common understanding is that the name can be used as long as the acronym is GPT or GPTs. Some authors do differentiate between the singular GPT and plural GPTs, as follows: a GPT is a concrete theory, with specified state spaces and tensor product, while GPTs is a collective name for the whole class of such theories. We will use the following nomenclature: by general probabilistic theories, shortened to GPTs, we will mean the whole framework including all possible state spaces.

\subsection{Organization of the review}
The review is organized as follows: in Section \ref{sec:what} we explain the basic concepts and operational motivations of preparations, `yes'-`no' questions and transformations. These will later on correspond to states, effects and channels. Note that we will start our construction from the state space (represented by a compact convex set), but we will later show that without the loss of generality one can start from an effect algebra or from order unit space and construct the same framework.

In Section \ref{sec:basic} we introduce the state spaces and effect algebras and we prove some basic results. This includes the duality between the state space and effect algebra, norms on state space and effect algebra and its connection to discrimination tasks. We also discuss the connection of GPTs to abstract effect algebras and order unit spaces and we show that these approaches are essentially equivalent to the one that we develop. At the end of Section \ref{sec:basic} we will introduce the diagrammatic notation, that we will use in subsequent calculations.

In Section \ref{sec:CT} we present the first example: classical theory. This is because classical theory will play an important role in later constructions, mainly in the definition of measurements in Section \ref{sec:channels}. Further examples of quantum theory and boxworld theory will be constructed in Sections \ref{sec:QT} and \ref{sec:boxworld} respectively.

In Section \ref{sec:tensor} we introduce bipartite state spaces and tensor products, as well as partial traces and result on monogamy of entanglement. We introduce tensor products before transformations (i.e., before channels and measurements), because tensor products are helpful when working with channels, due to the isomorphism between linear maps and elements of tensor products of vector spaces, that is reviewed in \ref{appendix:bilinear}.

Then in Section \ref{sec:channels} we introduce channels as transformations between state spaces and we define measurements as special case of channels. We do this to promote the use of channels instead of measurements whenever possible, since the formalism of channels is more suitable for the use of diagrammatic notation, which simplifies certain constructions and allows for easier use of ideas coming from category theory in quantum foundations.

In Section \ref{sec:compatibility} we investigate the concept of compatibility of channels and measurements. We will show, that several well known results in quantum theory and GPTs can be formulated as problems related to compatibility of channels. We also use compatibility to prove results about structure of channels and measurements.

In Sections \ref{sec:QT} and \ref{sec:boxworld} we construct two examples of GPTs: quantum theory and boxworld theory. We postpone the examples to these sections, because we want to present them as clear and concise theories, rater than as different constructions sprinkled in the other sections. But we encourage the reader to skip ahead and look up examples of some of the concepts when reading earlier sections.

In order to make the review as much self-contained as possible, we introduce some mathematical concepts that we need in the appendices. In \ref{appendix:cones} we review the notions of convex cones and ordered vector spaces, in \ref{appendix:duals} we introduce the concept of functionals and the hyperplane separation theorems and in \ref{appendix:bilinear} we introduce the isomorphism between tensor products of vector spaces, vector spaces of bilinear forms and vector spaces of linear maps.

\section{What are GPTs?} \label{sec:what}

%
%
The main objects that we will work with are going to be state spaces, effect algebras and channels. A state space of a theory is going to be identified with a set of equivalence classes of preparation procedures, and an effect is the equivalence class of `yes'-`no' questions that can be answered in an experiment. Before we explain what we mean by the equivalence classes, we will first introduce preparations procedures and `yes'-`no' questions.

A preparation procedure is a list of instructions that one performs to prepare a system in question at the beginning of an experiment. Note that the list of instructions may be conditioned by random events, e.g., the preparation procedure may differ on whether it rains or not. This might seem strange at first, but consider preparation of an experiment testing the tensile strength of a paper, that is to be performed outdoors. If it rains, the paper gets wet and we observe a different outcome of the experiment compared to if it did not rain. A `yes'-`no' question is a list of instructions that we perform after preparing a state to get either the answer `yes' or `no'. It is intuitive that more complex experiments can be build from `yes'-`no' questions, as for example the `yes' and `no' can be interpreted as 1 and 0 and we can reformulate a measurement that outputs a number as a series of `yes'-`no' questions determining the digits of the binary representation of the measured number.

We will say that two preparations are equivalent if the results of all possible `yes'-`no' questions are the same after the two preparations. Analogically, we will say that two `yes'-`no' questions are equivalent if they produce the same answer with respect to all possible (equivalence classes of) preparations.

Channels are going to be the (equivalence class of) list of instructions that we can either append at the end of preparation or, equivalently, prepend to the beginning of a `yes'-`no' question. This already yields a well-known duality: appending instructions to the preparation performs a transformation of the state of the system and it corresponds to the Schr\"{o}dinger picture of quantum theory. Prepending instructions to the beginning of an effect performs a transformation of the measurement and it corresponds to the Heisenberg picture of quantum theory.

Since GPTs are traditionally motivated as a framework for developing axiomatizations of quantum theory, we will provide five postulates about the properties of state spaces in the framework of GPTs. Four of these postulates are intuitive and hard to argue against, while the fifth will limit us to mathematically simpler, but still interesting scenarios.
\begin{definition} \label{def:what-stateSpace}
\emph{State space} is:
\begin{enumerate}[label = (S\arabic*), leftmargin=*]
\item\label{item:what-stateSpace-set} set of points,
\item\label{item:what-stateSpace-convex} convex,
\item\label{item:what-stateSpace-closed} closed in some physically motivated topology,
\item\label{item:what-stateSpace-bounded} bounded,
\item\label{item:what-stateSpace-Euclid} subset of a real, finite-dimensional vector space with Euclidean topology.
\end{enumerate}
\end{definition}
We will use these five postulates to construct the framework of GPTs for single state spaces. We will add additional postulates for bipartite and multipartite state spaces in Section \ref{sec:tensor}. \ref{item:what-stateSpace-Euclid} is the one postulate that simplifies the mathematics used, e.g., we can avoid using abstract notion of convexity in \ref{item:what-stateSpace-convex}. Let $K$ denote a state space, we postulate in \ref{item:what-stateSpace-convex} that the state space is convex, because if $x, y \in K$ are two states and $p \in [0,1]$ we want to be able to describe a scenario where we prepare $x$ with probability $p$ and $y$ with probability $1-p$. We use the convex combination $p x + (1-p) y$ to describe such scenario and we say that $K$ is convex if $px + (1-p)y \in K$ for all $x,y \in K$, $p \in [0,1]$. By requiring $K$ to be convex, we require that $px + (1-p)y$ is a well-defined state. We postulate in \ref{item:what-stateSpace-closed} that the state space is closed, because we assume that if we can prepare a state arbitrary close to some $x$ then we can also prepare $x$. \ref{item:what-stateSpace-bounded} is not necessarily needed, because if the state space would not be bounded, then there would be states that can not be distinguished by any effect and so we would have to factorize the state space to a bounded set. The following results connects Definition \ref{def:what-stateSpace} to the standard introduction of a state space in GPTs:
\begin{proposition} \label{prop:what-stateSpace-compact}
Every state space is a compact convex subset of a real, finite-dimensional vector space.
\end{proposition}
\begin{proof}
We only need to prove that every state space is compact, but this follows from \ref{item:what-stateSpace-closed}, \ref{item:what-stateSpace-bounded} and \ref{item:what-stateSpace-Euclid} since every closed and bounded subset of real finite-dimensional vector space is compact \cite[Theorem 27.3.]{Munkres-topology}.
\end{proof}

In quantum theory, the state space is the set of density operators, that is the set of positive semi-definite operators with trace normalized to one. It is common knowledge that the set of density operators is convex and compact. The underlying vector space is the Hilbert-Schmidt space of self-adjoint operators, which is real and finite-dimensional, given that the underlying Hilbert space is finite-dimensional. We will present quantum theory as an example of a GPT in Section \ref{sec:QT}, but we invite the reader to skip ahead and look up examples of the concepts presented in later sections.

\section{State spaces and effect algebras} \label{sec:basic}

%
%
We have already characterized all state spaces in Proposition \ref{prop:what-stateSpace-compact}. In this section, we will construct the effect algebra and the connection between a state space and its effect algebra. We will also investigate connections to other formalisms: abstract convex effect algebras and order-unit spaces. We will denote by $V$ a real, finite-dimensional vector space and we will denote by $K \subset V$ a compact, convex set, i.e., a state space. Let $X \subset V$, then we will denote $\linspan(X)$ the span of $X$, by $\aff(X)$ the affine hull of $X$, by $\conv(X)$ the convex hull of $X$ and $\cone(X)$ be the smallest cone containing $X$, for definition of cone see Definition \ref{def:cones-cone} in \ref{appendix:cones}. $\dim(V)$ will denote the dimension of $V$. Let $a,b \in \RR$, then we will use $(a,b)$ and $[a,b]$ to denote the open and closed intervals; $\Rp$ will denote the set of non-negative real numbers.

\subsection{State space} \label{subsec:basic-stateSpace}

\begin{definition}
Let $K$ be a state space and let $x \in K$. We say that $x$ is an \emph{extreme point} of $K$, or equivalently that $x$ is a \emph{pure state}, if for every $y,z \in K$ and $\lambda \in (0,1)$ such that $x = \lambda y + (1-\lambda)z$ we have $x = y = z$.
\end{definition}
Pure states are the states that can not be prepared by randomizing preparations of other states. It follows that a pure state $x$ must be preparable by a deterministic and non-randomized preparation procedure. Mixed states are the counterpart to pure states.
\begin{definition}
Let $x \in K$, then we say that $x$ is a \emph{mixed state} if $x$ is not a pure state. We say that $x$ is a \emph{mixture} of $y,z \in K$ if there is $\lambda \in [0,1]$ such that $x = \lambda y + (1-\lambda)z$.
\end{definition}
One way of constructing a mixture $\lambda x + (1-\lambda)y$ of states $x,y \in K$ is to run the experiment $N$ times and to prepare $x$ in $\lambda N$ of the runs and to prepare $y$ in $(1-\lambda)N$ of the runs (assuming $\lambda N$ and $(1-\lambda) N$ are whole numbers). Then the average state that was prepared is exactly the mixture. One can in principle object to constructing the mixture in this way as it was pointed out that knowing how a mixture was prepared is a non-trivial information about the system \cite{Popescu-mixedStates}, but one can bypass this by representing the information about preparation of the mixture as some additional classical information about the system; the classical information can be represented using classical theory that will be introduced in Section \ref{sec:CT}. From now on we will assume that one can prepare a mixture using the aforementioned construction.

The concept of pure state is generalized by the concept of face: a face $F \subset K$ is a convex set of states, such that every state from $F$ can be prepared only by randomizing preparations that are contained in $F$. In other words:
\begin{definition}
Let $F \subset K$ be a convex set such that if for $x, y \in K$ and $\lambda \in (0,1)$ we have $\lambda x + (1-\lambda)y$, then $x, y \in F$. Then we say that $F$ is a \emph{face} of $K$.
\end{definition}
Note that $K$ itself is a face of $K$. Pure states and faces play important roles in the geometry of state spaces, as demonstrated by the following results.

\begin{theorem}[Carath\'{e}odory] \label{thm:basic-stateSpace-Caratheodory}
Let $K \subset V$ be a state space and let $B \subset K$ be a set such that $K = \conv(B)$. Then any $x \in K$ can be expressed as a convex combination of at most $\dim(V) + 1$ points from $B$.
\end{theorem}
\begin{proof}
See \cite[Theorem 17.1]{Rockafellar-convex}.
\end{proof}

\begin{theorem} \label{thm:basic-stateSpace-convHull}
Let $K$ be a state space and let $\ext(K)$ be the set of pure states of $K$ then $K = \conv( \ext(K) )$.
\end{theorem}
\begin{proof}
See \cite[Theorem 18.5]{Rockafellar-convex}.
\end{proof}

Let $X \subset V$ be a convex set, then relative interior of $X$, denoted $\ri(X)$, is the topological interior of $X$ when considered as a subset of $\aff(X)$. For example, the relative interior of an interval $[0, 1] \subset \RR$ is $(0, 1)$, but also the relative interior of a line segment $L$ in arbitrary vector space is the open line segment contained in $L$. A relative interior of a set $\{ v \} \subset V$ is again $\{ v \}$. For finite-dimensional vector space $V$ equipped with the standard Euclidean topology, relative interior of a convex set $X \subset V$ can be defined purely using the convex structure of $X$ as follows:
\begin{equation}
\ri(X) = \{ x \in X : \forall y \in X, \exists \mu > 1, (1-\mu) y + \mu x \in X \}.
\end{equation}
see \cite[Section 6]{Rockafellar-convex}. We will introduce two additional classes of state spaces: polytopes and strictly convex state spaces. Both of them are a good source of examples and counter-examples in many calculations.
\begin{definition}
We say that a state space $K$ is a \emph{polytope} if $\ext(K)$ contains finitely many points, i.e., $K$ has finitely many pure states.
\end{definition}

\begin{definition}
We say that a state space $K$ is \emph{strictly convex} if for any $x, y \in K$ and $\lambda \in (0, 1)$ we have $\lambda x + (1-\lambda) y \in \ri(K)$. Equivalently, $K$ is \emph{strictly convex} if for every face $F$ of $K$ it holds that either $F = \{x\}$, or $F = K$.
\end{definition}

\subsection{Effect algebra}
An effect algebra is a list of all possible (equivalence classes of) `yes'-`no' questions. Moreover we also want to allow the case, when the answer is not only either `yes' or `no', but also a probability of the `yes' answer. This is quite natural, for example consider an unbiased coin. Will it land on the head if we flip it? The standard answer is yes, with probability 50\%. Another reason why we want to allow for probabilities is that we often only predict probabilities in quantum theory and we want our formalism to include quantum theory.

The answer to every `yes'-`no' question can be encoded by a number from the interval $[0,1]$, where we will interpret $p \in [0,1]$ as the probability of the outcome `yes'. As we will see later on, `yes'-`no' questions correspond to two-outcome (also called dichotomic) measurements, but for the time being, the description using only a single number will be sufficient for us, since it is straightforward to see that the probability of `no' is $1-p$. It follows that we can reduce the whole `yes'-`no' question to a single function $f: K \to [0,1]$.

We will require that we get the same result whether we mix the state or whether we mix the probabilities. This is a consistency requirement, as we have already assumed that we can prepare a mixture $\lambda x + (1-\lambda)y$, where $x,y \in K$ and $\lambda \in [0,1]$ by simply running the experiment several times and preparing either $x$ or $y$. Therefore we get
\begin{definition} \label{def:basic-EA-onK}
Let $K$ be a state space, then the \emph{effect algebra} over $K$ will be denoted $E(K)$. $E(K)$ is the set of all affine functions $f: K \to [0,1]$, i.e.,
\begin{equation}
0 \leq f(x) \leq 1
\end{equation}
for all $x \in K$ and
\begin{equation}
f(\lambda x + (1-\lambda) y) = \lambda f(x) + (1-\lambda) f(y)
\end{equation}
for all $x,y \in K$ and $\lambda \in [0,1]$.
\end{definition}

There is one very special element of $E(K)$ that will frequently appear in our calculation:
\begin{definition}
$1_K \in E(K)$ is the constant function given as $1_K(x) = 1$ for all $x \in K$.
\end{definition}

In the following we will construct several auxiliary notions that will appear in further calculations. More specifically, we will construct the cone generated by the effect algebra, the vector space spanned by the effect algebra and we will review some of their properties. For a short introduction to the theory of convex cones see \ref{appendix:cones}.

\begin{proposition} \label{prop:basic-EA-genAK}
Let $A(K)$ be the vector space of affine functions on $K$ and let $A(K)^+$ be the cone of positive affine functions on $K$, i.e.,
\begin{align}
A(K) &= \{ f: K \to \RR : f(\lambda x + (1-\lambda)y) = \lambda f(x) + (1-\lambda)f(y), \forall x,y \in K, \forall \lambda \in [0,1] \}, \\
A(K)^+ &= \{ f \in A(K) : f(x) \geq 0, \forall x \in K \}.
\end{align}
Then $A(K) = \linspan(E(K))$ and $A(K)^+ = \cone(E(K))$, i.e., $A(K)$ is the smallest vector space containing $E(K)$ and $A(K)^+$ is the smallest cone in $A(K)$ that contains $E(K)$.
\end{proposition}
\begin{proof}
We clearly have $E(K) \subset A(K)^+ \subset A(K)$ and so we only need to show that $A(K)$ is contained in $\linspan(E(K))$ and that $A(K)^+$ is contained in $\cone(E(K))$. Let $f \in A(K)^+$, note that since $K$ is compact we must have $M = \max_{x \in K} f(x) < \infty$ and so we can define $f' = \frac{1}{M} f$. Then $f' \in E(K)$ since $f'$ is affine function by construction and for any $x \in K$ we have $0 \leq f'(x) \leq 1$, thus we have $0 \leq f(x) \leq M$. It follows that any $f \in A(K)^+$ can be written as $f = M f'$ where $f' \in E(K)$ and $M \in \Rp$ and so $A(K)^+ \subset \cone(E(K))$. Now let $f \in A(K)$, then we must have $m = \min_{x \in K}f(x) > -\infty$ and we can write $f = (f + \abs{m} 1_K) - \abs{m} 1_K$. Note that $f + \abs{m} 1_K \in A(K)^+$ and $\abs{m} 1_K \in A(K)^+$ and so we have
\begin{equation}
A(K) \subset \linspan(A(K)^+) = \linspan( \cone( E(K) ) ) = \linspan( E(K) ),
\end{equation}
where we have used Lemma \ref{lemma:cones-coneVsSpan} from \ref{appendix:cones}.
\end{proof}
In the following lemma we will use the concepts of pointed and generating cones, for definitions see \ref{appendix:cones}.
\begin{lemma} \label{lemma:basic-EA-AKprop}
$A(K)^+$ is a convex, closed (in the Euclidean topology), pointed and generating cone.
\end{lemma}
\begin{proof}
Let $f,g \in A(K)^+$, then for $\lambda \in [0,1]$ also $\lambda f + (1-\lambda) g$ must be a positive function and so $\lambda f + (1-\lambda) g \in A(K)^+$. Let $\{f_n\}_{n=1}^\infty \subset A(K)^+$ be a Cauchy sequence, then for every $x \in K$ we must have $\lim_{n \to \infty}f_n(x) \geq 0$ simply because $f_n(x) \geq 0$. It follows that $\lim_{n \to \infty} f_n \in A(K)^+$. To see that $A(K)^+$ is pointed, simply observe that if $f \in A(K)^+ \cap (-A(K)^+)$, then for every $x \in K$ we have $0 \leq f(x) \leq 0$, so $f(x) = 0$ and we get $f = 0$. We have already showed in the proof of Proposition \ref{prop:basic-EA-genAK} that $\linspan(A(K)^+) = A(K)$ and so the cone $A(K)^+$ is generating.
\end{proof}

One can introduce a natural partial order to $A(K)$ in an intuitive manner: let $f \in A(K)$, then $f \geq 0$ if and only if $f(x) \geq 0$ for all $x \in K$. For $f,g \in A(K)$, we have $f \geq g$ whenever $f - g \geq 0$, i.e. whenever $f(x) \geq g(x)$ for all $x \in K$. Also we write $f \leq g$ whenever $g \geq f$. One can easily show that this ordering turns $A(K)$ into an ordered vector space.

\begin{lemma}
$A(K)$ with the ordering $\geq$ is a ordered vector space, i.e., for $f,g,h \in A(K)$ and $\lambda \in \Rp$ we have
\begin{itemize}
\item $f \geq g$ implies $f + h \geq g + h$;
\item $f \geq g$ implies $\lambda f \geq \lambda g$;
\item $f \geq f$;
\item $f \geq g$ and $g \geq f$ implies $f = g$;
\item $f \geq g$ and $g \geq h$ implies $f \geq h$.
\end{itemize}
\end{lemma}
\begin{proof}
It is a good exercise to prove the lemma directly. But observe that the the positive cone is exactly $A(K)^+$, which, as we know, is convex, pointed cone. And so it follows from Proposition \ref{prop:cones-orderFromCone} in \ref{appendix:cones} that the order generated by $A(K)^+$, which coincides with the order $\geq$ endows $A(K)$ with the structure of ordered vector space.
\end{proof}

We have introduced the effect algebra, one of the two main building blocks of every GPT, as a set of affine functions with outcomes from the set $[0,1]$, i.e. as functions $f \in A(K)$ such that $0 \leq f(x) \leq 1$. Clearly one can use the ordering $\geq$ on $A(K)$ to characterize the effects as follows:
\begin{proposition} \label{prop:basic-EA-interval}
We have
\begin{equation}
E(K) = \{ f \in A(K) : 0 \leq f \leq 1_K \}.
\end{equation}
\end{proposition}
\begin{proof}
The result follows from the definition: let $f \in A(K)$, then we have $0 \leq f \leq 1_K$ if and only if $0 \leq f(x) \leq 1$ for all $x \in K$.
\end{proof}
Let $f,g \in E(K)$ and define $h_M: K \to \RR$ as $h_M(x) = \max(f(x), g(x))$ for $x \in K$. Then we clearly have $0 \leq h_M(x) \leq 1$, but nevertheless in general $h_M \notin E(K)$. This is because in general $h_M$ is not an affine function; example when $h_M$ is not an affine function can easily be constructed in the boxworld theory presented in Section \ref{sec:boxworld}. In some cases, for example if $f \geq g$ or in the case of the classical theory presented in Section \ref{sec:CT}, $h_M$ is an affine function. Analogical results follow for the function $h_m$ defined as $h_m(x) = \min(f(x),g(x))$. There is one more easy to see result that follows from the definition of $\geq$.
\begin{lemma} \label{lemma:basic-EA-geqDecomp}
Let $f,g \in A(K)$. If $f \geq g$, then there is $h \in A(K)^+$, i.e., $h \geq 0$, such that $f = g + h$.
\end{lemma}
\begin{proof}
Let $f \geq g$, let $x \in K$ and define $h(x) = f(x) - g(x)$. Then $h \in A(K)$ as it is immediate that $h: K \to \RR$ is an affine function. Moreover we have $h(x) \geq 0$ because $f(x) \geq g(x)$ as a result of $f \geq g$, and so $h \in A(K)^+$. The result follows as we have $f = g + h$.
\end{proof}
Proposition \ref{prop:basic-EA-interval} and Lemma \ref{lemma:basic-EA-geqDecomp} represent the two use cases of the ordering $\geq$: we will use Proposition \ref{prop:basic-EA-interval} to express a condition that some element $f \in A(K)$ is an effect, i.e., that $f \in E(K)$, and we will use Lemma \ref{lemma:basic-EA-geqDecomp} to express the decomposition $f = g + h$ in a different way. This can be demonstrated by the following lemma:
\begin{lemma} \label{lemma:basic-EA-fPerp}
For every $f \in E(K)$ there is $f^\perp \in E(K)$ such that $f + f^\perp = 1_K$.
\end{lemma}
\begin{proof}
Since $f \in E(K)$, according to Proposition \ref{prop:basic-EA-interval} we have $f \leq 1_K$, from which according to Lemma \ref{lemma:basic-EA-geqDecomp} we have $1_K = f + f^\perp$ for some $f^\perp \in A(K)^+$. $1_K \geq f^\perp$ and $f^\perp \in E(K)$ follows. Another way to prove the statement is to show that $1_K - f \geq 0$, then take $f^\perp = 1_K - f$.
\end{proof}

The unit effect $1_K$ has one additional property, that we have already used in the proof of Proposition \ref{prop:basic-EA-genAK}: a suitable multiple of $1_K$ can be used to bound any element of $A(K)$ from above. This property can be generalized as follows:
\begin{definition}
Let $f \in A(K)^+$ be an element such that for every $g \in A(K)$ there is $\mu \in \Rp$ such that $g \leq \mu f$. Then we say that $f$ is the \emph{order unit} of $A(K)^+$.
\end{definition}

\begin{lemma}
$1_K$ is the order unit of $A(K)^+$.
\end{lemma}
\begin{proof}
Let $g \in A(K)$ and let $M = \max_{x \in K} g(x)$. Then we have $g \leq M 1_K$.
\end{proof}

\subsection{Duality between the state space and the effect algebra}
We have already showed how one can start with the state space and construct the effect algebra $E(K)$. Now we will show that given an effect algebra $E(K)$, we can reconstruct the state space $K$ from $E(K)$. In the following we will rely on the notions of linear functionals and dual vector space, see \ref{appendix:duals} for a short introduction to linear functionals, dual vector spaces and dual cones.

Let $A(K)^*$ denote the dual vector space to $A(K)$, that is, let $A(K)^*$ be the vector space of all linear functionals on $A(K)$. Let $\psi \in A(K)^*$ and $f \in A(K)$, then we will denote by $\< \psi, f \>$ the value that $\psi$ assigns to $f$, i.e., $\psi: f \mapsto \< \psi, f \>$. The presented notation is similar to the bra-ket notation for the inner product in quantum theory, but remember that $\< \psi, f \>$ is not an inner product of two vectors, because $f \in A(K)$ and $\psi \in A(K)^*$, i.e., $f$ and $\psi$ belong to different vector spaces. That being said, since we are working only in finite-dimensional spaces, the structure is very similar to an inner product and one can think of $\< \psi, f \>$ as a special type of inner product, which works only for a pair of $\psi \in A(K)^*$ and $f \in A(K)$, but does not work for a pair $\psi, \varphi \in A(K)^*$, neither for a pair $f,g \in A(K)$. This type of inner product-like structure is usually called pairing or duality in linear algebra textbooks.

The dual cone to $A(K)^+$ is
\begin{equation}
A(K)^{*+} = \{ \psi \in A(K)^* : \<\psi, f\> \geq 0, \forall f \in A(K)^+ \},
\end{equation}
i.e., $A(K)^{*+}$ is the cone of functionals that are positive on the cone of positive functions $A(K)^+$. In terms of \ref{appendix:duals}, $A(K)^{*+} = (A(K)^+)^*$.
\begin{lemma} \label{lemma:basic-dual-AKstarPointedGenerating}
$A(K)^{*+}$ is a convex, closed, pointed and generating cone.
\end{lemma}
\begin{proof}
The result follows from Lemma \ref{lemma:basic-EA-AKprop} and Propositions \ref{prop:duals-dualConvexClosed}, \ref{prop:duals-generatingToPointed} and \ref{prop:duals-pointedToGenerating} in \ref{appendix:duals}.
\end{proof}

We can naturally embed $K$ into $A(K)^{*+}$ using the following construction: let $x \in K$, $f,g \in A(K)$ and $\alpha, \beta \in \RR$, then we have
\begin{equation}
(\alpha f + \beta g)(x) = \alpha f(x) + \beta g(x)
\end{equation}
and so the expression $f(x)$ is linear in $f$. It follows that we can define a linear functional $x \in A(K)^{*+}$ by
\begin{equation}
\< x, f \> = f(x).
\end{equation}
We are abusing the notation by using the same symbol for the point $x \in K$ and the functional $x \in A(K)^{*+}$, but as we will shortly see, they are in fact isomorphic to each other. Also note that the functional $x \in A(K)^{*+}$ is positive by construction, since for every $f \in A(K)^+$ we must have $\< x, f \> = f(x) \geq 0$. Moreover note that $\< x, 1_K \> = 1$. We have, up to an isomorphism, $K \subset \{ \varphi \in A(K)^{*+} : \< \varphi, 1_K \> = 1 \}$. We will now show that also the other inclusion holds.

\begin{theorem} \label{thm:basic-dual-stateSpace}
For every state space $K$ we have
\begin{equation}
K = \{ \varphi \in A(K)^{*+} : \< \varphi, 1_K \> = 1 \}
\end{equation}
up to an isomorphism.
\end{theorem}
\begin{proof}
We will omit the isomorphism in the proof. Let $\psi \in \{ \varphi \in A(K)^{*+} : \< \varphi, 1_K \> = 1 \}$ and assume that $\psi \notin K$. Then according to the strict hyperplane separation theorem \ref{thm:duals-strictHyperplaneSeparation} there is an affine function $f \in A(K)$ such that
\begin{equation}
\< \psi, f \> < 0 < \< x, f \>
\end{equation}
for all $x \in K$; note that we have used that $A(K)^{**} = A(K)$ as proved in Proposition \ref{prop:duals-doubleDual}, since $f$ should actually be a functional on $A(K)^*$. Also note that $f$ is affine on $V = \{ \varphi \in A(K)^*: \< \varphi, 1_K \> = 1 \} = \aff(K)$. Since $\<x, f\> = f(x) > 0$ for all $x \in K$, we must have $f \in A(K)^+$. Then $\< \psi, f \> < 0$ is a contradiction with $\psi \in A(K)^{*+}$, so we must have $\psi \in K$.
\end{proof}

One has to be careful when working with a concrete theory, because although $K \subset A(K)^{*+}$ up to an isomorphism, one has to take this isomorphism into account when describing states from $K$ by vectors from $V = \aff(K)$ or by vectors from $A(K)^*$. It is usually preferred and more useful to describe states as vectors from $A(K)^*$, but note that $\dim(V) + 1 = \dim(A(K)^*)$ which has to be taken into account.

A standard approach is to do the following: first find a suitable representation of $K \subset V$, this is actually where we started to build the framework. Then for every $x \in K$, apply the map
\begin{equation} \label{eq:basic-dual-mapToK'}
x \mapsto
\begin{pmatrix}
x \\
1
\end{pmatrix}
\end{equation}
which simply adds the $1$ at the end of the vector representation of every state. Then the set
\begin{equation} \label{eq:basic-dual-K'}
K' = \left\lbrace
\begin{pmatrix}
x \\
1
\end{pmatrix}
: x \in K \right\rbrace
\end{equation}
is clearly isomorphic to $K$ and moreover
\begin{equation}
K' = \{ \varphi \in A(K)^{*+} : \< \varphi, 1_K \> = 1 \}
\end{equation}
even without the isomorphism. We can formalize this as follows:
\begin{proposition} \label{prop:basic-dual-rep}
Let $K \subset V$ be a state space, where $V = \aff(K)$ then:
\begin{enumerate}[label=(R\arabic*), leftmargin=*]
\item\label{item:basic-dual-rep-dim} $\dim(V) + 1 = \dim(A(K)^*)$.
\item\label{item:basic-dual-rep-iso} Let $K'$ be as given in \eqref{eq:basic-dual-K'}, then $K'$ is isomorphic to $K$.
\item\label{item:basic-dual-rep-functionals} $K' \subset A(K)^*$.
\item\label{item:basic-dual-rep-stateSpace} $K' = \{ \varphi \in A(K)^{*+} : \< \varphi, 1_K \> = 1 \}$.
\end{enumerate}
\end{proposition}
\begin{proof}
To prove \ref{item:basic-dual-rep-dim}, note that $\dim(A(K)^*) = \dim(A(K))$, we will show that $\dim(V) + 1 = \dim(A(K))$. Let $0 \in V$ be the zero vector and let $f \in A(K)$, then since $f$ is only affine function, not linear, we can have $f(0) \neq 0$. One can see that this is the only difference between affine and linear functions and any affine function such that $f(0) = 0$ is actually linear, i.e. if $f(0) = 0$ then $f \in V^*$. It follows that $f - f(0) 1_K \in V^*$ and so every $f \in A(K)$ can be described as an ordered pair $(\varphi, c)$, where $\varphi \in V^*$, $\varphi = f - f(0) 1_K$ and $c \in \RR$, $c = f(0)$. For any arbitrary ordered pair $(\varphi, c)$, let $x \in K$ and define
\begin{equation}
(\varphi, c)(x) = \varphi(x) + c 1_{K}(x).
\end{equation}
We get that $(\varphi, c) \in A(K)$, and so the vector space of ordered pairs $(\varphi, c)$ is isomorphic to $A(K)$. Note that the vector space of ordered pairs $(\varphi, c)$ has exactly one more dimension that $V$ and so \ref{item:basic-dual-rep-dim} follows. \ref{item:basic-dual-rep-iso} is straightforward, the map described in \eqref{eq:basic-dual-mapToK'} is affine and invertible. To prove \ref{item:basic-dual-rep-functionals}, let $x \in K$ and let $f \in A(K)$ correspond to the ordered pair $(\varphi, c)$, where $\varphi \in V^*$ and $c \in \RR$, then let
\begin{equation}
\left\langle
\begin{pmatrix}
x \\
1
\end{pmatrix},
(\varphi, c)
\right\rangle
= \varphi(x) + c.
\end{equation}
Notice that this expression is linear in $(\varphi, c)$ and so elements of $K'$ are functionals on $A(K)$, so $K' \subset A(K)^*$. Finally, to prove \ref{item:basic-dual-rep-stateSpace}, note that $K' \subset \{ \varphi \in A(K)^{*+} : \< \varphi, 1_K \> = 1 \}$ is immediate. We can now use the strict hyperplane separation theorem \ref{thm:duals-strictHyperplaneSeparation} in the same way as in the proof of Theorem \ref{thm:basic-dual-stateSpace} to show that $K' = \{ \varphi \in A(K)^{*+} : \< \varphi, 1_K \> = 1 \}$.
\end{proof}

\subsection{Connection to abstract convex effect algebras and order unit spaces}
One does not have to start building the formalism of GPTs from the state space $K$, but a different approach is to start with the abstract definition of convex effect algebra. We will show that starting from abstract effect algebras leads to the order unit spaces. Then, we will then show that the formalism of order unit spaces is equivalent to our framework.

\begin{definition}
An \emph{effect algebra} is a system $(E, 0, 1, +)$ where $E$ is a set, $0, 1 \in E$ and $+$ is a partially defined binary operation. Let $a, b \in E$ then we write $a \perp b$ if $a + b$ is defined. Let $a, b, c \in E$, then it must hold that:
\begin{enumerate}[label=(EA\arabic*), leftmargin=*]
\item If $a \perp b$ then also $b \perp a$ and $a +  b = b + a$.
\item If $a \perp b$ and $(a + b) \perp c$ then $b \perp c$, $a \perp (b + c)$ and $(a + b) + c = a + (b + c)$.
\item For every $a \in E$ there is $a' \in E$ such that $a \perp a'$ and $a + a' = 1$.
\item If $a \perp 1$ then $a = 0$.
\end{enumerate}
\end{definition}
Effect algebras were defined by Foulis and Bennet in 1994 \cite{FoulisBennet-EA}, but equivalent structure of D-posets was already presented by K\^{o}pka in 1992 \cite{Kopka-Dposets}. Effect algebras were introduced as a generalization of the projectors in quantum theory and they were heavily investigated, see \cite{DvurecenskijPulmannova-structures} for a review. Special class of effect algebras are convex effect algebras.
\begin{definition}
Let $E$ be an effect  algebra. $E$ is \emph{convex effect algebra} if for every $a \in E$ and $\lambda \in [0, 1]$ there is an element $\lambda a \in E$ such that for all $\lambda, \mu \in [0, 1]$ and $a, b \in E$ we have
\begin{enumerate}[label=(CEA\arabic*), leftmargin=*]
\item $\mu(\lambda a) = \lambda (\mu a)$.
\item If $\lambda + \mu \leq 1$, then $\lambda a \perp \mu a$ and $\lambda a + \mu a = (\lambda + \mu) a$.
\item If $a \perp b$, then $\lambda a \perp \lambda b$ and $\lambda a + \lambda b = \lambda (a + b)$.
\item $1a = a$.
\end{enumerate}
\end{definition}
Special class of convex effect algebras are effect algebras which are intervals in real ordered vector spaces. As we will see, these effect algebras are closely related to our definition of $E(K)$ as in Definition \ref{def:basic-EA-onK}.
\begin{proposition} \label{prop:basic-connection-OVStoEA}
Let $V$ be a real vector space with a pointed cone $C$, that is $C \cap -C = \{ 0 \}$. Let $v, w \in V$ and define the partial order $\geq$ as $v \geq w$ if and only if $v - w \in C$. This gives $V$ the structure of ordered vector space, see Proposition \ref{prop:cones-orderFromCone}. Let $u \in C$ then the interval
\begin{equation}
[0, u] = \{ v \in V: 0 \leq v \leq u \}
\end{equation}
is a convex effect algebra with the partially defined binary operation of sum of vectors and $1 = u$. The convex structure is given by multiplication of vectors by scalars.
\end{proposition}
\begin{proof}
The proof is straightforward. Let $a,b \in [0,u]$, then $a \perp b$ if and only if $a+b \in [0,u]$. Let $a,b,c \in [0, u]$, then $a+b = b+a$, and $a+b \leq u$ if and only if $b+a \leq u$. If $a+b \leq u$ and $(a+b)+c \leq u$ then we have $b+c \leq a+b+c \leq u$. We define $a' = u-a$ as the unique element such that $a + a' = u$ and $a' \in [0, u]$ if and only if $a \in [0, u]$. At last if $a+u \leq u$ then $a \leq 0$ but also $a \geq 0$ so we must have $a = 0$. This shows that $[0, u]$ is an effect algebra.

Keep $a,b \in [0, u]$ and let $\lambda, \mu \in [0, 1]$. Clearly we have $\lambda (\mu a) = \lambda \mu a = \mu (\lambda a)$. If $\lambda + \mu \leq 1$, then $\lambda a + \mu a = (\lambda + \mu) a \leq a \leq u$. If $a+b \leq u$, then also $\lambda a + \lambda b = \lambda (a+b) \leq a+b \leq u$. At last, $1a = a$ is trivial. This shows that $[0, u]$ is convex effect algebra.
\end{proof}
\begin{definition} \label{def:basic-connection-linEA}
Let $V$ be a real vector space with a convex, pointed cone $C$ and let $u \in C$, then we call $[0, u]$ \emph{linear effect algebra}.
\end{definition}
The following justifies why we used the term effect algebra in Definition \ref{def:basic-EA-onK}.
\begin{corollary} \label{coro:basic-connection-EKisLin}
Let $K$ be a state space and let $E(K)$ be the effect algebra over $K$ as given in Definition \ref{def:basic-EA-onK}. Then $E(K)$ is a linear effect algebra with $u = 1_K$, $C = A(K)^+$ and $V = A(K)$.
\end{corollary}
\begin{proof}
It follows from Proposition \ref{prop:basic-EA-interval} that we have $E(K) = [0, 1_K] \subset A(K)$.
\end{proof}
As we have seen, every linear effect algebra is convex effect algebra. The other implication holds as well.
\begin{theorem}
Every convex effect algebra is affinely isomorphic to a linear effect algebra.
\end{theorem}
\begin{proof}
See \cite{GudderPulmannova-convexEA} for the proof.
\end{proof}
Thus we see that convex effect algebras are isomorphic to linear effect algebras. We will now proceed to explain how linear effect algebras are isomorphic to order unit spaces. For a definition of ordered vector space, see Definition \ref{def:cones-OVS}.
\begin{definition}
\emph{Order unit space} $(V, \leq, u)$ is an ordered vector space $(V, \leq)$ with an order unit $u \in V$, which is an element such that for any $v \in V$ there is $\lambda \in \Rp$ such that $v \leq \lambda u$.
\end{definition}

\begin{proposition}
There is a one-to-one correspondence between order unit spaces and linear effect algebras.
\end{proposition}
\begin{proof}
We already know that if $(V, \leq, u)$ is an ordered vector space, then
\begin{equation}
E = [0, u] = \{ v \in V: 0 \leq v \leq u \}
\end{equation}
is a linear effect algebra, see Proposition \ref{prop:basic-connection-OVStoEA}. Let now $V$ be a vector space, $C \subset V$ a pointed cone, $u \in C$ and consider the linear effect algebra $E = [0, u]$. Without the loss of generality, we can assume that $V = \linspan(E)$; if $V \neq \linspan(E)$, then we can replace $V$ and $C$ with $V \cap \linspan(E)$ and $C \cap \linspan(E)$. $V = \linspan(E)$ also implies $V = \linspan(\cone(E))$, see Lemma \ref{lemma:cones-coneVsSpan}. It follows that $\cone(E)$ is generating, and so for every $v \in V$ there are $a,b \in E$ and $\lambda, \mu \in \Rp$ such that $v = \lambda a - \mu b$. We then have
\begin{equation}
v = \lambda a - \mu b \leq \lambda a \leq \lambda u,
\end{equation}
so $u$ is an order unit, and $(V, \leq, u)$ is an order unit space. One can also prove that $C = \cone(E)$.
\end{proof}
Thus we have showed that building an operational framework based on convex effect algebras, linear effect algebras and order unit spaces is equivalent. Now we will proceed to show that this is also equivalent to our framework based on state spaces. Hence one can, without the loss of generality, choose any of the possible starting points and obtain the same framework. We already know that given a state space $K$, $E(K)$ is a linear effect algebra, see Corollary \ref{coro:basic-connection-EKisLin}. We will now show that given a linear effect algebra $E$, one can construct a state space $S(E)$, such that $E$ is an effect algebra on $S(E)$, i.e., such that $E = E(S(E))$.

For simplicity, we will assume that the cone $C$ used to construct a linear effect algebra $[0,u]$ is closed and generating. If $C$ would not be generating, then we can without the loss of generality restrict to $\linspan(C)$. We assume that $C$ is closed for the sake of simplicity, but this assumption can also be operationally motivated as in Definition \ref{def:what-stateSpace}.
\begin{definition}
Let $V$ be real, finite-dimensional vector space, $C \subset V$ be convex, closed, pointed, and generating cone. Let $u \in C$ be an order unit and let $E = [0, u]$ be a linear effect algebra. Let $C^*$ be the dual cone to $C$, see Definition \ref{def:duals-dualCone}, and let
\begin{equation}
S(E) = \{ \psi \in C^* : \< \psi, u \> = 1 \}.
\end{equation}
We call $S(E)$ the \emph{state space} of $E$.
\end{definition}
Now we can treat $S(E)$ as a state space, we can construct $E(S(E))$, that is an effect algebra on $S(E)$. Using that $C^{**} = C$, see Proposition \ref{prop:duals-doubleDualCone}, one can show that
\begin{equation} \label{eq:basic-connection-ESE}
E(S(E)) = \{ a \in C : 0 \leq a \leq u \} = E.
\end{equation}
And so given a linear effect algebra $E$, we can construct the state space $S(E)$, but we can also use $S(E)$ to reconstruct $E$ using \eqref{eq:basic-connection-ESE}. Hence the framework one would get using linear effect algebras is equivalent to the framework we get using state spaces.

\subsection{Some useful results}
In this subsection we will present collection of simple results about state spaces, effect algebras and the underlying cones. All of these results are easy to prove and we invite the reader to try and do the proofs themselves.

\begin{lemma} \label{lemma:basic-results-subsets}
Let $V$ be a real, finite-dimensional vector space and let $K_A \subset V$ be a state space. Let $K_B \subset V$ be another state space such that $\linspan(K_B) = V$ and
\begin{equation}
K_B \subset K_A.
\end{equation}
Then
\begin{equation}
E(K_A) \subset E(K_B).
\end{equation}
\end{lemma}
\begin{proof}
Note that $\linspan(K_B) = \linspan(K_A)$ implies $A(K_A) = A(K_B)$. Let $f \in E(K_A)$ and let $x \in K_B$. From $K_B \subset K_A$ we get $x \in K_A$ and it follows that $0 \leq \<x, f\> \leq 1$ and so $f \in E(K_B)$.
\end{proof}

\begin{lemma}
Let $f, g \in A(K)^+$, then $f \geq g$ if and only if $1_K - g \geq 1_K - f$.
\end{lemma}
\begin{proof}
We have $f \geq g$ if and only if $-g \geq -f$ if and only if $1_K - g \geq 1_K - f$.
\end{proof}

\begin{lemma}
Let $f,g \in A(K)^+$ be such that $f \geq g$. If $\<x,f\> = 0$ for some $x \in K$, then also $\<x,g\> = 0$.
\end{lemma}
\begin{proof}
From $f \geq g$ we get $f - g \geq 0$ and we must have $\< x, f-g \> \geq 0$. Using also $g \geq 0$ we get $\<x, f\> \geq \<x, g\> \geq 0$. Since $\<x,f\> = 0$, we get $0 \geq \<x,g\> \geq 0$ and so $\<x,g\> = 0$.
\end{proof}

\begin{lemma}
Let $f,g \in E(K)$ be such that $f \geq g$. If $\<x, g\> = 1$ for some $x \in K$, then $\<x, f\> = 1$.
\end{lemma}
\begin{proof}
Let $x \in K$, then from $f \geq g$ and $f \in E(K)$ we get $1 \geq \<x, f\> \geq \<x, g\>$. If $\<x,g\> = 1$, then we have $1 \geq \<x, f\> \geq 1$ so we get $\<x,f\> = 1$.
\end{proof}

\begin{lemma} \label{lemma:basic-results-base}
For every $\psi \in A(K)^{*+}$ there is some $x \in K$ and $\lambda \in \Rp$ such that $\psi = \lambda x$.
\end{lemma}
\begin{proof}
Let $\psi \in A(K)^{*+}$ and let $x = \frac{1}{\<\psi,1_K\>} \psi$. Then we have $x \in A(K)^{*+}$ and $\<x, 1_K\>$ and so according to Theorem \ref{thm:basic-dual-stateSpace} we have $x \in K$. The result follows from $\psi = \<\psi, 1_K\> x$.
\end{proof}
The interpretation of Lemma \ref{lemma:basic-results-base} is that $K$ is a base of the cone $A(K)^{*+}$. Let $C$ be a cone, then a base of $C$ is a convex set $B \subset C$ such that for every $v \in C$ there are unique $x \in B$ and $\lambda \in \Rp$ such that $v = \lambda x$.

\begin{lemma} \label{lemma:basic-results-span}
For every $\psi \in A(K)^*$ there are $x,y \in K$ and $\lambda, \mu \in \Rp$ such that $\psi = \lambda x - \mu y$.
\end{lemma}
\begin{proof}
We know that $A(K)^{*+}$ is generating, see Lemma \ref{lemma:basic-dual-AKstarPointedGenerating}, so there are $\varphi, \xi \in A(K)^{*+}$ such that $\psi = \varphi - \xi$. According to Lemma \ref{lemma:basic-results-base} there are $x,y \in K$ and $\lambda, \mu \in \Rp$ such that $\varphi = \lambda x$, $\xi = \mu y$. The result follows.
\end{proof}

\begin{lemma}
Let $f \in A(K)$, then $f \in A(K)^+$ if and only if $\< \psi, f \> \geq 0$ for every $\psi \in A(K)^{*+}$.
\end{lemma}
\begin{proof}
One can either use the result of Proposition \ref{prop:duals-doubleDualCone} that $(A(K)^{*+})^* = A(K)^+$, i.e., that the dual cone of $A(K)^{*+}$ is again $A(K)^+$. Alternatively, using Lemma \ref{lemma:basic-results-base} we get that $\<\psi, f\> \geq 0$ for all $\psi \in A(K)^{*+}$ if and only if $\<x, f\> \geq 0$ for all $x \in K$. But if $\<x, f\> \geq 0$ for all $x \in K$, then $f \in A(K)^+$ by definition.
\end{proof}

\subsection{Order unit norm, base norm and discrimination tasks}
We have been so far avoiding specifying the topology on $K$, $A(K)$ and $A(K)^*$. This was not a burning issue as we are working with only finite-dimensional vector spaces. We will now introduce the topology by introducing norms to $A(K)$ and $A(K)^*$. We will assume that the reader is familiar with some basic facts about normed vector spaces, if not, then we recommend \cite{NaylorSell-operators}. We will start by introducing the norm to $A(K)$ since it is just the supremum norm for the functions.
\begin{proposition}
Let $K$ be a state space, $f \in A(K)$ and define
\begin{equation}
\norm{f} = \sup_{x \in K} \abs{\<x, f\>},
\end{equation}
then $\norm{\cdot}$ is a norm on $A(K)$.
\end{proposition}
\begin{proof}
Assume that $\norm{f} = 0$, then we have $\<x,f\> = 0$ for all $x \in K$ since $0 \leq \abs{\<x,f\>} \leq \norm{f} = 0$ and $f = 0$ follows. Let $\alpha \in \RR$, then clearly
\begin{equation}
\norm{\alpha f} = \sup_{x \in K} \abs{\alpha \<x, f\>} = \abs{\alpha} \sup_{x \in K} \abs{\<x, f\>} = \abs{\alpha} \norm{f}
\end{equation}
and so $\norm{\cdot}$ is homogeneous. Finally let $f,g \in A(K)$, then we have
\begin{equation}
\norm{f+g} = \sup_{x \in K} \abs{\<x, f\> + \<x, g\>} \leq \sup_{x \in K} \abs{\<x, f\>} + \sup_{y \in K} \abs{\<y, g\>} = \norm{f} + \norm{g}
\end{equation}
and so also triangle inequality holds.
\end{proof}
The following gives an equivalent expression for the norm on $A(K)$.
\begin{proposition} \label{prop:basic-norms-OUnormfromOU}
Let $f \in A(K)$, then
\begin{equation} \label{eq:basic-norms-OUnorm}
\norm{f} = \inf \{\lambda \in \Rp : -\lambda 1_K \leq f \leq \lambda 1_K \}.
\end{equation}
\end{proposition}
\begin{proof}
Let $f \in A(K)$ and $\lambda \in \Rp$ be such that $-\lambda 1_K \leq f \leq \lambda 1_K$. Then for every $x \in K$ we have $-\lambda \leq \<x,f\> \leq \lambda$, which implies $\abs{\<x,f\>} \leq \lambda$ and it follows that $\norm{f} \leq \lambda$, so we get $\norm{f} \leq \inf \{\lambda \in \Rp : -\lambda 1_K \leq f \leq \lambda 1_K \}$. Let $x \in K$, then by definition of the norm $\norm{\cdot}$ we have $\abs{\<x,f\>} \leq \norm{f}$, which is equivalent to
\begin{equation}
- \norm{f} \<x,1_K\> \leq \<x,f\> \leq \norm{f} \<x, 1_K\>.
\end{equation}
Since this holds for all $x \in K$, we get $- \norm{f} 1_K \leq f \leq \norm{f} 1_K$ and so we must have $\inf \{\lambda \in \Rp : -\lambda 1_K \leq f \leq \lambda 1_K \} \leq \norm{f}$.
\end{proof}
In every ordered vector space with an order unit one can use the expression on the right hand side of \eqref{eq:basic-norms-OUnorm} to define an order unit norm. Observe that the order unit norm is defined only using the geometry of the ordered vector space. Equation \eqref{eq:basic-norms-OUnorm} shows that the supremum norm $\norm{\cdot}$ coincides with the order unit norm, further strengthening the connection between geometry of $A(K)$ and its relation to $K$.

The following is an immediate result.
\begin{lemma}
Let $f \in A(K)^+$, then $f \in E(K)$ if and only if $\norm{f} \leq 1$.
\end{lemma}
\begin{proof}
If $f \in E(K)$, then $0 \leq \<x,f\> \leq 1$ for all $x \in K$ and so $\norm{f} \leq 1$. If $\norm{f} \leq 1$, then for every $x \in K$ we have $-1 \leq \<f,x\> \leq 1$, but since $f \in A(K)^+$ we have $0 \leq \<x,f\>$ and so we get $0 \leq \<x,f\> \leq 1$, which implies $f \in E(K)$.
\end{proof}

Since $A(K)$ is now a normed vector space, we can simply introduce the norm to $A(K)^*$ by using the standard norm for linear functionals.
\begin{proposition} \label{prop:basic-norms-baseNormDef}
Let $K$ be a state space, $\psi \in A(K)^*$ and define
\begin{equation}
\norm{\psi} = \sup_{f \in A(K), \norm{f} \leq 1} \abs{\<\psi, f\>},
\end{equation}
then $\norm{\cdot}$ is a norm on $A(K)^*$.
\end{proposition}
\begin{proof}
Let $\psi \in A(K)^*$ and let $\norm{\psi} = 0$, then we have $\<\psi, f\> = 0$ for all $f \in A(K)$ and so $\psi = 0$. Let $\alpha \in \RR$, then we have
\begin{equation}
\norm{\alpha \psi} = \sup_{f \in A(K), \norm{f} \leq 1} \abs{\alpha \<\psi,f\>} = \abs{\alpha} \sup_{f \in A(K), \norm{f} \leq 1} \abs{\<\psi,f\>} = \abs{\alpha} \norm{\psi}
\end{equation}
and so $\norm{\cdot}$ is homogeneous. Finally let $\psi, \varphi \in A(K)^*$, then we have
\begin{equation}
\norm{\psi + \varphi} = \sup_{f \in A(K), \norm{f} \leq 1} \abs{\<\psi+\varphi,f\>} \leq \sup_{f \in A(K), \norm{f} \leq 1} \abs{\<\psi,f\>} + \sup_{f \in A(K), \norm{f} \leq 1} \abs{\<\varphi,f\>} = \norm{\psi} + \norm{\varphi}
\end{equation}
and so $\norm{\cdot}$ is subadditive.
\end{proof}

\begin{lemma} \label{lemma:basic-norms-baseNormEAsup}
Let $\psi \in A(K)^*$, then we have
\begin{equation}
\norm{\psi} = \sup_{f \in E(K)} \abs{\<\psi, 2f - 1\>}.
\end{equation}
\end{lemma}
\begin{proof}
Let $g \in A(K)$ be such that $\norm{g} \leq 1$. According to Proposition \ref{prop:basic-norms-OUnormfromOU} we have
\begin{equation} \label{eq:basic-norms-baseNormEAsup-gBounds}
-1_K \leq g \leq 1_K
\end{equation}
Let now $f = \dfrac{1}{2} (1_K + g)$, then clearly $f \in A(K)$. Moreover from \eqref{eq:basic-norms-baseNormEAsup-gBounds} we get
\begin{equation}
0 \leq \dfrac{1}{2} (1_K + g) \leq 1_K
\end{equation}
and so $f \in E(K)$. Therefore every $g \in A(K)$ such that $\norm{g} \leq 1$ can be written as $g = 2f - 1_K$ where $f \in E(K)$. We get
\begin{equation}
\norm{\psi} = \sup_{g \in A(K), \norm{g} \leq 1} \abs{\<\psi, g\>} = \sup_{f \in E(K)} \abs{\<\psi, 2f - 1_K\>}
\end{equation}
which is the desired result.
\end{proof}

Also the norm on $A(K)^*$ has an equivalent geometrical expression.
\begin{proposition}
Let $\psi \in A(K)^*$, then
\begin{equation} \label{eq:basic-norms-baseNorm}
\norm{\psi} = \inf_{\lambda, \mu \in \Rp} \{ \lambda + \mu : \psi = \lambda x - \mu y, \, x,y \in K \}
\end{equation}
\end{proposition}
\begin{proof}
We will only lay out the key steps of the proof, see \cite{Jencova-baseNorms} for a complete proof. The key steps are: show that the expression $\inf_{\lambda, \mu \in \Rp} \{ \lambda + \mu : \psi = \lambda x - \mu y, \, x,y \in K \}$ defines a norm on $A(K)^*$, then show that the dual norm on $A(K)$ is the order unit norm, by using the result of Proposition \ref{prop:basic-norms-OUnormfromOU}. One then gets that $\norm{\cdot}$ is the double dual norm, and it is known that the double dual norm coincides with the original norm, see \cite[Theorem 4.3]{Rudin-functionalA}.
\end{proof}
The expression on the right hand side of \eqref{eq:basic-norms-baseNorm} is called base norm and we can introduce it in any ordered vector space with a fixed base of the positive cone. The base norm is useful, because it has operational meaning; for $x,y \in K$ the norm $\norm{x-y}$ determines the chance of discriminating the states $x$ and $y$. We will now define the discrimination task more precisely and we will show exactly how the base norm comes into play.

Let $x_0, x_1 \in K$ be two states and consider the following task: we will be given the state $x_0$ with probability (or relative frequency) $\lambda \in [0,1]$ or we will be given the state $x_1$ with probability $1-\lambda$. Our goal is to tell which of the states we were given with the highest possible accuracy, i.e., we want to maximize the probability of our answer being correct. We will call this task the discrimination of $x_0$ and $x_1$. Note that $\lambda$ and $1-\lambda$ are usually referred to as a priori probabilities for $x_0$ and $x_1$ respectively.

We are going to get the answer by performing a `yes'-`no' measurement in the following sense:  given a state $x \in \{x_0, x_1\}$ we will perform a `yes'-`no' measurement corresponding to some $f \in E(K)$. Then the probability of `yes' answer is $\<x, f\>$ and the probability of `no' answer is $1-\<x,f\>$. If we get the `yes' answer, we will predict that we were given $x = x_0$ and if we get the `no' answer, we will predict we were given $x = x_1$. Note that the assignment of `yes' answer to $x_0$ and `no' to $x_1$ is just a convention, we can also assign the `yes' answer to $x_1$ and `no' answer to $x_0$. There are four possible things that may happen: we either receive $x_0$ or $x_1$ and we either guess $x_0$ or $x_1$. Given a choice of $f \in E(K)$, we can assign probabilities to all possible outcomes, see Table \ref{table:basic-norms-discrimination}.

\begin{table}[ht]
\centering
\begin{tabular}{c|c|c}
& receive $x_0$ & receive $x_1$ \\ 
\hline 
guess $x_0$ & $\<x_0, f\>$ & $\<x_1, f\>$ \\ 
\hline 
guess $x_1$ & $1-\<x_0, f\>$ & $1-\<x_1, f\>$ \\ 
\end{tabular}
\caption{Probabilities of the four possible outcomes of the discrimination task of $x_0$ and $x_1$. \label{table:basic-norms-discrimination}}
\end{table}

We have to take into account that we will receive $x_0$ with probability $\lambda$ and $x_1$ with probability $1-\lambda$, so for a given $f \in E(K)$ the overall success probability $p_{\text{succ}}(f)$ is
\begin{equation}
p_{\text{succ}}(f) = \lambda \<x_0, f\> + (1-\lambda)(1-\<x_1,f\>).
\end{equation}
Our goal is to maximize $p_{\text{succ}}(f)$ over all possible choices of $f \in E(K)$, we get
\begin{equation} \label{eq:basic-norms-discrimination-pSucc}
p_{\text{succ}} = \sup_{f \in E(K)} p_{\text{succ}}(f) = (1-\lambda) + \sup_{f \in E(K)} \< \lambda x_0 - (1-\lambda) x_1, f\>,
\end{equation}
where we have used the bilinearity of $\< \cdot, \cdot \>$, i.e., we have used that $\lambda x_0 - (1-\lambda) x_1 \in A(K)^*$ and so $\lambda \< x_0, f\> - (1-\lambda) \<x_1, f\> = \< \lambda x_0 - (1-\lambda) x_1, f\>$. Now we want to rewrite \eqref{eq:basic-norms-discrimination-pSucc} in such way that we can use Lemma \ref{lemma:basic-norms-baseNormEAsup} to express the supremum over $f \in E(K)$ as a norm of some functional from $A(K)^*$. We get
\begin{align}
p_{\text{succ}} &= (1-\lambda) + \dfrac{1}{2} \sup_{f \in E(K)} \left( \< \lambda x_0 - (1-\lambda) x_1, 2f - 1\> + \< \lambda x_0 - (1-\lambda) x_1, 1\> \right) \\
&= (1-\lambda) + \dfrac{1}{2} \left( \< \lambda x_0 - (1-\lambda) x_1, 1\> +  \sup_{f \in E(K)} \< \lambda x_0 - (1-\lambda) x_1, 2f - 1\> \right) \\
&= \dfrac{1}{2} \left( 1 + \norm{\lambda x_0 - (1-\lambda) x_1} \right).
\end{align}
Thus we have proved the following:
\begin{theorem} \label{thm:basic-norms-discrimination}
Let $x_0, x_1 \in K$ be two states and consider the task of discriminating between $x_0$ and $x_1$. Let $\lambda \in [0,1]$ and $1-\lambda$ be the a priori probabilities for $x_0$ and $x_1$ respectively, then the maximal probability of successfully discriminating $x_0$ and $x_1$ is
\begin{equation}
p_{\text{succ}} = \dfrac{1}{2} \left( 1 + \norm{\lambda x_0 - (1-\lambda) x_1} \right).
\end{equation}
\end{theorem}

We see that $p_{\text{succ}}$ depends only on the norm $\norm{\lambda x_0 - (1-\lambda) x_1}$, which in turn shows that the functional norm $\norm{\cdot}$ introduced in Proposition \ref{prop:basic-norms-baseNormDef} is principle an observable quantity, not just a mathematical construct. Base norms were used to analyze discrimination tasks in the past, see e.g. \cite{Jencova-baseNorms,NuidaKimuraMiyadera-discrimination}. One can also connect the base norm of $\norm{x-y}$, where $x,y \in K$, to the task of finding the smallest $\mu \in [0,1]$ such that $(1-\mu) x + \mu x' = (1-\mu)y + \mu y'$ for some $x',y' \in K$, see \cite[proof of Theorem 3.2]{Gudder-convexStructures}.

\subsection{No-restriction hypothesis and restricted theories}
So far we have assumed that every effect from $E(K)$ corresponds to a well-defined `yes'-`no' question that, at least in principle, can be experimentally performed. This assumption is in literature called the no-restriction hypothesis. Theories without the no-restriction hypothesis were recently studied in \cite{JanottaLal-noRestriction,FilippovGudderHeinosaariLeppajarvi-restrictions} and as it was pointed out in \cite{FilippovGudderHeinosaariLeppajarvi-restrictions}, it is not trivial to consistently define a theory with restrictions.

A general approach would be to assume that there is a set of effects $E \subset E(K)$ and only `yes'-`no' questions corresponding to effects from $E$ are allowed. For this approach to be consistent, we must have $0, 1_K \in E$ and $E$ must be convex, i.e., $\conv(E) = E$. It is usually assumed that $E$ separates states, i.e., that for every $x,y \in K$ there is some $f \in E$ such that $\<x,f\> \neq \<y,f\>$.

Since $E \subset E(K)$, we also have $\cone(E) \subset A(K)^+$. Let $\cone(E)^*$ be the dual cone to $\cone(E)$ given as
\begin{equation}
\cone(E)^* = \{ \psi \in A(K)^* : \< \psi, f \> \geq 0, \forall f \in E \}.
\end{equation}
We then have $A(K)^{*+} \subset \cone(E)^*$ and we can define
\begin{equation}
S(E) = \{ \psi \in \cone(E)^*: \< \psi, 1_K \> = 1 \}.
\end{equation}
$S(E)$ is the state space corresponding to $E$ and we have $K \subset S(E)$. We can now see the restricted theory $(K, E)$ as a pair of state spaces $(K, S(E))$, such that $K \subset S(E)$. The interpretation then is that we can only prepare states from $K$, but we can only measure measurements that are well-defined on $S(E)$. In this sense, we can say that we are either restricted in the states that we can prepare or, equivalently, we can say that we are restricted in the `yes'-`no' questions we can ask about the system.

In even more general scenario, one can not only restrict the `yes'-`no' questions, but also the set of measurements can have additional restrictions, or the set of allowed transformations can be artificially restricted. Also one has to make sure that these restrictions are logically consistent in the way that every transformation is well defined in both Schr\"{o}dinger picture (as transformation on states) and Heisenberg picture (as transformation of the effects). For an in-depth treatment of restricted theories see \cite{FilippovGudderHeinosaariLeppajarvi-restrictions}.

\subsection{Diagrammatic notation}
In the upcoming sections we will work with more that one system, we will work with bipartite and tripartite systems, we will work with entangled states and entangled measurements and we will use transformations that map single systems to bipartite systems and vice-versa. We will use diagrammatic notation to make the calculations easier to understand, i.e., we will use diagrams to represent some complicated equations. Diagrammatic notation is often used in frameworks similar to GPTs, see e.g. \cite{ChiribellaDArioanoPerinotti-quantumCircuits,ChiribellaDArianoPerinotti-GPTpurification,ChiribellaDArianoPerinotti-derivationQT,BisioPerinotti-higherOrder,SelbyCoecke-leaks,BarnumLeeScandoloSelby-higherOrderInterference,WolfeSchmidSainzKunjwalSpekkens-commonCauseBoxes,SchmidSelbyPuseySpekkens-nonconModels,SchmidSelbySpekkens-causalInferentialTheories,SchmidFraserKunjwalSainzWolfeSpekkens-entanglement,SchmidHaoxingMudassarWitRossethoban-commonCauseChannels}. We are using a modified version of the \verb|quantikz| library \cite{quantikz} to typeset the diagrams.

Let $K$ be a state space and let $x \in K$, then we will use
\begin{equation}
\begin{quantikz}
& \prepareC{x} & \qw
\end{quantikz}
\end{equation}
to represent the equivalence class of preparations corresponding to $x$. If it will be needed to specify that $x$ belongs to $K$, we will use
\begin{equation}
\begin{quantikz}
& \prepareC{x} & \qw{K}
\end{quantikz}
\end{equation}
In a similar fashion, we will use
\begin{equation}
\begin{quantikz}
& \meterD{f}
\end{quantikz}
\end{equation}
and
\begin{equation}
\begin{quantikz}
& \meterD{f}{K}
\end{quantikz}
\end{equation}
to represent $f \in E(K)$. For the unit effect $1_K \in E(K)$ we will use the symbol
\begin{equation}
\begin{quantikz}
& \ground{}
\end{quantikz}
\end{equation}
since this is a generalization of the partial trace from quantum theory. We will use
\begin{equation}
\<x, f\> =
\begin{quantikz}[align equals at=1]
&[-15pt] \prepareC{x} & \meterD{f}
\end{quantikz}
\end{equation}
In other words, the closed diagram, i.e., the diagram with no free/unconnected legs, corresponds to a probability computed by pairing the corresponding state and effect. For the unit effect we have
\begin{equation}
\begin{quantikz}[align equals at=1]
& \prepareC{x} & \ground{}
\end{quantikz}
= 1.
\end{equation}
In a similar fashion, one can define equality between non-closed diagrams: let $x, y \in K$, then $x = y$ is the same as
\begin{equation}
\begin{quantikz}[align equals at=1]
&[-15pt] \prepareC{x} & \qw
\end{quantikz}
=
\begin{quantikz}[align equals at=1]
&[-15pt] \prepareC{y} & \qw
\end{quantikz}
\end{equation}
and for $f,g \in E(K)$ we have $f = g$ whenever
\begin{equation}
\begin{quantikz}[align equals at=1]
& \meterD{f}
\end{quantikz}
=
\begin{quantikz}[align equals at=1]
& \meterD{g}
\end{quantikz}
\end{equation}
One can also introduce equality of non-closed diagrams as follows: note that since $x,y \in K$ are equivalence classes of preparations, then we have $x = y$ whenever $\<x,f\> = \<y, f\>$ for all $f \in E(K)$. In diagrammatical notation, this reads:
\begin{equation}
\begin{quantikz}[align equals at=1]
&[-15pt] \prepareC{x} & \qw
\end{quantikz}
=
\begin{quantikz}[align equals at=1]
&[-15pt] \prepareC{y} & \qw
\end{quantikz}
\quad
\Leftrightarrow
\quad
\begin{quantikz}[align equals at=1]
&[-15pt] \prepareC{x} & \meterD{f}
\end{quantikz}
=
\begin{quantikz}[align equals at=1]
&[-15pt] \prepareC{y} & \meterD{f}
\end{quantikz},
\forall f \in E(K).
\end{equation}
In other words, two non-closed diagrams are equal whenever all of their possible closures are equal. This is consistent with our formalism of equivalence classes of preparations, effects and transformations. We will also use convex combinations of diagrams; for $\lambda \in [0,1]$ and $x, y \in K$ we define
\begin{equation}
\begin{quantikz}[align equals at=1]
& \lstick{$\lambda$} &[\prepfix] \prepareC{x} & \qw
\end{quantikz}
+
\begin{quantikz}[align equals at=1]
& \lstick{$(1-\lambda)$} &[\prepfix] \prepareC{y} & \qw
\end{quantikz}
=
\begin{quantikz}[align equals at=1]
&[\prepfix] \prepareC{\lambda x + (1-\lambda)y} & \qw
\end{quantikz}
\end{equation}
In analogical way, one can define convex combinations of effects and transformations and one can again understand the equality of convex combination of non-closed diagrams as equality of convex combinations of all possible closures. We will also allow an abuse of the diagrammatical notation and we will also use it for general elements of $A(K)$ or $A(K)^*$. For example, let $\varphi \in A(K)^*$ and $\alpha \in \RR$, then we can use
\begin{equation}
\begin{quantikz}
&\lstick{$\alpha$} &[\prepfix] \prepareC{\varphi} & \qw
\end{quantikz}
\end{equation}
to denote $\alpha \varphi \in A(K)^*$.

\section{Example: classical theory} \label{sec:CT}

%
%
In this section we will present the first example theory: classical theory. Classical theory provides the simplest possible example and so it is easy to grasp, but we will also need classical theory to introduce several important concepts, such as measurements and instruments. Also, there are many important results about classical theory that we will point out in subsequent sections.

A classical theory is a theory where we have $n$ independent pure states and their convex combinations. If we number the pure states $1, 2, \ldots, n$, then a classical theory is a theory where every state is a probability distribution over the set $\{1, \ldots, n\}$, i.e., every state is a set of numbers $(p_1, \ldots, p_n)$, $p_i \in \Rp$ for all $i \in \{1, \ldots, n\}$ and $\sum_{i=1}^n p_i = 1$. One can then see that a state is pure if and only if the corresponding probability distribution is concentrated at a single point, i.e., $p_j = 1$ for some $j \in \{1, \ldots, n \}$ and $p_i = 0$ for $i \neq j$. Moreover, it is easy to see that the pure states must be affinely independent and so the state space must be a simplex; a simplex $S_n$ is a convex hull of affinely independent points.

Let $V$ be a real finite-dimensional vector space, then $\{ v_0, v_1 \ldots, v_k \} \subset V$ are affinely independent if the only numbers $\alpha_i \in \RR$, $i \in \{ 0, 1, \ldots, k \}$ such that $\sum_{i=0}^k \alpha_i = 0$ and $\sum_{i=0}^k \alpha_i v_i = 0$ are $\alpha_i = 0$ for all $i \in \{ 0, \ldots, k \}$. If $\{ v_0, v_1 \ldots, v_k \}$ are affinely independent, then one can not express any of the vectors $v_i$ as an affine combination of the remaining vectors. Affine independence of $\{ v_0, v_1 \ldots, v_k \}$ is equivalent to linear independence of $\{ v_1 - v_0 \ldots, v_k - v_0 \}$. One can also see that if $\{ v_0, v_1 \ldots, v_k \}$ are affinely independent, then their affine hull is $k$-dimensional.

\begin{definition}
\emph{Classical theory} is a theory where the state space is a simplex $S_n$, $n \in \NN$. Any theory where the state space $K$ is not a simplex will be called non-classical.
\end{definition}
Let $n \in \NN$ and let $\{ s_1, s_2 \ldots, s_n \} \subset V$ be an affinely indent set of vectors, then
\begin{equation}
S_n = \conv ( \{ s_1, \ldots, s_n \} )
\end{equation}
is a simplex. For any $x \in S_n$ there are unique numbers $\lambda_i \in \RR$, $\lambda_i \geq 0$, for all $i \in \{1, \ldots, n\}$, $\sum_{i=1}^n \lambda_i = 1$ such that $x = \sum_{i=1}^n \lambda_i s_i$. The uniqueness of the numbers $\lambda_1, \ldots, \lambda_n$ follows from affine independence of $\{ s_1, \ldots, s_n \}$. The numbers $\lambda_1, \ldots, \lambda_n$ are exactly the probability distribution $(p_1, \ldots, p_n)$ corresponding to the state $x \in S_n$.

The effect algebra $E(S_n)$ is generated by the functions $b_1, \ldots, b_n$ that are given as
\begin{equation}
\< s_i, b_j \> = \delta_{ij}
\end{equation}
for all $i,j \in \{1, \ldots, n\}$, where $\delta_{ij}$ is the Kronecker delta, defined as
\begin{equation}
\delta_{ij} =
\begin{cases}
1 & i=j \\
0 & i \neq j
\end{cases}
\end{equation}
The functions $b_1, \ldots, b_n$ are well-defined, because $\{ s_1, \ldots, s_n \}$ are affinely independent. Note that for any other $f \in A(K)$ we have $f = \sum_{i=1}^n \< s_i, f \> b_i$. This is easy to see, let $x \in S_n$, $x = \sum_{i=1}^n \lambda_i s_i$, then for all $i \in \{ 1, \ldots, n \}$ we have $\lambda_i = \< x, b_i \>$
and so
\begin{equation}
\< x, f\> = \sum_{i=1}^n \lambda_i \< s_i, f \> = \sum_{i=1}^n \< x, b_i \> \< s_i, f \> = \left\langle x, \sum_{i=1}^n \<s_i, f\> b_i \right\rangle .
\end{equation}
It is obvious that $f \in A(K)^+$ only if $\< s_i, f\> \geq 0$, but it follows that then $f$ is a sum of the functions $b_1, \ldots, b_n$ with positive coefficients. Moreover $f \in E(S_n)$ if and only if it is a sum of the functions $b_1, \ldots, b_n$ with coefficients from the interval $[0,1]$. For example, one has
\begin{equation}
1_{S_n} = \sum_{i=1}^n b_i.
\end{equation}

We will now proceed with exploring the simplest cases. So let $n=1$, then $S_1 = \{ s \}$ and we have only a single state. Then $b = 1_{S_1}$ and $E(S_1) = [0,1]$ as every $f \in E(S_1)$ is of the form $\mu 1_{S_1}$ for some $\mu \in [0,1]$. So $S_1$ is the simplest possible state space as it contains only one state.

\begin{figure}[t]
\centering
\begin{subfigure}[t]{0.475\textwidth}
\centering
\includegraphics[width=\textwidth]{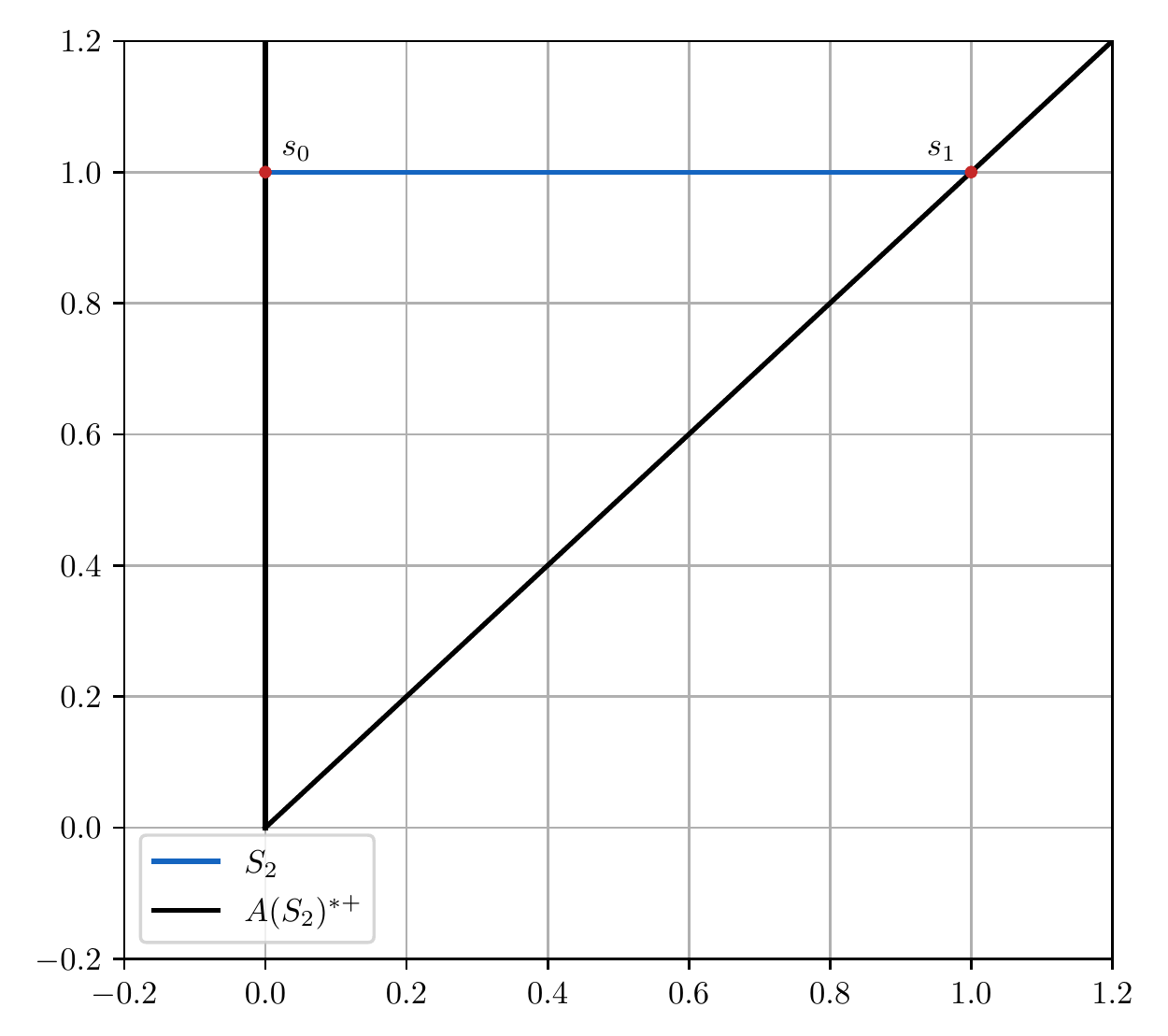}
\caption{Picture of the state space $S_2$ as subset of $A(S_2)^*$. The red points are the pure states $s_0$ and $s_1$, the blue line is the state space $S_2$, and the black lines are the boundary of the positive cone $A(S_2)^{*+}$.}
\label{fig:CT-S2-stateSpace}
\end{subfigure}
\hfill
\begin{subfigure}[t]{0.475\textwidth}
\centering
\includegraphics[width=\textwidth]{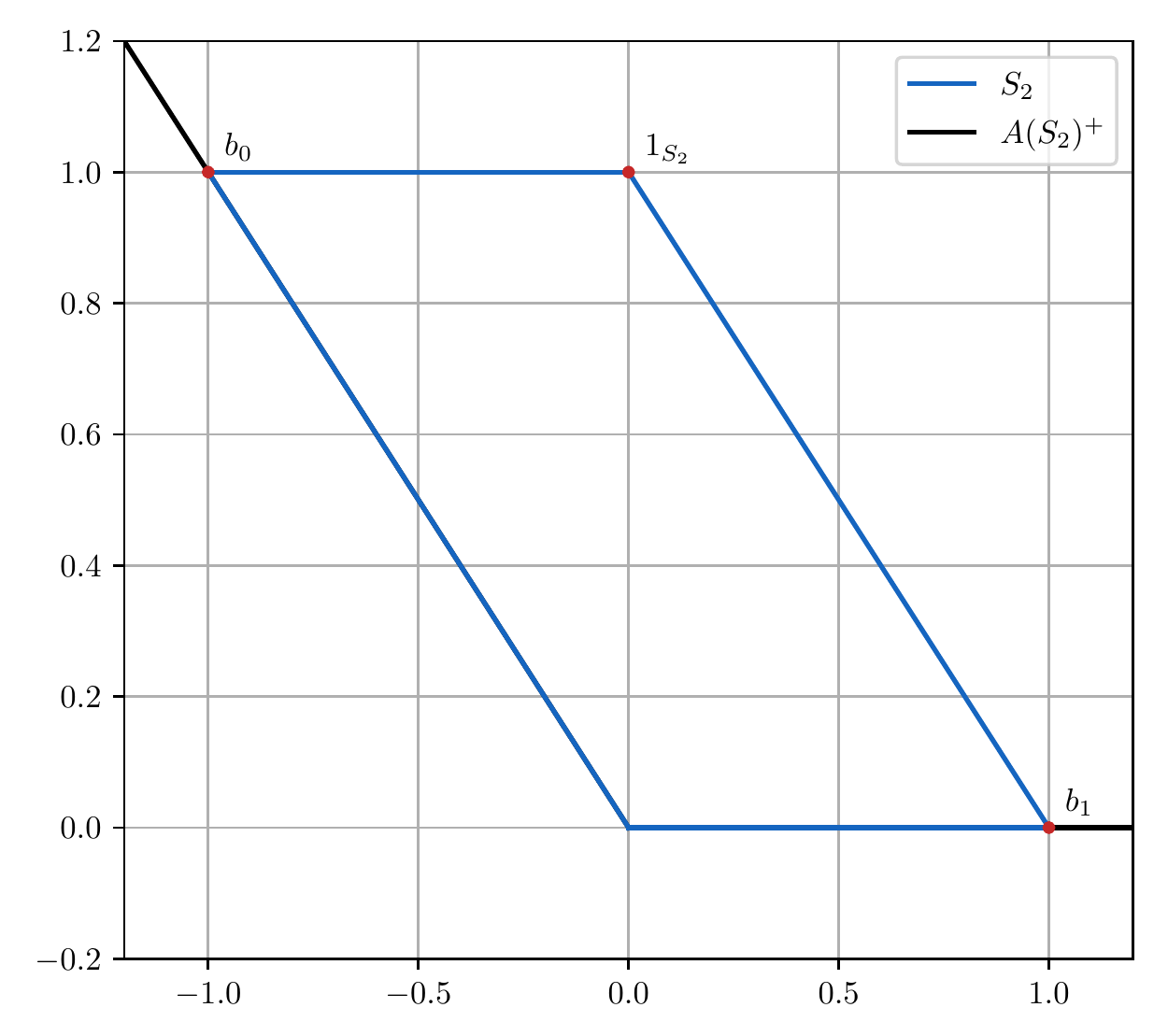}
\caption{Picture of the effect algebra $E(S_2)$ as subset of $A(S_2)$. The red points are the effects $b_0$, $b_1$, $1_{S_2}$, the blue lines are the boundary of the effect algebra $E(S_2)$ and the black lines are the boundary of the positive cone $A(S_2)^+$.}
\label{fig:CT-S2-EffectAlgebra}
\end{subfigure}
\caption{Pictures of the state space $S_2$ and effect algebra $E(S_2)$.}
\end{figure}

Let $n=2$, then the pure states are usually denoted $s_0$ and $s_1$, $S_2$ is isomorphic to the interval $[0,1]$ and the state space corresponds to a classical bit. In this case we can introduce the representation
\begin{align}
&s_0 =
\begin{pmatrix}
0 \\
1
\end{pmatrix},
&&s_1 =
\begin{pmatrix}
1 \\
1
\end{pmatrix},
\end{align}
where $s_0, s_1 \in A(S_2)^*$ are already represented as functionals. One can pick different representation; we find these most useful for calculations, but they produce skewed images. For $b_0, b_1, 1_{S_2} \in E(S_2)$ we have
\begin{align}
&b_0 = 
\begin{pmatrix}
-1 \\
1
\end{pmatrix},
&&b_1 = 
\begin{pmatrix}
1 \\
0
\end{pmatrix},
&&1_{S_2} = 
\begin{pmatrix}
0 \\
1
\end{pmatrix},
\end{align}
where the pairing is given by the usual Euclidean inner product. See Figure \ref{fig:CT-S2-stateSpace} for the picture of the state space $S_2$ and Figure \ref{fig:CT-S2-EffectAlgebra} for the picture of the effect algebra $E(S_2)$.

\begin{figure}[t]
\centering
\begin{subfigure}[t]{0.475\textwidth}
\centering
\includegraphics[width=\textwidth]{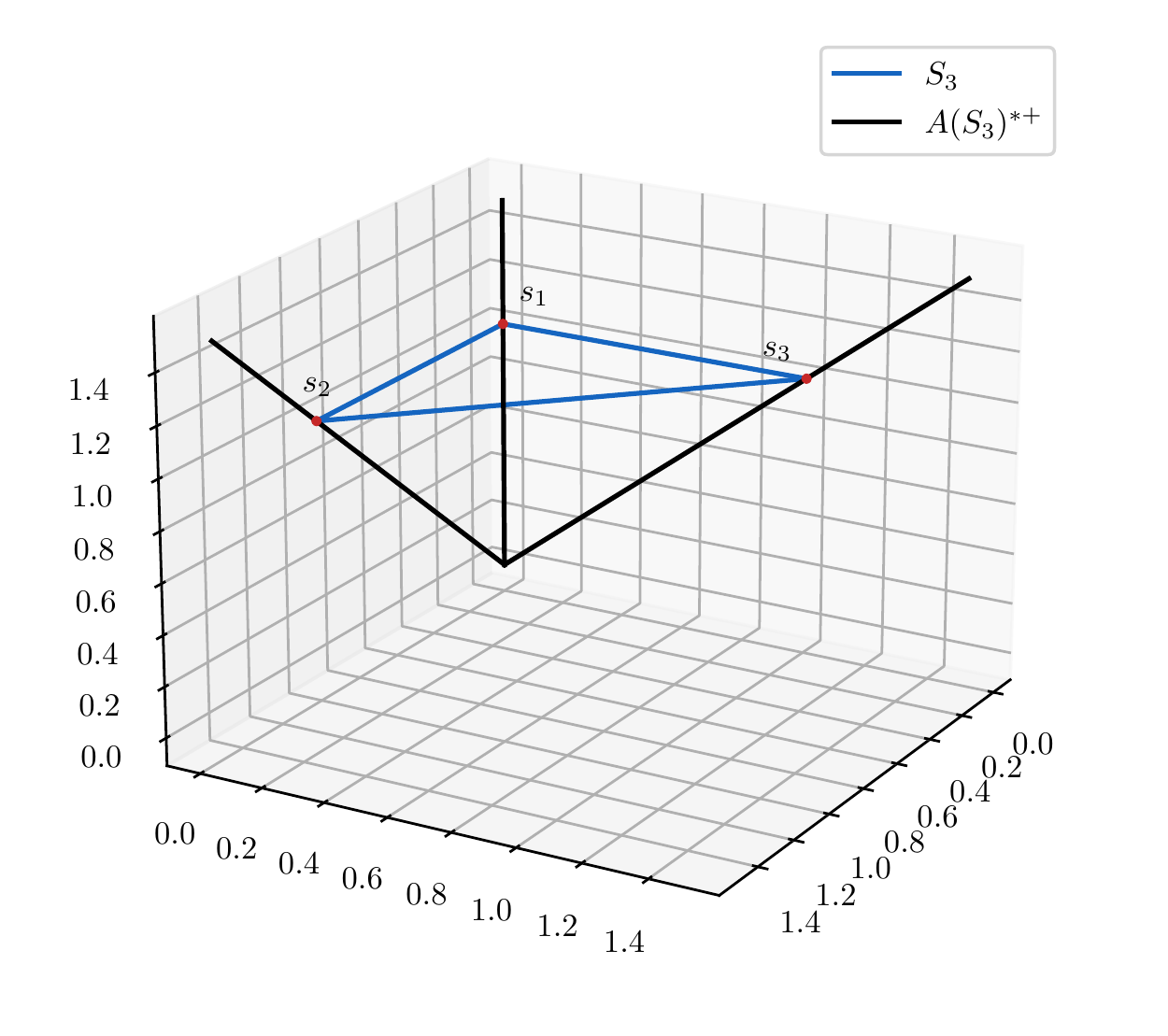}
\caption{Picture of the state space $S_3$ as subset of $A(S_3)^*$. The red points are the pure states $s_1$, $s_2$ and $s_3$, the blue lines are the edges of the state space $S_3$, and the black lines are the edges of the positive cone $A(S_3)^{*+}$.}
\label{fig:CT-S3-stateSpace}
\end{subfigure}
\hfill
\begin{subfigure}[t]{0.475\textwidth}
\centering
\includegraphics[width=\textwidth]{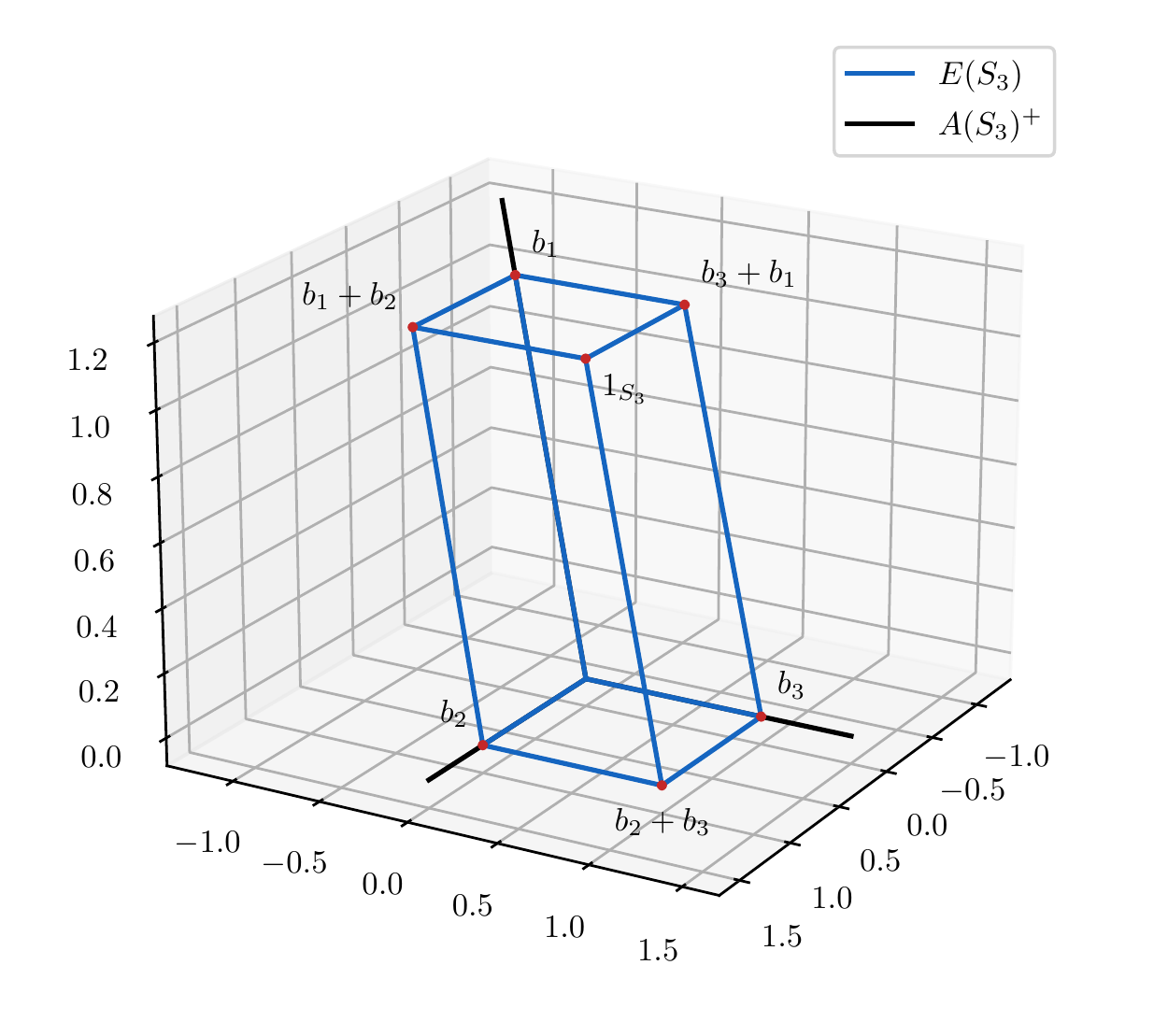}
\caption{Picture of the effect algebra $E(S_3)$ as subset of $A(S_3)$. The red points are the effects $b_1$, $b_2$, $b_3$, $b_1 + b_2$, $b_2 + b_3$, $b_3 + b_1$, $1_{S_2}$, blue lines are the edges of the effect algebra $E(S_3)$, and the black lines are the edges of the positive cone $A(S_3)^+$.}
\label{fig:CT-S3-EffectAlgebra}
\end{subfigure}
\caption{Pictures of the state space $S_3$ and effect algebra $E(S_3)$.}
\end{figure}

For $n=3$, the pure states are denoted $s_1, s_2, s_3$ and $S_3$ is a triangle. We will use the representation
\begin{align}
&s_1 =
\begin{pmatrix}
0 \\
0 \\
1
\end{pmatrix},
&&s_2 =
\begin{pmatrix}
1 \\
0 \\
1
\end{pmatrix},
&&s_3 =
\begin{pmatrix}
0 \\
1 \\
1
\end{pmatrix},
\end{align}
where $s_1, s_2, s_3 \in A(S_3)^*$ are again already represented as the corresponding functionals. For the elements of $E(S_3)$ we have
\begin{align}
&b_1 = 
\begin{pmatrix}
-1 \\
-1 \\
1
\end{pmatrix},
&&b_2 = 
\begin{pmatrix}
1 \\
0 \\
0
\end{pmatrix},
&&b_3 = 
\begin{pmatrix}
0 \\
1 \\
0
\end{pmatrix},
&&1_{S_3} = 
\begin{pmatrix}
0 \\
0 \\
1
\end{pmatrix}.
\end{align}
Note that the extreme points of $E(S_3)$ are not only $b_1, b_2, b_3, 1_{S_3}$ and $0$, but also $b_1+b_2,b_2+b_3,b_3+b_1$. See Figure \ref{fig:CT-S3-stateSpace} for the picture of the state space $S_3$ and Figure \ref{fig:CT-S3-EffectAlgebra} for the picture of the effect algebra $E(S_3)$.

\begin{figure}[t]
\centering
\includegraphics[width=0.85\textwidth]{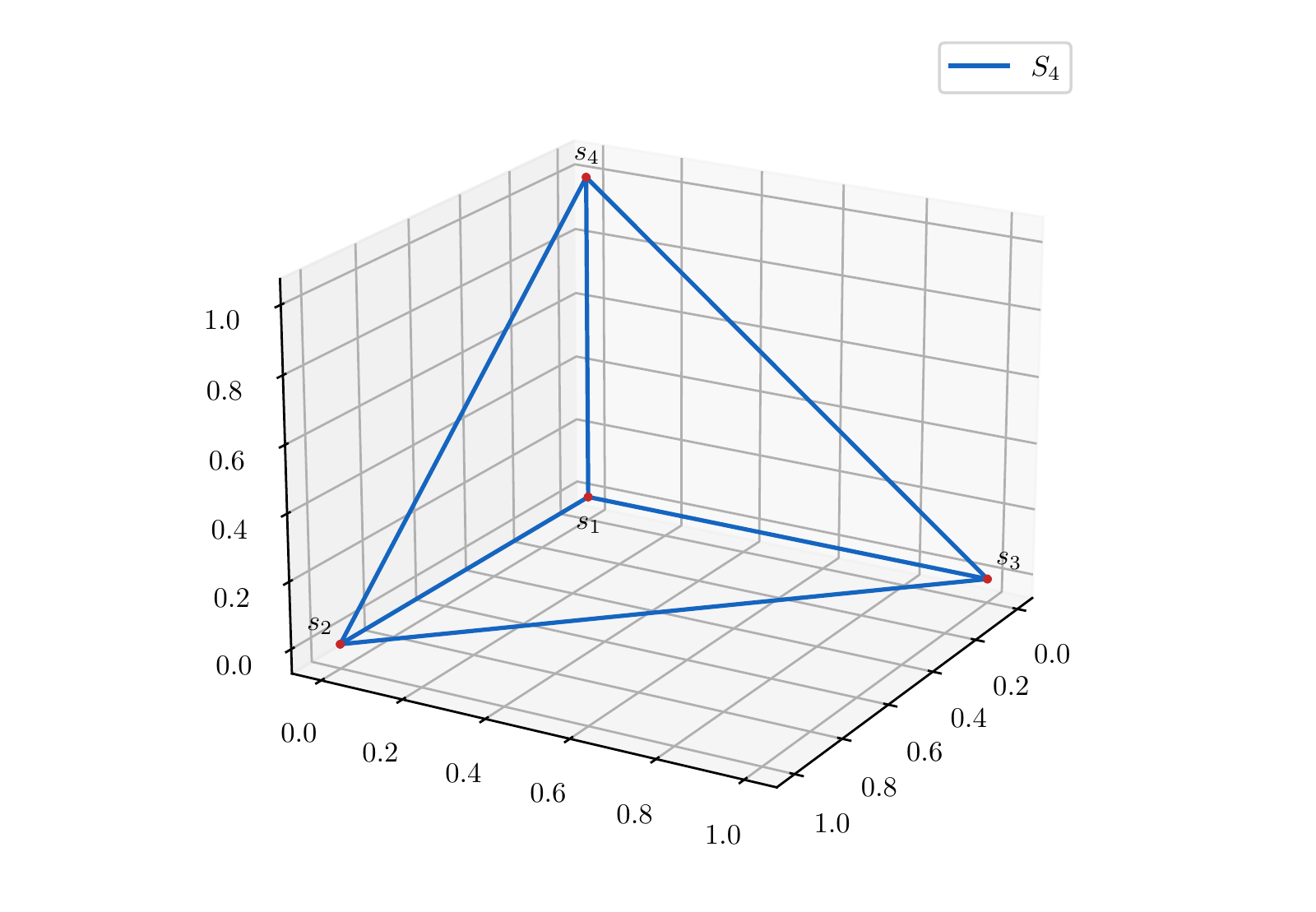}
\caption{Pictures of the state space $S_4$ as subset of $V = \aff(S_4)$. The red points are the pure states $s_1$, $s_2$, $s_3$ and $s_4$, and the blue lines are the edges of the state space $S_4$.}
\label{fig:CT-S4-stateSpace}
\end{figure}

For $n=4$, the pure states are
\begin{align}
&s_1 =
\begin{pmatrix}
0 \\
0 \\
0 \\
1
\end{pmatrix},
&&s_2 =
\begin{pmatrix}
1 \\
0 \\
0 \\
1
\end{pmatrix},
&&s_3 =
\begin{pmatrix}
0 \\
1 \\
0 \\
1
\end{pmatrix},
&&s_4 =
\begin{pmatrix}
0 \\
0 \\
1 \\
1
\end{pmatrix},
\end{align}
which are plotted in Figure \ref{fig:CT-S4-stateSpace}. The effect algebra $E(S_4)$ is generated by
\begin{align}
&b_1 =
\begin{pmatrix}
-1 \\
-1 \\
-1 \\
1
\end{pmatrix},
&&b_2 =
\begin{pmatrix}
1 \\
0 \\
0 \\
0
\end{pmatrix},
&&b_3 =
\begin{pmatrix}
0 \\
1 \\
0 \\
0
\end{pmatrix},
&&b_4 =
\begin{pmatrix}
0 \\
0 \\
1 \\
0
\end{pmatrix},
&&1_{S_4} =
\begin{pmatrix}
0 \\
0 \\
0 \\
1
\end{pmatrix}.
\end{align}
One can again see that the extreme points of $E(S_4)$ include not only $b_1, b_2, b_3, b_4, 1_{S_4}, 0$ but also $b_i + b_j$ for $i \neq j$ and $b_i + b_j + b_k$ for $i \neq j \neq k \neq i$, for $i,j,k \in \{ 1, \ldots, 4 \}$.

\section{Tensor products} \label{sec:tensor}

%
%
In this section we will introduce the concept of bipartite and multipartite systems. The idea is simple: given state spaces $K_A$ and $K_B$, we want to describe an experiment with two parties; traditionally called Alice and Bob. In the simplest scenario Alice prepares $x_A \in K_A$, applies $f_A \in E(K_A)$ and Bob prepares $x_B \in K_B$, applies $f_B \in E(K_B)$. Since both of the experiments are independent, the resulting joint probability for both Alice and Bob is a product of the respective probabilities, i.e., they both get the outcome `yes' with probability $\<x_A, f_A\> \<x_B, f_B\>$. In a more complex scenario, the preparation procedures of Alice and Bob can be correlated. Alice and Bob meet before the experiment and toss an unbiased coin, and based on the outcome prepare their states. For example if the coin lands on heads, Alice will prepare $x_A$ and Bob will prepare $x_B$, but if the coin lands on tails, Alice will prepare $y_A$ and Bob will prepare $y_B$. This is a valid preparation procedure and we have to have a way of describing it; we refer to this scenario as shared randomness, because the outcome of the coin toss is random information shared between Alice and Bob. In the most general scenario, Alice and Bob can share an entangled state, which can not be prepared using shared randomness.

\subsection{Bipartite scenarios}
Our aim si to find a state space $K_{AB}$ that will describe the bipartite scenario within our current framework. We will introduce several axioms that will fix basic properties of the bipartite state space.
\begin{definition}
A \emph{bipartite state space} $K_{AB}$ formed from state spaces $K_A$ and $K_B$ must satisfy:
\begin{enumerate}[label=(BP\arabic*), leftmargin=*]
\item\label{item:tensor-intro-stateSpace} $K_{AB}$ must be a valid state space.
\item\label{item:tensor-intro-sepStates} For every $x_A \in K_A$ and $x_B \in K_B$, there must be a bipartite state in $K_{AB}$ that describes the situation where Alice prepares $x_A$ and Bob prepares $x_B$. Moreover, the identification of the bipartite state with $x_A$ and $x_B$ is affine, meaning that if Alice (or Bob) prepares a mixture of states $x_{1,A}$ and $x_{2,A}$, then this results in a mixture of the respective bipartite states with the same coefficients.
\item\label{item:tensor-intro-sepEffects} For every $f_A \in E(K_A)$ and $f_B \in E(K_B)$, there must be a bipartite effect in $E(K_{AB})$ that describes the situation where Alice applies $f_A$ and Bob applies $f_B$. Moreover, the identification of the bipartite effect with $f_A$ and $f_B$ is linear, meaning that if Alice (or Bob) prepares a mixture or sum of effects $f_{1,A}$ and $f_{2,A}$, then this results in a mixture or sum of the respective bipartite effects.
\item\label{item:tensor-intro-unitEffect} The unit effect on $K_{AB}$ is equivalent to Alice applying $1_{K_A}$ and Bob applying $1_{K_B}$.
\item\label{item:tensor-intro-tomographicLocality} For every $x_{AB},y_{AB} \in K_{AB}$ there are $f_A \in E(K_A)$, $f_B \in E(K_B)$ such that when Alice and Bob prepare $x_{AB}$ and apply $f_A$ and $f_B$ respectively, then the resulting probability is different from the experiment where Alice and Bob prepare $y_{AB}$ and apply $f_A$ and $f_B$ respectively. In other words, applying effects locally is sufficient to distinguish all of the states in $K_{AB}$.
\end{enumerate}
\end{definition}
Note that \ref{item:tensor-intro-sepStates} together with the convexity coming from \ref{item:tensor-intro-stateSpace} implies that correlated preparations based on shared randomness between Alice and Bob are included in $K_{AB}$. Similar result follows for `yes'-	`no' questions that are performed based on shared randomness between Alice and Bob. To get more tangible results, we will use a well-know result from linear algebra that the dual of the vector space of bilinear forms is a tensor product of the original vector spaces, see \cite{Ryan-tensorProducts} or \ref{appendix:bilinear}. Consider first the scenario where Alice applies $f_A \in E(K_A)$ and Bob applies $f_B \in E(K_B)$. Mathematically speaking, $f_A$ is a linear functional on $A(K_A)^*$ and $f_B$ is a linear functional on $A(K_B)^*$, since $A(K_A)^{**} = A(K_A)$ and $A(K_B)^{**} = A(K_B)$, see Proposition \ref{prop:duals-doubleDual}. It then follows that the joint operation of Alice applying $f_A$ and Bob applying $f_B$ must behave as a bilinear functional on pairs of states.

Let $x_{AB} \in K_{AB}$, then according to Theorem \ref{thm:basic-dual-stateSpace} $x_{AB}$ is a linear functional on effects from $E(K_{AB})$. According to \ref{item:tensor-intro-sepEffects}, pairs of effects $f_A$ and $f_B$ must be included in $E(K_{AB})$. When acting on pairs of effects, $x_{AB}$ is essentially a bilinear functional and so according to Proposition \ref{prop:bilinear-isomorphisms} every $x_{AB}$ must correspond to an element of $A(K_A)^* \otimes A(K_B)^*$. According to \ref{item:tensor-intro-tomographicLocality} every $x_{AB}$ must be uniquely characterized by its action on pairs of effects, therefore every $x_{AB}$ must be equivalent to an element of $A(K_A)^* \otimes A(K_B)^*$. Thus we have proved the following:
\begin{lemma} \label{lemma:tensor-bipartite-tensorProd}
If $K_{AB}$ satisfies \ref{item:tensor-intro-stateSpace} - \ref{item:tensor-intro-tomographicLocality}, then
\begin{equation}
K_{AB} \subset A(K_A)^* \otimes A(K_B)^*.
\end{equation}
\end{lemma}

We are now going to construct the possible range of bipartite state spaces $K_{AB}$. This is only a possible range, because, as we will shortly see, $K_{AB}$ in general is not uniquely specified by the choices of $K_A$ and $K_B$. Let $x_A \in K_A$ and $x_B \in K_B$, then according to \ref{item:tensor-intro-sepStates} the pair of states must be represented in $K_{AB}$. Following the result of Lemma \ref{lemma:tensor-bipartite-tensorProd}, we are going to postulate that the state corresponding to Alice preparing $x_A$ and Bob preparing $x_B$ is $x_A \otimes x_B$. By using \ref{item:tensor-intro-sepStates} we get the smallest possible candidate for $K_{AB}$.
\begin{definition}
The \emph{minimal tensor product} of state spaces $K_A$ and $K_B$ is given as
\begin{equation}
K_A \tmin K_B = \conv \left( \{ x_A \otimes x_B : x_A \in K_A, x_B \in K_B \} \right).
\end{equation}
\end{definition}
One should check that $K_A \tmin K_B$ is a valid state space, i.e., that it is a compact convex subset of a real, finite-dimensional vector space. $K_A \tmin K_B$ clearly is convex and it clearly is  a subset of a real, finite-dimensional vector space. To see that $K_A \tmin K_B$ is compact, simply note that by definition $K_A \tmin K_B$ is closed and one can easily see that $K_A \tmin K_B$ is bounded, hence $K_A \tmin K_B$ is compact.

The counterpart to minimal tensor product is the maximal tensor product, that will be the largest possible candidate for $K_{AB}$. The maximal tensor product will be introduced with the help of \ref{item:tensor-intro-sepEffects} as largest possible set of states that is positive on pairs of effects. But to do so, we must first introduce the description for pairs of effects. For $f_A \in E(K_A)$ and $f_B \in E(K_B)$ we are going to denote the effect corresponding to Alice applying $f_A$ and Bob applying $f_B$ as $f_A \otimes f_B$. Then, analogical to the minimal tensor product of states, we get the minimal tensor product of effect algebras, given as
\begin{equation}
E(K_A) \tmin E(K_B) = \conv \left( \{ f_A \otimes f_B : f_A \in E(K_A), f_B \in E(K_B) \} \right).
\end{equation}
One should in principle check whether $E(K_A) \tmin E(K_B)$ is a valid effect algebra, i.e., that it is an interval in an order unit space. One can easily do this by verifying that $E(K_A) \tmin E(K_B)$ is a linear effect algebra, as introduced in Definition \ref{def:basic-connection-linEA}.

\begin{definition} \label{def:tensor-bipartite-maxProd}
The \emph{maximal tensor product} of state spaces $K_A$ and $K_B$ is given as
\begin{equation} \label{eq:tensor-bipartite-maxProdS}
K_A \tmax K_B = S(E(K_A) \tmin E(K_B)),
\end{equation}
which the same as
\begin{equation} \label{eq:tensor-bipartite-maxProdDef}
\begin{split}
K_A \tmax K_B = \{ \varphi \in A(K_A)^* \otimes A(K_B)^*: &\< \varphi, f_A \otimes f_B \> \geq 0, \forall f_A \in E(K_A), \forall f_B \in E(K_B), \\
&\< \varphi, 1_{K_A} \otimes 1_{K_B} \> = 1 \}
\end{split}
\end{equation}
\end{definition}
$K_A \tmax K_B$ is a state space by construction, since it was defined in \eqref{eq:tensor-bipartite-maxProdS} as the state space corresponding to the effect algebra $E(K_A) \tmin E(K_B)$. Analogically to the maximal tensor product of state spaces, we can define the maximal tensor product of effect algebras as
\begin{align}
E(K_A) \tmax E(K_B) &= E(K_A \tmin K_B) \\
&= \{ \psi \in A(K_A) \otimes A(K_B): 0 \leq \< x_A \otimes x_B, \psi \> \leq 1, \forall x_A \in K_A, \forall x_B \in K_B \}.
\end{align}
One may be temped to think that $K_A \tmax K_B$ includes way more states than $K_A \tmin K_B$ and so it must be way more useful in information-theoretic tasks, but this is not the case. Even though $K_A \tmax K_B$ has a very rich structure, the corresponding effect algebra $E(K_A \tmax K_B) = E(K_A) \tmin E(K_B)$ contains only separable effects, i.e., effects of the form $f_A \otimes f_B$ for $f_A \in E(K_A)$ and $f_B \in E(K_B)$ and their sums and convex combinations. It was observed in \cite[Section VII.]{Barrett-GPTinformation} that if we have access to only separable effects, then we can not implement neither teleportation nor superdense coding protocols. It was also observed in \cite{GrossMullerColbeckDahlsten-boxworldDynamics,AlSafiShort-boxworldDynamics,AlSafiRichens-reversibleDynamics} that for certain classes of state spaces, the set of reversible transformations on $K_A \tmax K_B$ is trivial. We also want to express $K_{AB}$ as some form of tensor product of $K_A$ and $K_B$. For this reason, we will change the notation; from now on we will use $K_A \treal K_B$ instead of $K_{AB}$.
\begin{definition}
Let $K_A$ and $K_B$ be state spaces, then we will denote $K_A \treal K_B$ the state space corresponding to the bipartite scenario.
\end{definition}
We will also denote
\begin{equation}
E(K_A) \treal E(K_B) = E(K_A \treal K_B).
\end{equation}

Note that the tensor product of state spaces $\treal$ is not a single object, but it is a placeholder for a rule that has to be specified by a given theory. In practice, $\treal$ is usually not defined on all possible state spaces, but only on a selected class of state spaces that are included in a given theory. For example, in quantum theory, we use a special rule for the quantum tensor product which is strictly different from the minimal and maximal tensor products, see Section \ref{sec:QT}. The quantum tensor product is constructed with the use of the underlying Hilbert spaces. For a general state space $K$, there is no underlying Hilbert space and so it is not clear how to extend the quantum tensor product to a general $K$. But this is not a problem, because quantum theory is defined only in terms of quantum state spaces.

The tensor product $K_A \treal K_B$ does not have an established name within the framework of GPTs. That is because in most applications it is sufficient to consider only bipartite scenarios and so the notation $K_{AB}$ is often used. But, as we will see, the notation $K_A \treal K_B$ is easier to work with in multipartite scenarios.

\begin{proposition} \label{prop:tensor-bipartite-inclusions}
For any valid bipartite state we must have
\begin{equation}
K_A \tmin K_B \subset K_A \treal K_B \subset K_A \tmax K_B.
\end{equation}
\end{proposition}
\begin{proof}
Note that $K_A \tmin K_B \subset K_A \treal K_B$ follows from \ref{item:tensor-intro-sepStates}, because every $x_{AB} \in K_A \tmin K_B$ can be written as $x_{AB} = \sum_{i=1}^N \lambda_i y_{i,A} \otimes y_{i,B}$, where $y_{i,A} \in K_A$, $y_{i,B} \in K_B$, $\lambda_i \in \Rp$ for all $i \in \{1, \ldots, N\}$ and $\sum_{i=1}^N \lambda_i = 1$. It follows from \ref{item:tensor-intro-sepStates} that $y_{i,A} \otimes y_{i,B} \in K_{AB}$; $x_{AB} \in K_{AB}$ follows by the convexity of $K_{AB}$.

Let $x_{AB} \in K_A \treal K_B$ and $f_A \in E(K_A)$, $f_B \in E(K_B)$. According to \ref{item:tensor-intro-sepEffects} $f_A \otimes f_B$ must be a well defined effect on $K_{AB}$, i.e., we must have $f_A \otimes f_B \in E(K_{AB})$. So we must have $\< x_{AB}, f_A \otimes f_B \> \geq 0$. It follows from \ref{item:tensor-intro-unitEffect} that $\<x_{AB}, 1_{K_A} \otimes 1_{K_B} \> = 1$ and so from \eqref{eq:tensor-bipartite-maxProdDef} we get $x_{AB} \in K_A \tmax K_B$.
\end{proof}
\begin{corollary}
We have
\begin{equation}
E(K_A) \tmin E(K_B) \subset E(K_{AB}) \subset E(K_A) \tmax E(K_B).
\end{equation}
\end{corollary}
\begin{proof}
The result follows from Proposition \ref{prop:tensor-bipartite-inclusions} and Lemma \ref{lemma:basic-results-subsets}, since we have
\begin{align}
E(K_A) \tmin E(K_B) &= E(K_A \tmax K_B), \\
E(K_A) \tmax E(K_B) &= E(K_A \tmin K_B).
\end{align}
\end{proof}

At last, we will introduce the concepts of separable and entangled states and we will discuss our constructions.
\begin{definition}
The states from $K_A \tmin K_B$ are called \emph{separable states}. The states from $K_A \treal K_B \setminus K_A \tmin K_B$, i.e., the states from $K_A \treal K_B$ that are not separable, are called \emph{entangled states}.
\end{definition}

One can, of course, ask whether we actually need all of the axioms \ref{item:tensor-intro-stateSpace} - \ref{item:tensor-intro-tomographicLocality}. \ref{item:tensor-intro-stateSpace} is necessary, as it only ensures that $K_A \treal K_B$ is a state space. \ref{item:tensor-intro-sepStates} and \ref{item:tensor-intro-sepEffects} are in some sense dual to each other and their goals are only to allows the natural scenarios where Alice and Bob do not interact and are unaware of each others existence. One could in principle drop \ref{item:tensor-intro-unitEffect}, since the unit effect of $E(K_A) \treal E(K_B)$ can be fixed by its action on separable states and this yields the unit effect of $E(K_A) \treal E(K_B)$ to be $1_{K_A} \otimes 1_{K_B}$ if separable states are generated by states of the form $x_A \otimes x_B$ for $x_A \in K_A$, $x_B \in K_B$. But one can also use \ref{item:tensor-intro-unitEffect} more explicitly to start the construction from order unit spaces corresponding to the effect algebras $E(K_A)$ and $E(K_B)$, hence we keep it in the list of assumptions.

At last, one can discuss \ref{item:tensor-intro-tomographicLocality}. This assumption is often called tomographic locality or local distinguishability and it was used as an axiom for the derivation of quantum theory in \cite{Hardy-derivationQT}. Without \ref{item:tensor-intro-tomographicLocality} we can have states in $K_{AB}$ that contain information hidden to Alice and Bob and which is only available when you can manipulate the whole state. It is being discussed whether physical theories should obey tomographic locality and theories without tomographic locality are actively researched \cite{DArianoErbaPerinotti-classicalEntanglement,DArianoErbaPerinotti-entanglement}.

\subsection{Multipartite scenarios}
In this section, we will investigate what additional assumptions one needs to make to describe scenarios including more than two parties. So let $K_A$, $K_B$, $K_C$ be state spaces. How do we then define the tripartite state space $K_{ABC}$? Clearly one option is to first form the bipartite state space $K_A \treal K_B$ and then add $K_C$, so that we get $(K_A \treal K_B) \treal K_C$. Other option is to first form $K_B \treal K_C$ and then add $K_A$ to obtain $K_A \treal (K_B \treal  K_C)$. It is natural to require that both of these construction yield the same result.
\begin{definition}
The tensor product of state spaces $\treal$ must be \emph{associative}, i.e., we must have
\begin{equation} \label{eq:tensor-multi-associative}
(K_A \treal K_B) \treal K_C = K_A \treal (K_B \treal  K_C).
\end{equation}
\end{definition}
We are now going to do three things: as first, we are going to show that if the tensor product of state spaces is associative, then so is the tensor product of effect algebras. Then we will prove that both the minimal tensor product $\tmin$ and the maximal tensor product $\tmax$ are associative. Finally, we are going to show a list of five identities that follow from the associativity of $\treal$ and that correspond to the pentagon diagram in category theory.

\begin{proposition} \label{prop:tensor-multi-Eassociative}
Let $K_A$, $K_B$, $K_C$ be state spaces. If $\treal$ is an associative tensor product of state spaces, then we have
\begin{equation}
(E(K_A) \treal E(K_B)) \treal E(K_C) = E(K_A) \treal (E(K_B) \treal E(K_C)).
\end{equation}
\end{proposition}
\begin{proof}
We have
\begin{align}
(E(K_A) \treal E(K_B)) \treal E(K_C) &= E(K_A \treal K_B) \treal E(K_C) = E((K_A \treal K_B) \treal K_C) \\
&= E(K_A \treal (K_B \treal K_C)) = E(K_A) \treal E(K_B \treal K_C) \\
&= E(K_A) \treal (E(K_B) \treal E(K_C)).
\end{align}
\end{proof}

\begin{proposition}
The minimal tensor product $\tmin$ is associative.
\end{proposition}
\begin{proof}
Let $K_A$, $K_B$, $K_C$ be state spaces. Denote
\begin{equation}
K_A \tmin K_B \tmin K_C = \conv \left( \{ x_A \otimes x_B \otimes x_C: x_A \in K_A, x_B \in K_B, x_C \in K_C \}. \right)
\end{equation}
We clearly have
\begin{align}
(K_A \tmin K_B) \tmin K_C &\subset K_A \tmin K_B \tmin K_C, \\
K_A \tmin (K_B \tmin K_C) &\subset K_A \tmin K_B \tmin K_C,
\end{align}
which one can show by simply writing out the general element of $(K_A \tmin K_B) \tmin K_C$ and $K_A \tmin (K_B \tmin K_C)$. Let $x_{ABC} \in K_A \tmin K_B \tmin K_C$, then
\begin{equation}
x_{ABC} = \sum_{i=1}^N \lambda_i y_{i,A} \otimes y_{i,B} \otimes y_{i,C}
\end{equation}
for some $y_{i,A} \in K_A$, $y_{i,B} \in K_B$, $y_{i,C} \in K_C$, $\lambda_i \in \Rp$ for $i \in \{1, \ldots, N\}$ and $\sum_{i=1}^n \lambda_i = 1$. Then $x_{ABC} \in (K_A \tmin K_B) \tmin K_C$ since $y_{i,A} \otimes y_{i,B} \in K_A \tmin K_B$ and $y_{i,C} \in K_C$. Also $x_{ABC} \in K_A \tmin (K_B \tmin K_C)$ as $y_{i,A} \in K_A$ and $y_{i,B} \otimes y_{i,C} \in K_B \tmin K_C$. So we have
\begin{equation}
(K_A \tmin K_B) \tmin K_C = K_A \tmin K_B \tmin K_C = K_A \tmin (K_B \tmin K_C).
\end{equation}
\end{proof}

\begin{corollary}
The maximal tensor product of effect algebras is associative.
\end{corollary}
\begin{proof}
The result follows from $E(K_A) \tmax E(K_B) = E(K_A \tmin K_B)$ and Proposition \ref{prop:tensor-multi-Eassociative}.
\end{proof}

\begin{proposition}
The maximal tensor product $\tmax$ is associative.
\end{proposition}
\begin{proof}
By definition, we have
\begin{equation}
\begin{split}
(K_A \tmax K_B) \tmax K_C = \{ \varphi \in A(K_A)^* \otimes A(K_B)^* \otimes A(K_C)^*: &\< \varphi, f_{AB} \otimes f_C \> \geq 0, \\
&\forall f_{AB} \in E(K_A) \tmin E(K_B), \forall f_C \in E(K_C), \\
&\< \varphi, 1_{K_A} \otimes 1_{K_B} \otimes 1_{K_C} \> = 1 \}.
\end{split}
\end{equation}
Since $f_{AB} \in E(K_A) \tmin E(K_B)$ can be written as a convex combination of the elements of the form $f_A \otimes f_B$, where $f_A \in E(K_A)$ and $f_B \in E(K_B)$, we get
\begin{equation}
\begin{split}
(K_A \tmax K_B) \tmax K_C = \{ \varphi \in A(K_A)^* \otimes A(K_B)^* \otimes A(K_C)^*: &\< \varphi, f_A \otimes f_B \otimes f_C \> \geq 0, \\
&\forall f_A \in E(K_A), \\
&\forall f_B \in E(K_B), \\
&\forall f_C \in E(K_C), \\
&\< \varphi, 1_{K_A} \otimes 1_{K_B} \otimes 1_{K_C} \> = 1 \}.
\end{split}
\end{equation}
It follows that since $f_B \otimes f_C \in E(K_B) \tmin E(K_C)$ for all $f_B \in E(K_B)$ and $f_C \in E(K_C)$, we get
\begin{equation} \label{eq:tensor-multi-maxAssoc-last}
\begin{split}
(K_A \tmax K_B) \tmax K_C = \{ \varphi \in A(K_A)^* \otimes A(K_B)^* \otimes A(K_C)^*: &\< \varphi, f_A \otimes f_{BC} \> \geq 0, \\
&\forall f_A \in E(K_A), \forall f_{BC} \in E(K_B) \tmin E(K_B), \\
&\< \varphi, 1_{K_A} \otimes 1_{K_B} \otimes 1_{K_C} \> = 1 \}.
\end{split}
\end{equation}
One can see that the right hand side of \eqref{eq:tensor-multi-maxAssoc-last} is exactly the definition of $K_A \tmax (K_B \tmax K_C)$ and so the result follows.
\end{proof}

\begin{corollary}
The minimal tensor product of effect algebras is associative.
\end{corollary}
\begin{proof}
The result follows from $E(K_A) \tmin E(K_B) = E(K_A \tmax K_B)$ and Proposition \ref{prop:tensor-multi-Eassociative}.
\end{proof}

\begin{proposition} \label{prop:tensor-multi-pentagon}
Let $K_A$, $K_B$, $K_C$, $K_D$ be state spaces. For a tensor product of state spaces $\treal$ it holds that
\begin{align}
(K_A \treal K_B) \treal (K_C \treal K_D) &= K_A \treal (K_B \treal (K_C \treal K_D)) = K_A \treal ((K_B \treal K_C) \treal K_D) \\
&= (K_A \treal (K_B \treal K_C)) \treal K_D = ((K_A \treal K_B) \treal K_C) \treal K_D
\end{align}
which can be also written as the following diagram of equalities:
\begin{equation} \label{eq:tensor-multi-pentagon}
\begin{split}
\begin{tikzpicture}
\node(P0)at (90:4.3125cm){$(K_A \treal K_B) \treal (K_C \treal K_D)$};
\node(P1)at (90+72:3.75cm){$K_A \treal (K_B \treal (K_C \treal K_D))$};
\node(P2)at (90+2*72:3.75cm) {$K_A \treal ((K_B \treal K_C) \treal K_D)$};
\node(P3)at (90+3*72:3.75cm){$(K_A \treal (K_B \treal K_C)) \treal K_D$};
\node(P4)at (90+4*72:3.75cm){$((K_A \treal K_B) \treal K_C) \treal K_D$};
\draw[
	line width=0.15mm,
	line cap=butt,
	double=white,
	double distance=0.75mm
	]
(P0) -- (P1)
(P1) -- (P2)
(P2) -- (P3)
(P3) -- (P4)
(P4) -- (P0);
\end{tikzpicture}
\end{split}
\end{equation}
\end{proposition}
\begin{proof}
The result follows by repeated application of \eqref{eq:tensor-multi-associative}.
\end{proof}
The importance of Proposition \ref{prop:tensor-multi-pentagon} is hidden in \eqref{eq:tensor-multi-pentagon}. \eqref{eq:tensor-multi-pentagon} corresponds to the pentagon diagram that is key in defining monoidal categories. Monoidal categories are very general mathematical structures that generalize tensor products of various objects. Our current framework can be formulated as a category of state spaces and \eqref{eq:tensor-multi-pentagon} shows that the tensor product of state spaces $\treal$ gives it the structure of a monoidal category. There is a slight caveat: as we have already explained, $\treal$ does not have to be defined for all state spaces, but only for selected state spaces. Therefore, to be more precise, one can show that a collection of selected state spaces for which $\treal$ is defined is a monoidal category. From now on we will assume that the tensor product $\treal$ is associative and always defined whenever needed.

\subsection{Diagrammatic notation for multipartite scenarios}
We will now explain, how to use diagrammatic notation in multipartite scenarios. So let $x_A \in K_A$ and $x_B$ in $K_B$, then we will use
\begin{equation}
\begin{quantikz}[row sep=\the\rowsep]
& \prepareC{x_A} & \qw{K_A} \\
& \prepareC{x_B} & \qw{K_B}
\end{quantikz}
\end{equation}
to denote the state $x_A \otimes x_B \in K_A \treal K_B$. We will use
\begin{equation}
\begin{quantikz}[row sep=\the\rowsep]
& \multiprepareC[2]{y_{AB}} & \qw{K_A} \\
& & \qw{K_B}
\end{quantikz}
\end{equation}
to denote a general $y_{AB} \in K_A \treal K_B$. Similarly for effects, we will use
\begin{equation}
\begin{quantikz}[row sep=\the\rowsep]
& \meterD{f_A}{K_A} \\
& \meterD{f_B}{K_B}
\end{quantikz}
\end{equation}
to denote the state $f_A \otimes f_B \in E(K_A) \treal E(K_B)$. We will use
\begin{equation}
\begin{quantikz}[row sep=\the\rowsep]
& \qw{K_A} & \multimeterD[2]{g_{AB}} \\
& \qw{K_B} &
\end{quantikz}
\end{equation}
to denote a general $g_{AB} \in E(K_A) \treal E(K_B)$. For $y_{AB} \in K_A \treal K_B$ and $g_{AB} \in E(K_A) \treal E(K_B)$ we then have
\begin{equation}
\begin{quantikz}[row sep=\the\rowsep, align equals at=1.5]
& \multiprepareC[2]{y_{AB}} & \qw{K_A} & \multimeterD[2]{g_{AB}} \\
& & \qw{K_B} &
\end{quantikz}
= \< y_{AB}, g_{AB}\>.
\end{equation}
In the future, we will mostly omit the wire labels that specify the respective state spaces.

\subsection{Partial trace and monogamy of entanglement}
Let $K_A$, $K_B$ be state spaces, let $x_{AB} \in K_A \treal K_B$ and let $f_B \in E(K_B)$. Can we define
\begin{equation} \label{eq:tensor-partial-xABfB}
\begin{quantikz}[row sep=\the\rowsep]
& \multiprepareC[2]{x_{AB}} & \qw \\
& & \meterD{f_B}
\end{quantikz}
\end{equation}
and does it have any meaning? We will first show that the object in \eqref{eq:tensor-partial-xABfB} has a valid mathematical meaning, then we will proceed with proving some of its properties as well as more general results about entanglement. Note that the inline equivalent of object in \eqref{eq:tensor-partial-xABfB} is $(\id_{K_A} \otimes f_B)(x_{AB})$, i.e.,
\begin{equation}
\begin{quantikz}[row sep=\the\rowsep, align equals at=1.5]
& \multiprepareC[2]{x_{AB}} & \qw \\
& & \meterD{f_B}
\end{quantikz}
= (\id_{K_A} \otimes f_B)(x_{AB}),
\end{equation}
where $\id_{K_A}$ denotes the identity map $\id_{K_A}: A(K_A)^* \to A(K_A)^*$ and $f_B$ si now treated as a linear map $f_B: A(K_B)^* \to \RR$. Then $\id_{K_A} \otimes f_B$ is a linear map $\id_{K_A} \otimes f_B: A(K_A)^* \otimes A(K_B)^* \to A(K_A)^*$.

The object in \eqref{eq:tensor-partial-xABfB} has an unused output wire in the $K_A$ system, so for any $g_A \in E(K_A)$ we can construct
\begin{equation}
\begin{quantikz}[row sep=\the\rowsep, align equals at=1.5]
& \multiprepareC[2]{x_{AB}} & \meterD{g_A} \\
& & \meterD{f_B}
\end{quantikz}
= \< x_{AB}, g_A \otimes f_B \>.
\end{equation}
In other words, the object in \eqref{eq:tensor-partial-xABfB} behaves as a functional on $A(K_A)$ and so we must have
\begin{equation}
\begin{quantikz}[row sep=\the\rowsep, align equals at=1.5]
& \multiprepareC[2]{x_{AB}} & \qw \\
& & \meterD{f_B}
\end{quantikz}
\in A(K_A)^*.
\end{equation}
To better demonstrate our point, assume that $x_{AB} = y_A \otimes y_B$ for some $y_A \in K_A$ and $y_B \in K_B$. We then have
\begin{equation} \label{eq:tensor-partial-yAyBfB}
\begin{quantikz}[row sep=\the\rowsep, align equals at=1.5]
& \multiprepareC[2]{x_{AB}} & \qw \\
& & \meterD{f_B}
\end{quantikz}
=
\begin{quantikz}[row sep=\the\rowsep, align equals at=1.5]
&[\prepfix] \prepareC{y_A} & \qw \\
& \prepareC{y_B} & \meterD{f_B}
\end{quantikz}
=
\begin{quantikz}[row sep=\the\rowsep, align equals at=1]
&\lstick{$\<y_B, f_B\>$} &[\prepfix] \prepareC{y_A} & \qw
\end{quantikz}.
\end{equation}
Since $A(K_A)^* \otimes A(K_B)^* = \linspan( \{ y_A \otimes y_B: y_A \in K_A, y_B \in K_B \}$, it follows that we can also define the object in \eqref{eq:tensor-partial-xABfB} by writing $x_{AB}$ as linear combination (with possibly non-positive coefficients) of product states $y_A \otimes y_B$ and using \eqref{eq:tensor-partial-yAyBfB}. Let us summarize the results so far.
\begin{proposition} \label{prop:tensor-partial-equivOneLeg}
Let $x_{AB} \in K_A \treal K_B$ and let $y_{i,A} \in K_A$, $y_{i,B} \in K_B$, $\alpha_i \in \RR$ for $i \in \{1, \ldots, N\}$ be such that $x_{AB} = \sum_{i=1}^N \alpha_i y_{i,A} \otimes y_{i,B}$. Let $f_B \in E(K_B)$ and let $\varphi_A \in A(K_A)^*$ be given for $g_A \in A(K_A)$ as
\begin{equation}
\< \varphi_A, g_A \> = \< x_{AB}, g_A \otimes f_B \>.
\end{equation}
Then we have
\begin{equation} \label{eq:tensor-partial-equivOneLeg}
\begin{quantikz}[row sep=\the\rowsep, align equals at=1.5]
& \multiprepareC[2]{x_{AB}} & \qw \\
& & \meterD{f_B}
\end{quantikz}
=
\sum_{i=1}^N \alpha_i \< y_{i,B}, f_B \>
\begin{quantikz}[row sep=\the\rowsep, align equals at=1]
&[\prepfix] \prepareC{y_{i,A}} & \qw
\end{quantikz}
=
\begin{quantikz}[row sep=\the\rowsep, align equals at=1]
&[\prepfix] \prepareC{\varphi_A} & \qw
\end{quantikz}
\end{equation}
\end{proposition}
\begin{proof}
Let $x = \sum_{i=1}^N \alpha_i y_{i,A} \otimes y_{i,B}$, then we have
\begin{equation}
\begin{quantikz}[row sep=\the\rowsep, align equals at=1.5]
& \multiprepareC[2]{x_{AB}} & \qw \\
& & \meterD{f_B}
\end{quantikz}
=
\sum_{i=1}^N \alpha_i
\begin{quantikz}[row sep=\the\rowsep, align equals at=1.5]
&[\prepfix] \prepareC{y_{i,A}} & \qw \\
& \prepareC{y_{i,B}} & \meterD{f_B}
\end{quantikz}
=
\sum_{i=1}^N \alpha_i \< y_{i,B}, f_B \>
\begin{quantikz}[row sep=\the\rowsep, align equals at=1]
&[\prepfix] \prepareC{y_{i,A}} & \qw
\end{quantikz}
\end{equation}
and so we have proved the first equality in \eqref{eq:tensor-partial-equivOneLeg}. To prove the second equality, first note that if $\psi_1, \psi_2 \in A(K)^*$ are such that for all $f \in A(K)$ we have $\< \psi_1, f\> = \< \psi_2, f\>$ then $\psi_1 = \psi_2$. Moreover, it is sufficient to check the equality only for all $f \in E(K)$, since $A(K) = \linspan(E(K))$. So now let $g_A \in E(K_A)$, then we have
\begin{equation}
\begin{quantikz}[row sep=\the\rowsep, align equals at=1.5]
& \multiprepareC[2]{x_{AB}} & \meterD{g_A} \\
& & \meterD{f_B}
\end{quantikz}
= \< x_{AB}, g_A \otimes f_B \> = \< \varphi_A, g_A \>
\end{equation}
and the second equality in \eqref{eq:tensor-partial-equivOneLeg} follows.
\end{proof}
One can easily prove many other results similar to Proposition \ref{prop:tensor-partial-equivOneLeg}, such as:
\begin{enumerate}
\item Let $x_{ABC} \in K_A \treal K_B \treal K_C$ and $f_C \in E(K_C)$, then
\begin{equation}
\begin{quantikz}[align equals at=2]
& \multiprepareC[3]{x_{ABC}} & \qw \\
& & \qw \\
& & \meterD{f_C}
\end{quantikz}
\in A(K_A)^* \otimes A(K_B)^*.
\end{equation}
\item Let $x_A \in K_A$ and $f_{AB} \in E(K_A) \treal E(K_B)$, then
\begin{equation}
\begin{quantikz}[row sep=\the\rowsep, align equals at=1.5]
& \prepareC{x_A} & \multimeterD[2]{f_{AB}} \\
& &
\end{quantikz}
\in A(K_B).
\end{equation}
\item Let $x_{AB} \in K_A \treal K_B$ and $f_{BC} \in E(K_B) \treal E(K_C)$, then
\begin{equation}
\begin{quantikz}[align equals at=2]
& \multiprepareC[2]{x_{AB}} & \qw \\
& & \multimeterD[2]{f_{BC}} \\
& &
\end{quantikz}
\in A(K_C) \otimes A(K_A)^*.
\end{equation}
\end{enumerate}

Let again $x_{AB} \in K_A \treal K_B$ and $f_B \in E(K_B)$ and note that for $g_A \in E(K_A)$ we have
\begin{equation} \label{eq:tensor-partial-isPositve}
\begin{quantikz}[row sep=\the\rowsep,align equals at=1.5]
& \multiprepareC[2]{x_{AB}} & \meterD{g_A} \\
& & \meterD{f_B}
\end{quantikz}
= \< x_{AB}, g_A \otimes f_B \> \geq 0.
\end{equation}
So we get
\begin{equation} \label{eq:tensor-partial-inPositiveCone}
\begin{quantikz}[row sep=\the\rowsep,align equals at=1.5]
& \multiprepareC[2]{x_{AB}} & \qw \\
& & \meterD{f_B}
\end{quantikz}
\in A(K_A)^{*+}
\end{equation}
from which the next result easily follows.
\begin{proposition} \label{prop:tensor-partial-multipleOfState}
Let $x_{AB} \in K_A \treal K_B$ and $f_B \in E(K_B)$, then there is $y_A \in K_A$ such that
\begin{equation}
\begin{quantikz}[row sep=\the\rowsep,align equals at=1.5]
& \multiprepareC[2]{x_{AB}} & \qw \\
& & \meterD{f_B}
\end{quantikz}
=
\< x_{AB}, 1_{K_A} \otimes f_B \>
\begin{quantikz}[row sep=\the\rowsep,align equals at=1]
&[\prepfix] \prepareC{y_A} & \qw
\end{quantikz}
\end{equation}
\end{proposition}
\begin{proof}
We already know that for every $\varphi \in A(K)^{*+}$ there must exist $y_A \in K_A$ and $\lambda \in \Rp$ such that $\varphi = \lambda y_A$, see Lemma \ref{lemma:basic-results-base}. So from \eqref{eq:tensor-partial-inPositiveCone} we get
\begin{equation}
\begin{quantikz}[row sep=\the\rowsep,align equals at=1.5]
& \multiprepareC[2]{x_{AB}} & \qw \\
& & \meterD{f_B}
\end{quantikz}
=
\lambda
\begin{quantikz}[row sep=\the\rowsep,align equals at=1]
&[\prepfix] \prepareC{y_A} & \qw
\end{quantikz}
\end{equation}
for some $\lambda \in \Rp$. By applying the unit effect to the free leg, we get
\begin{equation}
\begin{quantikz}[row sep=\the\rowsep,align equals at=1.5]
& \multiprepareC[2]{x_{AB}} & \ground{} \\
& & \meterD{f_B}
\end{quantikz}
=
\lambda
\begin{quantikz}[row sep=\the\rowsep,align equals at=1]
&[\prepfix] \prepareC{y_A} & \ground{}
\end{quantikz}
= \lambda
\end{equation}
which concludes the proof.
\end{proof}

\begin{corollary} \label{coro:tensor-partial-trace}
Let $x_{AB} \in K_A \treal K_B$, then
\begin{equation}
\begin{quantikz}[row sep=\the\rowsep,align equals at=1.5]
& \multiprepareC[2]{x_{AB}} & \qw \\
& & \ground{}
\end{quantikz}
\in K_A.
\end{equation}
\end{corollary}
\begin{proof}
Follows from Proposition \ref{prop:tensor-partial-multipleOfState}.
\end{proof}
The process of applying the unit effect to one leg of a bipartite state $x_{AB} \in K_A \treal K_B$, i.e., the map
\begin{equation}
\begin{quantikz}[align equals at=1.5]
& \multiprepareC[2]{x_{AB}} & \qw \\
& & \qw
\end{quantikz}
\mapsto
\begin{quantikz}[row sep=\the\rowsep,align equals at=1.5]
&[\prepfix] \multiprepareC[2]{x_{AB}} & \qw \\
& & \ground{}
\end{quantikz}
\end{equation}
is called partial trace. The name comes from quantum theory, where this construction corresponds to the partial trace over a subspace of the Hilbert space. Partial trace is an important concept, because it describes the local state that Alice (or Bob) have at their disposal when they work with the bipartite state $x_{AB} \in K_A \treal K_B$. Partial trace is also a key concept in monogamy of entanglement, which is the following result.
\begin{theorem} \label{thm:tensor-partial-monogamyTrace}
Let $x_{AB} \in K_A \treal K_B$ be such that
\begin{equation} \label{eq:tensor-partial-monogamyTrace}
\begin{quantikz}[row sep=\the\rowsep,align equals at=1.5]
& \multiprepareC[2]{x_{AB}} & \qw \\
& & \ground{}
\end{quantikz}
=
\begin{quantikz}[row sep=\the\rowsep,align equals at=1]
&[\prepfix] \prepareC{y_A} & \qw
\end{quantikz}
\end{equation}
where $y_A$ is a pure state. Then $x_{AB} = y_A \otimes z_B$ for some $z_B \in K_B$, i.e.,
\begin{equation}
\begin{quantikz}[row sep=\the\rowsep,align equals at=1.5]
& \multiprepareC[2]{x_{AB}} & \qw \\
& & \qw
\end{quantikz}
=
\begin{quantikz}[row sep=\the\rowsep,align equals at=1.5]
&[\prepfix] \prepareC{y_A} & \qw \\
&[\prepfix] \prepareC{z_B} & \qw
\end{quantikz}
\end{equation}
\end{theorem}
\begin{proof}
The proof can be found in \cite[Lemma 3.]{BarnumBarrettLeiferWilce-noBroadcasting}. We will provide exactly the same proof, only formulated in the language presented so far. Let $f_B \in E(K_B)$, then also $1_{K_B}-f_B \in E(K_B)$, see Lemma \ref{lemma:basic-EA-fPerp}. We have
\begin{equation}
\begin{quantikz}[row sep=\the\rowsep,align equals at=1.5]
& \multiprepareC[2]{x_{AB}} & \qw \\
& & \ground{}
\end{quantikz}
=
\begin{quantikz}[row sep=\the\rowsep,align equals at=1.5]
&[\prepfix] \multiprepareC[2]{x_{AB}} & \qw \\
& & \meterD{f_B}
\end{quantikz}
+
\begin{quantikz}[row sep=\the\rowsep,align equals at=1.5]
&[\prepfix] \multiprepareC[2]{x_{AB}} & \qw \\
& & \meterD{1_{K_B} - f_B}
\end{quantikz}
\end{equation}
which one can check by applying $g_A \in E(K_A)$ to the free leg and observing, that the equality holds. According to Proposition \ref{prop:tensor-partial-multipleOfState} we must have
\begin{equation} \label{eq:tensor-partial-monogamy-fB}
\begin{quantikz}[row sep=\the\rowsep,align equals at=1.5]
& \multiprepareC[2]{x_{AB}} & \qw \\
& & \meterD{f_B}
\end{quantikz}
=
\begin{quantikz}[row sep=\the\rowsep,align equals at=1]
&\lstick{$\< x_{AB}, 1_{K_A} \otimes f_B \>$} &[\prepfix] \prepareC{z_A} & \qw
\end{quantikz}
\end{equation}
and
\begin{equation}
\begin{quantikz}[row sep=\the\rowsep,align equals at=1.5]
& \multiprepareC[2]{x_{AB}} & \qw \\
& & \meterD{1_{K_B} - f_B}
\end{quantikz}
=
\begin{quantikz}[row sep=\the\rowsep,align equals at=1]
&\lstick{$\< x_{AB}, 1_{K_A} \otimes (1_{K_B} - f_B) \>$} &[\prepfix] \prepareC{w_A} & \qw
\end{quantikz}
\end{equation}
for some $z_A, w_A \in K_A$. So we have
\begin{equation}
\begin{quantikz}[row sep=\the\rowsep,align equals at=1.5]
& \multiprepareC[2]{x_{AB}} & \qw \\
& & \ground{}
\end{quantikz}
=
\begin{quantikz}[row sep=\the\rowsep,align equals at=1]
&\lstick{$\< x_{AB}, 1_{K_A} \otimes f_B \>$} &[\prepfix] \prepareC{z_A} & \qw
\end{quantikz}
+
\begin{quantikz}[row sep=\the\rowsep,align equals at=1]
&\lstick{$\< x_{AB}, 1_{K_A} \otimes (1_{K_B} - f_B) \>$} &[\prepfix] \prepareC{w_A} & \qw
\end{quantikz}
\end{equation}
Using \eqref{eq:tensor-partial-monogamyTrace} we get
\begin{equation}
\begin{quantikz}[row sep=\the\rowsep,align equals at=1]
&[\prepfix] \prepareC{y_A} & \qw
\end{quantikz}
=
\begin{quantikz}[row sep=\the\rowsep,align equals at=1]
&\lstick{$\< x_{AB}, 1_{K_A} \otimes f_B \>$} &[\prepfix] \prepareC{z_A} & \qw
\end{quantikz}
+
\begin{quantikz}[row sep=\the\rowsep,align equals at=1]
&\lstick{$(1-\< x_{AB}, 1_{K_A} \otimes f_B \>)$} &[\prepfix] \prepareC{w_A} & \qw
\end{quantikz}
\end{equation}
Since $y_A$ is a pure state, we must have $z_A = w_A = y_A$. This is an important point, because in general $w_A$ and $z_A$ would depend on the choice of $f_B$, but since $y_A$ is a pure state, we have $z_A = w_A = y_A$, and so for all $f_B \in E(K_B)$ we get the same $z_A$ and $w_A$. Let us denote
\begin{equation}
\begin{quantikz}[row sep=\the\rowsep,align equals at=1.5]
& \multiprepareC[2]{x_{AB}} & \ground{} \\
& & \qw
\end{quantikz}
=
\begin{quantikz}[row sep=\the\rowsep,align equals at=1]
&[\prepfix] \prepareC{z_B} & \qw
\end{quantikz}
\end{equation}
where $z_B \in K_B$. Let now $g_A \in E(K_A)$, then using \eqref{eq:tensor-partial-monogamy-fB} we get
\begin{equation} \label{eq:tensor-partial-monogamy-separation}
\begin{quantikz}[row sep=\the\rowsep,align equals at=1.5]
& \multiprepareC[2]{x_{AB}} & \meterD{g_A} \\
& & \meterD{f_B}
\end{quantikz}
=
\begin{quantikz}[row sep=\the\rowsep,align equals at=1]
&\lstick{$\< x_{AB}, 1_{K_A} \otimes f_B \>$} &[\prepfix] \prepareC{y_A} & \meterD{g_A}
\end{quantikz}
=
\begin{quantikz}[row sep=\the\rowsep,align equals at=1.5]
&[\prepfix] \prepareC{y_A} & \meterD{g_A} \\
&[\prepfix] \prepareC{z_B} & \meterD{f_B}
\end{quantikz}
\end{equation}
where we have used that
\begin{equation}
\< x_{AB}, 1_{K_A} \otimes f_B \> =
\begin{quantikz}[row sep=\the\rowsep,align equals at=1.5]
&[\prepfix] \multiprepareC[2]{x_{AB}} & \ground{} \\
& & \meterD{f_B}
\end{quantikz}
=
\begin{quantikz}[row sep=\the\rowsep,align equals at=1]
&[\prepfix] \prepareC{z_B} & \meterD{f_B}
\end{quantikz}
\end{equation}
It follows from \eqref{eq:tensor-partial-monogamy-separation} that for any $g_A \in E(K_A)$ and $f_B \in E(K_B)$ we have
\begin{equation}
\< x_{AB}, g_A \otimes f_B \> = \< y_A \otimes z_B, g_A \otimes f_B \>
\end{equation}
and so we must have $x_{AB} = y_A \otimes z_B$ as a result of tomographic locality of the tensor product.
\end{proof}

\subsection{Existence of entanglement}
We have already argued that the minimal and maximal tensor products are the smallest possible and largest possible choice of the bipartite state space. In Proposition \ref{prop:tensor-bipartite-inclusions} we showed that for any two state spaces $K_A$, $K_B$, we have $K_A \tmin K_B \subset K_A \tmax K_B$. If we would have $K_A \tmin K_B = K_A \tmax K_B$, then the choice of the bipartite state space would be unique, but also all bipartite states would be separable and there would be no entangled states in $K_A \treal K_B$. It is intuitive to expect that entanglement does not exist in classical theory. One can easily prove the following, slightly more general result.
\begin{proposition} \label{prop:tensor-existence-notWithSn}
Let $K$ be a state space and let $S_n$ be a simplex, i.e., a classical state space. Then
\begin{equation}
K \tmin S_n = K \tmax S_n.
\end{equation}
and $S_n \tmin K = S_n \tmax K$.
\end{proposition}
\begin{proof}
Clearly if $K \tmin S_n = K \tmax S_n$ then also $S_n \tmin K = S_n \tmax K$ because the definitions of minimal and maximal tensor products are symmetric. So let $S_n = \conv( \{ s_1, \ldots, s_n \} )$ be a simplex with pure states $s_1, \ldots, s_n$. Let $b_1, \ldots, b_n \in E(S_n)$ be the effects such that $\< s_i, b_j \> = \delta_{ij}$ for all $i,j \in \{1, \ldots, n\}$. Also remember that $\{ s_1, \ldots, s_n \}$ is a basis of $A(S_n)^*$ and $\{ b_1, \ldots, b_n \}$ is a basis of $A(S_n)$. Let $y \in K \tmax S_n$, then there are $\{ v_1, \ldots, v_n \} \subset A(K)^*$ such that $y = \sum_{i=1}^n v_i \otimes s_i$, see Lemma \ref{lemma:bilinear-productSum}. Since we have
\begin{equation}
\begin{quantikz}[align equals at=1.5]
& \multiprepareC[2]{y} & \qw{K} \\
& & \meterD{b_i}{S_n}
\end{quantikz}
=
\begin{quantikz}[row sep=\the\rowsep,align equals at=1]
&[\prepfix] \prepareC{v_i} & \qw{K}
\end{quantikz}
\end{equation}
it follows from Proposition \ref{prop:tensor-partial-multipleOfState} that $v_i \in A(K)^{*+}$, i.e., $v_i = \lambda_i x_i$ for some $x_i \in K$, $\lambda_i \in \Rp$ for all $i \in \{1, \ldots, n \}$. So we have
\begin{equation}
y = \sum_{i=1}^n \lambda_i x_i \otimes s_i \in K \tmin S_n.
\end{equation}
Hence we have proved that $K \tmax S_n \subset K \tmin S_n$, from which the result follows.
\end{proof}

One can now ask, whether Proposition \ref{prop:tensor-existence-notWithSn} gives also sufficient condition for non-existence of entangled states. This problem was in the context of tensor products of the underlying cones already investigated in \cite{NamiokaPhelps-cones, Barker-cones} but it was only recently solved in \cite{AubrunLamiPalazuelosPlavala-cones}.

\begin{theorem} \label{thm:tensor-existence-nonClassical}
Let $K_A$, $K_B$ be state spaces, then we have $K_A \tmin K_B = K_A \tmax K_B$ if and only if at least one of the state spaces is a simplex, i.e., if and only if we have $K_A = S_n$ or $K_B = S_n$.
\end{theorem}
\begin{proof}
See \cite{AubrunLamiPalazuelosPlavala-cones}.
\end{proof}

\section{Channels, measurements and instruments} \label{sec:channels}

%
%
We finally get to describe transformations of systems. There are in principle three different types of transformations: channels, measurements and instruments. Channels map states to states and they describe some manipulation of the system, e.g., time-evolution. Measurements map states of a given system to probability distributions over measurement outcomes, they describe the measurement process in the sense that they give us the probabilities of occurrence of the outcomes. Instruments describe the measurement process by mapping a state to weighted set of post-measurement states.

We will argue that measurements and instruments are special kinds of channels. This may appear as counter-intuitive at first, since physically channels and measurements are different object. We already know from Section \ref{sec:CT} that probability distributions correspond to classical state spaces, and so measurement as a map from states to probability distributions can be described as a channel from a given state space $K$ to classical state space $S_n$. Similarly, instruments can be described as channels from $K$ to $K \treal S_n$.

\subsection{Channels}
Channel is a transformation of a system that can be either appended to a preparation procedure, or prepended to a measurement procedure, such that mixtures are preserved. Let us unpack this statement: since channel should transform a state to something measurable, it must map states of one system to states of other system. Moreover, we require that channels preserve mixtures, which just implies that a channel is an affine map between state spaces.

\begin{definition}
Let $K_A$, $K_B$ be state spaces. \emph{Channel} $\Phi$ from $K_A$ to $K_B$ is an affine map $\Phi: K_A \to K_B$, i.e., for all $x_A, y_A \in K_A$ and $\lambda \in [0,1]$ we have
\begin{equation}
\Phi(\lambda x_A + (1-\lambda) y_A) = \lambda \Phi(x_A) + (1-\lambda) \Phi(y_A).
\end{equation}
We will denote the set of all channels $\Phi: K_A \to K_B$ by $\chan(K_A, K_B)$. We will use the shorthand $\chan(K)$ for channels $\Phi: K \to K$, i.e., $\chan(K) = \chan(K,K)$.
\end{definition}
Since $A(K_A)^* = \linspan(K_A)$ and $A(K_B)^* = \linspan(K_B)$, we can easily extend $\Phi$ to a linear map $\Phi: A(K_A)^* \to A(K_B)^*$ as follows: let $v_A \in A(K_A)^*$, then according to Lemma \ref{lemma:basic-results-span} we have $v_A = \lambda x_A - \mu y_A$ for some $x_A, y_A \in K_A$ and $\lambda, \mu \in \Rp$. Then we have $\Phi(v_A) = \lambda \Phi(x_A) - \mu \Phi(y_A)$. One can check that then $\Phi: A(K_A)^* \to A(K_B)^*$ is a linear map. Since $\Phi: A(K_A)^{*+} \to A(K_B)^{*+}$, the map $\Phi$ is called positive. We will now present examples of channels one can find in every GPT.
\begin{example}
Let $K$ be a state space and let $\id_K \in \chan(K)$ be the identity map, given as $\id_K(x) = x$ for all $x \in K$. It is straightforward to check that $\id_K$ is a channel and that the induced linear map $\id_K: A(K)^* \to A(K)^*$ is positive linear map. $\id_K$ is usually called the identity map, the identity channel, or just identity.
\end{example}

\begin{example} \label{exm:channels-channels-constant}
Let $K_A$, $K_B$ be state spaces, let $x_B \in K_B$ be a fixed state and define a channel $\tau_x \in \chan(K_A, K_B)$ as $\tau_x(y_A) = x_B$ for all $y_A \in K_A$. To see that $\tau_x$ is a channel, we need to verify that it is affine. So let $y_A, z_A \in K_A$, $\lambda \in [0,1]$, then we have
\begin{equation}
\lambda \tau_x(y_A) + (1-\lambda) \tau_x(z_A) = \lambda x_B + (1-\lambda) x_B = x_B = \tau_x(\lambda y_A + (1-\lambda) z_A)
\end{equation}
and so $\tau_x$ is affine and a channel. $\tau_x$ is usually called the constant channel. When extended to a linear map $\tau_x: A(K_A)^* \to A(K_B)^*$, we get $\tau_x(v_A) = \< v_A, 1_{K_A} \> x_B$ for $v_A \in A(K_A)^*$. This is easy to derive, for every $v_A \in A(K_A)^*$ there are $y_A, z_A \in K_A$ and $\lambda, \mu \in \Rp$ such that $v_A = \lambda y_A - \mu z_A$ and by linearity we get
\begin{equation}
\tau_x(v_A) = \tau_x (\lambda y_A - \mu z_A) = \lambda \tau_x(y_A) - \mu \tau_x(z_A) = (\lambda - \mu) x_B
\end{equation}
and the result follows from $\< v_A, 1_{K_A} \> = \lambda - \mu$.
\end{example}

\begin{example} \label{exm:channels-channels-partialTrace}
Let $K_A$, $K_B$ be state spaces and let $\id_{K_A} \otimes 1_{K_B} \in \chan(K_A \treal K_B, K_A)$ be the partial trace map, i.e., for $x_{AB} \in K_A \treal K_B$ we have
\begin{equation}
\id_{K_A} \otimes 1_{K_B}:
\begin{quantikz}[row sep=\the\rowsep,align equals at=1.5]
&[\prepfix] \multiprepareC[2]{x_{AB}} & \qw{K_A} \\
& & \qw{K_B}
\end{quantikz}
\mapsto
\begin{quantikz}[row sep=\the\rowsep,align equals at=1.5]
&[\prepfix] \multiprepareC[2]{x_{AB}} & \qw{K_A} \\
& & \ground{}{K_B}
\end{quantikz}
\end{equation}

To see that $\id_{K_A} \otimes 1_{K_B}$ is a channel note that we have already showed that $(\id_{K_A} \otimes 1_{K_B})(x_{AB}) \in K_A$ in Corollary \ref{coro:tensor-partial-trace}, we only need to argue that $\id_{K_A} \otimes 1_{K_B}$ is affine. So let $x_{AB}, y_{AB} \in K_A \treal K_B$, $\lambda \in [0,1]$ and $f_A \in E(K_A)$, then
\begin{align}
\< (\id_{K_A} \otimes 1_{K_B})(\lambda x_{AB} + (1-\lambda) y_{AB}), f_A \> &= \< \lambda x_{AB} + (1-\lambda) y_{AB}, f_A \otimes 1_{K_B} \> \\
&=  \lambda \< x_{AB}, f_A \otimes 1_{K_B} \> + (1-\lambda) \< y_{AB}, f_A \otimes 1_{K_B} \> \\
&= \lambda \< (\id_{K_A} \otimes 1_{K_B})(x_{AB}), f_A \> + (1-\lambda) \< (\id_{K_A} \otimes 1_{K_B})(y_{AB}), f_A \>
\end{align}
and so
\begin{equation}
(\id_{K_A} \otimes 1_{K_B})(\lambda x_{AB} + (1-\lambda) y_{AB}) = \lambda (\id_{K_A} \otimes 1_{K_B})(x_{AB}) + (1-\lambda) (\id_{K_A} \otimes 1_{K_B})(y_{AB})
\end{equation}
follows.
\end{example}

\begin{lemma}
Let $\Phi_1, \Phi_2 \in \chan(K_A, K_B)$ be channels and let $\lambda \in [0,1]$, then also their convex combination $\lambda \Phi_1 + (1-\lambda) \Phi_2$, given for $x_A \in K_A$ as $(\lambda \Phi_1 + (1-\lambda) \Phi_2)(x_A) = \lambda \Phi_1(x_A) + (1-\lambda) \Phi_2(x_A)$ is also a channel.
\end{lemma}
\begin{proof}
Let $x_A \in K_A$, since $\Phi_1, \Phi_2$ are channels, we have $\Phi_1(x_A) \in K_B$ and $\Phi_2(x_A) \in K_B$, so $\lambda \Phi_1(x_A) + (1-\lambda) \Phi_2(x_A) \in K_B$ follows by convexity of $K_B$. So $\lambda \Phi_1 + (1-\lambda) \Phi_2 \in \chan(K_A, K_B)$, it is straightforward to verify that $\lambda \Phi_1 + (1-\lambda) \Phi_2$ is also affine.
\end{proof}

\begin{example}
Let $x \in K$ be a fixed point and let $\tau_x \in \chan(K)$ be the corresponding constant channel and let $\lambda \in [0,1]$, then we have $\lambda \id_K + (1-\lambda) \tau_x \in \chan(K)$, given for $y \in K$ as $(\lambda \id_K + (1-\lambda) \tau_x)(y) = \lambda y + (1-\lambda) x$.
\end{example}

We will use
\begin{equation}
\begin{quantikz}
&\qw{K_A} &\gate{\Phi} &\qw &\qw{K_B}
\end{quantikz}
\end{equation}
to denote the channel $\Phi\in \chan(K_A, K_B)$. Let $x_A \in K_A$, then
\begin{equation}
\begin{quantikz}[align equals at=1]
&\prepareC{\Phi(x_A)} &\qw
\end{quantikz}
=
\begin{quantikz}[align equals at=1]
&[\prepfix]\prepareC{x_A} &\gate{\Phi} &\qw
\end{quantikz}
\end{equation}
denotes the state $\Phi(x_A) \in K_B$. Note that for channels $\Phi_1 \in \chan(K_A, K_B)$, $\Phi_2 \in  \chan(K_B,K_C)$ we will use
\begin{equation}
\begin{quantikz}[align equals at=1]
&\gate{\Phi_1} &\gate{\Phi_2} &\qw
\end{quantikz}
=
\begin{quantikz}[align equals at=1]
&\gate{\Phi_2 \circ \Phi_1} &\qw
\end{quantikz}
\end{equation}
where $\Phi_2 \circ \Phi_1 \in \chan(K_A, K_C)$, $(\Phi_2 \circ \Phi_1)(x_A) = \Phi_2(\Phi_1(x_A))$ for $x_A \in K_A$, i.e. we use $\circ$ to denote the composition (also called concatenation) of channels. The identity channel $\id_K \in \chan(K)$ will be represented by a plain wire, i.e.,
\begin{equation}
\begin{quantikz}[align equals at = 1]
&\gate{\id_K} &\qw{}
\end{quantikz}
=
\begin{quantikz}[align equals at = 1]
&\qw{}
\end{quantikz}
\end{equation}
Let $\Phi \in \chan(K_A, K_B)$ and $f_B \in E(K_B)$, then we can construct
\begin{equation} \label{eq:channels-channels-HeisenbergPicture}
\begin{quantikz}[align equals at=1]
&\gate{\Phi} &\meterD{f_B}
\end{quantikz}
\in E(K_A).
\end{equation}
The object in \eqref{eq:channels-channels-HeisenbergPicture} belongs to $E(K_A)$ because for every $x_A \in K_A$ we have
\begin{equation}
\begin{quantikz}[align equals at=1]
&\prepareC{x_A} &\gate{\Phi} &\meterD{f_B}
\end{quantikz}
=
\begin{quantikz}[align equals at=1]
&[\prepfix] \prepareC{\Phi(x_A)} &\meterD{f_B}
\end{quantikz}
\in [0,1]
\end{equation}
since $\Phi(x_A) \in K_B$. We will denote
\begin{equation} \label{eq:channels-channels-adjointDef}
\begin{quantikz}[align equals at=1]
&\gate{\Phi} &\meterD{f_B}
\end{quantikz}
=
\begin{quantikz}[align equals at=1]
&\meterD{\Phi^*(f_B)}
\end{quantikz}
\end{equation}
where $\Phi^*: E(K_B) \to E(K_A)$ is the induced map. One can easily check that it extends to a linear map $\Phi^*:A(K_B) \to A(K_A)$.
\begin{definition}
Let $\Phi \in \chan(K_A, K_B)$ be a channel, then the \emph{adjoint map} $\Phi^*:E(K_B) \to E(K_A)$ is a linear map defined by \eqref{eq:channels-channels-adjointDef}, or equivalently, $\Phi^*: E(K_B) \to E(K_A)$ is the unique linear map such that for all $x_A \in K_A$ and $f_B \in E(K_B)$ we have
\begin{equation}
\< \Phi(x_A), f_B \> = \< x_A, \Phi^*(f_B) \>.
\end{equation}
\end{definition}
We already said that a channel can be seen both as appending instruction to preparations, but also as prepending instructions to measurements. The original channel $\Phi \in \chan(K_A, K_B)$ was mapping states to states and so it was appending instructions to a preparation procedure; we usually refer to this as the Schr\"{o}dinger picture. The adjoint map $\Phi^*:E(K_B) \to E(K_A)$ is prepending instructions to measurement procedures; we usually refer to this as the Heisenberg picture. Both of the maps $\Phi$ and $\Phi^*$ are different descriptions of the same thing. The following is an important and often used result about the adjoint map of a channel.
\begin{proposition} \label{prop:channels-channels-unital}
Let $K_A$, $K_B$ be state spaces and let $\Phi \in \chan(K_A, K_B)$ be a channel. Then the adjoint map $\Phi^*: E(K_A) \to E(K_B)$ is unital, i.e., we have $\Phi^* (1_{K_B}) = 1_{K_A}$.
\end{proposition}
\begin{proof}
Let $x_A \in K_A$, then we have
\begin{equation}
\< x_A, \Phi^*(1_{K_B}) \> = \< \Phi_A(x_A), 1_{K_B} \> = 1
\end{equation}
and so we must have $\Phi^* (1_{K_B}) = 1_{K_A}$.
\end{proof}

We will now construct a useful mathematical representation of channels. Let $\Phi \in \chan(K_A, K_B)$. Since $\Phi$ can be extended to a linear map $\Phi: A(K_A)^* \to A(K_B)^*$, it follows from Proposition \ref{prop:bilinear-isomorphisms} that this linear map corresponds to a vector from $A(K_A) \otimes A(K_B)^*$. And so, by omitting the isomorphism, we can write $\Phi \in A(K_A) \otimes A(K_B)^*$. It then follows that there are $g_{i,A} \in A(K_A)$ and $w_{i,B} \in A(K_B)^*$, $i \in \{1, \ldots, n\}$ such that $\Phi = \sum_{i=1}^n g_{i,A} \otimes w_{i,B}$. Then for $v_A \in A(K_A)^*$ and $f_B \in A(K_B)$ we have
\begin{equation}
\< \Phi(v_A), f_B \> = \sum_{i=1}^n \< v_A, g_{i,A} \> \< w_{i,B}, f_B \>.
\end{equation}
It follows that for $v_A \in A(K_A)^*$ we have
\begin{equation}
\Phi(v_A) = \sum_{i=1}^n \< v_A, g_{i,A} \> w_{i,B}.
\end{equation}
For $x_A \in K_A$ and $f_B \in E(K_B)$ we get $\< \Phi(x_A), f_B \> \geq 0$ which means that $\Phi \in A(K_A) \otimes A(K_B)^*$ must be positive in some sense. We can use this property together with Proposition \ref{prop:channels-channels-unital} to characterize all channels as s subset of $A(K_A) \otimes A(K_B)^*$.
\begin{proposition} \label{prop:channels-channels-tensorSubset}
Let $K_A$, $K_B$ be state spaces, then
\begin{equation}
\chan(K_A, K_B) = \{ \Phi \in A(K_A)^+ \tmax A(K_B)^{*+} : \Phi^*(1_{K_B}) = 1_{K_A} \},
\end{equation}
where
\begin{equation}
A(K_A)^+ \tmax A(K_B)^{*+} = \{ v \in A(K_A) \otimes A(K_B)^* : \< v, x_A \otimes f_B \> \geq 0, \forall x_A \in A(K_A)^{*+}, \forall f_B \in A(K_B)^+ \}.
\end{equation}
\end{proposition}
\begin{proof}
We will first prove that if $\Phi \in \chan(K_A, K_B)$ is a channel, then $\Phi \in A(K_A)^+ \tmax A(K_B)^{*+}$. Let $x_A \in K_A$, $f_B \in E(K_B)$, then we have
\begin{equation}
\< \Phi, x_A \otimes f_B \> = \< \Phi(x_A), f_B \> \geq 0,
\end{equation}
where we have used the isomorphism between linear maps and elements of tensor product, see Proposition \ref{prop:bilinear-isomorphisms}. It follows that we have $\Phi \in A(K_A)^+ \tmax A(K_B)^{*+}$ and since we already know that $\Phi^*(1_{K_B}) = 1_{K_A}$, see Proposition \ref{prop:channels-channels-unital}, we get
\begin{equation}
\Phi \in \{ \Psi \in A(K_A)^+ \tmax A(K_B)^{*+} : \Psi^*(1_{K_B}) = 1_{K_A} \}.
\end{equation}

Now let $\Phi \in A(K_A)^+ \tmax A(K_B)^{*+}$ be such that $\Phi^*(1_{K_B}) = 1_{K_A}$, and let $x_A \in K_A$. Then we can define $v_B \in A(K_B)^*$ as the unique element such that for all $f_B \in E(K_B)$ we have
\begin{equation}
\< v_B, f_B \> = \< \Phi, x_A \otimes f_B \>.
\end{equation}
We have $\< v_B, f_B \> \geq 0$ and so $v_B \in A(K_B)^{*+}$. Moreover we also have $\< v_B, 1_{K_B} \> = 1$ and so it follows from Theorem \ref{thm:basic-dual-stateSpace} that $v_B \in K_B$. Hence we can define $\Phi(x_B) = v_B$ and so $\Phi$ corresponds to a map $K_A \to K_B$; one can easily check that $\Phi$ defined like this is affine map. So it follows that $\Phi \in  \chan(K_A, K_B)$.
\end{proof}

The result above is extremely important, because it shows that we can treat the set of channels $\chan(K_A, K_B)$ as a state space. One can easily check that $\chan(K_A, K_B)$ is a base of a positive cone $A(K_A)^+ \tmax A(K_B)^{*+} \cap \linspan(\chan(K_A,K_B))$. This is an important result, because it follows that if we would be interested in, for example, discrimination of channels, we can use the result of Theorem \ref{thm:basic-norms-discrimination}. It also follows that we do not have to develop a separate theory of channels, or a separate theory of superchannels, that is maps that map channels to channels, all of these theories are already included in our formalism.

\begin{corollary}
Let $S^A_n$, $S^B_m$ be simplexes, given by their extreme points
\begin{align}
&S^A_n = \conv ( \{ s_{1,A} , \ldots, s_{n,A}  \} ), 
&&S^B_n = \conv ( \{ s_{1,B} , \ldots, s_{m,B}  \} ).
\end{align}
Then
\begin{equation} \label{eq:channels-channels-simplexTensor}
\chan(S^A_n, S^B_m) = \{ \Phi \in A(S^A_n)^+ \tmin A(S^B_m)^{*+} : \Phi^*(1_{S^A_n}) = 1_{S^B_m} \}
\end{equation}
and for $s_A \in S^A_n$ we have
\begin{equation}
\Phi(s_A) = \sum_{i=1}^n \sum_{j=1}^m \nu_{ij} \< s_A, b_{i,A} \> s_{j,B},
\end{equation}
where $b_{i,A} \in E(S_n)$ are the functions such that $\< b_{i,A}, s_{k,A} \> = \delta_{ik}$ for $i,k \in \{ 1, \ldots, n \}$. $\nu_{ij} \in \Rp$ are such that $\sum_{j=1}^m \nu_{ij} = 1$ for all $i \in \{ 1, \ldots, n \}$.
\end{corollary}
\begin{proof}
\eqref{eq:channels-channels-simplexTensor} follows from Propositions \ref{prop:channels-channels-tensorSubset} and \ref{prop:tensor-existence-notWithSn}. Note that
\begin{equation}
A(S^A_n) \tmin A(S^B_m) = \conv \cone ( \{ b_{i,A} \otimes s_{j, B} : i \in \{1, \ldots, n\}, j \in \{1, \ldots, m \} \} ),
\end{equation}
so we must have
\begin{equation}
\Phi = \sum_{i=1}^n \sum_{j=1}^m \nu_{ij} b_{i,A} \otimes s_{j,B}.
\end{equation}
We then have
\begin{equation}
\Phi^*(1_{S_n^B}) = \sum_{i=1}^n \sum_{j=1}^m \nu_{ij} \< s_{j,B}, 1_{S_n^B} \> b_{i,A} = \sum_{i=1}^n \sum_{j=1}^m \nu_{ij} b_{i,A}
\end{equation}
Since we must have $\Phi^*(1_{S_n^B}) = 1_{S_m^A} = \sum_{i=1}^n b_{i,A}$, we get $\sum_{j=1}^m \nu_{ij} = 1$.
\end{proof}

\subsection{Measurements}
As we have already pointed out, measurements are maps that map states to probability distributions; we will consider only probability distributions over finitely many possible outcomes. We have already discussed in Section \ref{sec:CT} that such probability distributions are in one-to-one correspondence with states on a classical state space $S_n$. Hence a measurement is a channel from a state space $K$ to $S_n$.
\begin{definition}
$n$-outcome \emph{measurement} is a channel $m \in \chan(K, S_n)$.
\end{definition}
We will simply use the word measurement when the number of outcomes will not be important. We immediately have the following:
\begin{proposition} \label{prop:channels-measurements-minTensor}
Let $m \in \chan(K, S_n)$ be a measurement, then $m \in A(K)^+ \tmin A(S_n)^{*+}$, where
\begin{equation}
A(K)^+ \tmin A(S_n)^{*+} = \conv \cone( \{ f \otimes s: f \in E(K), s \in S_n \} ).
\end{equation}
\end{proposition}
\begin{proof}
The result follows from Proposition \ref{prop:channels-channels-tensorSubset}, since we must have $A(K)^+ \tmax A(S_n)^{*+} = A(K)^+ \tmin A(S_n)^{*+}$ which follows from Proposition \ref{prop:tensor-existence-notWithSn}.
\end{proof}

Let $s_1, \ldots, s_n$ be the pure states in $S_n$, so that we have $S_n = \conv( \{ s_1, \ldots, s_n \} )$ and $A(S_n)^{*+} = \linspan( \{ s_1, \ldots, s_n \} )$. Let $m \in \chan(K, S_n)$ be a measurement, then according to Proposition \ref{prop:channels-measurements-minTensor} there must be $f_i \in E(K)$, $i \in \{1, \ldots, n\}$ such that
\begin{equation}
m = \sum_{i=1}^n f_i \otimes s_i.
\end{equation}
Then for $x \in K$ we have
\begin{equation}
m(x) = \sum_{i=1}^n \< x, f_i \> s_i.
\end{equation}
According to Proposition \ref{prop:channels-channels-unital} we must have $\< m(x), 1_{S_n} \> = 1$, which implies $\sum_{i=1}^n \< x, f_i \> =1$ for all $x \in K$ and so $\sum_{i=1}^n f_i = 1_{S_n}$. Thus we have proved the following
\begin{proposition} \label{prop:channels-measurements-effects}
$n$-outcome measurement $m \in \chan(K, S_n)$ is uniquely defined by effects $\{ f_1, \ldots, f_n \} \subset E(K)$ such that $\sum_{i=1}^n f_i = 1_{S_n}$.
\end{proposition}
\begin{proof}
The only thing that remains to be proved is the uniqueness of the set $\{ f_1, \ldots, f_n \}$. Let $\{b_1, \ldots, b_n\} \subset E(S_n)$ be the effects such that $\< s_i, b_j \> = \delta_{ij}$ for all $i,j \in \{1, \ldots, n\}$. Let $m \in \chan(K, S_n)$ and let $x \in K$, then we have $\<x, m^*(b_i) \> = \< m(x), b_i \> = \< x, f_i \>$ and so $f_i = m^*(b_i)$, where $m^*$ is adjoint map of $m$. So $f_i$ is uniquely given by $m$.
\end{proof}
The result above shows an equivalence between our operational definition of a measurement and the definition that is often used in quantum information theory, where measurements are often introduced as collections of effect. It follows from Proposition \ref{prop:channels-measurements-effects} that the two definitions are equivalent and we can without loss of generality either describe a measurement as a collection $f_1, \ldots, f_n$, where $f_i \in E(K)$, for $i \in \{1, \ldots, n\}$ and $\sum_{i=1}^n f_i = 1_K$, or as a map $m: K \to S_n$, $m(x) = \sum_{i=1}^n f_i(x) s_i$.

Let $n=2$ and consider the two-outcome measurement, i.e., channels $m_2 \in \chan(K, S_2)$. In this case, $m_2$ is uniquely specified by the two effects $\{ f, g \} \subset E(K)$. Since we must have $f+g=1_K$, we have $g = 1_K - f$ and so $m_2$ is uniquely specified by $f, 1_K - f$, or equivalently, $m_2$ is uniquely specified by a choice of $f \in E(K)$. Hence we get:
\begin{corollary}
The set of two-outcome measurements is isomorphic to $E(K)$.
\end{corollary}
This was an expected result, because we have introduced $E(K)$ as the set of classes of equivalence of all possible `yes'-`no' questions. A `yes'-`no' question is nothing else than a two-outcome measurement with some labels assigned to the outcomes. And so the result above only shows that our framework is consistent.

\subsection{Instruments}
Consider the scenario where we are not only interested in the statistics of a measurement, but also in the post-measurement state, that is in the state of the system after performing the measurement. There are several things to consider: the map from input state to post-measurement state should be a channel, because the post-measurement state has to be a well-defined state. But the map from input state to post-measurement state can depend on the outcome of the measurement; for example in quantum theory it is natural that the post-measurement state is described by an eigenvector corresponding to the observed eigenvalue.

\begin{definition}
\emph{Instrument} is a channel $\iI \in \chan(K, S_n \treal K)$ that describes the measurement and the resulting post-measurement state.
\end{definition}
One can clearly generalize an instrument to the case when post-measurement state belongs to a different system, in that case it would be $\iI \in \chan(K_A,  S_n \treal K_B)$. In diagrammatic notation, an instrument $\iI$ is represented as
\begin{equation}
\begin{quantikz}[row sep={\the\originsep,between origins}]
&\gate[3,nwires={1,3}]{\iI} &\qw{S_n} \\
& & \\
& & \qw{K}
\end{quantikz}
\end{equation}
Let $x \in K$ and let $\iI: K \to S_n$ be an instrument, then  $\iI(x) \in S_n \treal K$ and so we have $\iI(x) = \sum_{j=1}^n \lambda_j s_j \otimes y_j$ where $s_1, \ldots, s_n$ are the extreme points of $S_n$, $\lambda_j \in \Rp$, $y_j \in K$ for all $j \in \{1, \ldots, n\}$ and $\sum_{j=1}^n \lambda_j = 1$. We then have
\begin{equation}
\begin{quantikz}[row sep={\the\originsep,between origins}, align equals at=2]
& &\gate[3,nwires={1,3}]{\iI} &\meterD{b_j}{S_n} \\
&\prepareC{x} & & \\
& & & \qw{K}
\end{quantikz}
=
\begin{quantikz}[align equals at=1]
&\lstick{$\lambda_j$} &[\prepfix] \prepareC{y_j} &\qw
\end{quantikz}
\end{equation}
where $b_j \in E(S_n)$ is the effect such that $\< s_k, b_j \> = \delta_{jk}$ for all $k \in \{1, \ldots, n\}$. Denote
\begin{equation}
\begin{quantikz}[row sep={\the\originsep,between origins}, align equals at=2]
&\gate[3,nwires={1,3}]{\iI} &\meterD{b_j}{S_n} \\
& & \\
& & \qw{K}
\end{quantikz}
=
\begin{quantikz}[align equals at=1]
&\gate{\iI_j} & \qw
\end{quantikz}
\end{equation}
then $\iI_j: K \to A(K)^{*+}$ is a map that maps state from $K$ to elements of $A(K)^{*+}$. The maps $\iI_j$ are exactly the maps that assign the post-measurement state to an input state $x$ and the normalization $\< \iI_j(x), 1_{K} \> = \lambda_i$ is exactly the probability of measuring the outcome $j$. Note that
\begin{equation}
\begin{quantikz}[row sep={\the\originsep,between origins}, align equals at=2]
&\gate[3,nwires={1,3}]{\iI} &\ground{}{S_n} \\
& & \\
& & \qw{K}
\end{quantikz}
=
\begin{quantikz}[align equals at=1]
&\gate{\Phi} & \qw
\end{quantikz}
\end{equation}
where $\Phi \in \chan(K)$ is a channel and we clearly have $\Phi = \sum_{j=1}^n \iI_j$. Moreover
\begin{equation}
\begin{quantikz}[row sep={\the\originsep,between origins}, align equals at=2]
&\gate[3,nwires={1,3}]{\iI} &\qw{S_n} \\
& & \\
& & \ground{}{K}
\end{quantikz}
=
\begin{quantikz}[align equals at=1]
&\gate{m} & \qw
\end{quantikz}
\end{equation}
where $m \in \chan(K, S_n)$ is a measurement and for $x \in K$ we have $m(x) = \sum_{j=1}^n \< \iI_j(x), 1_{K} \> s_j$. So it follows that we can express $\iI$ as
\begin{equation} \label{eq:channels-instruments-instrumentForm}
\iI = \sum_{j=1}^n s_j \otimes \iI_j
\end{equation}
and for $x \in K$ we have $\iI(x) = \sum_{j=1}^n s_j \otimes \iI_j(x)$. Thus we have proved the following:
\begin{proposition} \label{prop:channels-instruments-instrumentForm}
Every instrument $\iI: K \to S_n \otimes K$ is of the form given by \eqref{eq:channels-instruments-instrumentForm}, i.e., there are affine maps $\iI_j: K \to \cone(K)$ for $j \in \{1, \ldots, n \}$ such that $\iI = \sum_{j=1}^n s_j \otimes \iI_j$.
\end{proposition}
Another way to prove Proposition \ref{prop:channels-instruments-instrumentForm} would be to use Propositions \ref{prop:channels-channels-tensorSubset} and \ref{prop:tensor-existence-notWithSn}. Note that some authors call the maps $\iI_j$ instruments and instead of working with $\iI$ they define a collection of instruments $\iI_1, \ldots, \iI_n$, such that $\sum_{j=1}^n \iI_j$ is a channel. These two approaches are equivalent.

\subsection{Preparations, measure-and-prepare channels}
So far, we have identified states and preparation procedures, but one can actually describe preparation procedures as channels $\Pe \in \chan(S_1, K)$. Note that $S_1 = \{ s \}$ and so the channel acts as $\Pe(s) = x$ for some $x \in K$. It immediately follows that that the set of all preparations $\chan(S_1, K)$ is isomorphic to $K$.

Extending the idea of preparations, one can define a conditional preparation as channel $\Pe \in \chan(S_n, K)$. A conditional preparation $\Pe \in \chan(S_n, K)$ is essentially a device that can prepare $n$ different states and the classical input determines, which of the $n$ states will be prepared. Analogically to the result of Proposition \ref{prop:channels-measurements-effects}, one can easily prove that for every conditional preparation $\Pe \in \chan(S_n, K)$, there are states $x_1, \ldots, x_n \in K$ such that for $s \in S_n$ we have
\begin{equation} \label{eq:channels-MnP-prepForm}
\Pe(s) = \sum_{i=1}^n \< s, b_i \> x_i.
\end{equation}
Preparations and conditional preparations are thus dual to effects and measurements, but they are often not discussed because measurements are more often used in practical applications.

Let $\Pe \in \chan(S_n, K_B)$ be a conditional preparation and let $m \in \chan(K_A, S_n)$ be a measurement, then we can construct
\begin{equation} \label{eq:channels-MnP-MnPchannel}
\begin{quantikz}[align equals at=1]
& \gate{m} \qw{K_A} & \gate{\Pe} \qw{S_n} & \qw{K_B}
\end{quantikz}
=
\begin{quantikz}[align equals at=1]
& \gate{\Phi_{\MP}} \qw{K_A} & \qw{K_B}
\end{quantikz}
\end{equation}
where $\Phi_{\MP} \in \chan(K_A, K_B)$ is a channel.
\begin{definition} \label{def:channels-MnP}
Let $\Phi_{\MP} \in \chan(K_A, K_B)$ be channel of the form given by \eqref{eq:channels-MnP-MnPchannel}, i.e., such that there are $m \in \chan(K_A, S_n)$ and $\Pe \in \chan(S_n, K_B)$ such that $\Phi_{\MP} = \Pe \circ m$, then we call $\Phi_{\MP}$ \emph{measure-and-prepare channel}.
\end{definition}
We will later see that not all channels are measure-and-prepare. One can easily prove the following structural result for measure-and-prepare channels:
\begin{proposition}
Let $\Phi_{\MP}$ be measure-and-prepare channel, then there are effects $f_1, \ldots, f_n \in E(K)$ and states $x_1, \ldots, x_n \in K$ such that for any $y \in K$ we have $\Phi_{\MP}(y) = \sum_{i=1}^n \< y, f_i \> x_i$.
\end{proposition}
\begin{proof}
Let $\Phi_{\MP}$ be given as in \eqref{eq:channels-MnP-MnPchannel}. According to Proposition \ref{prop:channels-measurements-effects} we have  $m(y) = \sum_{i=1}^n \< y, f_i \> s_i$ and according to \eqref{eq:channels-MnP-prepForm} we have $\Pe(s) = \sum_{i=1}^n \< s, b_i \> x_i$. Then for any $y \in K$ we have
\begin{equation}
\Phi_{\MP}(y) = \Pe \left( \sum_{i=1}^n \< y, f_i \> s_i \right) = \sum_{i=1}^n \< y, f_i \> x_i.
\end{equation}
\end{proof}

\begin{corollary}
Any measurement $m \in \chan(K, S_n)$ is a measure-and-prepare channel.
\end{corollary}
\begin{proof}
The only trick needed to prove the result is to take $\Pe \in \chan(S_n, S_n)$ to be the conditional preparation given as $\Pe(s_i) = s_i$ for all $i \in \{1, \ldots, n\}$, in other words $\Pe = \id_{S_n}$.
\end{proof}

\begin{corollary}
Let $\tau_x \in \chan(K_A, K_B)$ be a constant channel as in Example \ref{exm:channels-channels-constant}, i.e., for all $y_A \in K_A$ we have $\tau_x(y_A) = x_B$. Then $\tau_B$ is measure-and-prepare.
\end{corollary}
\begin{proof}
We already know that we have $\tau_x(y_A) = \< y_A, 1_{K_A} \> x_B$ which just implies that $\tau_x = \Pe_1 \circ m_1$, where $m_1 \in \chan(K_A, S_1)$ is the trivial, single-outcome measurement given as $m_1(y_A) = \< y_A, 1_{K_A} \> s$, where $S_1 = \{s\}$, and $\Pe_1 \in \chan(S_1, K_B)$ is the preparation $\Pe_1(s) = x_B$.
\end{proof}

\subsection{Completely-positive channels}
So far we have only considered channels acting on a single system, but in principle we can also apply channels to parts of larger systems. For example let $x_{AB} \in K_A \treal K_B$ and let $\Phi \in \chan(K_B, K_C)$, then we should be allowed to construct
\begin{equation} \label{eq:channels-CP-oneLeg}
\begin{quantikz}[row sep=\the\rowsep, align equals at=1.5]
&\multiprepareC[2]{x_{AB}} & \qw & \qw{K_A} \\
& & \gate{\Phi} \qw{K_B} & \qw{K_C}
\end{quantikz}
=
(\id_{K_A} \otimes \Phi)(x_{AB}).
\end{equation}
It is rather simple to define the object in \eqref{eq:channels-CP-oneLeg} mathematically: we know that there are $y_{i,A} \in A(K_A)^{*+}$, $y_{i,B} \in A(K_B)^{*+}$ and $\alpha_i \in \RR$, $i \in \{1, \ldots, n\}$ such that $x_{AB} = \sum_{i=1}^n \alpha_i y_{i,A} \otimes y_{i,B}$. Then, by linearity, we get
\begin{equation}
(\id_{K_A} \otimes \Phi)(x_{AB}) = \sum_{i=1}^n \alpha_i y_{i,A} \otimes \Phi(y_{i,B}).
\end{equation}
We should require that $(\id_{K_A} \otimes \Phi)(x_{AB}) \in K_A \treal K_C$, which is a new requirement that we have not taken into account so far.
\begin{definition}
Let $\Phi \in \chan(K_B, K_C)$ be a channel, then we say that $\Phi$ is \emph{completely positive} (or CP for short) with respect to $K_A \treal K_B \to K_A \treal K_C$, if for all $x_{AB} \in K_A \treal K_B$ we have $(\id_{K_A} \otimes \Phi)(x_{AB}) \in K_A \treal K_C$, i.e., if $\id_{K_A} \otimes \Phi \in \chan(K_A \treal K_B, K_A \treal K_C)$.
\end{definition}
It is known that not all channels are completely positive, a well-known example is the partial transposition map in quantum theory. Also note that a channel $\Phi: K_B \to K_C$ is sometimes called positive map, because it preserves the positivity of elements of $A(K_B)^{*+}$. One has to be careful to not mistake positivity and complete positivity as these are two different notions.

One has to be careful with respect to what choice of tensor products $K_A \treal K_B$ and $K_A \treal K_C$ is the complete positivity defined. In quantum theory and in general in theories where the tensor product $\treal$ is fixed, we usually implicitly assume that complete positivity is defined with respect to the chosen tensor product $\treal$. But since in general there is no unique choice of $\treal$, we have to always specify with respect to what choice of tensor product we are defining complete positivity. We will now prove several results about positivity and complete positivity of channels.

\begin{proposition}
The identity channel $\id_{K_B} \in \chan(K_B)$ is completely positive with respect to any choice of tensor product $\treal$, i.e., with respect to any $K_A \treal K_B \to K_A \treal K_B$.
\end{proposition}
\begin{proof}
Let $x_{AB} \in K_A \treal K_B$, then we have $(\id_{K_A} \otimes \id_{K_B}) \in \chan(K_A \treal K_B)$ and $(\id_{K_A} \otimes \id_{K_B})(x_{AB}) = x_{AB}$, from which the result immediately follows.
\end{proof}

\begin{proposition} \label{prop:channels-CP-measurements}
A measurement $m \in \chan(K_B, S_n)$ is completely positive with respect to any $K_A \treal K_B \to K_A \treal K_C$.
\end{proposition}
\begin{proof}
Let $m$ be given as $m = \sum_{i=1}^n f_i \otimes s_i$. Let  $x_{AB} \in K_A \treal K_B$, then
\begin{equation}
(\id_{K_A} \otimes m)(x_{AB}) = \sum_{i=1}^n (\id_{K_A} \otimes f_i)(x_{AB}) \otimes s_i
\end{equation}
and according to Proposition \ref{prop:tensor-partial-multipleOfState} we get
\begin{equation}
(\id_{K_A} \otimes m)(x_{AB}) = \sum_{i=1}^n \< x_{AB}, 1_{K_A} \otimes f_i \> y_i \otimes s_i \in K_A \tmin S_n
\end{equation}
where $y_i \in K_A$, $i \in \{1, \ldots, n\}$.
\end{proof}

\begin{corollary}
All measure-and-prepare channels $\Phi_{\MP} \in \chan(K_B, K_C)$ are completely positive with respect to any $K_A \treal K_B \to K_A \treal K_C$.
\end{corollary}
\begin{proof}
Since by definition $\Phi_{\MP} = \Pe \circ m$, where $m \in \chan(K_B, S_n)$ is a measurement and $\Pe \in \chan(S_n, K_C)$ is a conditional preparation, the result follows from Proposition \ref{prop:channels-CP-measurements}.
\end{proof}
Complete positivity of measurements and measure-and-prepare channels is not surprising, since measurements and measure-and-prepare channels are entanglement-breaking channels.
\begin{definition} \label{def:channels-CP-ENTBreaking}
Channel $\Phi \in \chan(K_B, K_C)$ is \emph{entanglement-breaking} with respect to $K_A \treal K_B \to K_A \treal K_C$ if for any $K_A$ and $x_{AB} \in K_A \treal K_B$ we have
\begin{equation}
\begin{quantikz}[row sep=\the\rowsep, align equals at=1.5]
&\multiprepareC[2]{x_{AB}} & \qw & \qw \\
& & \gate{\Phi} & \qw
\end{quantikz}
= (\id_{K_A} \otimes \Phi)(x_{AB}) \in K_A \tmin K_C,
\end{equation}
i.e., $(\id_{K_A} \otimes \Phi)(x_{AB})$ is always a separable state.
\end{definition}

\begin{lemma} \label{lemma:channels-CP-ENTBreakingCP}
Let $\Phi \in \chan(K_B, K_C)$ be entanglement-breaking channel with respect to $K_A \treal K_B \to K_A \treal K_C$. Then $\Phi$ is completely positive with respect to $K_A \treal K_B \to K_A \treal K_C$.
\end{lemma}
\begin{proof}
The result follows from $K_A \tmin K_C \subset K_A \treal K_C$.
\end{proof}

So far we have shown that some channels are completely positive for any tensor product $K_A \treal K_B$, now we will investigate complete positivity with respect to the minimal and maximal tensor products. These results will showcase that the choice of the tensor product $K_A \treal K_B$ affects complete positivity of channels.
\begin{proposition} \label{prop:channels-CP-minTensor}
Let $\Phi \in \chan(K_B, K_C)$, then $\Phi$ is completely positive with respect to $K_A \tmin K_B \to K_A \tmin K_C$.
\end{proposition}
\begin{proof}
The proof si simple. Let $y_{AB} \in K_A \tmin K_B$, then there are $\lambda_i \in [0,1]$, $x_{i,A} \in K_A$, $x_{i,B} \in K_B$ for $i \in \{1, \ldots, n\}$ and $\sum_{i=1}^n \lambda_i = 1$ such that $y_{AB} = \sum_{i=1}^n \lambda_i x_{i, A} \otimes x_{i, B}$ and we have
\begin{equation}
(\id_{K_A} \otimes \Phi)(y_{AB}) = \sum_{i=1}^n \lambda_i x_{i, A} \otimes \Phi(x_{i, B}) \in K_A \tmin K_C.
\end{equation}
\end{proof}

\begin{proposition} \label{prop:channels-CP-maxTensor}
Let $\Phi \in \chan(K_B, K_C)$, then $\Phi$ is completely positive with respect to $K_A \tmax K_B \to K_A \tmax K_C$.
\end{proposition}
\begin{proof}
Let $x_{AB} \in K_A \tmax K_B$, let $f_A \in E(K_A)$ and $f_C \in E(K_C)$, then we have
\begin{equation}
\begin{quantikz}[row sep=\the\rowsep,align equals at=1.5]
\multiprepareC[2]{x_{AB}} & \qw & \meterD{f_A} \\
& \gate{\Phi} & \meterD{f_C}
\end{quantikz}
=
\begin{quantikz}[row sep=\the\rowsep,align equals at=1.5]
\multiprepareC[2]{x_{AB}} & \meterD{f_A} \\
& \meterD{\Phi^*(f_C)}
\end{quantikz}
\geq 0,
\end{equation}
because $\Phi^*: E(K_C) \to E(K_B)$. It then follows that $(\id_{K_A} \otimes \Phi)(x_{AB}) \in K_A \tmax K_C$.
\end{proof}
The last result may look strange at first, as one would expect that the more entangled states there are, the bigger the difference between positive and completely positive maps will be. But this is not the case and the underlying reason is that a channel $\Phi \in \chan(K_B, K_C)$ is completely positive with respect to $K_A \treal K_B \to K_A \treal K_C$ if and only if the adjoint map $\Phi^*: E(K_C) \to E(K_B)$ is completely positive with respect to $E(K_A \treal K_C) \to E(K_A \treal K_B)$. In the case of complete positivity of $\Phi \in \chan(K_B, K_C)$ with respect to $K_A \tmax K_B \to K_A \tmax K_C$, the adjoint map $\Phi^*$ is needed to be completely positive with respect to $E(K_A) \tmin E(K_C) \to E(K_A) \tmin E(K_B)$, which is easy to show analogically to the proof of Proposition \ref{prop:channels-CP-minTensor}.

\subsection{Post-processing preorder of channels}
In this section we are going to introduce the post-processing preorder of quantum channels. The main idea is simple: let $K_A, K_B, K_C$ be state spaces and let $\Phi \in \chan(K_A, K_B)$, $\Psi \in \chan(K_A, K_C)$, $\Lambda \in \chan(K_B, K_C)$ be channel. Let
\begin{equation} \label{eq:channels-postProc-preorder}
\begin{quantikz}[align equals at=1]
&\gate{\Psi} &\qw{}
\end{quantikz}
=
\begin{quantikz}[align equals at=1]
&\gate{\Phi} &\gate{\Lambda} &\qw{}
\end{quantikz}
\end{equation}
Then we can say that $\Phi$ is a better channel then $\Psi$, because we can always obtain $\Psi$ from $\Phi$. We will formalize this in the following definition.
\begin{definition} \label{def:channels-postProc}
Let $K_A, K_B, K_C$ be state spaces and let $\Phi \in \chan(K_A, K_B)$, $\Psi \in \chan(K_A, K_C)$ be channels. The we say that $\Psi$ is a \emph{post-processing} of $\Phi$ and we write
\begin{equation}
\Psi \prec \Phi
\end{equation}
if there is a channel $\Lambda \in \chan(K_B, K_C)$ such that \eqref{eq:channels-postProc-preorder} holds.
\end{definition}

If we restrict only to channels $\Phi \in \chan(K)$, then the post-processing relation gives rise to a preorder.
\begin{proposition} \label{prop:channels-postProc-preorder}
The post-processing relation $\prec$ is an preorder on the set of channel mapping $K$ to $K$, i.e., $\chan(K)$.
\end{proposition}
\begin{proof}
We need to prove that $\prec$ is reflexive, i.e., that $\Phi \prec \Phi$ and transitive, i.e. that $\Phi_1 \prec \Phi_2$ and $\Phi_2 \prec \Phi_3$ implies $\Phi_1 \prec \Phi_3$. Let $\Phi \in \chan(K)$, then clearly $\Phi \prec \Phi$ as we have $\Phi = \Phi \circ \id_K$ and so $\prec$ is reflexive. To show that $\prec$ is transitive, let $\Phi_1, \Phi_2, \Phi_3 \in \chan(K)$ be such that $\Phi_1 \prec \Phi_2$ and $\Phi_2 \prec \Phi_3$. Then there are $\Lambda_{1}, \Lambda_{2} \in \chan(K)$ such that
\begin{equation}
\begin{quantikz}[align equals at=1]
&\gate{\Phi_1} &\qw{}
\end{quantikz}
=
\begin{quantikz}[align equals at=1]
&\gate{\Phi_2} &\gate{\Lambda_{1}} &\qw{}
\end{quantikz}
\end{equation}
and
\begin{equation}
\begin{quantikz}[align equals at=1]
&\gate{\Phi_2} &\qw{}
\end{quantikz}
=
\begin{quantikz}[align equals at=1]
&\gate{\Phi_3} &\gate{\Lambda_{2}} &\qw{}
\end{quantikz}
\end{equation}
We get
\begin{equation}
\begin{quantikz}[align equals at=1]
&\gate{\Phi_1} &\qw{}
\end{quantikz}
=
\begin{quantikz}[align equals at=1]
&\gate{\Phi_2} &\gate{\Lambda_{1}} &\qw{}
\end{quantikz}
=
\begin{quantikz}[align equals at=1]
&\gate{\Phi_3} &\gate{\Lambda_{2}} &\gate{\Lambda_{1}} &\qw{}
\end{quantikz}
\end{equation}
and $\Phi_1 \prec \Phi_3$ follows.
\end{proof}

One of the important results of the post-processing preorder is that it showcases an important difference between channels and measurements: while a post-processing greatest channel exists for every state space $K$, post-processing greatest measurement exists only if $K$ is a simplex. A post-processing greatest channel $\Phi \in \chan(K)$ is a channel such that if $\Phi \prec \Psi$ for some $\Psi \in \chan(K)$, then also $\Psi \prec \Phi$. The identity channel $\id_K \in \chan(K)$ is post-processing greatest and the proof is immediate. We will postpone the proof that post-processing greatest measurement exist if and only if $K$ is a simplex to Section \ref{sec:compatibility}.
\begin{proposition} \label{prop:channels-postProc-id}
Let $\Phi \in \chan(K_A, K_B)$, then $\Phi \prec \id_{K_A}$, and so $\id_{K_A} \in \chan(K_A)$ is a post-processing greatest channel in $\chan(K_A, K_B)$ for any state space $K_B$.
\end{proposition}
\begin{proof}
We have $\Phi = \id_{K_A} \circ \Phi$ and so the result follows.
\end{proof}

\section{Compatibility of channels} \label{sec:compatibility}

%
%
Several notable non-classical features of quantum theory are connected to the non-commutativity of operators. Compatibility is one of the possible operational generalizations of non-commutativity of operators, and it is the generalization that most frequently appears in other applications of quantum information theory, see \cite{HeinosaariMiyaderaZiman-compatibility} for a review. Compatibility is usually only introduced for measurements, but we are going to introduce compatibility of channels. As before, we will easily recover results about compatibility of measurements as a special cases of the results for channels.

Consider the following scenario: let Alice and Bob be two parties with state spaces $K_A$ and $K_B$ and imagine that Alice wants to message to Bob. Alice encodes the message into a state $x_A \in K_A$ and she uses a fixed channel $\Phi \in \chan(K_A, K_B)$ to send the message to Bob, Bob would receive $\Phi(x_A)$ and then proceed to decode the message. Our task is to intercept the message and to learn about it as much as possible. One thing that we can do is to replace the channel $\Phi \in \chan(K_A, K_B)$ with a different channel $\Psi \in \chan(K_A, K_B \treal K_C)$, where $K_C$ is a system we control. The idea is that the channel $\Psi$ is meant to extract as much information as we can get while keeping the state that Bob receives practically unchanged. So we have to require that
\begin{equation} \label{eq:compatibility-intro-Bob}
\begin{quantikz}[row sep={\the\originsep,between origins}, align equals at=2]
& \gate[3, nwires={1,3}]{\Psi} & \qw{K_B} \\
&\qw{K_A} & \\
& & \ground{}{K_C}
\end{quantikz}
=
\begin{quantikz}[align equals at = 1]
&\gate[1]{\Phi} \qw{K_A} & \qw{K_B}
\end{quantikz}
\end{equation}
so that Bob receives the intended message. The channel $\Psi_C \in \chan(K_A, K_C)$, given as
\begin{equation}
\begin{quantikz}[row sep={\the\originsep,between origins}, align equals at=2]
& \gate[3, nwires={1,3}]{\Psi} & \ground{}{K_B} \\
&\qw{K_A} & \\
& & \qw{K_C}
\end{quantikz}
=
\begin{quantikz}[align equals at = 1]
&\gate[1]{\Psi_C} \qw{K_A} & \qw{K_C}
\end{quantikz}
\end{equation}
is the information about the encoded message that we are able to extract. In principle, we would want to choose the channel $\Psi_C \in \chan(K_A, K_C)$ to give us as much information as possible about the input state, but this does not have to be always possible because of the condition \eqref{eq:compatibility-intro-Bob}. As we will see, there is a certain trade-off between how much information about the input state is encoded in the output of $\Phi \in \chan(K_A, K_B)$ and how much we can extract using $\Psi_C \in \chan(K_A, K_C)$. Notice that we are (in some intuitive sense) attempting to get the outcome of both $\Phi \in \chan(K_A, K_B)$ and $\Psi_C \in \chan(K_A, K_C)$ at the same time using the bigger channel $\Psi \in \chan(K_A, K_B \treal K_C)$.
\begin{definition} \label{def:compatibility-channels-def}
Let $K_A$, $K_B$, $K_C$ be state spaces and let $\Phi_1 \in \chan(K_A, K_B)$, $\Phi_2 \in \chan(K_A, K_C)$ be channels. We say that $\Phi_1$ and $\Phi_2$ are \emph{compatible} if and only if there is a channel $\Phi \in \chan(K_A, K_B \treal K_C)$ such that
\begin{equation} \label{eq:compatibility-def-1}
\begin{quantikz}[row sep={\the\originsep,between origins}, align equals at=2]
& \gate[3, nwires={1,3}]{\Phi} & \qw{K_B} \\
&\qw{K_A} & \\
& & \ground{}{K_C}
\end{quantikz}
=
\begin{quantikz}[align equals at = 1]
&\gate[1]{\Phi_1} \qw{K_A} & \qw{K_B}
\end{quantikz}
\end{equation}
and
\begin{equation} \label{eq:compatibility-def-2}
\begin{quantikz}[row sep={\the\originsep,between origins}, align equals at=2]
& \gate[3, nwires={1,3}]{\Phi} & \ground{}{K_B} \\
&\qw{K_A} & \\
& & \qw{K_C}
\end{quantikz}
=
\begin{quantikz}[align equals at = 1]
&\gate[1]{\Phi_2} \qw{K_A} & \qw{K_C}
\end{quantikz}
\end{equation}
hold. The channel $\Phi$ is usually called the \emph{joint channel} of $\Phi_1$ and $\Phi_2$, or the compatibilizer.
\end{definition}
The problem of deciding whether two channels are compatible or not can seem complicated at first, but it is not so. One can in principle rewrite it as a problem of conic programming and get a resource theory \cite{UolaKraftShangYuGuhne-conicResourceTheories} of compatibility of channels. The underlying conic programming problems are not easily solvable in general, but in the case of quantum theory they are equivalent to quantum marginal problems \cite{Plavala-channels,HaapasaloKraftMiklinUola-marginalProblem,GirardPlavalaSikora-jordan}, which are just semi-definite programming problems. The following is an intuitive result saying that constant channels are compatible with every other channel.
\begin{proposition}
Let $K_A$, $K_B$, $K_C$ be state spaces, let $\Phi \in \chan(K_A, K_B)$ be a channel, let $y_C \in K_C$ be a fixed state and let $\tau_{y_C} \in \chan(K_A, K_C)$ be a constant channel, i.e., for every $x_A \in K_A$ we have $\tau_{y_C}(x_A) = y_C$. Then $\Phi$ and $\tau_{y_C}$ are compatible.
\end{proposition}
\begin{proof}
The proof is straightforward: we will construct the joint channel $\Phi \in \chan(K_A, K_B \treal K_C)$. Let $x_A \in K_A$ and let $\Psi(x_A) = \Phi(x_A) \otimes y_C$, or in diagrammatic notation
\begin{equation}
\begin{quantikz}[row sep={\the\originsep,between origins}, align equals at=2]
& &\gate[3, nwires={1,3}]{\Psi} & \qw{K_B} \\
&\prepareC{x_A} & & \\
& & & \qw{K_C}
\end{quantikz}
=
\begin{quantikz}[row sep=\the\rowsep, align equals at=1.5]
&[\prepfix]\prepareC{x_A} &\gate{\Phi} &\qw{K_B} \\
& &\prepareC{y_C} &\qw{K_C}
\end{quantikz}
\end{equation}
It is straightforward to verify that $\Psi$ is a channel and that it is the joint channel of $\Phi$ and $\tau_{y_C}$.
\end{proof}

One special case of compatibility of channels is the self-compatibility of a channel with itself. This is the scenario where we assume $K_B = K_C$ and $\Phi_1 = \Phi_2$.
\begin{definition}
Let $K_A$, $K_B$ be state spaces and let $\Phi \in \chan(K_A, K_B)$. We say that $\Phi$ is \emph{self-compatible} if there is a channel $\Psi \in \chan(K_A, K_B \treal K_B)$ such that
\begin{equation}
\begin{quantikz}[row sep={\the\originsep,between origins}, align equals at=2]
& \gate[3, nwires={1,3}]{\Psi} & \qw{} \\
& & \\
& & \ground{}
\end{quantikz}
=
\begin{quantikz}[align equals at = 1]
&\gate[1]{\Phi} & \qw{}
\end{quantikz}
\end{equation}
and
\begin{equation}
\begin{quantikz}[row sep={\the\originsep,between origins}, align equals at=2]
& \gate[3, nwires={1,3}]{\Psi} & \ground{} \\
& & \\
& & \qw{}
\end{quantikz}
=
\begin{quantikz}[align equals at = 1]
&\gate[1]{\Phi} & \qw{}
\end{quantikz}
\end{equation}
hold.
\end{definition}
It is natural to assume that measurements are self-compatible, because the outcome of a measurement is some classical information about the system and it is intuitive that we can copy classical information. We will prove a stronger version of this result.
\begin{proposition} \label{prop:compatibility-MnP}
Let $\Phi_{\MP} \in \chan(K_A, K_B)$ be a measure-and-prepare channel, see Definition \ref{def:channels-MnP}. Then $\Phi_{\MP}$ is self-compatible.
\end{proposition}
\begin{proof}
Let $\Phi_{\MP} \in \chan(K_A, K_B)$ be a measure-and-prepare channel, then there are a measurement $m \in \chan(K_A, S_n)$ and preparation $\Pe \in \chan(S_n, K_C)$ such that
\begin{equation}
\begin{quantikz}[align equals at=1]
&\gate{\Phi_{\MP}} &\qw{}
\end{quantikz}
=
\begin{quantikz}[align equals at=1]
&\gate{m} &\gate{\Pe} &\qw{}
\end{quantikz}
\end{equation}
Define $\Psi_D \in \chan(S_n, S_n \tmin S_n)$ by $\Psi_D(s_i) = s_i \otimes s_i$. Note that $\Phi_D$ is well-defined, because $s_1, \ldots, s_n$ form a basis of $A(S_n)^*$. For any $s \in S_n$ there are numbers $\lambda_1, \ldots, \lambda_n \in \Rp$, $\sum_{i=1}^n \lambda_i = 1$, such that $s = \sum_{i=1}^n \lambda_i s_i$ and we have $\Phi_D(s) = \sum_{i=1}^n \lambda_i s_i \otimes s_i$. For any $s \in S_n$ we also have
\begin{equation} \label{eq:compatibility-MnP-D1}
\begin{quantikz}[row sep={\the\originsep,between origins}, align equals at=2]
& & \gate[3, nwires={1,3}]{\Psi_D} & \qw{} \\
&\prepareC{s} & & \\
& & & \ground{}
\end{quantikz}
=
\begin{quantikz}[align equals at = 1]
&\prepareC{s} &\qw{}
\end{quantikz}
\end{equation}
and
\begin{equation} \label{eq:compatibility-MnP-D2}
\begin{quantikz}[row sep={\the\originsep,between origins}, align equals at=2]
& & \gate[3, nwires={1,3}]{\Psi_D} & \ground{} \\
&\prepareC{s} & & \\
& & & \qw{}
\end{quantikz}
=
\begin{quantikz}[align equals at = 1]
&\prepareC{s} &\qw{}
\end{quantikz}
\end{equation}
which is straightforward to verify. We are now ready to construct the joint channel $\Psi \in \chan(K_A, K_B \treal K_B)$ as
\begin{equation}
\begin{quantikz}[row sep={\the\originsep,between origins}, align equals at=2]
&\gate[3, nwires={1,3}]{\Psi} &\qw{} \\
& & \\
& &\qw{}
\end{quantikz}
=
\begin{quantikz}[row sep={\the\originsep,between origins}, align equals at=2]
& &\gate[3, nwires={1,3}]{\Psi_D} &\gate{\Pe} &\qw{} \\
&\gate{m} & & \\
& & &\gate{\Pe} &\qw{}
\end{quantikz}
\end{equation}
It is straightforward to check that $\Psi$ is the joint channel using \eqref{eq:compatibility-MnP-D1} and \eqref{eq:compatibility-MnP-D2}.
\end{proof}

For measurements, the definition of compatibility simplifies:
\begin{proposition} \label{prop:compatibility-measuremets}
Let $K$ be a state space and let $m_1 \in \chan(K, S_{n_1})$ and $m_2 \in \chan(K, S_{n_2})$ be measurements given as
\begin{align}
&m_1 = \sum_{i=1}^{n_1} f_i \otimes s_i,
&&m_2 = \sum_{j=1}^{n_2} g_j \otimes s_j,
\end{align}
where $\{f_i\}_{i=1}^{n_1} \subset E(K)$, $\{g_j\}_{j=1}^{n_2} \subset E(K)$, see Proposition \ref{prop:channels-measurements-effects}. Then $m_1$ and $m_2$ are compatible if and only if there are effects $\{ h_{ij} \}_{i,j=1}^{n_1, n_2} \subset E(K)$ such that $\sum_{i=1}^{n_1} \sum_{j=1}^{n_2} h_{ij} = 1_K$
and
\begin{align} \label{eq:compatibility-measurements-sums}
&\sum_{j=1}^{n_2} h_{ij} = f_i,
&&\sum_{i=1}^{n_1} h_{ij} = g_j.
\end{align}
\end{proposition}
\begin{proof}
Assume that $m_1 \in \chan(K, S_{n_1})$ and $m_2 \in \chan(K, S_{n_2})$ are compatible, then the joint channel is $m \in \chan(K, S_{n_1} \tmin S_{n_2})$, given as $m = \sum_{i=1}^{n_1} \sum_{j=1}^{n_2} h_{ij} \otimes s_i \otimes s_j$. Then we have
\begin{equation}
m_1 =
\begin{quantikz}[row sep={\the\originsep,between origins}, align equals at=2]
&\gate[3, nwires={1,3}]{m} &\qw{} \\
& & \\
& &\ground{}
\end{quantikz}
=
\sum_{i=1}^{n_1} \sum_{j=1}^{n_2} h_{ij} \otimes s_i  
\end{equation}
from where we get $\sum_{j=1}^{n_2} h_{ij} = f_i$ for all $i \in \{1, \ldots, n_1\}$. $\sum_{i=1}^{n_1} h_{ij} = g_j$ follows analogically. Now assume that there are effects $\{ h_{ij} \}_{i,j=1}^{n_1, n_2} \subset E(K)$ such that $\sum_{i=1}^{n_1} \sum_{j=1}^{n_2} h_{ij} = 1_K$ and \eqref{eq:compatibility-measurements-sums} hold, then let $m \in \chan(K, S_{n_1} \tmin S_{n_2})$ be given as $m = \sum_{i=1}^{n_1} \sum_{j=1}^{n_2} h_{ij} \otimes s_i \otimes s_j$. It is easy to verify that $m$ is the joint measurement of $m_1$ and $m_2$.
\end{proof}

In the following we will investigate several aspects of incompatibility of channels and measurements, we will present the known results about existence of incompatibility in non-classical theories and we will show how incompatibility interacts with entanglement.

\subsection{No-broadcasting theorem and existence of incompatible measurements}
We are going to investigate compatibility of the identity channel $\id_K \in \chan(K)$. So let $K$ be a state space, then we want to find a channel $\Phi \in \chan(K, K \treal K)$ such that for every $x \in K$ we have
\begin{equation}
\begin{quantikz}[row sep={\the\originsep,between origins}, align equals at=2]
& & \gate[3, nwires={1,3}]{\Phi} & \qw{} \\
&\prepareC{x} & & \\
& & & \ground{}
\end{quantikz}
=
\begin{quantikz}[align equals at = 1]
&[\prepfix] \prepareC{x} &\qw{}
\end{quantikz}
\end{equation}
and
\begin{equation}
\begin{quantikz}[row sep={\the\originsep,between origins}, align equals at=2]
& & \gate[3, nwires={1,3}]{\Phi} & \ground{} \\
&\prepareC{x} & & \\
& & & \qw{}
\end{quantikz}
=
\begin{quantikz}[align equals at = 1]
&[\prepfix] \prepareC{x} &\qw{}
\end{quantikz}
\end{equation}
or, purely in terms of channels,
\begin{equation}
\begin{quantikz}[row sep={\the\originsep,between origins}, align equals at=2]
& \gate[3, nwires={1,3}]{\Phi} & \qw{} \\
& & \\
& & \ground{}
\end{quantikz}
=
\begin{quantikz}[align equals at = 1]
&\qw{}
\end{quantikz}
\end{equation}
and
\begin{equation}
\begin{quantikz}[row sep={\the\originsep,between origins}, align equals at=2]
& \gate[3, nwires={1,3}]{\Phi} & \ground{} \\
& & \\
& & \qw{}
\end{quantikz}
=
\begin{quantikz}[align equals at = 1]
&\qw{}
\end{quantikz}
\end{equation}
where the plain wire represents the identity channel, i.e.,
\begin{equation}
\begin{quantikz}[align equals at = 1]
&\qw{}
\end{quantikz}
=
\begin{quantikz}[align equals at = 1]
&\gate{\id_K} &\qw{}
\end{quantikz}
\end{equation}
Assume that $x$ is pure, then according to Theorem \ref{thm:tensor-partial-monogamyTrace} we must have
\begin{equation} \label{eq:compatibility-noBroadcast-action}
\Phi(x) = x \otimes x.
\end{equation}
It follows that $\Phi \in \chan(K, K \tmin K)$ is the universal cloning channel for pure states, also called the universal broadcasting channel. Note that the universal broadcasting channel $\Phi$ is uniquely specified by \eqref{eq:compatibility-noBroadcast-action} and by the convexity of $K$. It is known that we can copy classical information, i.e., that in classical theory, the universal cloning machine exists.
\begin{proposition} \label{prop:compatibility-noBroadcast-classical}
Let $S_n$ be a simplex, i.e., the state space of classical theory, and let $s_1, \ldots, s_n \in S_n$ be the extreme points of $S_n$. Then the identity channel $\id_{S_n} \in \chan(S_n)$ is self-compatible and the joint channel is $\Phi_D \in \chan(S_n, S_n \tmin S_n)$ given as $\Phi_D(s_i) = s_i \otimes s_i$.
\end{proposition}
\begin{proof}
The result follows from Proposition \ref{prop:compatibility-MnP}, since the identity channel $\id_{S_n} \in \chan(S_n)$ is measure-and-prepare.
\end{proof}

We are now going to formulate a well-known result that for non-classical theories we can not construct the universal broadcasting channel, this result is known as no-broadcasting theorem \cite{BarnumBarrettLeiferWilce-noBroadcasting,BarnumBarrettLeiferWilce-noBroadcastingPRL}. Let $K$ be a non-classical state space, i.e., not a simplex. Let $x_1, \ldots, x_n \in K$ be pure states such that they form basis of $A(K)^*$, such set always exists because the set of all pure states of $K$ must be overcomplete. Let $y \in K$ be pure state, $y \neq x_i$ for all $i \in \{1, \ldots, n\}$, then there are $\alpha_1, \ldots, \alpha_n \in \RR$, $\sum_{i=1}^n \alpha_i = 1$, such that $y = \sum_{i=1}^n \alpha_i x_i$. Since $y$ and $x_1, \ldots, x_n$ are pure states, then according to \eqref{eq:compatibility-noBroadcast-action} we must have
\begin{equation}
\Phi(y) = \sum_{i=1}^n \alpha_i \Phi(x_i) = \sum_{i=1}^n \alpha_i x_i \otimes x_i
\end{equation}
and
\begin{equation}
\Phi(y) = y \otimes y = \sum_{i,j = 1}^n \alpha_i \alpha_j x_i \otimes x_j.
\end{equation}
Comparing the two terms we get
\begin{equation}
\sum_{i, j = 1}^n (\alpha_i \delta_{ij} - \alpha_i \alpha_j) x_i \otimes x_j = 0.
\end{equation}
Since $\{x_i \otimes x_j\}_{i,j=1}^n$ is a basis of $A(K)^* \otimes A(K)^*$ we get $\alpha_i \delta_{ij} - \alpha_i \alpha_j = 0$ for all $i,j \in \{1, \ldots, n\}$. For $i=j$ we get $\alpha_i = \alpha_i^2$ and so $\alpha_i \in \{0, 1 \}$ for all $i \in \{1, \ldots, n\}$. For $i \neq j$, we must have $\alpha_i \alpha_j = 0$ and so either $\alpha_i = 0$ or $\alpha_j = 0$. It follows that only one of the numbers $\alpha_i$ can be non-zero, without loss of generality we argue that we must have $\alpha_1 \neq 0$ and $\alpha_j = 0$ for all $j \in \{2, \ldots, n\}$. We then have $y = \sum_{i=1}^n \alpha_i x_i = x_1$ which is a contradiction with $y \neq x_i$ for all $i \in \{1, \ldots, n\}$. Thus we have proved:
\begin{theorem} \label{thm:compatibility-noBroadcasting-simplex}
Let $K$ be a state space and let $\id_K \in \chan(K)$ be the identity channel. The following statements are all equivalent:
\begin{enumerate}[label=(NB\arabic*), leftmargin=*]
\item $\id_K$ is a measure-and-prepare channel;
\item $\id_K$ is self-compatible channel;
\item there exists universal broadcasting channel $\Phi \in \chan(K, K \treal K)$;
\item $K = S_n$ is a simplex.
\end{enumerate}
\end{theorem}
\begin{proof}
See above, or \cite{BarnumBarrettLeiferWilce-noBroadcasting,BarnumBarrettLeiferWilce-noBroadcastingPRL}.
\end{proof}
This is another characterization of classical state spaces, the first one we presented was in terms of existence of entanglement, see Theorem \ref{thm:tensor-existence-nonClassical}. There are two immediate corollaries of Theorem \ref{thm:compatibility-noBroadcasting-simplex}.
\begin{corollary}
Let $K$ be a non-classical state-space, then there exist pair of incompatible channels $\Phi_1, \Phi_2 \in \chan(K)$.
\end{corollary}
\begin{proof}
Take $\Phi_1 = \Phi_2 = \id_K$.
\end{proof}

\begin{corollary}
Let $K$ be a non-classical state-space, then not all channels in $\chan(K)$ are measure-and-prepare.
\end{corollary}
\begin{proof}
If $\id_K$ was measure-and-prepare, then according to Proposition \ref{prop:compatibility-MnP} it would be self-compatible.
\end{proof}

One can also investigate whether we can broadcast at least some convex subset $B \subset K$, i.e., whether there is a channel $\Psi \in \chan(K, K \treal K)$ such that for every $x \in B$ we have
\begin{equation} \label{eq:comaptiblity-noBroadcast-BsubK-1}
\begin{quantikz}[row sep={\the\originsep,between origins}, align equals at=2]
& & \gate[3, nwires={1,3}]{\Psi} & \qw{} \\
&\prepareC{x} & & \\
& & & \ground{}
\end{quantikz}
=
\begin{quantikz}[align equals at = 1]
&[\prepfix] \prepareC{x} &\qw{}
\end{quantikz}
\end{equation}
and
\begin{equation} \label{eq:comaptiblity-noBroadcast-BsubK-2}
\begin{quantikz}[row sep={\the\originsep,between origins}, align equals at=2]
& & \gate[3, nwires={1,3}]{\Psi} & \ground{} \\
&\prepareC{x} & & \\
& & & \qw{}
\end{quantikz}
=
\begin{quantikz}[align equals at = 1]
&[\prepfix] \prepareC{x} &\qw{}
\end{quantikz}
\end{equation}

The following is a reformulation of the main result of \cite{BarnumBarrettLeiferWilce-noBroadcasting,BarnumBarrettLeiferWilce-noBroadcastingPRL}.
\begin{theorem}
Let $B \subset K$, then there exists channel $\Psi \in \chan(K, K \treal K)$ satisfying \eqref{eq:comaptiblity-noBroadcast-BsubK-1} and \eqref{eq:comaptiblity-noBroadcast-BsubK-2} if and only if there is a measure-and-prepare channel $\Phi_{\MP} \in \chan(K)$ such that for all $x \in B$ we have $\Phi_{MP}(x) = x$.
\end{theorem}
\begin{proof}
We are only going to show that if such measure-and-prepare channel $\Phi_{\MP}$ exists, then there also exists channel $\Psi$ satisfying \eqref{eq:comaptiblity-noBroadcast-BsubK-1} and \eqref{eq:comaptiblity-noBroadcast-BsubK-2}. For the proof of the other implication see \cite{BarnumBarrettLeiferWilce-noBroadcasting,BarnumBarrettLeiferWilce-noBroadcastingPRL}. So let $\Phi_{\MP} \in \chan(K)$ be a measure-and-prepare channel such that for every $x \in B$ we have $\Phi_{\MP}(x) = x$. Then according to Proposition \ref{prop:compatibility-MnP} $\Phi_{\MP}$ is self-compatible, so there exists a channel $\Psi \in \chan(K, K \treal K)$ such that \eqref{eq:compatibility-def-1} and \eqref{eq:compatibility-def-2} are satisfied. Let $x \in B$, we then have
\begin{equation}
\begin{quantikz}[row sep={\the\originsep,between origins}, align equals at=2]
& & \gate[3, nwires={1,3}]{\Psi} & \qw{} \\
&\prepareC{x} & & \\
& & & \ground{}
\end{quantikz}
=
\begin{quantikz}[align equals at = 1]
&[\prepfix] \prepareC{x} &\gate{\Phi_{\MP}} &\qw{}
\end{quantikz}
=
\begin{quantikz}[align equals at = 1]
&[\prepfix] \prepareC{x} &\qw{}
\end{quantikz}
\end{equation}
and
\begin{equation}
\begin{quantikz}[row sep={\the\originsep,between origins}, align equals at=2]
& & \gate[3, nwires={1,3}]{\Psi} & \ground{} \\
&\prepareC{x} & & \\
& & & \qw{}
\end{quantikz}
=
\begin{quantikz}[align equals at = 1]
&[\prepfix] \prepareC{x} &\gate{\Phi_{\MP}} &\qw{}
\end{quantikz}
=
\begin{quantikz}[align equals at = 1]
&[\prepfix] \prepareC{x} &\qw{}
\end{quantikz}
\end{equation}
so \eqref{eq:comaptiblity-noBroadcast-BsubK-1} and \eqref{eq:comaptiblity-noBroadcast-BsubK-2} are satisfied.
\end{proof}

Another way to extend the results we have obtained so far is to restrict the set of channels we investigate: instead of asking whether there exist some pair of incompatible channels $\Phi_1 \in \chan(K_A, K_B)$ and $\Phi_2 \in \chan(K_A, K_C)$, we can ask whether there exists a pair of incompatible two-outcome measurements $m_1 \in \chan(K, S_2)$ and $m_2 \in \chan(K, S_2)$. It was shown in \cite{Plavala-simplex,Kuramochi-simplex} that such pair of incompatible two-outcome measurements exists whenever $K$ is not a simplex.
\begin{theorem} \label{thm:compatibility-noBroadcasting-measurementSimplex}
There exists a pair of incompatible two-outcome measurements $m_1 \in \chan(K, S_2)$ and $m_2 \in \chan(K, S_2)$ whenever $K$ is not a simplex.
\end{theorem}
\begin{proof}
See \cite{Plavala-simplex} for a constructive proof.
\end{proof}
One can also show that existence of incompatible measurements is related to existence of entanglement between appropriate cones, see \cite{NamiokaPhelps-cones}.

\subsection{Preorder of channels and compatibility}
In this section we will explore the connection between the post-processing preorder and compatibility of channels; most of the results are inspired by \cite{HeinosaariMiyaderaZiman-compatibility}. We will also prove that post-processing greatest measurement exists only if $K$ is a simplex. The first result is immediate.
\begin{proposition} \label{prop:compatibility-postProc-compatibility}
Let $K_A, K_B, K_C, K_D$ be state spaces and let $\Phi_1 \in \chan(K_A, K_B)$, $\Phi_2 \in \chan(K_A, K_C)$ and $\Phi_3 \in \chan(K_A, K_D)$ be channels such that $\Phi_3 \prec \Phi_2$. If $\Phi_1$ and $\Phi_2$ are compatible, then also $\Phi_1$ and $\Phi_3$ are compatible.
\end{proposition}
\begin{proof}
Since $\Phi_1$ and $\Phi_2$ are compatible, there is a channel $\Psi \in \chan(K_A, K_B \treal K_C)$ such that
\begin{equation}
\begin{quantikz}[row sep={\the\originsep,between origins}, align equals at=2]
& \gate[3, nwires={1,3}]{\Psi} & \qw{} \\
& & \\
& & \ground{}
\end{quantikz}
=
\begin{quantikz}[align equals at = 1]
&\gate{\Phi_1} &\qw{}
\end{quantikz}
\end{equation}
and
\begin{equation}
\begin{quantikz}[row sep={\the\originsep,between origins}, align equals at=2]
& \gate[3, nwires={1,3}]{\Psi} & \ground{} \\
& & \\
& & \qw{}
\end{quantikz}
=
\begin{quantikz}[align equals at = 1]
&\gate{\Phi_2} &\qw{}
\end{quantikz}
\end{equation}
Since $\Phi_3 \prec \Phi_2$ there is a channel $\Lambda \in \chan(K_C, K_D)$ such that
\begin{equation}
\begin{quantikz}[align equals at=1]
&\gate{\Phi_3} &\qw{}
\end{quantikz}
=
\begin{quantikz}[align equals at=1]
&\gate{\Phi_2} &\gate{\Lambda} &\qw{}
\end{quantikz}
\end{equation}
We then have
\begin{equation}
\begin{quantikz}[row sep={\the\originsep,between origins}, align equals at=2]
& \gate[3, nwires={1,3}]{\Psi} &\qw{} &\qw{} \\
& & & \\
& &\gate{\Lambda} &\ground{}
\end{quantikz}
=
\begin{quantikz}[row sep={\the\originsep,between origins}, align equals at=2]
& \gate[3, nwires={1,3}]{\Psi} & \qw{} \\
& & \\
& & \ground{}
\end{quantikz}
=
\begin{quantikz}[align equals at = 1]
&\gate{\Phi_1} &\qw{}
\end{quantikz}
\end{equation}
and
\begin{equation}
\begin{quantikz}[row sep={\the\originsep,between origins}, align equals at=2]
& \gate[3, nwires={1,3}]{\Psi} &\qw{} &\ground{} \\
& & & \\
& &\gate{\Lambda} &\qw{}
\end{quantikz}
=
\begin{quantikz}[row sep={\the\originsep,between origins}, align equals at=2]
& \gate[3, nwires={1,3}]{\Psi} &\ground{} & \\
& & & \\
& &\qw{} &\gate{\Lambda} &\qw{}
\end{quantikz}
=
\begin{quantikz}[align equals at = 1]
&\gate{\Phi_2} &\gate{\Lambda} &\qw{}
\end{quantikz}
=
\begin{quantikz}[align equals at=1]
&\gate{\Phi_3} &\qw{}
\end{quantikz}
\end{equation}
and so $\Phi_1$ and $\Phi_3$ are also compatible.
\end{proof}

\begin{corollary} \label{coro:compatibility-postProc-id}
If a channel $\Phi_1 \in \chan(K_A, K_B)$ is compatible with $\id_{K_A} \in \chan(K_A)$, then it is also compatible with any other channel $\Phi_2 \in \chan(K_A, K_C)$.
\end{corollary}
\begin{proof}
The result follows from Proposition \ref{prop:compatibility-postProc-compatibility} and Proposition \ref{prop:channels-postProc-id}.
\end{proof}

\begin{corollary}
Let $S_n$ be a simplex, let $K_A$, $K_B$ be any state spaces, and let $\Pe_1 \in \chan(S_n, K_A)$ and $\Pe_2 \in \chan(S_n, K_B)$ be conditional preparations. Then $\Pe_1$ and $\Pe_2$ are compatible.
\end{corollary}
\begin{proof}
We know from Theorem \ref{thm:compatibility-noBroadcasting-simplex} that $\id_{S_n}: S_n \to S_n$ is self-compatible. Since $\Pe_1 \prec \id_{S_n}$ it follows from Proposition \ref{prop:compatibility-postProc-compatibility} that $\Pe_1$ and $\id_{S_n}$ are compatible. Repeating the same argument for $\Pe_2 \prec \id_{S_n}$ we get that $\Pe_1$ and $\Pe_2$ are compatible.
\end{proof}

The following result is a generalization of known result that two measurements are compatible if and only if they are both post-proccesings of a single measurement, see \cite{FilippovHeinosaariLeppajarvi-compatibility}.
\begin{proposition} \label{prop:compatibility-postProc-fromSelfCompat}
Let $\Phi \in \chan(K_A, K_B)$ be a self-compatible channel and let $\Phi_1 \in \chan(K_A, K_C)$ and $\Phi_2 \in \chan(K_A, K_D)$ be channels such that $\Phi_1 \prec \Phi$ and $\Phi_2 \prec \Phi$. Then $\Phi_1$ and $\Phi_2$ are compatible.
\end{proposition}
\begin{proof}
Since $\Phi \in \chan(K_A, K_B)$ is self-compatible and $\Phi_1 \prec \Phi$, it follows from Proposition \ref{prop:compatibility-postProc-compatibility} that $\Phi$ and $\Phi_1$ are compatible. Repeating the argument for $\Phi_2 \prec \Phi$ we get that $\Phi_1$ and $\Phi_2$ are compatible.
\end{proof}
Using the obtained results, we can prove that a post-processing greatest measurement $m \in \chan(K, S_n)$ exists only if $K$ is a simplex.
\begin{proposition} \label{prop:compatibility-postProc-measurements}
There exists a post-processing greatest measurement $m \in \chan(K, S_n)$ if and only if $K$ is a simplex.
\end{proposition}
\begin{proof}
Assume that a post-processing greatest measurement $m \in \chan(K, S_n)$ exists. Then for any two measurements $m_1 \in \chan(K, S_{n_1})$ and $m_2 \in \chan(K, S_{n_2})$ we have $m_1 \prec m$ and $m_2 \prec m$. Since $m$ is measurement, it is self-compatible, see Proposition \ref{prop:compatibility-MnP} . It then follows from Proposition \ref{prop:compatibility-postProc-fromSelfCompat} that $m_1$ and $m_2$ are compatible. Since $m_1$ and $m_2$ were arbitrary, it follows that all measurements $m_1 \in \chan(K, S_{n_1})$ and $m_2 \in \chan(K, S_{n_2})$ are compatible, but then according to Theorem \ref{thm:compatibility-noBroadcasting-measurementSimplex} $K$ is a simplex.
\end{proof}
When comparing Proposition \ref{prop:channels-postProc-id} to Proposition \ref{prop:compatibility-postProc-measurements}, one finds a significant difference between channels and measurements. This difference is going to manifest itself in determining the steerability of a state in the following subsection.

\subsection{Incompatibility witnesses, steering, and Bell non-locality}
It is rather easy to prove that two channels are compatible, we just need to provide the joint channel. But how do we prove that two channels are incompatible? Let
\begin{equation}
\chan^2(K_A; K_B, K_C) = \{ (\Phi_1, \Phi_2) : \Phi_1 \in \chan(K_A, K_B), \Phi_2 \in \chan(K_A, K_B) \}
\end{equation}
be the set of pairs of channels. Since $\chan(K_A, K_B)$ and $\chan(K_A, K_C)$ are convex sets, it is easy to see that $\chan^2(K_A; K_B, K_C)$ is also a convex set, with the convex combination defined as
\begin{equation}
\lambda (\Phi_1, \Phi_2) + (1-\lambda) (\Psi_1, \Psi_2) = (\lambda \Phi_1 + (1-\lambda) \Psi_1, \lambda \Phi_2 + (1-\lambda) \Psi_2),
\end{equation}
where $(\Phi_1, \Phi_2), (\Psi_1, \Psi_2) \in \chan^2(K_A; K_B, K_C)$ and $\lambda \in [0,1]$. Let $\chan \! \chan^2(K_A; K_B, K_C)$ be the set of compatible pairs of channels, i.e., $(\Phi_1, \Phi_2) \in \chan \! \chan^2 (K_A; K_B, K_C)$
only if $\Phi_1$ and $\Phi_2$ are compatible; we clearly have $\chan \! \chan^2(K_A; K_B, K_C) \subset \chan^2(K_A; K_B, K_C)$. One can also easily show that $\chan \! \chan^2(K_A; K_B, K_C)$ is a convex set: let $(\Phi_1, \Phi_2), (\Psi_1, \Psi_2) \in \chan \! \chan^2(K_A; K_B, K_C)$ with joint channels $\Phi, \Psi \in \chan(K_A, K_B \treal K_C)$ respectively. Then it is easy to see that we have
\begin{equation}
\begin{quantikz}[row sep={\the\originsep,between origins}, align equals at=2]
& \gate[3, nwires={1,3}]{\lambda \Phi + (1-\lambda)\Psi} &\ground{} \\
& & \\
& &\qw{}
\end{quantikz}
=
\begin{quantikz}[align equals at=1]
&\lstick{$\lambda$} &\gate{\Phi_1} &\qw{}
\end{quantikz}
+
\begin{quantikz}[align equals at=1]
&\lstick{$(1-\lambda)$} &\gate{\Psi_1} &\qw{}
\end{quantikz}
\end{equation}
and
\begin{equation}
\begin{quantikz}[row sep={\the\originsep,between origins}, align equals at=2]
& \gate[3, nwires={1,3}]{\lambda \Phi + (1-\lambda)\Psi} &\qw{} \\
& & \\
& &\ground{}
\end{quantikz}
=
\begin{quantikz}[align equals at=1]
&\lstick{$\lambda$} &\gate{\Phi_2} &\qw{}
\end{quantikz}
+
\begin{quantikz}[align equals at=1]
&\lstick{$(1-\lambda)$} &\gate{\Psi_2} &\qw{}
\end{quantikz}
\end{equation}
where we have used the linearity of the partial trace, see Example \ref{exm:channels-channels-partialTrace}. We can now use the hyperplane separation theorem \ref{thm:duals-hyperplaneSeparation} to find affine functions $W: \chan (K_A; K_B, K_C) \to \RR$ such that for all pairs of compatible channels $(\Phi_1, \Phi_2) \in \chan \! \chan^2(K_A; K_B, K_C)$ we have $W(\Phi_1, \Phi_2) \geq 0$. Then if $(\Psi_1, \Psi_2) \in \chan^2(K_A; K_B, K_C)$ such that $W(\Psi_1, \Psi_2) < 0$ then $\Psi_1$ and $\Psi_2$ must be incompatible.
\begin{definition}
Let $W: \chan (K_A; K_B, K_C) \to \RR$ be an affine function, such that $W(\Phi_1, \Phi_2) \geq 0$ for all $(\Phi_1, \Phi_2) \in \chan \! \chan^2(K_A; K_B, K_C)$ and such that there exist $(\Psi_1, \Psi_2) \in \chan^2(K_A; K_B, K_C)$ such that $W(\Psi_1, \Psi_2) < 0$. Then we call $W$ the \emph{incompatibility witness} for $(\Psi_1, \Psi_2)$.
\end{definition}
Incompatibility witnesses are an ideal tool to prove incompatibility of pairs fo channels. The only problem is: how does one find the suitable incompatibility witness for a given pair of channels? There are several know constructions: for measurements, one can use a relation between compatible measurements and entanglement breaking channels to obtain an incompatibility witness \cite{Jencova-incomaptibility}, or one can use discrimination tasks with partial immediate information \cite{CarmeliHeinosaariToigo-postMeasStateDiscrimination,CarmeliHeinosaariToigo-incWitness,CarmeliHeinosaariMiyaderaToigo-incWitnessChannels}.

We will present several ideas on how to prove incompatibility of pair of channels. We will not formulate the ideas as incompatibility witnesses, but rather as a more general and operationally motivated strategies. We will roughly follow the ideas presented in \cite{Plavala-channels} and we will comment on how these strategies connect to steering and Bell non-locality. We will start by presenting a construction that does not work, but it introduces the main concept that will be used later on.

Let $\Phi_1 \in \chan(K_A, K_B)$ and $\Phi_2 \in \chan(K_A, K_C)$ be compatible channels and let $\Phi \in \chan(K_A, K_B \treal K_C)$ be the joint channel. Then for every $x_A \in K_A$ there is some $y_{BC} \in K_B \treal K_C$ such that
\begin{equation} \label{eq:compatibility-witness-tr1}
\begin{quantikz}
[row sep={\the\originsep,between origins}, align equals at=2]
&\multiprepareC[3, nwires={1,3}]{y_{BC}} &\ground{} \\
& & \\
& &\qw{}
\end{quantikz}
=
\begin{quantikz}[align equals at=1]
&[\prepfix]\prepareC{x_A} &\gate{\Phi_1} &\qw{}
\end{quantikz}
\end{equation}
and
\begin{equation} \label{eq:compatibility-witness-tr2}
\begin{quantikz}
[row sep={\the\originsep,between origins}, align equals at=2]
&\multiprepareC[3, nwires={1,3}]{y_{BC}} &\qw{} \\
& & \\
& &\ground{}
\end{quantikz}
=
\begin{quantikz}[align equals at=1]
&[\prepfix]\prepareC{x_A} &\gate{\Phi_2} &\qw{}
\end{quantikz}
\end{equation}
This is immediate as one can always take $y_{BC} = \Phi(x_A)$ and then \eqref{eq:compatibility-witness-tr1} and \eqref{eq:compatibility-witness-tr2} follow from Definition \ref{def:compatibility-channels-def}. So then if for some pair of channels $\Phi_1$ and $\Phi_2$ such $y_{BC} \in K_B \treal K_C$ does not exist, then the channels must be incompatible and hence we have obtained a test of incompatibility. Unfortunately this test never works, because we can also take $y_{BC} = \Phi_1(x_A) \otimes \Phi_2(x_A)$ and \eqref{eq:compatibility-witness-tr1} and \eqref{eq:compatibility-witness-tr2} are satisfied even if $\Phi_1$ and $\Phi_2$ are incompatible. We have two opportunities to overcome this problem: we can either use a set of states $\{ x_{1,A}, \ldots, x_{n,A} \} \subset K_A$, or we can use an entangled state $x_{AD} \in K_A \treal K_D$. We can also combine both approaches. We will proceed with formulating the possible strategies to proving the incompatibility of channels: we will formulate our results as necessary conditions for compatibility of channels, violation of these conditions gives proofs of incompatibility of the channels in question. We will also show that using entangled states leads to steering \cite{UolaCostaNguyenGuhne-steering} and Bell non-locality \cite{BrunnerCavalcantiPironioScaraniWehner-BellNonlocality}, two well known phenomena in quantum information theory.

\begin{proposition} \label{prop:compatibility-witness-xi}
Let $\{ x_{1,A}, \ldots, x_{n,A} \} \subset K_A$ be states such that for some $\alpha_i \in \RR$, $\sum_{i=1}^n \alpha_i = 0$ we have $\sum_{i=1}^n \alpha_i x_{i,A} = 0$. Let $\Phi_1 \in \chan(K_A, K_B)$ and $\Phi_2 \in \chan(K_A, K_C)$ be compatible channels, then there are $y_{i,BC} \in K_B \treal K_C$, $i \in \{1, \ldots, n\}$, such that
\begin{equation} \label{eq:compatibility-witness-xi-tr1}
\begin{quantikz}
[row sep={\the\originsep,between origins}, align equals at=2]
&\multiprepareC[3, nwires={1,3}]{y_{i,BC}} &\ground{} \\
& & \\
& &\qw{}
\end{quantikz}
=
\begin{quantikz}[align equals at=1]
&[\prepfix]\prepareC{x_{i,A}} &\gate{\Phi_1} &\qw{}
\end{quantikz}
\end{equation}
and
\begin{equation} \label{eq:compatibility-witness-xi-tr2}
\begin{quantikz}
[row sep={\the\originsep,between origins}, align equals at=2]
&\multiprepareC[3, nwires={1,3}]{y_{i,BC}} &\qw{} \\
& & \\
& &\ground{}
\end{quantikz}
=
\begin{quantikz}[align equals at=1]
&[\prepfix]\prepareC{x_{i,A}} &\gate{\Phi_2} &\qw{}
\end{quantikz}
\end{equation}
for all $i \in \{1, \ldots, n\}$, and
\begin{equation} \label{eq:compatibility-witness-xi-sum0}
\sum_{i=1}^n \alpha_i y_{i,BC} = 0.
\end{equation}
\end{proposition}
\begin{proof}
Let $\Phi \in \chan(K_A, K_B \treal K_C)$ be the joint channel of $\Phi_1$ and $\Phi_2$. Let $y_{i,BC} = \Phi(x_{i,A})$, then \eqref{eq:compatibility-witness-xi-tr1} and \eqref{eq:compatibility-witness-xi-tr2} follow. Moreover we have
\begin{equation}
\sum_{i=1}^n \alpha_i y_{i,BC} = \sum_{i=1}^n \alpha_i \Phi(x_{i,A}) = \Phi \left( \sum_{i=1}^n \alpha_i x_{i,A} \right) = \Phi(0) = 0
\end{equation}
so also \eqref{eq:compatibility-witness-xi-sum0} holds.
\end{proof}

The main idea of Proposition \ref{prop:compatibility-witness-xi} is that for a given pair of channels $\Phi_1 \in \chan(K_A, K_B)$ and $\Phi_2 \in \chan(K_A, K_C)$ we can take some test subset $\{ x_{1,A}, \ldots, x_{n,A} \} \subset K_A$ and test whether there is some set $\{ y_{1,BC}, \ldots, y_{n,BC}\} \subset K_B \treal K_C$ such that \eqref{eq:compatibility-witness-xi-tr1}, \eqref{eq:compatibility-witness-xi-tr2} and \eqref{eq:compatibility-witness-xi-sum0} are satisfied.
\begin{definition}
Let $\{ x_{1,A}, \ldots, x_{n,A} \} \subset K_A$ be states and let $\Phi_1 \in \chan(K_A, K_B)$ and $\Phi_2 \in \chan(K_A, K_C)$ be channels. We say that $x_{1,A}, \ldots, x_{n,A}$ \emph{certify incompatibility} of $\Phi_1$ and $\Phi_2$ if there does not exist any set of states $\{y_{1,BC}, \ldots, y_{n,BC} \} \subset K_B \treal K_C$ that would satisfy \eqref{eq:compatibility-witness-xi-tr1}, \eqref{eq:compatibility-witness-xi-tr2} and \eqref{eq:compatibility-witness-xi-sum0} for every set of numbers $\{\alpha_1, \ldots, \alpha_n\} \subset \RR$, such that $\sum_{i=1}^n \alpha_i = 0$ and $\sum_{i=1}^n \alpha_i x_{i,A} = 0$.
\end{definition}

Note that we have already used the idea behind Proposition \ref{prop:compatibility-witness-xi} to prove Theorem \ref{thm:compatibility-noBroadcasting-simplex}. The main idea there is that $\id_K$ is self-compatible only if the set of pure states $\{ x_{1,A}, \ldots, x_{n,A} \}$ is affinely independent. This is because testing incompatibility on affinely independent states is essentially the same as trying to find incompatible channels on simplex. This is formalized as follows:
\begin{proposition}
Let $\{ x_{1,A}, \ldots, x_{n,A} \} \subset K_A$ be affinely independent points and let $\Phi_1 \in \chan(K_A, K_B)$, $\Phi_2 \in \chan(K_A, K_C)$, then $x_{1,A}, \ldots, x_{n,A}$ do not certify the incompatibility of $\Phi_1$ and $\Phi_2$.
\end{proposition}
\begin{proof}
Since $\{ x_{1,A}, \ldots, x_{n,A} \}$ are affinely independent, we have that for any $\alpha_i \in \RR$, $i \in \{1, \ldots, n\}$, such that $\sum_{i=1}^n \alpha_i = 0$ and $\sum_{i=1}^n \alpha_i x_i = 0$ we must have $\alpha_i = 0$ for all $i \in \{1, \ldots, n\}$. It then follows that \eqref{eq:compatibility-witness-xi-sum0} holds for any set $\{y_{1,BC}, \ldots, y_{n,BC}\} \subset K_B \treal K_C$. So we can simply take $y_{i,BC} = \Phi_1(x_{i,A}) \otimes \Phi_2(x_{i,A})$ for all $i \in \{1, \ldots, n\}$. We have already argued that \eqref{eq:compatibility-witness-xi-sum0} holds and it is straightforward to check that also \eqref{eq:compatibility-witness-xi-tr1} and \eqref{eq:compatibility-witness-xi-tr2} hold.
\end{proof}

We will also show that for large enough collection of states $\{ x_{1,A}, \ldots, x_{n,A} \} \subset K_A$, the test coming from Proposition \ref{prop:compatibility-witness-xi} must be conclusive, in the sense that it must either prove or disprove compatibility of the two channels.
\begin{proposition} \label{prop:compatibility-witness-xi-sufficient}
Let $\{ x_{1,A}, \ldots, x_{n,A} \} \subset K_A$ contain all pure states, i.e., all extreme points of $K_A$. Let $\Phi_1 \in \chan(K_A, K_B)$ and $\Phi_2 \in \chan(K_A, K_C)$, then $\Phi_1$ and $\Phi_2$ are incompatible if and only if $x_{1,A}, \ldots, x_{n,A}$ certify incompatibility of $\Phi_1$ and $\Phi_2$.
\end{proposition}
\begin{proof}
For simplicity, let $\{ x_{1,A}, \ldots, x_{n,A} \} \subset K_A$ be the set of all pure states and assume that $x_{1,A}, \ldots, x_{n,A}$ do not certify incompatibility of $\Phi_1$ and $\Phi_2$. Then there are $\{ y_{1,BC}, \ldots, y_{n, BC} \} \subset K_B \treal K_C$ that satisfy \eqref{eq:compatibility-witness-xi-tr1}, \eqref{eq:compatibility-witness-xi-tr2} and \eqref{eq:compatibility-witness-xi-sum0} for every set of numbers $\{\alpha_1, \ldots, \alpha_n\} \subset \RR$, such that $\sum_{i=1}^n \alpha_i = 0$ and $\sum_{i=1}^n \alpha_i x_{i,A} = 0$. Define a channel $\Phi \in \chan(K_A, K_B \treal K_C)$ by $\Phi(x_{i,A}) = y_{i,BC}$ and extended to the rest of $K_A$ by convexity. To see that $\Phi$ is well defined and convex, it is sufficient to check that for any $z_A \in K_A$ and any two affine decompositions of $z_A$ given as
\begin{align}
z_A &= \sum_{i=1}^n \beta^1_i x_{i,A} \label{eq:compatibility-witness-xi-sufficient-beta1} \\
z_A &= \sum_{i=1}^n \beta^2_i x_{i,A} \label{eq:compatibility-witness-xi-sufficient-beta2}
\end{align}
we have
\begin{equation} \label{eq:compatibility-witness-xi-sufficient-equal}
\sum_{i=1}^n \beta^1_i \Phi(x_{i,A}) = \sum_{i=1}^n \beta^2_i \Phi(x_{i,A}).
\end{equation}
Applying the unit effect $1_{K_A}$ to \eqref{eq:compatibility-witness-xi-sufficient-beta1} and \eqref{eq:compatibility-witness-xi-sufficient-beta2} yields
$\sum_{i=1}^n \beta^1_i = \sum_{i=1}^n \beta^2_i = 1$. We have $\sum_{i=1}^n (\beta^1_i - \beta^2_i) x_{i,A} = 0$ and so $\beta^1_i - \beta^2_i = \alpha_i$ is a set such that $\sum_{i=1}^n \alpha_i = 0$ and $\sum_{i=1}^n \alpha_i x_{i,A} = 0$. From \eqref{eq:compatibility-witness-xi-sum0} we get $\sum_{i=1}^n (\beta^1_i - \beta^2_i) \Phi(x_{i,A}) = 0$, it follows that \eqref{eq:compatibility-witness-xi-sufficient-equal} holds and that $\Phi$ is well-defined. It follows from \eqref{eq:compatibility-witness-xi-tr1} and \eqref{eq:compatibility-witness-xi-tr2} that $\Phi$ is a joint channel of $\Phi_1$ and $\Phi_2$. Hence $\Phi_1$ and $\Phi_2$ are compatible.
\end{proof}
Under the rug, we have assumed in Proposition \ref{prop:compatibility-witness-xi-sufficient} that $K_A$ has finitely many extreme points, i.e., that $K_A$ is a polytope. But this assumption is not necessary and one can easily extend the result of Proposition \ref{prop:compatibility-witness-xi-sufficient} to all state spaces using Carath\'{e}odory theorem \ref{thm:basic-stateSpace-Caratheodory}. One can also find a relation between certifying incompatibility and post-processing preorder of channels.
\begin{proposition}
Let $\{ x_{1,A}, \ldots, x_{n,A} \} \subset K_A$ and let $\Phi_1 \in \chan(K_A, K_B)$, $\Phi_2 \in \chan(K_A, K_C)$ and $\Phi_3  \in \chan(K_A, K_D)$ be such that $\Phi_3 \prec \Phi_2$. Then if $x_{1,A}, \ldots, x_{n,A}$ certify incompatibility of $\Phi_1$ and $\Phi_3$, then $x_{1,A}, \ldots, x_{n,A}$ also certify incompatibility of $\Phi_1$ and $\Phi_2$.
\end{proposition}
\begin{proof}
The proof follows by contradiction: assume that $x_{1,A}, \ldots, x_{n,A}$ do not certify incompatibility of $\Phi_1$ and $\Phi_2$, and let $\{ y_{1,BC}, \ldots, y_{n,BC} \} \subset K_B \treal K_C$ be the corresponding states satisfying \eqref{eq:compatibility-witness-xi-tr1}, \eqref{eq:compatibility-witness-xi-tr2}, and \eqref{eq:compatibility-witness-xi-sum0}. Since $\Phi_3 \prec \Phi_2$, there is a channel $\Psi \in \chan(K_C, K_D)$ such that
\begin{equation}
\begin{quantikz}[align equals at=1]
&\gate{\Phi_3} &\qw{}
\end{quantikz}
=
\begin{quantikz}[align equals at=1]
&\gate{\Phi_2} &\gate{\Psi} &\qw{}
\end{quantikz}
\end{equation}
Now let $\{ z_{1, BD}, \ldots, z_{n,BD} \} \subset K_B \treal K_D$ be defined as
\begin{equation}
\begin{quantikz}
[row sep={\the\originsep,between origins}, align equals at=2]
&\multiprepareC[3, nwires={1,3}]{z_{i,BD}} &\qw{} \\
& & \\
& &\qw{}
\end{quantikz}
=
\begin{quantikz}
[row sep={\the\originsep,between origins}, align equals at=2]
&[\prepfix]\multiprepareC[3, nwires={1,3}]{y_{i,BC}} &\qw{} &\qw{} \\
& & & \\
& &\gate{\Psi} &\qw{}
\end{quantikz}
\end{equation}
Then it is straightforward to check that $\{ z_{1, BD}, \ldots, z_{n,BD} \}$ satisfies \eqref{eq:compatibility-witness-xi-tr1}, \eqref{eq:compatibility-witness-xi-tr2}, and \eqref{eq:compatibility-witness-xi-sum0} for channels $\Phi_1$ and $\Phi_3$, so $x_{1,A}, \ldots, x_{n,A}$ can not certify the incompatibility of $\Phi_1$ and $\Phi_3$.
\end{proof}

Now we will proceed to testing incompatibility of channels using entangled states.
\begin{proposition} \label{prop:compatibility-witness-ent}
Let $x_{AD} \in K_A \treal K_D$ and let $\Phi_1 \in \chan(K_A, K_B)$, $\Phi_2 \in \chan(K_A, K_C)$ be compatible channels. Then there is $y_{BCD} \in K_B \treal K_C \treal K_D$ such that
\begin{equation} \label{eq:compatibility-witness-ent-tr1}
\begin{quantikz}
[row sep=\the\rowsep, align equals at=2]
&\multiprepareC[3]{y_{BCD}} &\qw{K_B} \\
& &\ground{}{K_C} \\
& &\qw{K_D}
\end{quantikz}
=
\begin{quantikz}[row sep=\the\rowsep, align equals at=1.5]
&[\prepfix]\multiprepareC[2]{x_{AD}} &\gate{\Phi_1} \qw{K_A} &\qw{K_B} \\
& &\qw &\qw{K_D}
\end{quantikz}
\end{equation}
and
\begin{equation} \label{eq:compatibility-witness-ent-tr2}
\begin{quantikz}[align equals at=2]
&\multiprepareC[3]{y_{BCD}} &\ground{}{K_B} \\
& &\qw{K_C} \\
& &\qw{K_D}
\end{quantikz}
=
\begin{quantikz}[row sep=\the\rowsep, align equals at=1.5]
&[\prepfix]\multiprepareC[2]{x_{AD}} &\gate{\Phi_2} \qw{K_A} &\qw{K_C} \\
& &\qw &\qw{K_D}
\end{quantikz}
\end{equation}
hold.
\end{proposition}
\begin{proof}
Take
\begin{equation}
\begin{quantikz}[align equals at=2]
&\multiprepareC[3]{y_{BCD}} &\qw{K_B} \\
& &\qw{K_C} \\
& &\qw{K_D}
\end{quantikz}
=
\begin{quantikz}[row sep=\the\rowsep, align equals at=3]
& &\gate[3, nwires={1,3}]{\Phi} &\qw{K_B} \\
&[\prepfix]\multiprepareC[3]{x_{AD}} &\qw{K_A} & \\
& & &\qw{K_C} \\
& &\qw &\qw{K_D}
\end{quantikz}
\end{equation}
where $\Phi \in \chan(K_A, K_B \treal K_C)$ is the joint channel of $\Phi_1$ and $\Phi_2$. Then \eqref{eq:compatibility-witness-ent-tr1} and \eqref{eq:compatibility-witness-ent-tr2} follow immediately.
\end{proof}
We can again use Proposition \ref{prop:compatibility-witness-ent} as a strategy to prove incompatibility of channels $\Phi_1$ and $\Phi_2$. We can simply select a state space $K_D$ and $x_{AD} \in K_A \treal K_D$ and check whether suitable $y_{BCD} \in K_B \treal K_C \treal K_D$ exists. But note that the state $x_{AD}$ can not be chosen randomly, but $x_{AD}$ must be an entangled state for the test to be meaningful.
\begin{proposition}
Let $x_{AD} \in K_A \tmin K_D$ be a separable state. Then for any $\Phi_1 \in \chan(K_A, K_B)$, $\Phi_2 \in \chan(K_A, K_C)$ there exists $y_{BCD} \in K_B \treal K_C \treal K_D$ satisfying \eqref{eq:compatibility-witness-ent-tr1} and \eqref{eq:compatibility-witness-ent-tr2}.
\end{proposition}
\begin{proof}
It is sufficient to prove that such $y_{BCD}$ exists for product states, i.e., for $x_{AD} = x_A \otimes z_D$. The proof for general separable states follows from the linearity of \eqref{eq:compatibility-witness-ent-tr1} and \eqref{eq:compatibility-witness-ent-tr2}. For $x_{AD} = x_A \otimes z_D$ take $y_{BCD} = \Phi_1(x_A) \otimes \Phi_2(x_A) \otimes z_D$. It is straightforward to show that \eqref{eq:compatibility-witness-ent-tr1} and \eqref{eq:compatibility-witness-ent-tr2} hold.
\end{proof}

Proposition \ref{prop:compatibility-witness-ent} allows us to construct entanglement-assisted incompatibility tests. More specifically, it was shown in \cite{Plavala-channels} that if we would replace channels by measurements, these tests would exactly correspond to steering \cite{UolaCostaNguyenGuhne-steering}. Hence the following definitions:
\begin{definition}
We say that channels $\Phi_1 \in \chan(K_A, K_B)$ and $\Phi_2 \in \chan(K_A, K_C)$ \emph{steer} the state $x_{AD} \in K_A \treal K_D$ if there is no $y_{BCD} \in K_B \treal K_C \treal K_D$ that would satisfy \eqref{eq:compatibility-witness-ent-tr1} and \eqref{eq:compatibility-witness-ent-tr2}.
\end{definition}

\begin{definition}
We say that the state $x_{AD} \in K_A \treal K_D$ is \emph{steerable by channels} if there are channels $\Phi_1 \in \chan(K_A, K_B)$ and $\Phi_2 \in \chan(K_A, K_C)$ that steer $x_{AD}$, i.e., such that there is no $y_{BCD} \in K_B \treal K_C \treal K_D$ that would satisfy \eqref{eq:compatibility-witness-ent-tr1} and \eqref{eq:compatibility-witness-ent-tr2}.
\end{definition}
One can define an equivalent notion of steerability by measurements.
\begin{definition}
We say that the state $x_{AD} \in K_A \treal K_D$ is \emph{steerable by measurements} if there are measurements $m_1 \in \chan(K_A, S_{n_1})$ and $m_2 \in \chan(K_A, S_{n_2})$ that steer $x_{AD}$, i.e., such that there is no $y_{BCD} \in S_{n_1} \treal S_{n_2} \treal K_D$ that would satisfy \eqref{eq:compatibility-witness-ent-tr1} and \eqref{eq:compatibility-witness-ent-tr2}.
\end{definition}

Is is known that in quantum theory two channels are incompatible if and only if they steer some state; it is an open question whether the same also holds for all GPTs. We will again get a relation between steering and post-processing preorder of channels.
\begin{proposition} \label{prop:compatibility-witness-ent-preorder}
Let $\Phi_1 \in \chan(K_A, K_B)$, $\Phi_2 \in \chan(K_A, K_C)$, and $\Phi_3 \in \chan(K_A, K_D)$ be channels such that $\Phi_3 \prec \Phi_2$. If $\Phi_1$ and $\Phi_2$ do not steer $x_{AE} \in K_A \treal K_E$, then also $\Phi_1$ and $\Phi_3$ do not steer $x_{AE}$.
\end{proposition}
\begin{proof}
If $\Phi_1$ and $\Phi_2$ do not steer $x_{AE} \in K_A \treal K_E$, then there is a state $y_{BCE} \in K_B \treal K_C \treal K_E$ such that
\begin{equation}
\begin{quantikz}
[row sep=\the\rowsep, align equals at=2]
&\multiprepareC[3]{y_{BCE}} &\qw{K_B} \\
& &\ground{}{K_C} \\
& &\qw{K_E}
\end{quantikz}
=
\begin{quantikz}[row sep=\the\rowsep, align equals at=1.5]
&[\prepfix]\multiprepareC[2]{x_{AE}} &\gate{\Phi_1} \qw{K_A} &\qw{K_B} \\
& &\qw{} &\qw{K_E}
\end{quantikz}
\end{equation}
and
\begin{equation}
\begin{quantikz}
[row sep=\the\rowsep, align equals at=2]
&\multiprepareC[3]{y_{BCE}} &\ground{}{K_B} \\
& &\qw{K_C} \\
& &\qw{K_E}
\end{quantikz}
=
\begin{quantikz}[row sep=\the\rowsep, align equals at=1.5]
&[\prepfix]\multiprepareC[2]{x_{AE}} &\gate{\Phi_2} \qw{K_A} &\qw{K_C} \\
& &\qw{} &\qw{K_E}
\end{quantikz}
\end{equation}
Since $\Phi_3 \prec \Phi_2$, there is a channel $\Psi \in \chan(K_C, K_D)$ such that
\begin{equation}
\begin{quantikz}[align equals at=1]
&\gate{\Phi_3} &\qw{}
\end{quantikz}
=
\begin{quantikz}[align equals at=1]
&\gate{\Phi_2} &\gate{\Psi} &\qw{}
\end{quantikz}
\end{equation}
Take $z_{BDE} \in K_B \treal K_D \treal K_E$ given as
\begin{equation}
\begin{quantikz}
[row sep=\the\rowsep, align equals at=2]
&\multiprepareC[3]{z_{BDE}} &\qw{K_B} & &[-7.5pt]\multiprepareC[3]{y_{BCE}} &\qw{} &\qw{K_B} \\
& &\qw{K_D} &\lstick{$=$} & &\gate{\Psi} \qw{K_C} &\qw{K_D} \\
& &\qw{K_E} & & &\qw{} &\qw{K_E}
\end{quantikz}
\end{equation}
It is straightforward to show that $z_{BCE}$ satisfies \eqref{eq:compatibility-witness-ent-tr1} and \eqref{eq:compatibility-witness-ent-tr2} and so $\Phi_1$ and $\Phi_3$ do not steer $x_{AE}$.
\end{proof}

\begin{corollary}
Let $\Phi_1 \in \chan(K_A, K_B)$, $\Phi_2 \in \chan(K_A, K_C)$, and $\Phi_3 \in \chan(K_A, K_D)$ be channels such that $\Phi_3 \prec \Phi_2$. If $\Phi_1$ and $\Phi_3$ steer $x_{AE} \in K_A \treal K_E$, then also $\Phi_1$ and $\Phi_2$ steer $x_{AE}$.
\end{corollary}
\begin{proof}
The result follows from Proposition \ref{prop:compatibility-witness-ent-preorder}. If $\Phi_1$ and $\Phi_2$ do not steer $x_{AE}$, then also $\Phi_1$ and $\Phi_3$ do not steer $x_{AE}$. So if $\Phi_1$ and $\Phi_3$ steer $x_{AE}$, then also $\Phi_1$ and $\Phi_2$ must steer $x_{AE}$.
\end{proof}

Let $\Phi_1 \in \chan(K_A, K_B)$ and $\Phi_2 \in \chan(K_A, K_C)$ be channels, then we already know that $\Phi_1 \prec \id_{K_A}$ and $\Phi_2 \prec \id_{K_A}$, where $\id_{K_A} \in \chan(K_A)$ is the identity channel. As a consequence we have the following.
\begin{proposition}
A state $x_{AD}$ is steerable by channels if and only if two copies of the identity channel $\id_{K_A} \in \chan(K_A)$ steer $x_{AD}$.
\end{proposition}
\begin{proof}
If two copies of the identity channel $\id_{K_A} \in \chan(K_A)$ steer $x_{AD}$, then $x_{AD}$ is steerable. So assume now that two copies of $\id_{K_A}$ do not steer $x_{AD}$. Let $\Phi_1 \in \chan(K_A, K_B)$ and $\Phi_2 \in \chan(K_A, K_C)$ be channels, then $\Phi_1 \prec \id_{K_A}$ and $\Phi_2 \prec \id_{K_A}$ and according to Proposition \ref{prop:compatibility-witness-ent-preorder} $\Phi_1$ and $\Phi_2$ can not steer $x_{AD}$.
\end{proof}
It is now easy to see the difference between measurements and channels. Since for non-classical state space $K_A$ there does not exist post-processing greatest measurement $m$, we can not easily determine whether a state $x_{AD} \in K_A \treal K_D$ is steerable by some measurements $m_1$, $m_2$, because we would have to check for all possible pairs of measurements. In fact, one only needs to check for measurements $m_1$ and $m_2$ such that if $m_1 \prec m'_1$, then also $m'_1 \prec m_1$ and if $m_2 \prec m'_2$, then also $m'_2 \prec m_2$, i.e., we only need to consider post-processing maximal measurement $m_1$ and $m_2$, see \cite{HeinosaariMiyaderaZiman-compatibility}. If $K_A$ is a polytope, this reduces to finite number of pairs measurements.

Since we have discussed steering, it is natural to expect that we can formalize Bell non-locality \cite{BrunnerCavalcantiPironioScaraniWehner-BellNonlocality} in this fashion. The main idea is that in steering, we are applying channels $\Phi_1 \in \chan(K_A, K_B)$ and $\Phi_2 \in \chan(K_A, K_C)$ to only the one leg of $x_{AD} \in K_A \treal K_D$. Given another pair of channel $\Psi_1 \in \chan(K_D, K_E)$ and $\Psi_2 \in \chan(K_D, K_F)$, we can apply them to the other leg of $x_{AD}$.
\begin{proposition} \label{prop:compatibility-witness-Bell}
Let $x_{AD} \in K_A \treal K_D$ be a state and let $\Phi_1 \in \chan(K_A, K_B)$, $\Phi_2 \in \chan(K_A, K_C)$, $\Psi_1 \in \chan(K_D, K_E)$ and $\Psi_2 \in \chan(K_D, K_F)$ be channels. Assume that $\Phi_1$, $\Phi_2$ are compatible and that $\Psi_1$, $\Psi_2$ are compatible. Then there is a state $y_{BCEF} \in K_B \treal K_C \treal K_E \treal K_F$ such that
\begin{equation} \label{eq:compatibility-witness-Bell-tr11}
\begin{quantikz}
[row sep=\the\rowsep, align equals at=2.5]
&\multiprepareC[4]{y_{BCEF}} &\qw{K_B} \\
& &\ground{}{K_C} \\
& &\qw{K_E} \\
& &\ground{}{K_F}
\end{quantikz}
=
\begin{quantikz}[row sep=\the\rowsep, align equals at=1.5]
&[\prepfix]\multiprepareC[2]{x_{AD}} &\gate{\Phi_1} \qw{K_A} &\qw{K_B} \\
& &\gate{\Psi_1} \qw{K_D} &\qw{K_E}
\end{quantikz}
\end{equation}
and
\begin{equation} \label{eq:compatibility-witness-Bell-tr12}
\begin{quantikz}
[row sep=\the\rowsep, align equals at=2.5]
&\multiprepareC[4]{y_{BCEF}} &\qw{K_B} \\
& &\ground{}{K_C} \\
& &\ground{}{K_E} \\
& &\qw{K_F}
\end{quantikz}
=
\begin{quantikz}[row sep=\the\rowsep, align equals at=1.5]
&[\prepfix]\multiprepareC[2]{x_{AD}} &\gate{\Phi_1} \qw{K_A} &\qw{K_B} \\
& &\gate{\Psi_2} \qw{K_D} &\qw{K_F}
\end{quantikz}
\end{equation}
and
\begin{equation} \label{eq:compatibility-witness-Bell-tr21}
\begin{quantikz}
[align equals at=2.5]
&\multiprepareC[4]{y_{BCEF}} &\ground{}{K_B} \\
& &\qw{K_C} \\
& &\qw{K_E} \\
& &\ground{}{K_F}
\end{quantikz}
=
\begin{quantikz}[row sep=\the\rowsep, align equals at=1.5]
&[\prepfix]\multiprepareC[2]{x_{AD}} &\gate{\Phi_2} \qw{K_A} &\qw{K_C} \\
& &\gate{\Psi_1} \qw{K_D} &\qw{K_E}
\end{quantikz}
\end{equation}
and
\begin{equation} \label{eq:compatibility-witness-Bell-tr22}
\begin{quantikz}
[row sep=\the\rowsep, align equals at=2.5]
&\multiprepareC[4]{y_{BCEF}} &\ground{}{K_B} \\
& &\qw{K_C} \\
& &\ground{}{K_E} \\
& &\qw{K_F}
\end{quantikz}
=
\begin{quantikz}[row sep=\the\rowsep, align equals at=1.5]
&[\prepfix]\multiprepareC[2]{x_{AD}} &\gate{\Phi_2} \qw{K_A} &\qw{K_B} \\
& &\gate{\Psi_2} \qw{K_D} &\qw{K_F}
\end{quantikz}
\end{equation}
hold.
\end{proposition}
\begin{proof}
Let $\Phi \in \chan(K_A, K_B \treal K_C)$ and $\Psi \in \chan(K_D, K_E \treal K_F)$ be the joint channels of $\Phi_1$, $\Phi_2$ and $\Psi_1$, $\Psi_2$ respectively. Take
\begin{equation}
\begin{quantikz}[align equals at=2.5]
&\multiprepareC[4]{y_{BCEF}} &\qw{K_B} \\
& &\qw{K_C} \\
& &\qw{K_E} \\
& &\qw{K_F}
\end{quantikz}
=
\begin{quantikz}[row sep=\the\rowsep, align equals at=3.5]
& &\gate[3, nwires={1,3}]{\Phi} &\qw{K_B} \\
&[\prepfix]\multiprepareC[4]{x_{AD}} &\qw{K_A} & \\
& & &\qw{K_C} \\
& &\gate[3, nwires={1,3}]{\Psi} &\qw{K_E} \\
& &\qw{K_D} & \\
& & &\qw{K_F}
\end{quantikz}
\end{equation}
It is straightforward to verify that \eqref{eq:compatibility-witness-Bell-tr11} - \eqref{eq:compatibility-witness-Bell-tr22} hold.
\end{proof}

One can again show that if we would replace channels by measurements, we would simply get the standard definition of Bell non-locality \cite{Plavala-channels}. Hence the following definitions:
\begin{definition}
We say that $x_{AD}$ is \emph{Bell non-local with respect to $\Phi_1 \in \chan(K_A, K_B)$, $\Phi_2 \in \chan(K_A, K_C)$, $\Psi_1 \in \chan(K_D, K_E)$ and $\Psi_2 \in \chan(K_D, K_F)$} if there is no $y_{BCEF} \in K_B \treal K_C \treal K_E \treal K_F$ such that \eqref{eq:compatibility-witness-Bell-tr11} - \eqref{eq:compatibility-witness-Bell-tr22} are satisfied.
\end{definition}

\begin{definition}
We say that $x_{AD}$ is \emph{Bell non-local state} if there are channel $\Phi_1 \in \chan(K_A, K_B)$, $\Phi_2 \in \chan(K_A, K_C)$, $\Psi_1 \in \chan(K_D, K_E)$ and $\Psi_2 \in \chan(K_D, K_F)$ with respect to which $x_{AD}$ is Bell non-local.
\end{definition}

One can again prove that we need entanglement to get Bell non-locality.
\begin{proposition}
Let $x_{AD} \in K_A \tmin K_D$ be a separable state, then $x_{AD}$ is Bell local, i.e., $x_{AD}$ is not Bell non-local.
\end{proposition}
\begin{proof}
Since the conditions \eqref{eq:compatibility-witness-Bell-tr11} - \eqref{eq:compatibility-witness-Bell-tr22} are linear, it is sufficient to take $x_{AD} = z_A \otimes w_D$ where $z_A \in K_A$ and $w_D \in K_D$, for a general state the result follow by taking convex combinations. Let $\Phi_1 \in \chan(K_A, K_B)$, $\Phi_2 \in \chan(K_A, K_C)$, $\Psi_1 \in \chan(K_D, K_E)$ and $\Psi_2 \in \chan(K_D, K_F)$ be channels, then we can take
\begin{equation}
\begin{quantikz}[align equals at=2.5]
&\multiprepareC[4]{y_{BCEF}} &\qw{K_B} \\
& &\qw{K_C} \\
& &\qw{K_E} \\
& &\qw{K_F}
\end{quantikz}
=
\begin{quantikz}[row sep=\the\rowsep,align equals at=2.5]
&[\prepfix]\prepareC{z_A} &\gate{\Phi_1} \qw{K_A} &\qw{K_B} \\
&[\prepfix]\prepareC{z_A} &\gate{\Phi_2} \qw{K_A} &\qw{K_C} \\
&[\prepfix]\prepareC{w_D} &\gate{\Psi_1} \qw{K_D} &\qw{K_E} \\
&[\prepfix]\prepareC{w_D} &\gate{\Psi_2} \qw{K_D} &\qw{K_F} 
\end{quantikz}
\end{equation}
It is straightforward to show that \eqref{eq:compatibility-witness-Bell-tr11} - \eqref{eq:compatibility-witness-Bell-tr22} are satisfied.
\end{proof}
One can again prove relations between post-processing preorder of channels and Bell non-locality. We will not do so, since they are straightforward to formulate. At last, we would want to comment on the connection between steering and Bell non-locality. For measurements, it is easy to show that steering is necessary for Bell non-locality, but for channels this is not so, see \cite{Plavala-channels} for a counter-example. One can also combine the approaches and consider scenarios with sets of entangled states $\{ x_{1,AD}, \ldots, x_{n,AD}\} \subset K_A \treal K_D$; the generalization is straightforward and we will not investigate it.

\section{Example: quantum theory} \label{sec:QT}

%
%
In this section we will review quantum theory as an example of a GPT. We will be brief as most of the things we will cover are considered basic knowledge in quantum information theory. If the reader is not familiar with quantum theory, we recommend \cite{HeinosaariZiman-MLQT}.

Let $\Ha$ be a finite-dimensional complex Hilbert space. We will use the bra–ket notation to denote the vectors as $\ket{\psi}$, the inner product of $\ket{\psi}, \ket{\varphi} \in \Ha$ is then denoted $\braket{\varphi|\psi}$ and it is linear in the second argument, i.e., $\braket{\varphi|\psi}$ is linear in $\ket{\psi}$. $\lin(\Ha)$ will denote the complex vector space of operators $X: \Ha \to \Ha$, $\I$ will denote the identity operator. $\bound_H(\Ha)$ will denote the real vector space of self-adjoint operators. Let $X \in \bound_H(\Ha)$, then $\Tr(X)$ will denote the trace of $X$. We say that $X$ is positive semi-definite and we write $X \geq 0$ if for all $\ket{\psi} \in \Ha$ we have $\bra{\psi}X\ket{\psi} \geq 0$. $\bound_H^+(\Ha)$ will denote the set of all positive semi-definite operators; note that $\bound_H^+(\Ha)$ is a convex, pointed and generating cone.

\subsection{State space and effect algebra}
The state space in quantum theory is the set of density operators
\begin{equation}
\dens(\Ha) = \{ \rho \in \bound_H^+(\Ha) : \Tr(\rho) = 1 \}.
\end{equation}
The pure states are rank-1 projectors, i.e., the pure states of $\dens(\Ha)$ are projectors $\ketbra{\psi}$ for $\ket{\psi} \in \Ha$, $\norm{\psi}^2 = \braket{\psi|\psi} = 1$. The vector space of affine functions $A(\dens(\Ha))$ is isomorphic to $\bound_H(\Ha)$, for $X \in \bound_H(\Ha)$ the corresponding function $f_X \in A(\dens(\Ha))$ is given as $f_X(\rho) = \Tr(\rho X)$. From now on we will omit the isomorphism between $f_X$ and $X$ and use $X$ to refer to the function $f_X$ and we will write
\begin{equation}
A(\dens(\Ha)) = \bound_H(\Ha).
\end{equation}
The cone of positive functions  $A(\dens(\Ha))^+$ is isomorphic to $\bound_H^+(\Ha)$, because let $X \in \bound_H(\Ha)$, then $\Tr(\rho X) \geq 0$ if and only if $\Tr( \ketbra{\psi} X) = \bra{\psi} X \ket{\psi} \geq 0$ for all $\ket{\psi} \in \Ha$, $\norm{\psi} = 1$. It follows that $X \in A(\dens(\Ha))^+$ if and only if $X \geq 0$, and so
\begin{equation}
A(\dens(\Ha))^+ = \bound_H^+(\Ha),
\end{equation}
again omitting the isomorphism between $A(\dens(\Ha))$ and $\bound_H(\Ha)$. Now we will characterize the effect algebra $E(\dens(\Ha))$. Since for every $\rho \in \dens(\Ha)$ we have $\Tr(\rho) = 1$, we have $\Tr(\rho X) \leq 1$ if and only if $\Tr(\rho (\I - X)) \geq 0$, which is equivalent to $\I - X \geq 0$. It follows that $0 \leq \Tr(\rho X) \leq 1$ if and only if $0 \leq X \leq \I$. We have
\begin{equation}
E(\dens(\Ha)) = \effect(\Ha) = \{ X \in \bound_H(\Ha): 0 \leq X \leq \I \}
\end{equation}
up to the isomorphism. It is well-known that $\bound_H(\Ha)$ is a Hilbert space with the Hilbert-Schmidt inner product given as $\Tr(XY)$ for $X,Y \in \bound_H(\Ha)$. It follows that the dual of $A(\dens(\Ha)) = \bound_H(\Ha)$ is again going to be $\bound_H(\Ha)$, since Hilbert spaces are self-dual. So we have
\begin{equation}
A(\dens(\Ha))^* = \bound_H(\Ha).
\end{equation}
This is the reason why in the standard approach to quantum information theory we do not distinguish between $A(\dens(\Ha))$ and $A(\dens(\Ha))^*$, because they are isomorphic. The isomorphism between $A(\dens(\Ha))$ and $A(\dens(\Ha))^*$ (and as we will shortly see, also between $A(\dens(\Ha))^+$ and $A(\dens(\Ha))^{*+}$) is a very important aspect of quantum theory. One can again use simple arguments based on rank-1 projectors to show that
\begin{equation}
A(\dens(\Ha))^{*+} = \bound_H^+(\Ha).
\end{equation}
It is also straightforward to see that $\dens(\Ha)$ is base of the cone $\bound_H^+(\Ha)$, as it should be.

The base norm corresponds to the trace norm given as $\Tr(\abs{X})$, where $\abs{X} = \sqrt{X^2}$. The order unit norm corresponds to the operator norm $\norm{X}$. Also note that the constructed theory satisfies no-restriction hypothesis.

\subsection{Tensor product}
In quantum theory, tensor products of state spaces are induced by the tensor products of the underlying Hilbert spaces, i.e., let $\Ha_A$ and $\Ha_B$ be Hilbert spaces, then we define
\begin{equation} \label{eq:QT-tensor-prod}
\dens(\Ha_A) \treal \dens(\Ha_B) = \dens(\Ha_A \otimes \Ha_B).
\end{equation}
We will now show that \eqref{eq:QT-tensor-prod} defines a valid tensor product of state spaces. Note that we have $\bound_H(\Ha_A) \otimes \bound_H(\Ha_B) = \bound_H(\Ha_A \otimes \Ha_B)$ so
\begin{equation}
\linspan(\dens(\Ha_A) \treal \dens(\Ha_B)) = \bound_H(\Ha_A \otimes \Ha_B) = \bound_H(\Ha_A) \otimes \bound_H(\Ha_B) = A(\dens(\Ha_A))^* \otimes A(\dens(\Ha_B))^*
\end{equation}
as we should have. It remains to show that
\begin{equation} \label{eq:QT-tensor-inclusions}
\dens(\Ha_A) \tmin \dens(\Ha_B) \subset \dens(\Ha_A \otimes \Ha_B) \subset \dens(\Ha_A) \tmax \dens(\Ha_B).
\end{equation}
Let $\rho_A \in \dens(\Ha_A)$ and $\rho_B \in \dens(\Ha_B)$, then $\rho_A \otimes \rho_B \geq 0$, i.e., $\rho_A \otimes \rho_B$ is s positive semi-definite operator. It follows that $\rho_A \otimes \rho_B \in \dens(\Ha_A \otimes \Ha_B)$ and we get $\dens(\Ha_A) \tmin \dens(\Ha_B) \subset \dens(\Ha_A \otimes \Ha_B)$. Now let $\rho_{AB} \in \dens(\Ha_A \otimes \Ha_B)$ and let $E_A \in \effect(\Ha_A)$, $E_B \in \effect(\Ha_B)$, then we have $E_A \otimes E_B \geq 0$ and $\Tr(\rho_{AB} (E_A \otimes E_B)) \geq 0$. It follows that $\rho_{AB} \in \dens(\Ha_A) \tmax \dens(\Ha_B)$ and we get $\dens(\Ha_A \otimes \Ha_B) \subset \dens(\Ha_A) \tmax \dens(\Ha_B)$. Thus we have proved both inclusion in \eqref{eq:QT-tensor-inclusions} and so $\dens(\Ha_A) \treal \dens(\Ha_B)$ is a well-defined bipartite state space. Note that both of the inclusions are strict. Let for simplicity $\Ha_A = \Ha_B = \Ha$, then it is well-known that entangled states, such as the maximally entangled state $\ketbra{\phi^+}$, $\ket{\phi^+} = \frac{1}{\sqrt{\dH}} \sum_{i=1}^{\dH} \ket{ii}$, exist, so we have $\dens(\Ha_A) \tmin \dens(\Ha_B) \neq \dens(\Ha_A \otimes \Ha_B)$. It is also well-known that the partial transpose of $\ketbra{\phi^+}$, denoted $\ketbra{\phi^+}^\Gamma$, is not positive semi-definite, hence $\ketbra{\phi^+}^\Gamma \notin \dens(\Ha_A \otimes \Ha_B)$. But we know from Proposition \ref{prop:channels-CP-maxTensor} that every positive map is completely positive with respect to the maximal tensor product, so we must have $\ketbra{\phi^+}^\Gamma \in \dens(\Ha_A) \tmax \dens(\Ha_B)$. Therefore we have $\dens(\Ha_A \otimes \Ha_B) \neq \dens(\Ha_A) \tmax \dens(\Ha_B)$.

We will now show that the tensor product defined in \eqref{eq:QT-tensor-prod} is associative. Let $\Ha_A, \Ha_B, \Ha_C$ be finite-dimensional complex Hilbert spaces, then we have
\begin{align}
(\dens(\Ha_A) \treal \dens(\Ha_B)) \treal \dens(\Ha_C) &= \dens(\Ha_A \otimes \Ha_B) \treal \dens(\Ha_C) = \dens(\Ha_A \otimes \Ha_B \otimes \Ha_C) \\
&= \dens(\Ha_A) \treal \dens(\Ha_B \otimes \Ha_C) = \dens(\Ha_A) \treal (\dens(\Ha_B) \treal \dens(\Ha_C))
\end{align}
as a result of associativity of the tensor product of Hilbert spaces.

Let $\Ha_A, \Ha_B$ be Hilbert spaces, then the partial trace is defined as the unique map $\Tr_A : \bound_H(\Ha_A \otimes \Ha_B) \to \bound_H(\Ha_B)$ such that for $X_{AB} \in \bound_H(\Ha_A \otimes \Ha_B)$ and all $Y_B \in \bound_H(\Ha_B)$ we have
\begin{equation}
\Tr(\Tr_A(X_{AB}) Y_B) = \Tr( X_{AB} (\I_A \otimes Y_B)).
\end{equation}
This exactly corresponds to the definition of partial trace in Example \ref{exm:channels-channels-partialTrace}.

\subsection{Channels}
Since we have a well-defined tensor product in quantum theory, we usually work only with completely-positive channels in quantum theory. It is well-known that the set of completely positive channels $\Phi \in \chan(\dens(\Ha_A), \dens(\Ha_B)$ is isomorphic to the set of Choi matrices $\J(\Ha_A, \Ha_B)$ given as
\begin{equation}
\J(\Ha_A, \Ha_B) = \left\lbrace X \in \bound_H^+(\Ha_A \otimes \Ha_B) : \Tr_B(X) = \frac{\I_A}{\dim(\Ha_A)} \right\rbrace.
\end{equation}
We have normalized the trace of the Choi matrices to $1$, so we have $\J(\Ha_A, \Ha_B) \subset \dens(\Ha_A \otimes \Ha_B)$. This is to be compared to the characterization of all positive channels provided in Proposition \ref{prop:channels-channels-tensorSubset}, as one can clearly see that the set of completely positive channels is strictly smaller than the set of positive channels. Note that all measurements are automatically completely positive, because measurements are completely positive with respect to any tensor product, see Proposition \ref{prop:channels-CP-measurements}. Hence the characterization of measurements derived in Proposition \ref{prop:channels-measurements-effects} still holds without any modifications.

\subsection{Compatibility of channels}
In quantum theory, the channels $\Phi_1 \in \J(\Ha_A, \Ha_B)$ and $\Phi_2 \in \J(\Ha_A, \Ha_C)$ are said to be compatible if there is a joint channel $\Phi \in \J(\Ha_A, \Ha_B \otimes \Ha_C)$ such that for all $\rho_A \in \dens(\Ha_A)$ we have
\begin{align}
&\Tr_C ( \Phi(\rho_A) ) = \Phi_1 (\rho_A),
&&\Tr_B ( \Phi(\rho_A) ) = \Phi_2 (\rho_A).
\end{align}
This is exactly the same definition as Definition \ref{def:compatibility-channels-def}, except that positive channels are replaced by completely positive channels. Note that all of the results we have proved for compatibility of positive channels are easily generalizable to completely positive channels.

\section{Example: boxworld theory} \label{sec:boxworld}

%
%
In this section we will review a theory usually refer to as boxworld. Boxworld was introduced in \cite{Barrett-GPTinformation} to describe a theory of black boxes that have finite number of inputs and finite number of outputs. We will investigate the case of boxes with one input bit and one output bit. There are going to be four extreme points, denoted $s_{ij}$, $i,j \in \{0,1\}$, corresponding to one of the four possible scenarios: the box $s_{00}$, $s_{11}$ always outputs $0$, $1$, respectively, no matter the input, the box $s_{01}$ outputs its input unchanged, and the box $s_{10}$ outputs $1$ if the input is $0$ and outputs $0$ if the input is $1$, see also Table \ref{table:boxworld-inOut}.

\begin{table}[ht]
\centering
\begin{tabular}{c|c|c}
& $1 \mapsto 0$ & $1 \mapsto 1$ \\ 
\hline 
$0 \mapsto 0$ & $s_{00}$ & $s_{01}$ \\ 
\hline 
$0 \mapsto 1$ & $s_{10}$ & $s_{11}$ \\ 
\end{tabular}
\caption{The extreme points of the simplest boxworld state space described in terms of how they handle the input $0$ and how they handle the input $1$. $s_{ij}$ maps $0$ to $i$ and $1$ to $j$, where $i,j \in \{0, 1\}$. \label{table:boxworld-inOut}}
\end{table}

The four extreme points $s_{ij}$ are not affinely independent, but we have
\begin{equation} \label{eq:boxworld-cross}
\dfrac{1}{2} ( s_{00} + s_{11} ) = \dfrac{1}{2} ( s_{10} +s_{01} ),
\end{equation}
which is easy to check for both possible inputs. It follows that the state space we are dealing with is a square.

\subsection{State space and effect algebra}
Let
\begin{align}
&s_{00} =
\begin{pmatrix}
0 \\
0 \\
1
\end{pmatrix},
&&s_{10} =
\begin{pmatrix}
1 \\
0 \\
1
\end{pmatrix},
&&s_{01} =
\begin{pmatrix}
0 \\
1 \\
1
\end{pmatrix},
&&s_{11} =
\begin{pmatrix}
1 \\
1 \\
1
\end{pmatrix},
\end{align}
and $S = \conv(\{ s_{00}, s_{10}, s_{01}, s_{11} \})$. $S$ is the square state space that we will be investigating. Note that we have
\begin{equation} \label{eq:boxworld-stateSpace-s11}
s_{11} = s_{10} + s_{01} - s_{00}
\end{equation}
which one can prove from \eqref{eq:boxworld-cross} or from the definition of the pure states. The vector space of affine functions is $A(S)$ and we have $\dim(A(S)) = 3$. Let
\begin{align}
&f_x =
\begin{pmatrix}
1 \\
0 \\
0
\end{pmatrix},
&&1_S - f_x =
\begin{pmatrix}
-1 \\
0 \\
1
\end{pmatrix},
&&f_y =
\begin{pmatrix}
0 \\
1 \\
0
\end{pmatrix},
&&1_S - f_y =
\begin{pmatrix}
0 \\
-1 \\
1
\end{pmatrix},
&&1_S =
\begin{pmatrix}
0 \\
0 \\
1
\end{pmatrix}
\end{align}
then $f_x, 1_S - f_x, f_y, 1_S - f_y, 1_S \in A(S)$, where the pairing is given by the usual Euclidean inner product. The cone of the positive functions $A(S)^+$ and the effect algebra $E(S)$ are generated by $f_x, 1_S - f_x, f_y, 1_S - f_y$. We then have
\begin{equation}
E(S) = \conv( \{ 0, f_x, 1_S - f_x, f_y, 1_S - f_y, 1_S \} ).
\end{equation}
The state space $S$ together with the positive cone $A(S)^{*+}$ is depicted in Figure \ref{fig:boxworld-stateSpace} and the effect algebra $E(S)$ together with the positive cone $A(S)^+$ is depicted in Figure \ref{fig:boxworld-effectAlgebra}.

\begin{figure}[t]
\centering
\begin{subfigure}[t]{0.475\textwidth}
\centering
\includegraphics[width=\textwidth]{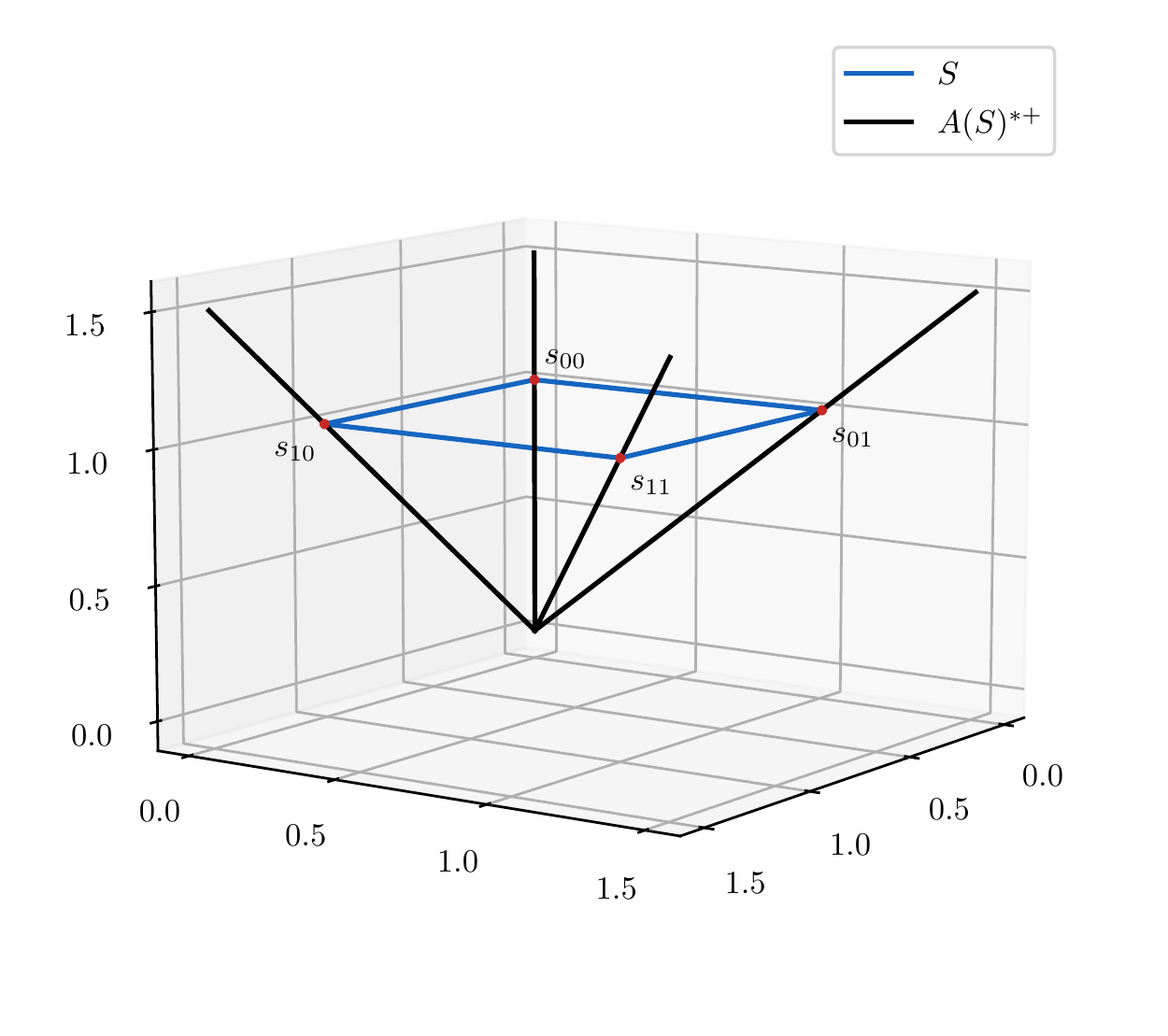}
\caption{Picture of the state space $S$ as a subset of $A(S)^*$. The red points are the pure states $s_{00}$, $s_{10}$, $s_{01}$, and $s_{11}$, the blue lines are the edges of the state space $S$, and the black lines are the edges of the positive cone $A(S)^{*+}$.}
\label{fig:boxworld-stateSpace}
\end{subfigure}
\hfill
\begin{subfigure}[t]{0.475\textwidth}
\centering
\includegraphics[width=\textwidth]{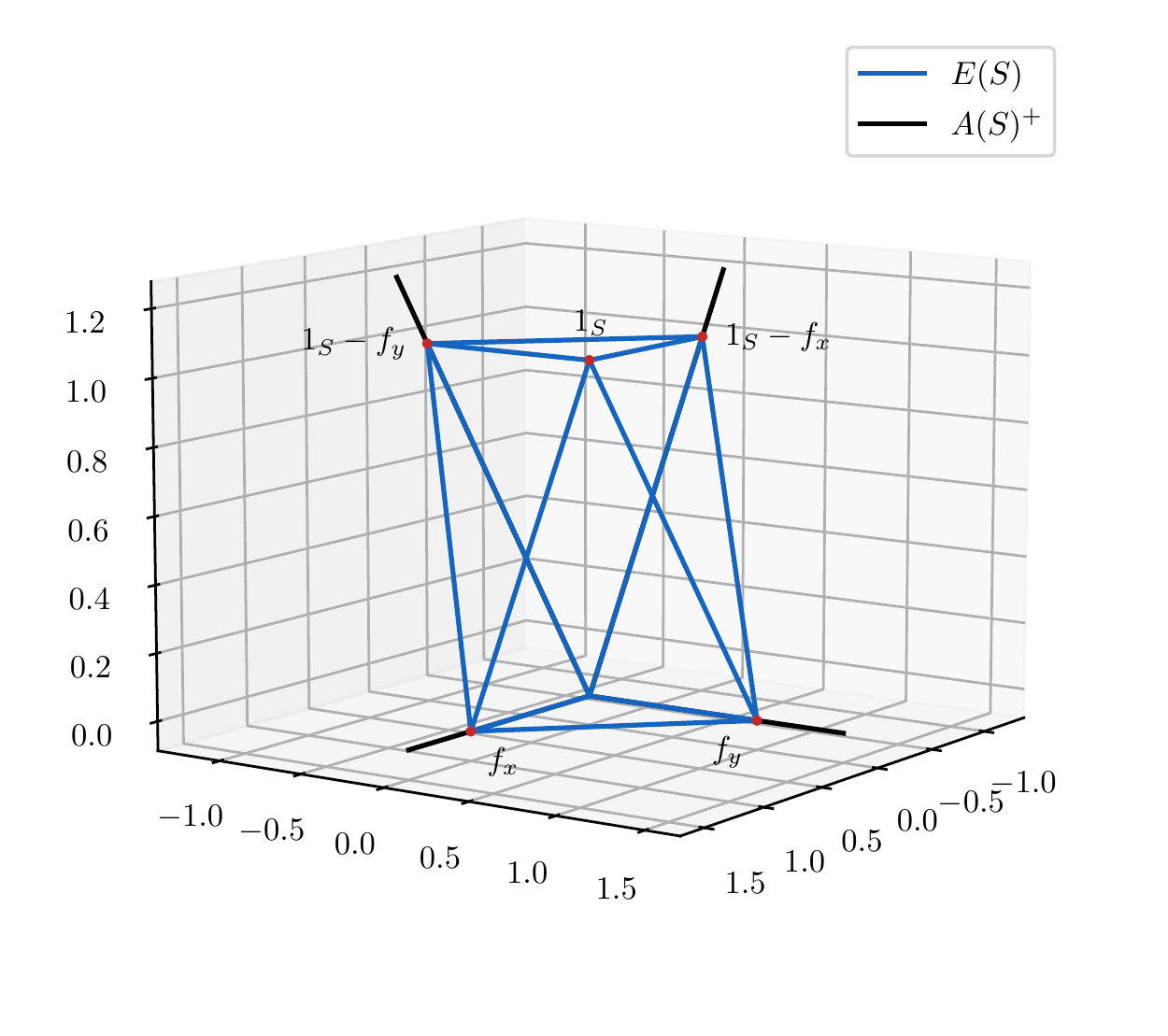}
\caption{Picture of the effect algebra $E(S)$ as a subset of $A(S)$. The red points are the effects $f_x$, $1_S - f_x$, $f_y$, $1_S - f_y$, and $1_S$, the blue lines are the edges of the effect algebra $E(S)$, and the black lines are the edges of the positive cone $A(S)^+$.}
\label{fig:boxworld-effectAlgebra}
\end{subfigure}
\caption{Pictures of the state space $S$ and effect algebra $E(S)$.}
\end{figure}

\subsection{Tensor product}
Since boxworld theory is a hypothetical theory, there is no physical principle that would select a specific tensor product. Therefore, we will investigate the minimal and maximal tensor products. The minimal tensor product is $S \tmin S$ and it contains $16$ pure states; it is given as
\begin{equation}
S \tmin S = \conv( \{ s_{ij} \otimes s_{kl} : i,j,k,l \in \{0,1 \} \} ).
\end{equation}
Let us characterize the maximal tensor product. Since $S \tmax S \subset A(S)^* \otimes A(S)^*$ and since $s_{00}, s_{10}, s_{01}$ is a basis of $A(S)^*$, we can express any state $x \in S \tmax S$ as
\begin{equation}
x = \sum_{I,J \in \{00, 10, 01\}} \alpha_{IJ} s_I \otimes s_J,
\end{equation}
where $I,J$ are multi-indexes. Since $\< x, 1_S \otimes 1_S \> = 1$ we must have $\sum_{I,J \in \{00, 10, 01\}} \alpha_{IJ} = 1$. We know from Definition \ref{def:tensor-bipartite-maxProd} that $x \in S \tmax S$ if and only if we have $\< x, g_A \otimes g_B \> \geq 0$ for all $g_A, g_B \in E(S)$. But since $E(S)$ is generated by $f_x$, $1_S - f_x$, $f_y$, and $1_S - f_y$, it is sufficient to check only for $g_A, g_B \in \{ f_x, 1_S - f_x, f_y, 1_S - f_y \}$, this yields $16$ conditions. One can explicitly write down all of the $16$ conditions and find the most general form of the coefficients $\alpha_{IJ}$. This can be carried out numerically and one can find that the pure entangled states can be characterized in terms of correlations between Alice and Bob \cite{BarrettLindenMassarPironioPopescuDavid-nonlocal,JanottaGogolinBarrettBrunner-nonlocal,BrunnerKaplanLeverrierSkrypczyk-dimensions} and that they maximally violate the CHSH inequality, see also \cite{JencovaPlavala-PRbox} for the construction of such states.

Another option is to express $x \in S \tmax S$ as
\begin{equation} \label{eq:boxworld-tensor-xDef}
x = s_{00} \otimes v_0 + s_{10} \otimes v_1 + s_{01} \otimes v_2.
\end{equation}
From $\<x, 1_S \otimes 1_S\> = 1$ we get
\begin{equation} \label{eq:boxworld-tensor-trace}
\< v_0 + v_1 + v_2, 1_S \> = 1.
\end{equation}
For $g \in E(S)$ we get
\begin{align}
\< x, f_x \otimes g \> &= \< v_1, g\>, \label{eq:boxworld-tensor-maxfx} \\
\< x, (1_S - f_x) \otimes g \> &= \< v_0 + v_2, g\>, \label{eq:boxworld-tensor-max1-fx} \\
\< x, f_y \otimes g \> &= \< v_2, g\>, \label{eq:boxworld-tensor-maxfy} \\
\< x, (1_S - f_x) \otimes g \> &= \< v_0 + v_1, g\>. \label{eq:boxworld-tensor-max1-fy}
\end{align}
We can now express the positivity conditions $\< x, g_A \otimes g_B \> \geq 0$ in terms of $v_0$, $v_1$, and $v_2$. \eqref{eq:boxworld-tensor-maxfx} and \eqref{eq:boxworld-tensor-maxfy} imply $v_1, v_2 \in A(S)^{*+}$. \eqref{eq:boxworld-tensor-max1-fx} and \eqref{eq:boxworld-tensor-max1-fy} imply $v_0 + v_2 \in A(S)^{*+}$ and $v_0 + v_1 \in A(S)^{*+}$. Note that $v_0$ does not have to be an element of the positive cone $A(S)^{*+}$, this was not implied by any of the positivity conditions and, as we will see, it will not be.

Since $v_1, v_2 \in A(S)^{*+}$, there must be $\lambda_1, \lambda_2 \in \Rp$ and $y_1, y_2 \in S$ such that $v_1 = \lambda_1 y_1$, $v_2 = \lambda_2 y_2$. Since $v_0 \in A(S)$, there must be $\mu, \mu' \in \Rp$ and $z, z' \in S$ such that $v_0 = \mu z - \mu' z'$. The positivity conditions \eqref{eq:boxworld-tensor-max1-fx} and \eqref{eq:boxworld-tensor-max1-fy} then become
\begin{align}
\lambda_1 y_1 + \mu z - \mu' z' &\geq 0, \label{eq:boxworld-tensor-y1Cond} \\
\lambda_2 y_2 + \mu z - \mu' z' &\geq 0, \label{eq:boxworld-tensor-y2Cond}
\end{align}
and as a result of \eqref{eq:boxworld-tensor-trace} we must have
\begin{equation} \label{eq:boxworld-tensor-sumTrace}
\lambda_1 + \lambda_2 + \mu - \mu' = 1.
\end{equation}
Now we can prove the following lemmata:
\begin{lemma} \label{lemma:boxworld-tensor-ineqsToSep}
Let $x \in S \tmax S$ be a state given by \eqref{eq:boxworld-tensor-xDef}, i.e.,
\begin{equation}
x = s_{00} \otimes (\mu z - \mu' z') + \lambda_1 s_{10} \otimes y_1 + \lambda_2 s_{01} \otimes y_2,
\end{equation}
where $y_1, y_2, z, z' \in S$ and $\lambda_1, \lambda_2, \mu, \mu' \in \Rp$. If $\lambda_1 y_1 \geq \mu' z'$ and $\lambda_2 y_2 \geq \mu' z'$, then $x$ is separable, i.e., $x \in S \tmin S$.
\end{lemma}
\begin{proof}
Using \eqref{eq:boxworld-stateSpace-s11} we get
\begin{equation}
x = \mu s_{00} \otimes z + s_{10} \otimes (\lambda_1 y_1 - \mu' z') + s_{01} \otimes (\lambda_2 y_2 - \mu' z') + \mu'  s_{11} \otimes z'
\end{equation}
from which the result easily follows.
\end{proof}

\begin{lemma}
Let $x \in S \tmax S$ be a state given by \eqref{eq:boxworld-tensor-xDef}, i.e.,
\begin{equation}
x = s_{00} \otimes (\mu z - \mu' z') + \lambda_1 s_{10} \otimes y_1 + \lambda_2 s_{01} \otimes y_2,
\end{equation}
where $y_1, y_2, z, z' \in S$ and $\lambda_1, \lambda_2, \mu, \mu' \in \Rp$. $x$ is entangled, i.e., $x \notin S \tmin S$, only if the coefficients $\lambda_1$, $\lambda_2$, $\mu$, $\mu'$ are all non-zero.
\end{lemma}
\begin{proof}
If $\mu' = 0$ then $x$ is obviously separable. If $\lambda_1 = 0$ (or $\lambda_2 = 0$), then it follows from \eqref{eq:boxworld-tensor-y1Cond} (resp. from \eqref{eq:boxworld-tensor-y2Cond}) that $\mu z - \mu' z' \in A(S)^{*+}$ and so $x$ is separable. If $\mu = 0$, then then it follows from \eqref{eq:boxworld-tensor-y1Cond} and \eqref{eq:boxworld-tensor-y2Cond} that we have $\lambda_1 y_1 \geq \mu' z'$ and $\lambda_2 y_2 \geq \mu' z'$. The result follows from Lemma \ref{lemma:boxworld-tensor-ineqsToSep}.
\end{proof}

\begin{lemma}
Let $x \in S \tmax S$ be a state given by \eqref{eq:boxworld-tensor-xDef}, i.e.,
\begin{equation}
x = s_{00} \otimes (\mu z - \mu' z') + \lambda_1 s_{10} \otimes y_1 + \lambda_2 s_{01} \otimes y_2,
\end{equation}
where $y_1, y_2, z, z' \in S$ and $\lambda_1, \lambda_2, \mu, \mu' \in \Rp$. If $y_1 = y_2$, then $x$ is separable, i.e., $x \in S \tmin S$.
\end{lemma}
\begin{proof}
Let $y_1 = y_2 = y$ and without the loss of generality assume that $\lambda_1 \leq \lambda_2$. We have
\begin{align}
x &= s_{00} \otimes (\mu z - \mu' z') + \lambda_1 (s_{10} + s_{01}) \otimes y + (\lambda_2 - \lambda_1) s_{01} \otimes y \\
&= s_{00} \otimes (\lambda_1 y + \mu z - \mu' z') + \lambda_1 s_{11} \otimes y + (\lambda_2 - \lambda_1) s_{01} \otimes y.
\end{align}
It follows from \eqref{eq:boxworld-tensor-y1Cond} that $x$ is separable.
\end{proof}

One can reduce the positivity conditions \eqref{eq:boxworld-tensor-y1Cond} and \eqref{eq:boxworld-tensor-y2Cond} to just one condition. Take \eqref{eq:boxworld-tensor-y1Cond} and denote $\lambda_3 y_3 = \lambda_1 y_1 + \mu z - \mu' z'$, then we get
\begin{equation}
x = s_{00} \otimes (\lambda_3 y_3 - \lambda_1 y_1) + \lambda_1 s_{10} \otimes y_1 + \lambda_2 s_{01} \otimes y_2
\end{equation}
with the positivity condition $\lambda_2 y_2 + \lambda_3 y_3 - \lambda_1 y_1 \geq 0$. Note that the other positivity condition is trivial as we have $\lambda_1 y_1 + \lambda_3 y_3 - \lambda_1 y_1 = \lambda_3 y_3 \geq 0$. Moreover it follows from the normalization condition \eqref{eq:boxworld-tensor-sumTrace} that $\lambda_3 = 1-\lambda_2$. Thus we obtain the following:
\begin{proposition}
Every bipartite state $x \in S \tmax S$ is characterized by states $y_1, y_2, y_3 \in S$ and numbers $\lambda_1, \lambda_2 \in [0,1]$ such that
\begin{equation} \label{eq:boxworld-tensor-y3Cond}
\lambda_2 y_2 + (1-\lambda_2) y_3 - \lambda_1 y_1 \geq 0
\end{equation}
and the state is given as
\begin{equation} \label{eq:boxworld-tensor-xYform}
x = s_{00} \otimes ((1-\lambda_2) y_3 - \lambda_1 y_1) + \lambda_1 s_{10} \otimes y_1 + \lambda_2 s_{01} \otimes y_2.
\end{equation}
\end{proposition}
\begin{proof}
We have already showed that every state $x \in S \tmax S$ is of this form. Going in the other direction, it is straightforward to check that for any states $y_1, y_2, y_3 \in S$ and numbers $\lambda_1, \lambda_2 \in [0,1]$ satisfying \eqref{eq:boxworld-tensor-y3Cond}, \eqref{eq:boxworld-tensor-xYform} gives a valid bipartite state.
\end{proof}

\eqref{eq:boxworld-tensor-y3Cond} implies that there is $y_4 \in S$ such that $\lambda_2 y_2 + (1-\lambda_2) y_3 = \lambda_1 y_1 + (1-\lambda_1) y_4$. The states $y_i$, $i \in \{1, \ldots, 4\}$ form a tetragon (polygon with four vertexes) inside $S$. These polygons corresponds to positive maps $\Psi: A(S)^{*+} \to A(S)^{*+}$ defined as
\begin{align}
&\dfrac{1}{2} \Psi(s_{00}) = \lambda_1 y_1,
&&\dfrac{1}{2} \Psi(s_{10}) = \lambda_2 y_2,
&&\dfrac{1}{2} \Psi(s_{01}) = (1-\lambda_2) y_3.
\end{align}
Note that $\Psi$ is in general not a channel, because for $s \in S$, in general $\< \Psi(s), 1_S \> \neq 1$. Instead of that we have a weaker condition $\dfrac{1}{2} \< \Psi(s_{10}) + \Psi(s_{01}), 1_S \> = 1$. Also note that using \eqref{eq:boxworld-cross} we get
\begin{equation}
\dfrac{1}{2} \Psi(s_{11}) = \dfrac{1}{2} \left( \Psi(s_{10}) + \Psi(s_{01}) - \Psi(s_{00}) \right) = \lambda_2 y_2 + (1-\lambda_2) y_3 - \lambda_1 y_1 = (1-\lambda_1) y_4.
\end{equation}
Let $x_0 \in S \tmax S$ be given as
\begin{equation} \label{eq:boxworld-tensor-x0}
x_0 = \dfrac{1}{2} \left( s_{00} \otimes (s_{01} - s_{00}) + s_{10} \otimes s_{00} + s_{01} \otimes s_{10} \right)
\end{equation}
then we have
\begin{equation}
x = (\id_S \otimes \Psi)(x_0) = \dfrac{1}{2} \left( s_{00} \otimes ( \Psi(s_{01}) - \Psi(s_{00})) + s_{10} \otimes \Psi(s_{00}) + s_{01} \otimes \Psi(s_{10}) \right)
\end{equation}
or in diagrams
\begin{equation} \label{eq:boxworld-tensor-xPsi}
\begin{quantikz}[row sep=\the\rowsep, align equals at=1.5]
&\multiprepareC[2]{x} &\qw \\
& &\qw
\end{quantikz}
=
\begin{quantikz}[row sep=\the\rowsep, align equals at=1.5]
&[\prepfix]\multiprepareC[2]{x_0} &\qw &\qw \\
& &\gate{\Psi} &\qw
\end{quantikz}
\end{equation}
We will now formalize our results.
\begin{proposition} \label{prop:boxworld-tensor-PsiRel}
For every state $x \in S \tmax S$, there is a positive map $\Psi: A(S)^{*+} \to A(S)^{*+}$ such that
\begin{equation} \label{eq:boxworld-tensor-PsiCond}
\dfrac{1}{2} \< \Psi(s_{10}) + \Psi(s_{01}), 1_S \> = 1.
\end{equation}
such that \eqref{eq:boxworld-tensor-xPsi} holds, and, vice-versa, for every positive map $\Psi: A(S)^{*+} \to A(S)^{*+}$ satisfying \eqref{eq:boxworld-tensor-PsiCond} there is a state $x \in S \tmax S$ such that \eqref{eq:boxworld-tensor-xPsi} holds.
\end{proposition}
\begin{proof}
We already know that to every $x \in S \tmax S$ we can find the corresponding map $\Psi: A(S)^{*+} \to A(S)^{*+}$ such that \eqref{eq:boxworld-tensor-xPsi} holds. So let $\Psi: A(S)^{*+} \to A(S)^{*+}$ be a positive map such that \eqref{eq:boxworld-tensor-PsiCond} holds. Since $\Psi$ is positive, it follows from Proposition \ref{prop:channels-CP-maxTensor} that $(\id_S \otimes \Psi)(x_0) \in A(S \tmax S)^{*+}$ and we only need to check the normalization. We have
\begin{equation}
\begin{quantikz}[row sep=\the\rowsep, align equals at=1.5]
&\multiprepareC[2]{x_0} &\qw &\ground{} \\
& &\gate{\Psi} &\qw
\end{quantikz}
= \dfrac{1}{2} ( \Psi(s_{01}) - \Psi(s_{00}) + \Psi(s_{00}) + \Psi(s_{10}) ) = \dfrac{1}{2} ( \Psi(s_{01}) + \Psi(s_{10}) )
\end{equation}
and the normalization of $(\id_S \otimes \Psi)(x_0)$ follows from \eqref{eq:boxworld-tensor-PsiCond}. Therefore we have $(\id_S \otimes \Psi)(x_0) \in S \tmax S$.
\end{proof}
One can spot certain similarity between the state $x_0 \in S \tmax S$ and the maximally entangled state $\ketbra{\phi^+} \in \dens(\Ha \otimes \Ha)$ in the sense that both states are used to construct a correspondence between entangled states and positive maps, or completely positive channels. In fact, this similarity is not a coincidence, but it stems from a shared property of both $S$ and $\dens(\Ha)$: isomorphism between $A(K)^+$ and $A(K)^{*+}$. We have already argued that the cones $A(\dens(\Ha))^+$ and $A(\dens(\Ha))^{*+}$ are isomorphic in Section \ref{sec:QT}. One can construct similar isomorphism between $A(S)^+$ and $A(S)^{*+}$ as follows: let $\iota: A(S) \to A(S)^*$ be defined as
\begin{align}
&\iota(f_x) = s_{00},
&&\iota(f_y) = s_{10},
&&\iota(1_S - f_y) = s_{01}.
\end{align}
Then clearly $\iota: A(S)^+ \to A(S)^{*+}$, i.e., $\iota$ is a positive map and it is also straightforward to show that $\iota$ is invertible. Hence the cones $A(S)^+$ and $A(S)^{*+}$ are isomorphic. Then one can use the result on the structure of channels from Proposition \ref{prop:channels-channels-tensorSubset} together with the isomorphism between $A(K)^+$ and $A(K)^{*+}$ to construct the correspondence.

\subsection{Channels}
Since we use either minimal or maximal tensor product, we know from Propositions \ref{prop:channels-CP-minTensor} and \ref{prop:channels-CP-maxTensor} that all positive channels are completely positive. The channels that are often used are the isomorphisms of the state space. These are the channels that are invertible and the inverse map is a channel as well. The isomorphisms correspond to rotations and reflections of the state space. For example, consider the channel $R: S \to S$ given as
\begin{align}
&R(s_{00}) = s_{10},
&&R(s_{10}) = s_{11},
&&R(s_{01}) = s_{00}.
\end{align}
It then follows that
\begin{equation}
R(s_{11}) = R(s_{10}) + R(s_{01}) - R(s_{00}) = s_{11} + s_{00} - s_{10} = s_{01}.
\end{equation}
It easily follows that $R^4 = \id_S$ and so we get that the channel $R$ is an isomorphism. Another such isomorphism is $M:S \to S$ given as
\begin{align}
&M(s_{00}) = s_{11},
&&M(s_{10}) = s_{10},
&&M(s_{01}) = s_{01}.
\end{align}
Then we have $M(s_{11}) = s_{00}$. We again have $M^2 = \id_S$. These isomorphism generate the whole group of isomorphisms of $S$. One can also relate the isomorphisms of $S$ to the pure entangled states in $S \tmax S$ by using the result of Proposition \ref{prop:boxworld-tensor-PsiRel}.

\subsection{Compatibility of channels}
Compatibility of the measurements on the square state space was investigated before \cite{BuschHeinosaariSchultzStevens-compatibility, JencovaPlavala-maxInc} and one can show that the two-outcome measurements corresponding to the effects $f_x$ and $f_y$ are maximally incompatible, meaning that they are as incompatible as mathematically possible. This is closely related to the maximal violations of CHSH inequality, see \cite{PlavalaZiman-PRbox,JencovaPlavala-PRbox}.

\section*{Acknowledgement}
The author is thankful to Teiko Heinosaari and Matthias Kleinmann for comments on the early version of the manuscript and to Thomas Bullock and Peter Morgan for helpful discussions. The author acknowledges the support by the Deutsche Forschungsgemeinschaft (DFG, GermanResearch Foundation - 447948357) and the ERC (Consolidator Grant 683107/TempoQ).

{
	\interlinepenalty=10000
	\bibliographystyle{elsarticle-num}
	\bibliography{citations}

\begin{thebibliography}{100}
\expandafter\ifx\csname url\endcsname\relax
  \def\url#1{\texttt{#1}}\fi
\expandafter\ifx\csname urlprefix\endcsname\relax\def\urlprefix{URL }\fi
\expandafter\ifx\csname href\endcsname\relax
  \def\href#1#2{#2} \def\path#1{#1}\fi

\bibitem{Barrett-GPTinformation}
J.~Barrett, {Information processing in generalized probabilistic theories},
  Physical Review A 75~(3) (2007) 032304.
\newblock \href {http://arxiv.org/abs/0508211} {\path{arXiv:0508211}}, \href
  {https://doi.org/10.1103/PhysRevA.75.032304}
  {\path{doi:10.1103/PhysRevA.75.032304}}.

\bibitem{BarnumGaeblerWilce-steering}
H.~Barnum, C.~P. Gaebler, A.~Wilce, {Ensemble Steering, Weak Self-Duality, and
  the Structure of Probabilistic Theories}, Foundations of Physics 43~(12)
  (2013) 1411--1427.
\newblock \href {http://arxiv.org/abs/0912.5532} {\path{arXiv:0912.5532}},
  \href {https://doi.org/10.1007/s10701-013-9752-2}
  {\path{doi:10.1007/s10701-013-9752-2}}.

\bibitem{Banik-steering}
M.~Banik, {Measurement incompatibility and
  Schr{\"{o}}dinger-Einstein-Podolsky-Rosen steering in a class of
  probabilistic theories}, Journal of Mathematical Physics 56~(5) (2015)
  052101.
\newblock \href {http://arxiv.org/abs/1502.05779} {\path{arXiv:1502.05779}},
  \href {https://doi.org/10.1063/1.4919546} {\path{doi:10.1063/1.4919546}}.

\bibitem{KarGhoshChoudharyBanik-uncertainty}
G.~Kar, S.~Ghosh, S.~Choudhary, M.~Banik, {Role of Measurement Incompatibility
  and Uncertainty in Determining Nonlocality}, Mathematics 4~(3) (2016) 52.
\newblock \href {https://doi.org/10.3390/math4030052}
  {\path{doi:10.3390/math4030052}}.

\bibitem{Plavala-channels}
M.~Pl{\'{a}}vala, {Conditions for the compatibility of channels in general
  probabilistic theory and their connection to steering and Bell nonlocality},
  Physical Review A 96~(5) (2017) 052127.
\newblock \href {http://arxiv.org/abs/1707.08650} {\path{arXiv:1707.08650}},
  \href {https://doi.org/10.1103/PhysRevA.96.052127}
  {\path{doi:10.1103/PhysRevA.96.052127}}.

\bibitem{PlavalaZiman-PRbox}
M.~Pl{\'{a}}vala, M.~Ziman, {Popescu-Rohrlich box implementation in general
  probabilistic theory of processes}, Physics Letters A 384~(16) (2020) 126323.
\newblock \href {http://arxiv.org/abs/1708.07425} {\path{arXiv:1708.07425}},
  \href {https://doi.org/10.1016/j.physleta.2020.126323}
  {\path{doi:10.1016/j.physleta.2020.126323}}.

\bibitem{JencovaPlavala-PRbox}
A.~Jen{\v{c}}ov{\'{a}}, M.~Pl{\'{a}}vala, {Structure of quantum and classical
  implementations of the Popescu-Rohrlich box}, Physical Review A 102~(4)
  (2020) 042208.
\newblock \href {http://arxiv.org/abs/1907.08933} {\path{arXiv:1907.08933}},
  \href {https://doi.org/10.1103/PhysRevA.102.042208}
  {\path{doi:10.1103/PhysRevA.102.042208}}.

\bibitem{BhattacharyaSahaGuhaBanik-nonlocality}
S.~S. Bhattacharya, S.~Saha, T.~Guha, M.~Banik, {Nonlocality without
  entanglement: Quantum theory and beyond}, Physical Review Research 2~(1)
  (2020) 012068.
\newblock \href {http://arxiv.org/abs/1908.10676} {\path{arXiv:1908.10676}},
  \href {https://doi.org/10.1103/PhysRevResearch.2.012068}
  {\path{doi:10.1103/PhysRevResearch.2.012068}}.

\bibitem{CzekajHorodeckiTylec-bipartiteEffects}
L.~Czekaj, M.~Horodecki, T.~Tylec, {Bell measurement ruling out supraquantum
  correlations}, Physical Review A 98~(3) (2018) 032117.
\newblock \href {http://arxiv.org/abs/1802.09510} {\path{arXiv:1802.09510}},
  \href {https://doi.org/10.1103/PhysRevA.98.032117}
  {\path{doi:10.1103/PhysRevA.98.032117}}.

\bibitem{PopescuRohrlich-PRbox}
S.~Popescu, D.~Rohrlich, {Quantum nonlocality as an axiom}, Foundations of
  Physics 24~(3) (1994) 379--385.
\newblock \href {https://doi.org/10.1007/BF02058098}
  {\path{doi:10.1007/BF02058098}}.

\bibitem{WisemanDohertyJones-nonlocal}
H.~M. Wiseman, S.~J. Jones, A.~C. Doherty, {Steering, entanglement,
  nonlocality, and the Einstein-Podolsky-Rosen paradox}, Physical Review
  Letters 98~(14) (2007) 140402.
\newblock \href {http://arxiv.org/abs/0612147} {\path{arXiv:0612147}}, \href
  {https://doi.org/10.1103/PhysRevLett.98.140402}
  {\path{doi:10.1103/PhysRevLett.98.140402}}.

\bibitem{DahlstenGarnerVedral-uncertainty}
O.~C.~O. Dahlsten, A.~J.~P. Garner, V.~Vedral, {The uncertainty principle
  enables non-classical dynamics in an interferometer}, Nature Communications
  5~(1) (2014) 4592.
\newblock \href {http://arxiv.org/abs/1206.5702} {\path{arXiv:1206.5702}},
  \href {https://doi.org/10.1038/ncomms5592} {\path{doi:10.1038/ncomms5592}}.

\bibitem{SahaOszmaniecCzekajHorodeckiHorodecki-uncertainty}
D.~Saha, M.~Oszmaniec, L.~Czekaj, M.~Horodecki, R.~Horodecki, {Operational
  foundations for complementarity and uncertainty relations}, Physical Review A
  101~(5) (2020) 052104.
\newblock \href {http://arxiv.org/abs/1809.03475} {\path{arXiv:1809.03475}},
  \href {https://doi.org/10.1103/PhysRevA.101.052104}
  {\path{doi:10.1103/PhysRevA.101.052104}}.

\bibitem{TakakuraMiyadera-uncertainty}
R.~Takakura, T.~Miyadera, {Preparation uncertainty implies measurement
  uncertainty in a class of generalized probabilistic theories}, Journal of
  Mathematical Physics 61~(8) (2020) 082203.
\newblock \href {http://arxiv.org/abs/2006.02092} {\path{arXiv:2006.02092}},
  \href {https://doi.org/10.1063/5.0017854} {\path{doi:10.1063/5.0017854}}.

\bibitem{TakakuraMyiadera-entropicUncertainty}
R.~Takakura, T.~Miyadera, {Entropic Uncertainty Relations in a Class of
  Generalized Probabilistic Theories} (2020).
\newblock \href {http://arxiv.org/abs/2006.05671} {\path{arXiv:2006.05671}}.

\bibitem{SunLiangZhouKwekYu=uncertaintyRelations}
L.~L. Sun, X.~Zhou, L.-c. Kwek, S.~Yu, {No Disturbance Without Uncertainty
  Under Generalized Probability Theory} (2020).

\bibitem{BarnumBarrettLeiferWilce-noBroadcasting}
H.~Barnum, J.~Barrett, M.~Leifer, A.~Wilce, {Cloning and Broadcasting in
  Generic Probabilistic Theories} (2006).
\newblock \href {http://arxiv.org/abs/0611295} {\path{arXiv:0611295}}.

\bibitem{BarnumBarrettLeiferWilce-noBroadcastingPRL}
H.~Barnum, J.~Barrett, M.~Leifer, A.~Wilce, {Generalized No-broadcasting
  theorem}, Physical Review Letters 99~(24) (2007) 240501.
\newblock \href {http://arxiv.org/abs/0707.0620} {\path{arXiv:0707.0620}},
  \href {https://doi.org/10.1103/PhysRevLett.99.240501}
  {\path{doi:10.1103/PhysRevLett.99.240501}}.

\bibitem{BuschHeinosaariSchultzStevens-compatibility}
P.~Busch, T.~Heinosaari, J.~Schultz, N.~Stevens, {Comparing the degrees of
  incompatibility inherent in probabilistic physical theories}, Europhysics
  Letters 103~(1) (2013) 10002.
\newblock \href {http://arxiv.org/abs/1210.4142} {\path{arXiv:1210.4142}},
  \href {https://doi.org/10.1209/0295-5075/103/10002}
  {\path{doi:10.1209/0295-5075/103/10002}}.

\bibitem{Plavala-simplex}
M.~Pl{\'{a}}vala, {All measurements in a probabilistic theory are compatible if
  and only if the state space is a simplex}, Physical Review A 94~(4) (2016)
  042108.
\newblock \href {http://arxiv.org/abs/1608.05614} {\path{arXiv:1608.05614}},
  \href {https://doi.org/10.1103/PhysRevA.94.042108}
  {\path{doi:10.1103/PhysRevA.94.042108}}.

\bibitem{JencovaPlavala-maxInc}
A.~Jen{\v{c}}ov{\'{a}}, M.~Pl{\'{a}}vala, {Conditions on the existence of
  maximally incompatible two-outcome measurements in general probabilistic
  theory}, Physical Review A 96~(2) (2017) 022113.
\newblock \href {http://arxiv.org/abs/1703.09447} {\path{arXiv:1703.09447}},
  \href {https://doi.org/10.1103/PhysRevA.96.022113}
  {\path{doi:10.1103/PhysRevA.96.022113}}.

\bibitem{FilippovHeinosaariLeppajarvi-compatibility}
S.~N. Filippov, T.~Heinosaari, L.~Lepp{\"{a}}j{\"{a}}rvi, {Necessary condition
  for incompatibility of observables in general probabilistic theories},
  Physical Review A 95~(3) (2017) 032127.
\newblock \href {http://arxiv.org/abs/1609.08416} {\path{arXiv:1609.08416}},
  \href {https://doi.org/10.1103/PhysRevA.95.032127}
  {\path{doi:10.1103/PhysRevA.95.032127}}.

\bibitem{HeinosaariLeppajarviPlavala-noFreeInformation}
T.~Heinosaari, L.~Lepp{\"{a}}j{\"{a}}rvi, M.~Pl{\'{a}}vala,
  {No-free-information principle in general probabilistic theories}, Quantum 3
  (2018) 157.
\newblock \href {http://arxiv.org/abs/1808.07376} {\path{arXiv:1808.07376}},
  \href {https://doi.org/10.22331/q-2019-07-08-157}
  {\path{doi:10.22331/q-2019-07-08-157}}.

\bibitem{Jencova-incomaptibility}
A.~Jen{\v{c}}ov{\'{a}}, {Incompatible measurements in a class of general
  probabilistic theories}, Physical Review A 98 (2018) 012133.
\newblock \href {http://arxiv.org/abs/1705.08008} {\path{arXiv:1705.08008}},
  \href {https://doi.org/10.1103/PhysRevA.98.012133}
  {\path{doi:10.1103/PhysRevA.98.012133}}.

\bibitem{Kuramochi-simplex}
Y.~Kuramochi, {Compatibility of any pair of 2-outcome measurements
  characterizes the Choquet simplex}, Positivity (2020).
\newblock \href {http://arxiv.org/abs/1912.00563} {\path{arXiv:1912.00563}},
  \href {https://doi.org/10.1007/s11117-020-00742-0}
  {\path{doi:10.1007/s11117-020-00742-0}}.

\bibitem{CzekajSainzSelbyHorodecki-compositeMeasurements}
{\L}.~Czekaj, A.~B. Sainz, J.~Selby, M.~Horodecki, {Correlations constrained by
  composite measurements} (2020).
\newblock \href {http://arxiv.org/abs/2009.04994} {\path{arXiv:2009.04994}}.

\bibitem{BluhmJencovaNechita-spectrahedra}
A.~Bluhm, A.~Jen{\v{c}}ov{\'{a}}, I.~Nechita, {Incompatibility in general
  probabilistic theories, generalized spectrahedra, and tensor norms} (2020).
\newblock \href {http://arxiv.org/abs/2011.06497} {\path{arXiv:2011.06497}}.

\bibitem{Spekkens-contextuality}
R.~W. Spekkens, {Contextuality for preparations, transformations, and unsharp
  measurements}, Physical Review A 71~(5) (2005) 052108.
\newblock \href {http://arxiv.org/abs/0406166} {\path{arXiv:0406166}}, \href
  {https://doi.org/10.1103/PhysRevA.71.052108}
  {\path{doi:10.1103/PhysRevA.71.052108}}.

\bibitem{Spekkens-toyTheory}
R.~W. Spekkens, {Evidence for the epistemic view of quantum states: A toy
  theory}, Physical Review A 75~(3) (2007) 032110.
\newblock \href {http://arxiv.org/abs/0401052} {\path{arXiv:0401052}}, \href
  {https://doi.org/10.1103/PhysRevA.75.032110}
  {\path{doi:10.1103/PhysRevA.75.032110}}.

\bibitem{ChiribellaYuan-contextuality}
G.~Chiribella, X.~Yuan, {Measurement sharpness cuts nonlocality and
  contextuality in every physical theory} (2014).
\newblock \href {http://arxiv.org/abs/1404.3348} {\path{arXiv:1404.3348}}.

\bibitem{SchmidSelbyWolfeKunjwalSpekkens-noncontextuality}
D.~Schmid, J.~H. Selby, E.~Wolfe, R.~Kunjwal, R.~W. Spekkens, {Characterization
  of Noncontextuality in the Framework of Generalized Probabilistic Theories},
  PRX Quantum 2~(1) (2021) 010331.
\newblock \href {http://arxiv.org/abs/1911.10386} {\path{arXiv:1911.10386}},
  \href {https://doi.org/10.1103/PRXQuantum.2.010331}
  {\path{doi:10.1103/PRXQuantum.2.010331}}.

\bibitem{SchmidSelbyPuseySpekkens-nonconModels}
D.~Schmid, J.~H. Selby, M.~F. Pusey, R.~W. Spekkens, {A structure theorem for
  generalized-noncontextual ontological models} (2020).
\newblock \href {http://arxiv.org/abs/2005.07161} {\path{arXiv:2005.07161}}.

\bibitem{WeilenmannColbeck-causalStructures}
M.~Weilenmann, R.~Colbeck, {Analysing causal structures in generalised
  probabilistic theories}, Quantum 4 (2020) 236.
\newblock \href {http://arxiv.org/abs/1812.04327} {\path{arXiv:1812.04327}},
  \href {https://doi.org/10.22331/q-2020-02-27-236}
  {\path{doi:10.22331/q-2020-02-27-236}}.

\bibitem{ScandoloSalazarKorbiczHorodecki-objectivity}
C.~M. Scandolo, R.~Salazar, J.~K. Korbicz, P.~Horodecki, {The origin of
  objectivity in all fundamental causal theories} (2018).
\newblock \href {http://arxiv.org/abs/1805.12126} {\path{arXiv:1805.12126}}.

\bibitem{GrossMullerColbeckDahlsten-boxworldDynamics}
D.~Gross, M.~M{\"{u}}ller, R.~Colbeck, O.~C.~O. Dahlsten, {All Reversible
  Dynamics in Maximally Nonlocal Theories are Trivial}, Physical Review Letters
  104~(8) (2010) 080402.
\newblock \href {http://arxiv.org/abs/0910.1840} {\path{arXiv:0910.1840}},
  \href {https://doi.org/10.1103/PhysRevLett.104.080402}
  {\path{doi:10.1103/PhysRevLett.104.080402}}.

\bibitem{AlSafiShort-boxworldDynamics}
S.~W. Al-Safi, A.~J. Short, {Reversible dynamics in strongly non-local Boxworld
  systems}, Journal of Physics A: Mathematical and Theoretical 47~(32) (2014)
  325303.
\newblock \href {http://arxiv.org/abs/1312.3931} {\path{arXiv:1312.3931}},
  \href {https://doi.org/10.1088/1751-8113/47/32/325303}
  {\path{doi:10.1088/1751-8113/47/32/325303}}.

\bibitem{AlSafiRichens-reversibleDynamics}
S.~W. Al-Safi, J.~Richens, {Reversibility and the structure of the local state
  space}, New Journal of Physics 17~(12) (2015) 123001.
\newblock \href {http://arxiv.org/abs/1508.03491} {\path{arXiv:1508.03491}},
  \href {https://doi.org/10.1088/1367-2630/17/12/123001}
  {\path{doi:10.1088/1367-2630/17/12/123001}}.

\bibitem{BranfordDahlstenGarner-dynamicsInGPTs}
D.~Branford, O.~C.~O. Dahlsten, A.~J.~P. Garner, {On Defining the Hamiltonian
  Beyond Quantum Theory}, Foundations of Physics 48~(8) (2018) 982--1006.
\newblock \href {http://arxiv.org/abs/1808.05404} {\path{arXiv:1808.05404}},
  \href {https://doi.org/10.1007/s10701-018-0205-9}
  {\path{doi:10.1007/s10701-018-0205-9}}.

\bibitem{GalleyMasanes-dynamics}
T.~D. Galley, L.~Masanes, {How dynamics constrains probabilities in general
  probabilistic theories} (2020).
\newblock \href {http://arxiv.org/abs/2002.05088} {\path{arXiv:2002.05088}}.

\bibitem{GarnerDahlstenNakataMurao-phaseInGPTs}
A.~J.~P. Garner, O.~C.~O. Dahlsten, Y.~Nakata, M.~Murao, V.~Vedral, {A
  framework for phase and interference in generalized probabilistic theories},
  New Journal of Physics 15~(9) (2013) 093044.
\newblock \href {https://doi.org/10.1088/1367-2630/15/9/093044}
  {\path{doi:10.1088/1367-2630/15/9/093044}}.

\bibitem{DahlstenGarnerThompsonGuVedral-particleExchange}
O.~C.~O. Dahlsten, A.~J.~P. Garner, J.~Thompson, M.~Gu, V.~Vedral, {Particle
  exchange in post-quantum theories} (2013).
\newblock \href {http://arxiv.org/abs/1307.2529} {\path{arXiv:1307.2529}}.

\bibitem{GalleyGiacominiSelby-gravity}
T.~D. Galley, F.~Giacomini, J.~H. Selby, {A no-go theorem on the nature of the
  gravitational field beyond quantum theory} (2020).
\newblock \href {http://arxiv.org/abs/2012.01441} {\path{arXiv:2012.01441}}.

\bibitem{UdudecBarnumEmerson-interference}
C.~Ududec, H.~Barnum, J.~Emerson, {Three Slit Experiments and the Structure of
  Quantum Theory}, Foundations of Physics 41~(3) (2011) 396--405.
\newblock \href {http://arxiv.org/abs/0909.4787} {\path{arXiv:0909.4787}},
  \href {https://doi.org/10.1007/s10701-010-9429-z}
  {\path{doi:10.1007/s10701-010-9429-z}}.

\bibitem{BarnumLeeScandoloSelby-higherOrderInterference}
H.~Barnum, C.~Lee, C.~Scandolo, J.~Selby, {Ruling out Higher-Order Interference
  from Purity Principles}, Entropy 19~(6) (2017) 253.
\newblock \href {http://arxiv.org/abs/1704.05106} {\path{arXiv:1704.05106}},
  \href {https://doi.org/10.3390/e19060253} {\path{doi:10.3390/e19060253}}.

\bibitem{Kleinmann-multipleSlit}
M.~Kleinmann, {Sequences of projective measurements in generalized
  probabilistic models}, Journal of Physics A: Mathematical and Theoretical
  47~(45) (2014) 455304.
\newblock \href {http://arxiv.org/abs/1402.3583} {\path{arXiv:1402.3583}},
  \href {https://doi.org/10.1088/1751-8113/47/45/455304}
  {\path{doi:10.1088/1751-8113/47/45/455304}}.

\bibitem{DakicPaterekBrukner-densityCubes}
B.~Daki{\'{c}}, T.~Paterek, {\v{C}}.~Brukner, {Density cubes and higher-order
  interference theories}, New Journal of Physics 16~(2) (2014) 023028.
\newblock \href {http://arxiv.org/abs/1308.2822} {\path{arXiv:1308.2822}},
  \href {https://doi.org/10.1088/1367-2630/16/2/023028}
  {\path{doi:10.1088/1367-2630/16/2/023028}}.

\bibitem{LeeSelby-multipleSlit}
C.~M. Lee, J.~H. Selby, {Higher-Order Interference in Extensions of Quantum
  Theory}, Foundations of Physics 47~(1) (2017) 89--112.
\newblock \href {http://arxiv.org/abs/1510.03860} {\path{arXiv:1510.03860}},
  \href {https://doi.org/10.1007/s10701-016-0045-4}
  {\path{doi:10.1007/s10701-016-0045-4}}.

\bibitem{LeeSelby-phaseKickback}
C.~M. Lee, J.~H. Selby, {Generalised phase kick-back: the structure of
  computational algorithms from physical principles}, New Journal of Physics
  18~(3) (2016) 033023.
\newblock \href {http://arxiv.org/abs/1510.04699} {\path{arXiv:1510.04699}},
  \href {https://doi.org/10.1088/1367-2630/18/3/033023}
  {\path{doi:10.1088/1367-2630/18/3/033023}}.

\bibitem{BarnumMullerUdudec-interferenceDerivatonQT}
H.~Barnum, M.~P. M{\"{u}}ller, C.~Ududec, {Higher-order interference and
  single-system postulates characterizing quantum theory}, New Journal of
  Physics 16~(12) (2014) 123029.
\newblock \href {http://arxiv.org/abs/1403.4147} {\path{arXiv:1403.4147}},
  \href {https://doi.org/10.1088/1367-2630/16/12/123029}
  {\path{doi:10.1088/1367-2630/16/12/123029}}.

\bibitem{HorvatDakic-interference}
S.~Horvat, B.~Daki{\'{c}}, {Interference as an information-theoretic game}
  (2020).
\newblock \href {http://arxiv.org/abs/2003.12114} {\path{arXiv:2003.12114}}.

\bibitem{BarnumLeeSelby-oraclesComputation}
H.~Barnum, C.~M. Lee, J.~H. Selby, {Oracles and Query Lower Bounds in
  Generalised Probabilistic Theories}, Foundations of Physics 48~(8) (2018)
  954--981.
\newblock \href {http://arxiv.org/abs/1704.05043} {\path{arXiv:1704.05043}},
  \href {https://doi.org/10.1007/s10701-018-0198-4}
  {\path{doi:10.1007/s10701-018-0198-4}}.

\bibitem{LeeHoban-proofs}
C.~M. Lee, M.~J. Hoban, {Bounds on the power of proofs and advice in general
  physical theories}, Proceedings of the Royal Society A: Mathematical,
  Physical and Engineering Sciences 472~(2190) (2016) 20160076.
\newblock \href {http://arxiv.org/abs/1510.04702} {\path{arXiv:1510.04702}},
  \href {https://doi.org/10.1098/rspa.2016.0076}
  {\path{doi:10.1098/rspa.2016.0076}}.

\bibitem{LeeSelby-Grover}
C.~M. Lee, J.~H. Selby, {Deriving Grover's lower bound from simple physical
  principles}, New Journal of Physics 18~(9) (2016) 093047.
\newblock \href {http://arxiv.org/abs/1604.03118} {\path{arXiv:1604.03118}},
  \href {https://doi.org/10.1088/1367-2630/18/9/093047}
  {\path{doi:10.1088/1367-2630/18/9/093047}}.

\bibitem{LeeBarrett-computation}
C.~M. Lee, J.~Barrett, {Computation in generalised probabilisitic theories},
  New Journal of Physics 17~(8) (2015) 083001.
\newblock \href {http://arxiv.org/abs/1412.8671} {\path{arXiv:1412.8671}},
  \href {https://doi.org/10.1088/1367-2630/17/8/083001}
  {\path{doi:10.1088/1367-2630/17/8/083001}}.

\bibitem{Garner-computation}
A.~J.~P. Garner, {Interferometric Computation Beyond Quantum Theory},
  Foundations of Physics 48~(8) (2018) 886--909.
\newblock \href {http://arxiv.org/abs/1610.04349} {\path{arXiv:1610.04349}},
  \href {https://doi.org/10.1007/s10701-018-0142-7}
  {\path{doi:10.1007/s10701-018-0142-7}}.

\bibitem{KrummMuller-computation}
M.~Krumm, M.~P. M{\"{u}}ller, {Quantum computation is the unique reversible
  circuit model for which bits are balls}, npj Quantum Information 5~(1) (2019)
  7.
\newblock \href {http://arxiv.org/abs/1804.05736} {\path{arXiv:1804.05736}},
  \href {https://doi.org/10.1038/s41534-018-0123-x}
  {\path{doi:10.1038/s41534-018-0123-x}}.

\bibitem{BarrettbeaudrapHobanLee-computation}
J.~Barrett, N.~de~Beaudrap, M.~J. Hoban, C.~M. Lee, {The computational
  landscape of general physical theories}, npj Quantum Information 5~(1) (2019)
  41.
\newblock \href {http://arxiv.org/abs/1702.08483} {\path{arXiv:1702.08483}},
  \href {https://doi.org/10.1038/s41534-019-0156-9}
  {\path{doi:10.1038/s41534-019-0156-9}}.

\bibitem{TakagiRegula-resourceTheories}
R.~Takagi, B.~Regula, {General Resource Theories in Quantum Mechanics and
  Beyond: Operational Characterization via Discrimination Tasks}, Physical
  Review X 9~(3) (2019) 031053.
\newblock \href {http://arxiv.org/abs/1901.08127} {\path{arXiv:1901.08127}},
  \href {https://doi.org/10.1103/PhysRevX.9.031053}
  {\path{doi:10.1103/PhysRevX.9.031053}}.

\bibitem{LamiRegulaTakagiFerrari-resourceTheories}
L.~Lami, B.~Regula, R.~Takagi, G.~Ferrari, {Framework for resource
  quantification in infinite-dimensional general probabilistic theories},
  Physical Review A 103~(3) (2021) 032424.
\newblock \href {http://arxiv.org/abs/2009.11313} {\path{arXiv:2009.11313}},
  \href {https://doi.org/10.1103/PhysRevA.103.032424}
  {\path{doi:10.1103/PhysRevA.103.032424}}.

\bibitem{BarnumWilce-informationPrcessing}
H.~Barnum, A.~Wilce, {Information processing in convex operational theories},
  Electronic Notes in Theoretical Computer Science 270~(1) (2011) 3--15.
\newblock \href {http://arxiv.org/abs/0908.2352} {\path{arXiv:0908.2352}},
  \href {https://doi.org/10.1016/j.entcs.2011.01.002}
  {\path{doi:10.1016/j.entcs.2011.01.002}}.

\bibitem{BarnumBarrettLeiferWilce-teleportation}
{Howard Barnum}, {Jonathan Barrett}, {Matthew Leifer}, {Alexander Wilce},
  {Teleportation in general probabilistic theories}, 2012, pp. 25--47.
\newblock \href {http://arxiv.org/abs/0805.3553} {\path{arXiv:0805.3553}},
  \href {https://doi.org/10.1090/psapm/071/600}
  {\path{doi:10.1090/psapm/071/600}}.

\bibitem{MullerUdudec-computation}
M.~P. M{\"{u}}ller, C.~Ududec, {Structure of Reversible Computation Determines
  the Self-Duality of Quantum Theory}, Physical Review Letters 108~(13) (2012)
  130401.
\newblock \href {http://arxiv.org/abs/1110.3516} {\path{arXiv:1110.3516}},
  \href {https://doi.org/10.1103/PhysRevLett.108.130401}
  {\path{doi:10.1103/PhysRevLett.108.130401}}.

\bibitem{MullerDahlstenVedral-coinTossing}
M.~P. M{\"{u}}ller, O.~C.~O. Dahlsten, V.~Vedral, {Unifying Typical
  Entanglement and Coin Tossing: on Randomization in Probabilistic Theories},
  Communications in Mathematical Physics 316~(2) (2012) 441--487.
\newblock \href {http://arxiv.org/abs/1107.6029} {\path{arXiv:1107.6029}},
  \href {https://doi.org/10.1007/s00220-012-1605-x}
  {\path{doi:10.1007/s00220-012-1605-x}}.

\bibitem{Chiribella-dilation}
G.~Chiribella, {Dilation of states and processes in operational-probabilistic
  theories}, Electronic Proceedings in Theoretical Computer Science 172~(Qpl)
  (2014) 1--14.
\newblock \href {http://arxiv.org/abs/1412.8102} {\path{arXiv:1412.8102}},
  \href {https://doi.org/10.4204/EPTCS.172.1} {\path{doi:10.4204/EPTCS.172.1}}.

\bibitem{BarnumDahlstenLeiferToner-bitCommitment}
H.~Barnum, O.~C. Dahlsten, M.~Leifer, B.~Toner, {Nonclassicality without
  entanglement enables bit commitment}, in: 2008 IEEE Information Theory
  Workshop, IEEE, 2008, pp. 386--390.
\newblock \href {http://arxiv.org/abs/0803.1264} {\path{arXiv:0803.1264}},
  \href {https://doi.org/10.1109/ITW.2008.4578692}
  {\path{doi:10.1109/ITW.2008.4578692}}.

\bibitem{CzekajHorodeckiHorodeckiHorodecki-informationContent}
L.~Czekaj, M.~Horodecki, P.~Horodecki, R.~Horodecki, {Information content of
  systems as a physical principle}, Physical Review A 95~(2) (2017) 022119.
\newblock \href {http://arxiv.org/abs/1403.4643} {\path{arXiv:1403.4643}},
  \href {https://doi.org/10.1103/PhysRevA.95.022119}
  {\path{doi:10.1103/PhysRevA.95.022119}}.

\bibitem{FilippovHeinosaariLeppajarvi-simulability}
S.~N. Filippov, T.~Heinosaari, L.~Lepp{\"{a}}j{\"{a}}rvi, {Simulability of
  observables in general probabilistic theories}, Physical Review A 97~(6)
  (2018) 062102.
\newblock \href {http://arxiv.org/abs/1803.11006} {\path{arXiv:1803.11006}},
  \href {https://doi.org/10.1103/PhysRevA.97.062102}
  {\path{doi:10.1103/PhysRevA.97.062102}}.

\bibitem{BaeKimKwek-discrimination}
J.~Bae, D.-G. Kim, L.-C. Kwek, {Structure of Optimal State Discrimination in
  Generalized Probabilistic Theories}, Entropy 18~(2) (2016) 39.
\newblock \href {http://arxiv.org/abs/1707.02521} {\path{arXiv:1707.02521}},
  \href {https://doi.org/10.3390/e18020039} {\path{doi:10.3390/e18020039}}.

\bibitem{SelbySikora-money}
J.~H. Selby, J.~Sikora, {How to make unforgeable money in generalised
  probabilistic theories}, Quantum 2 (2018) 103.
\newblock \href {http://arxiv.org/abs/1803.10279} {\path{arXiv:1803.10279}},
  \href {https://doi.org/10.22331/q-2018-11-02-103}
  {\path{doi:10.22331/q-2018-11-02-103}}.

\bibitem{SikoraSelby-bitCommitment}
J.~Sikora, J.~Selby, {Simple proof of the impossibility of bit commitment in
  generalized probabilistic theories using cone programming}, Physical Review A
  97~(4) (2018) 042302.
\newblock \href {http://arxiv.org/abs/1711.02662} {\path{arXiv:1711.02662}},
  \href {https://doi.org/10.1103/PhysRevA.97.042302}
  {\path{doi:10.1103/PhysRevA.97.042302}}.

\bibitem{SikoraSelby-coinFlipping}
J.~Sikora, J.~H. Selby, {Impossibility of coin flipping in generalized
  probabilistic theories via discretizations of semi-infinite programs},
  Physical Review Research 2~(4) (2020) 043128.
\newblock \href {http://arxiv.org/abs/1901.04876} {\path{arXiv:1901.04876}},
  \href {https://doi.org/10.1103/PhysRevResearch.2.043128}
  {\path{doi:10.1103/PhysRevResearch.2.043128}}.

\bibitem{LamiPalazuelosWinter-dataHiding}
L.~Lami, C.~Palazuelos, A.~Winter, {Ultimate Data Hiding in Quantum Mechanics
  and Beyond}, Communications in Mathematical Physics 361~(2) (2018) 661--708.
\newblock \href {http://arxiv.org/abs/1703.03392} {\path{arXiv:1703.03392}},
  \href {https://doi.org/10.1007/s00220-018-3154-4}
  {\path{doi:10.1007/s00220-018-3154-4}}.

\bibitem{YoshidaAraiHayashi-discrimination}
Y.~Yoshida, H.~Arai, M.~Hayashi, {Perfect Discrimination in Approximate Quantum
  Theory of General Probabilistic Theories}, Physical Review Letters 125~(15)
  (2020) 150402.
\newblock \href {http://arxiv.org/abs/2004.04949} {\path{arXiv:2004.04949}},
  \href {https://doi.org/10.1103/PhysRevLett.125.150402}
  {\path{doi:10.1103/PhysRevLett.125.150402}}.

\bibitem{BanikSahaGuhaAgrawalBhattachryaRoyMajumdar-informationSymmetry}
M.~Banik, S.~Saha, T.~Guha, S.~Agrawal, S.~S. Bhattacharya, A.~Roy, A.~S.
  Majumdar, {Constraining the state space in any physical theory with the
  principle of information symmetry}, Physical Review A 100~(6) (2019) 060101.
\newblock \href {http://arxiv.org/abs/1905.09413} {\path{arXiv:1905.09413}},
  \href {https://doi.org/10.1103/PhysRevA.100.060101}
  {\path{doi:10.1103/PhysRevA.100.060101}}.

\bibitem{SahaBhattacharyaGuhaHalderBanik-communication}
S.~Saha, S.~S. Bhattacharya, T.~Guha, S.~Halder, M.~Banik, {Advantage of
  Quantum Theory over Nonclassical Models of Communication}, Annalen der Physik
  532~(12) (2020) 2000334.
\newblock \href {http://arxiv.org/abs/1806.09474} {\path{arXiv:1806.09474}},
  \href {https://doi.org/10.1002/andp.202000334}
  {\path{doi:10.1002/andp.202000334}}.

\bibitem{SahaGuhaBhattacharyaBanik-distributedComputation}
S.~Saha, T.~Guha, S.~S. Bhattacharya, M.~Banik, {Distributed Computing Model:
  Classical vs. Quantum vs. Post-Quantum} (2020).
\newblock \href {http://arxiv.org/abs/2012.05781} {\path{arXiv:2012.05781}}.

\bibitem{ShortWehner-entropy}
A.~J. Short, S.~Wehner, {Entropy in general physical theories}, New Journal of
  Physics 12~(3) (2010) 033023.
\newblock \href {http://arxiv.org/abs/0909.4801} {\path{arXiv:0909.4801}},
  \href {https://doi.org/10.1088/1367-2630/12/3/033023}
  {\path{doi:10.1088/1367-2630/12/3/033023}}.

\bibitem{KimuraNuidaImai-entropyDistiguishability}
G.~Kimura, K.~Nuida, H.~Imai, {Distinguishability measures and entropies for
  general probabilistic theories}, Reports on Mathematical Physics 66~(2)
  (2010) 175--206.
\newblock \href {http://arxiv.org/abs/0910.0994} {\path{arXiv:0910.0994}},
  \href {https://doi.org/10.1016/S0034-4877(10)00025-X}
  {\path{doi:10.1016/S0034-4877(10)00025-X}}.

\bibitem{BarnumBarrettClarkLeiferSpekkensStepanikWilceWilke-entropy}
H.~Barnum, J.~Barrett, L.~O. Clark, M.~Leifer, R.~Spekkens, N.~Stepanik,
  A.~Wilce, R.~Wilke, {Entropy and information causality in general
  probabilistic theories}, New Journal of Physics 14~(12) (2012) 129401.
\newblock \href {http://arxiv.org/abs/0909.5075} {\path{arXiv:0909.5075}},
  \href {https://doi.org/10.1088/1367-2630/14/12/129401}
  {\path{doi:10.1088/1367-2630/14/12/129401}}.

\bibitem{KimuraIshiguroFukui-entropyHolevo}
G.~Kimura, J.~Ishiguro, M.~Fukui, {Entropies in general probabilistic theories
  and their application to the Holevo bound}, Physical Review A 94~(4) (2016)
  042113.
\newblock \href {http://arxiv.org/abs/1604.08009} {\path{arXiv:1604.08009}},
  \href {https://doi.org/10.1103/PhysRevA.94.042113}
  {\path{doi:10.1103/PhysRevA.94.042113}}.

\bibitem{Takakura-entropy}
R.~Takakura, {Entropy of mixing exists only for classical and quantum-like
  theories among the regular polygon theories}, Journal of Physics A:
  Mathematical and Theoretical 52~(46) (2019) 465302.
\newblock \href {http://arxiv.org/abs/1803.07388} {\path{arXiv:1803.07388}},
  \href {https://doi.org/10.1088/1751-8121/ab4a2e}
  {\path{doi:10.1088/1751-8121/ab4a2e}}.

\bibitem{ChiribellaScandolo-thermodynamics}
G.~Chiribella, C.~M. Scandolo, {Entanglement and thermodynamics in general
  probabilistic theories}, New Journal of Physics 17~(10) (2015) 103027.
\newblock \href {http://arxiv.org/abs/1504.07045} {\path{arXiv:1504.07045}},
  \href {https://doi.org/10.1088/1367-2630/17/10/103027}
  {\path{doi:10.1088/1367-2630/17/10/103027}}.

\bibitem{ChiribellaScandolo-microcanonicalThermodynamics}
G.~Chiribella, C.~M. Scandolo, {Microcanonical thermodynamics in general
  physical theories}, New Journal of Physics 19~(12) (2017) 123043.
\newblock \href {http://arxiv.org/abs/1608.04460} {\path{arXiv:1608.04460}},
  \href {https://doi.org/10.1088/1367-2630/aa91c7}
  {\path{doi:10.1088/1367-2630/aa91c7}}.

\bibitem{KrummBarnumBarrettMuller-thermodynamics}
M.~Krumm, H.~Barnum, J.~Barrett, M.~P. M{\"{u}}ller, {Thermodynamics and the
  structure of quantum theory}, New Journal of Physics 19~(4) (2017) 043025.
\newblock \href {http://arxiv.org/abs/1608.04461} {\path{arXiv:1608.04461}},
  \href {https://doi.org/10.1088/1367-2630/aa68ef}
  {\path{doi:10.1088/1367-2630/aa68ef}}.

\bibitem{ChiribellaScandolo-diagonalization}
G.~Chiribella, C.~M. Scandolo, {Operational axioms for diagonalizing states},
  Electronic Proceedings in Theoretical Computer Science 195~(Qpl) (2015)
  96--115.
\newblock \href {http://arxiv.org/abs/1506.00380} {\path{arXiv:1506.00380}},
  \href {https://doi.org/10.4204/EPTCS.195.8} {\path{doi:10.4204/EPTCS.195.8}}.

\bibitem{BarnumHilgert-spectral}
H.~Barnum, J.~Hilgert, {Strongly symmetric spectral convex bodies are Jordan
  algebra state spaces} (2019).
\newblock \href {http://arxiv.org/abs/1904.03753} {\path{arXiv:1904.03753}}.

\bibitem{Gudder-spectralEA}
S.~Gudder, {Contexts in Convex and Sequential Effect Algebras}, Electronic
  Proceedings in Theoretical Computer Science 287~(Qpl 2018) (2019) 191--211.
\newblock \href {http://arxiv.org/abs/1901.10640} {\path{arXiv:1901.10640}},
  \href {https://doi.org/10.4204/EPTCS.287.11}
  {\path{doi:10.4204/EPTCS.287.11}}.

\bibitem{JencovaPlavala-spectralEA}
A.~Jen{\v{c}}ov{\'{a}}, M.~Pl{\'{a}}vala, {On the properties of spectral effect
  algebras}, Quantum 3 (2019) 148.
\newblock \href {http://arxiv.org/abs/1811.12407} {\path{arXiv:1811.12407}},
  \href {https://doi.org/10.22331/q-2019-06-03-148}
  {\path{doi:10.22331/q-2019-06-03-148}}.

\bibitem{Hardy-derivationQT}
L.~Hardy, {Quantum Theory From Five Reasonable Axioms} (2001).
\newblock \href {http://arxiv.org/abs/0101012} {\path{arXiv:0101012}}.

\bibitem{ChiribellaDArianoPerinotti-derivationQT}
G.~Chiribella, G.~M. D'Ariano, P.~Perinotti, {Informational derivation of
  quantum theory}, Physical Review A 84~(1) (2011) 012311.
\newblock \href {http://arxiv.org/abs/1011.6451} {\path{arXiv:1011.6451}},
  \href {https://doi.org/10.1103/PhysRevA.84.012311}
  {\path{doi:10.1103/PhysRevA.84.012311}}.

\bibitem{PfisterWehner-discreteGPTs}
C.~Pfister, S.~Wehner, {An information-theoretic principle implies that any
  discrete physical theory is classical}, Nature Communications 4~(1) (2013)
  1851.
\newblock \href {http://arxiv.org/abs/1210.0194} {\path{arXiv:1210.0194}},
  \href {https://doi.org/10.1038/ncomms2821} {\path{doi:10.1038/ncomms2821}}.

\bibitem{Kleinmann-emergenceQT}
M.~Kleinmann, T.~J. Osborne, V.~B. Scholz, A.~H. Werner, {Typical Local
  Measurements in Generalized Probabilistic Theories: Emergence of Quantum
  Bipartite Correlations}, Physical Review Letters 110~(4) (2013) 040403.
\newblock \href {http://arxiv.org/abs/1205.3358} {\path{arXiv:1205.3358}},
  \href {https://doi.org/10.1103/PhysRevLett.110.040403}
  {\path{doi:10.1103/PhysRevLett.110.040403}}.

\bibitem{RichensSelbyAlSafi-entanglement}
J.~G. Richens, J.~H. Selby, S.~W. Al-Safi, {Entanglement is necessary for
  emergent classicality}, Physical Review Letters 119 (2017) 080503.
\newblock \href {http://arxiv.org/abs/1705.08028} {\path{arXiv:1705.08028}},
  \href {https://doi.org/10.1103/PhysRevLett.119.080503}
  {\path{doi:10.1103/PhysRevLett.119.080503}}.

\bibitem{Wilce-derivationQT}
A.~Wilce, {A Royal Road to Quantum Theory (or Thereabouts)}, Entropy 20~(4)
  (2018) 227.
\newblock \href {http://arxiv.org/abs/arXiv:1507.06278}
  {\path{arXiv:arXiv:1507.06278}}, \href {https://doi.org/10.3390/e20040227}
  {\path{doi:10.3390/e20040227}}.

\bibitem{MasanesMuller-derivatonQT}
L.~Masanes, M.~P. M{\"{u}}ller, {A derivation of quantum theory from physical
  requirements}, New Journal of Physics 13~(6) (2011) 063001.
\newblock \href {http://arxiv.org/abs/1004.1483} {\path{arXiv:1004.1483}},
  \href {https://doi.org/10.1088/1367-2630/13/6/063001}
  {\path{doi:10.1088/1367-2630/13/6/063001}}.

\bibitem{LeeSelby-decoherenceToQT}
C.~M. Lee, J.~H. Selby, {A no-go theorem for theories that decohere to quantum
  mechanics}, Proceedings of the Royal Society A: Mathematical, Physical and
  Engineering Sciences 474~(2214) (2018) 20170732.
\newblock \href {http://arxiv.org/abs/1701.07449} {\path{arXiv:1701.07449}},
  \href {https://doi.org/10.1098/rspa.2017.0732}
  {\path{doi:10.1098/rspa.2017.0732}}.

\bibitem{vandeWetering-derivatonQT}
J.~van~de Wetering, {An effect-theoretic reconstruction of quantum theory},
  Compositionality 1 (2019) 1.
\newblock \href {http://arxiv.org/abs/1801.05798} {\path{arXiv:1801.05798}},
  \href {https://doi.org/10.32408/compositionality-1-1}
  {\path{doi:10.32408/compositionality-1-1}}.

\bibitem{vandeWetering-sequential}
J.~van~de Wetering, {Sequential product spaces are Jordan algebras}, Journal of
  Mathematical Physics 60~(6) (2019) 062201.
\newblock \href {http://arxiv.org/abs/1803.11139} {\path{arXiv:1803.11139}},
  \href {https://doi.org/10.1063/1.5093504} {\path{doi:10.1063/1.5093504}}.

\bibitem{MazurekPuseyReschSpekkens-experiment}
M.~D. Mazurek, M.~F. Pusey, K.~J. Resch, R.~W. Spekkens, {Experimentally
  bounding deviations from quantum theory in the landscape of generalized
  probabilistic theories} (2017).
\newblock \href {http://arxiv.org/abs/1710.05948} {\path{arXiv:1710.05948}}.

\bibitem{WeilenmannColbeck-selfTesting}
M.~Weilenmann, R.~Colbeck, {Self-Testing of Physical Theories, or, Is Quantum
  Theory Optimal with Respect to Some Information-Processing Task?}, Physical
  Review Letters 125~(6) (2020) 060406.
\newblock \href {http://arxiv.org/abs/2003.00349} {\path{arXiv:2003.00349}},
  \href {https://doi.org/10.1103/PhysRevLett.125.060406}
  {\path{doi:10.1103/PhysRevLett.125.060406}}.

\bibitem{GarnerMuller-retractsToQT}
A.~J.~P. Garner, M.~P. M{\"{u}}ller, {Characterization of the probabilistic
  models that can be embedded in quantum theory} (2020).
\newblock \href {http://arxiv.org/abs/2004.06136} {\path{arXiv:2004.06136}}.

\bibitem{Mccrimmon-Jordan}
K.~McCrimmon, {Jordan algebras and their applications}, Bulletin of the
  American Mathematical Society 84~(4) (1978) 612--628.
\newblock \href {https://doi.org/10.1090/S0002-9904-1978-14503-0}
  {\path{doi:10.1090/S0002-9904-1978-14503-0}}.

\bibitem{JordanNeumannWigner-algebras}
P.~Jordan, J.~v.~Neumann, E.~Wigner, {On an Algebraic Generalization of the
  Quantum Mechanical Formalism}, The Annals of Mathematics 35~(1) (1934) 29.
\newblock \href {https://doi.org/10.2307/1968117} {\path{doi:10.2307/1968117}}.

\bibitem{ChiribellaDArianoPerinotti-GPTpurification}
G.~Chiribella, G.~M. D'Ariano, P.~Perinotti, {Probabilistic theories with
  purification}, Physical Review A 81~(6) (2010) 062348.
\newblock \href {http://arxiv.org/abs/0908.1583} {\path{arXiv:0908.1583}},
  \href {https://doi.org/10.1103/PhysRevA.81.062348}
  {\path{doi:10.1103/PhysRevA.81.062348}}.

\bibitem{BisioPerinotti-higherOrder}
A.~Bisio, P.~Perinotti, {Theoretical framework for higher-order quantum
  theory}, Proceedings of the Royal Society A: Mathematical, Physical and
  Engineering Sciences 475~(2225) (2019) 20180706.
\newblock \href {http://arxiv.org/abs/1806.09554} {\path{arXiv:1806.09554}},
  \href {https://doi.org/10.1098/rspa.2018.0706}
  {\path{doi:10.1098/rspa.2018.0706}}.

\bibitem{Perinotti-cellularAutomata}
P.~Perinotti, {Cellular automata in operational probabilistic theories},
  Quantum 4 (2020) 294.
\newblock \href {http://arxiv.org/abs/1911.11216} {\path{arXiv:1911.11216}},
  \href {https://doi.org/10.22331/q-2020-07-09-294}
  {\path{doi:10.22331/q-2020-07-09-294}}.

\bibitem{ChoJacobsWesterbaanWesterbann-effectus}
K.~Cho, B.~Jacobs, B.~Westerbaan, A.~Westerbaan, {An Introduction to Effectus
  Theory} (2015).
\newblock \href {http://arxiv.org/abs/1512.05813} {\path{arXiv:1512.05813}}.

\bibitem{Munkres-topology}
J.~R. Munkres, Topology, Featured Titles for Topology Series, Prentice Hall,
  Incorporated, 2000.

\bibitem{Popescu-mixedStates}
S.~Popescu, {Quantum states and knowledge: Between pure states and density
  matrices} (2018).
\newblock \href {http://arxiv.org/abs/1811.05472} {\path{arXiv:1811.05472}}.

\bibitem{Rockafellar-convex}
R.~T. Rockafellar, Convex Analysis, Princeton landmarks in mathematics and
  physics, Princeton University Press, 1997.

\bibitem{FoulisBennet-EA}
D.~J. Foulis, M.~K. Bennett, {Effect algebras and unsharp quantum logics},
  Foundations of Physics 24~(10) (1994) 1331--1352.
\newblock \href {https://doi.org/10.1007/BF02283036}
  {\path{doi:10.1007/BF02283036}}.

\bibitem{Kopka-Dposets}
F.~K{\^{o}}pka, {D-posets of fuzzy sets}, Tatra Mountains Mathematical
  Publications 1~(1) (1992) 83--87.

\bibitem{DvurecenskijPulmannova-structures}
A.~Dvurecenskij, S.~Pulmannov{\'a}, New Trends in Quantum Structures,
  Mathematics and Its Applications, Springer Netherlands, 2013.

\bibitem{GudderPulmannova-convexEA}
S.~Gudder, S.~Pulmannov{\'{a}}, {Representation theorem for convex effect
  algebras}, Commentationes Mathematicae Universitatis Carolinae 39~(4) (1998)
  645--660.

\bibitem{NaylorSell-operators}
A.~Naylor, G.~Sell, Linear Operator Theory in Engineering and Science, Applied
  Mathematical Sciences, Springer New York, 1982.

\bibitem{Jencova-baseNorms}
A.~Jen{\v{c}}ov{\'{a}}, {Base norms and discrimination of generalized quantum
  channels}, Journal of Mathematical Physics 55 (2014) 022201.
\newblock \href {http://arxiv.org/abs/1308.4030} {\path{arXiv:1308.4030}},
  \href {https://doi.org/10.1063/1.4863715} {\path{doi:10.1063/1.4863715}}.

\bibitem{Rudin-functionalA}
W.~Rudin, Functional Analysis, International series in pure and applied
  mathematics, McGraw-Hill, 1991.

\bibitem{NuidaKimuraMiyadera-discrimination}
K.~Nuida, G.~Kimura, T.~Miyadera, {Optimal observables for minimum-error state
  discrimination in general probabilistic theories}, Journal of Mathematical
  Physics 51~(9) (2010).
\newblock \href {http://arxiv.org/abs/0906.5419} {\path{arXiv:0906.5419}},
  \href {https://doi.org/10.1063/1.3479008} {\path{doi:10.1063/1.3479008}}.

\bibitem{Gudder-convexStructures}
S.~Gudder, {Convex structures and operational quantum mechanics},
  Communications in Mathematical Physics 29~(3) (1973) 249--264.
\newblock \href {https://doi.org/10.1007/BF01645250}
  {\path{doi:10.1007/BF01645250}}.

\bibitem{JanottaLal-noRestriction}
P.~Janotta, R.~Lal, {Generalized probabilistic theories without the
  no-restriction hypothesis}, Physical Review A 87~(5) (2013) 052131.
\newblock \href {http://arxiv.org/abs/1412.8524} {\path{arXiv:1412.8524}},
  \href {https://doi.org/10.1103/PhysRevA.87.052131}
  {\path{doi:10.1103/PhysRevA.87.052131}}.

\bibitem{FilippovGudderHeinosaariLeppajarvi-restrictions}
S.~N. Filippov, S.~Gudder, T.~Heinosaari, L.~Lepp{\"{a}}j{\"{a}}rvi,
  {Operational Restrictions in General Probabilistic Theories}, Foundations of
  Physics 50~(8) (2020) 850--876.
\newblock \href {http://arxiv.org/abs/1912.08538} {\path{arXiv:1912.08538}},
  \href {https://doi.org/10.1007/s10701-020-00352-6}
  {\path{doi:10.1007/s10701-020-00352-6}}.

\bibitem{ChiribellaDArioanoPerinotti-quantumCircuits}
G.~Chiribella, G.~M. D'Ariano, P.~Perinotti, {Quantum Circuit Architecture},
  Physical Review Letters 101~(6) (2008) 060401.
\newblock \href {http://arxiv.org/abs/0712.1325} {\path{arXiv:0712.1325}},
  \href {https://doi.org/10.1103/PhysRevLett.101.060401}
  {\path{doi:10.1103/PhysRevLett.101.060401}}.

\bibitem{SelbyCoecke-leaks}
J.~Selby, B.~Coecke, {Leaks: Quantum, Classical, Intermediate and More},
  Entropy 19~(4) (2017) 174.
\newblock \href {http://arxiv.org/abs/1701.07404} {\path{arXiv:1701.07404}},
  \href {https://doi.org/10.3390/e19040174} {\path{doi:10.3390/e19040174}}.

\bibitem{WolfeSchmidSainzKunjwalSpekkens-commonCauseBoxes}
E.~Wolfe, D.~Schmid, A.~B. Sainz, R.~Kunjwal, R.~W. Spekkens, {Quantifying
  Bell: the Resource Theory of Nonclassicality of Common-Cause Boxes}, Quantum
  4 (2020) 280.
\newblock \href {http://arxiv.org/abs/1903.06311} {\path{arXiv:1903.06311}},
  \href {https://doi.org/10.22331/q-2020-06-08-280}
  {\path{doi:10.22331/q-2020-06-08-280}}.

\bibitem{SchmidSelbySpekkens-causalInferentialTheories}
D.~Schmid, J.~H. Selby, R.~W. Spekkens, {Unscrambling the omelette of causation
  and inference: The framework of causal-inferential theories} (2020).
\newblock \href {http://arxiv.org/abs/2009.03297} {\path{arXiv:2009.03297}}.

\bibitem{SchmidFraserKunjwalSainzWolfeSpekkens-entanglement}
D.~Schmid, T.~C. Fraser, R.~Kunjwal, A.~B. Sainz, E.~Wolfe, R.~W. Spekkens,
  {Why standard entanglement theory is inappropriate for the study of Bell
  scenarios} (2020).
\newblock \href {http://arxiv.org/abs/2004.09194} {\path{arXiv:2004.09194}}.

\bibitem{SchmidHaoxingMudassarWitRossethoban-commonCauseChannels}
D.~Schmid, H.~Du, M.~Mudassar, G.~C.-d. Wit, D.~Rosset, M.~J. Hoban,
  {Postquantum common-cause channels: the resource theory of local operations
  and shared entanglement} (2020).
\newblock \href {http://arxiv.org/abs/2004.06133} {\path{arXiv:2004.06133}}.

\bibitem{quantikz}
A.~Kay,
  \href{https://royalholloway.figshare.com/articles/dataset/Quantikz/7000520/4}{Quantikz}
  (2018).
\newblock \href {http://arxiv.org/abs/1809.03842} {\path{arXiv:1809.03842}},
  \href {https://doi.org/10.17637/rh.7000520.v4}
  {\path{doi:10.17637/rh.7000520.v4}}.
\newline\urlprefix\url{https://royalholloway.figshare.com/articles/dataset/Quantikz/7000520/4}

\bibitem{Ryan-tensorProducts}
R.~Ryan, Introduction to Tensor Products of Banach Spaces, Springer Monographs
  in Mathematics, Springer London, 2002.

\bibitem{DArianoErbaPerinotti-classicalEntanglement}
G.~M. D'Ariano, M.~Erba, P.~Perinotti, {Classical theories with entanglement},
  Physical Review A 101~(4) (2020) 042118.
\newblock \href {http://arxiv.org/abs/1909.07134} {\path{arXiv:1909.07134}},
  \href {https://doi.org/10.1103/PhysRevA.101.042118}
  {\path{doi:10.1103/PhysRevA.101.042118}}.

\bibitem{DArianoErbaPerinotti-entanglement}
G.~M. D'Ariano, M.~Erba, P.~Perinotti, {Classicality without local
  discriminability: Decoupling entanglement and complementarity}, Physical
  Review A 102~(5) (2020) 052216.
\newblock \href {http://arxiv.org/abs/2008.04011} {\path{arXiv:2008.04011}},
  \href {https://doi.org/10.1103/PhysRevA.102.052216}
  {\path{doi:10.1103/PhysRevA.102.052216}}.

\bibitem{NamiokaPhelps-cones}
I.~Namioka, R.~Phelps, {Tensor products of compact convex sets}, Pacific
  Journal of Mathematics 31~(2) (1969) 469--480.
\newblock \href {https://doi.org/10.2140/pjm.1969.31.469}
  {\path{doi:10.2140/pjm.1969.31.469}}.

\bibitem{Barker-cones}
G.~P. Barker, {Theory of cones}, Linear Algebra and its Applications 39 (1981)
  263--291.
\newblock \href {https://doi.org/10.1016/0024-3795(81)90310-4}
  {\path{doi:10.1016/0024-3795(81)90310-4}}.

\bibitem{AubrunLamiPalazuelosPlavala-cones}
G.~Aubrun, L.~Lami, C.~Palazuelos, M.~Plavala, Entangleability of cones (2019).
\newblock \href {http://arxiv.org/abs/1911.09663} {\path{arXiv:1911.09663}}.

\bibitem{HeinosaariMiyaderaZiman-compatibility}
T.~Heinosaari, T.~Miyadera, M.~Ziman, {An Invitation to Quantum
  Incompatibility}, Journal of Physics A: Mathematical and Theoretical 49~(12)
  (2015) 123001.
\newblock \href {http://arxiv.org/abs/1511.07548} {\path{arXiv:1511.07548}},
  \href {https://doi.org/10.1088/1751-8113/49/12/123001}
  {\path{doi:10.1088/1751-8113/49/12/123001}}.

\bibitem{UolaKraftShangYuGuhne-conicResourceTheories}
R.~Uola, T.~Kraft, J.~Shang, X.-D. Yu, O.~G{\"{u}}hne, {Quantifying Quantum
  Resources with Conic Programming}, Physical Review Letters 122~(13) (2019)
  130404.
\newblock \href {http://arxiv.org/abs/1812.09216} {\path{arXiv:1812.09216}},
  \href {https://doi.org/10.1103/PhysRevLett.122.130404}
  {\path{doi:10.1103/PhysRevLett.122.130404}}.

\bibitem{HaapasaloKraftMiklinUola-marginalProblem}
E.~Haapasalo, T.~Kraft, N.~Miklin, R.~Uola, {Quantum marginal problem and
  incompatibility} (2019).
\newblock \href {http://arxiv.org/abs/1909.02941} {\path{arXiv:1909.02941}}.

\bibitem{GirardPlavalaSikora-jordan}
M.~Girard, M.~Plávala, J.~Sikora, Jordan products of quantum channels and
  their compatibility (2020).
\newblock \href {http://arxiv.org/abs/2009.03279} {\path{arXiv:2009.03279}}.

\bibitem{CarmeliHeinosaariToigo-postMeasStateDiscrimination}
C.~Carmeli, T.~Heinosaari, A.~Toigo, {State discrimination with postmeasurement
  information and incompatibility of quantum measurements}, Physical Review A
  98~(1) (2018) 012126.
\newblock \href {http://arxiv.org/abs/1804.09693} {\path{arXiv:1804.09693}},
  \href {https://doi.org/10.1103/PhysRevA.98.012126}
  {\path{doi:10.1103/PhysRevA.98.012126}}.

\bibitem{CarmeliHeinosaariToigo-incWitness}
C.~Carmeli, T.~Heinosaari, A.~Toigo, {Quantum Incompatibility Witnesses},
  Physical Review Letters 122~(13) (2019) 130402.
\newblock \href {http://arxiv.org/abs/1812.02985} {\path{arXiv:1812.02985}},
  \href {https://doi.org/10.1103/PhysRevLett.122.130402}
  {\path{doi:10.1103/PhysRevLett.122.130402}}.

\bibitem{CarmeliHeinosaariMiyaderaToigo-incWitnessChannels}
C.~Carmeli, T.~Heinosaari, T.~Miyadera, A.~Toigo, {Witnessing incompatibility
  of quantum channels}, Journal of Mathematical Physics 60~(12) (2019) 122202.
\newblock \href {http://arxiv.org/abs/1906.10904} {\path{arXiv:1906.10904}},
  \href {https://doi.org/10.1063/1.5126496} {\path{doi:10.1063/1.5126496}}.

\bibitem{UolaCostaNguyenGuhne-steering}
R.~Uola, A.~C.~S. Costa, H.~C. Nguyen, O.~G{\"{u}}hne, {Quantum steering},
  Reviews of Modern Physics 92~(1) (2020) 015001.
\newblock \href {http://arxiv.org/abs/1903.06663} {\path{arXiv:1903.06663}},
  \href {https://doi.org/10.1103/RevModPhys.92.015001}
  {\path{doi:10.1103/RevModPhys.92.015001}}.

\bibitem{BrunnerCavalcantiPironioScaraniWehner-BellNonlocality}
N.~Brunner, D.~Cavalcanti, S.~Pironio, V.~Scarani, S.~Wehner, {Bell
  nonlocality}, Reviews of Modern Physics 86~(2) (2014) 419--478.
\newblock \href {http://arxiv.org/abs/1303.2849} {\path{arXiv:1303.2849}},
  \href {https://doi.org/10.1103/RevModPhys.86.419}
  {\path{doi:10.1103/RevModPhys.86.419}}.

\bibitem{HeinosaariZiman-MLQT}
T.~Heinosaari, M.~Ziman, The Mathematical Language of Quantum Theory. From
  Uncertainty to Entanglement, Cambridge University Press, 2012.

\bibitem{BarrettLindenMassarPironioPopescuDavid-nonlocal}
J.~Barrett, N.~Linden, S.~Massar, S.~Pironio, S.~Popescu, D.~Roberts, {Nonlocal
  correlations as an information-theoretic resource}, Physical Review A 71~(2)
  (2005) 022101.
\newblock \href {http://arxiv.org/abs/0404097} {\path{arXiv:0404097}}, \href
  {https://doi.org/10.1103/PhysRevA.71.022101}
  {\path{doi:10.1103/PhysRevA.71.022101}}.

\bibitem{JanottaGogolinBarrettBrunner-nonlocal}
P.~Janotta, C.~Gogolin, J.~Barrett, N.~Brunner, {Limits on nonlocal
  correlations from the structure of the local state space}, New Journal of
  Physics 13 (2011).
\newblock \href {http://arxiv.org/abs/1012.1215} {\path{arXiv:1012.1215}},
  \href {https://doi.org/10.1088/1367-2630/13/6/063024}
  {\path{doi:10.1088/1367-2630/13/6/063024}}.

\bibitem{BrunnerKaplanLeverrierSkrypczyk-dimensions}
N.~Brunner, M.~Kaplan, A.~Leverrier, P.~Skrzypczyk, {Dimension of physical
  systems, information processing, and thermodynamics}, New Journal of Physics
  16~(12) (2014) 123050.
\newblock \href {http://arxiv.org/abs/1401.4488} {\path{arXiv:1401.4488}},
  \href {https://doi.org/10.1088/1367-2630/16/12/123050}
  {\path{doi:10.1088/1367-2630/16/12/123050}}.

\bibitem{MacLaneBirkhoff-algebra}
S.~Lane, G.~Birkhoff, Algebra, AMS Chelsea Publishing Series, Chelsea
  Publishing Company, 1999.

\bibitem{BoydVandenberghe-convex}
S.~Boyd, L.~Vandenberghe, Convex Optimization, Berichte {\"u}ber verteilte
  messysteme, Cambridge University Press, 2004.

\end{thebibliography}
}

\appendix

\makeatletter
\gdef\theproposition{\@Alph\c@section.\arabic{proposition}}
\makeatother

\section{Convex cones and ordered vector spaces} \label{appendix:cones}

%
%
\begin{definition} \label{def:cones-cone}
Let $V$ denote a real, finite-dimensional vector space. A \emph{cone} $C \subset V$ is a set such that for any $v \in C$ and $\lambda \in \Rp$ we have $\lambda v \in C$. Let $X \subset V$, then $\cone(X)$ is the smallest cone containing $X$, i.e., $\cone(X) = \{ \lambda v : v \in X, \lambda \in \Rp \}$.
\end{definition}
A cone is a subset of $V$ that is invariant to scaling, i.e., to multiplication by $\lambda \in \Rp$. The following is a simple lemma about the interplay between linear hulls and conic hulls.
\begin{lemma} \label{lemma:cones-coneVsSpan}
Let $X \subset V$, then $\linspan( \cone(X) ) = \cone( \linspan(X) ) = \linspan(X)$.
\end{lemma}
\begin{proof}
The proof is straightforward. Since $X \subset \cone(X)$, it follows that $\linspan(X) \subset \linspan(\cone(X))$. So let $v \in \linspan(\cone(X))$, then there are $w_i \in \cone(X)$ and $\alpha_i \in \RR$ for $i \in \{1, \ldots, n\}$ such that $v = \sum_{i=1}^n \alpha_i w_i$. But since $w_i \in \cone(X)$, there are $\lambda_i \in \Rp$ and $x_i \in X$ such that $w_i = \lambda_i x_i$ for all $i \in \{1, \ldots, n\}$, so we get $v = \sum_{i=1}^n \alpha_i \lambda_i x_i$, which implies $v \in \linspan(X)$ and $\linspan(\cone(X)) \subset \linspan(X)$ follows. So we have $\linspan(\cone(X)) = \linspan(X)$. To show that $\cone(\linspan(X)) = \linspan(X)$, simply observe that for any $v \in \linspan(X)$ and $\lambda \in \Rp$ we must have $\lambda v \in \linspan(X)$, so $\linspan(X)$ already is a cone.
\end{proof}

\begin{definition}
Let $V$ be a real, finite-dimensional vector space equipped with the Euclidean topology and let $C \subset V$ be a cone. We say that:
\begin{itemize}
\item $C$ is \emph{convex} if $C$ is a convex set, i.e., $\conv(C) = C$;
\item $C$ is \emph{closed} if $C$ is a closed set in the Euclidean topology on $V$;
\item $C$ is \emph{pointed} if $C \cap -C = \{ 0 \}$;
\item $C$ is \emph{generating} if $\linspan(C) = V$, i.e., if $C - C = V$.
\end{itemize}
\end{definition}
Some authors refer to convex, closed, pointed, generating cones as proper cones. We are interested in cones because a suitable cone $C \subset V$ gives the structure of ordered vector space to $V$.
\begin{definition} \label{def:cones-OVS}
\emph{Ordered vector space} is a vector space $V$ equipped with a binary relation $\leq$ such that for all $v, w, x \in V$, $\lambda \in \Rp$ we have that
\begin{enumerate}[label=(OVS\arabic*), leftmargin=*]
\item\label{item:cones-OVS-reflexive} $\leq$ is reflexive, i.e., $v \leq v$;
\item\label{item:cones-OVS-antiSymmetric} $\leq$ is anti-symmetric, i.e., $v \leq w$ and $w \leq v$ implies $v = w$;
\item\label{item:cones-OVS-transitive} $\leq$ is transitive, i.e., $v \leq w$ and $w \leq x$ implies $v \leq x$;
\item\label{item:cones-OVS-cancelative} $\leq$ respects the addition on $V$, i.e., $v \leq w$ implies $v + x \leq w + x$;
\item\label{item:cones-OVS-multiplicative} $\leq$ respects the multiplication by positive scalars, i.e., $v \leq w$ implies $\lambda v \leq \lambda w$.
\end{enumerate}
\end{definition}

\begin{definition}
Let $(V, \leq)$ be an ordered vector space. We say that $V$ is a \emph{directed set} under the ordering $\leq$ if for every $v_1, v_2 \in V$ there is $w \in V$ such that $v_1 \leq w$ and $v_2 \leq w$.
\end{definition}
We will show that we can construct a natural cone from the order $\leq$ and that we can construct an order on $V$ given a suitable cone $C \subset V$. Let $V$ be a real, finite-dimensional vector space, equipped with the relation $\leq$, so that $(V, \leq)$ is an ordered vector space. Let
\begin{equation} \label{eq:cones-positiveCone}
C = \{ v \in V : 0 \leq v \}
\end{equation}
be the set of positive elements, we will show that $C$ is a convex, pointed, generating cone. Note that we will use $\geq$ instead of $\leq$ when better suited, we have $v \geq w$ whenever $w \leq v$.
\begin{proposition}
Let $V$ be an ordered vector space and let $C \subset V$ be a positive cone as given by \eqref{eq:cones-positiveCone}, then $C$ is a convex and pointed cone.
\end{proposition}
\begin{proof}
Let $x \in C$ and $\lambda \in \Rp$, then $0 \leq x$ and $0 \leq \lambda x$ follows from \ref{item:cones-OVS-multiplicative} so $C$ is a cone. Let $x, y \in C$ and $\lambda \in [0,1]$, then we have since $0 \leq x$ and $0 \leq y$. From \ref{item:cones-OVS-multiplicative} we get $0 \leq \lambda x$ and $0 \leq (1-\lambda) y$ and we get $0 \leq \lambda x + (1-\lambda) y$ from \ref{item:cones-OVS-cancelative} and \ref{item:cones-OVS-transitive}. It follows that $C$ is convex. Assume that $x \in C$ and $x \in -C$, i.e., that we have $0 \leq x$ and $0 \leq -x$. Using \ref{item:cones-OVS-cancelative} we get $x \leq 0$ and then from \ref{item:cones-OVS-antiSymmetric} we get $x = 0$. It follows that $C \cap -C = \{ 0 \}$, so $C$ is pointed.
\end{proof}

\begin{proposition}
Let $(V, \leq)$ be an ordered vector space such that $V$ is a directed set under the ordering $\leq$. Then the positive cone $C \subset V$ given by \eqref{eq:cones-positiveCone} is generating.
\end{proposition}
\begin{proof}
Let $v \in V$, then for the pair of elements $v, -v$ there must be $w \in V$ such that $v \leq w$ and $-v \leq w$. This implies that $0 \leq w - v$ and $0 \leq w + v$, so $w - v \in C$ and $w + v \in C$. Since we have
\begin{equation}
v = \dfrac{w + v}{2} - \dfrac{w - v}{2},
\end{equation}
it follows that $v \in \linspan(C)$.
\end{proof}
Thus we have showed that an ordered vector space $(V, \leq)$ contains  the positive cone $C$. Now, we will start by assuming that we have a suitable cone $C \subset V$ and we will show that then we can construct the ordering $\leq$. Therefore we will show complete equivalence between ordered vector spaces and vector spaces with cones.

Let $C \subset V$ be a cone, then we can invert the logic of \eqref{eq:cones-positiveCone} and say that for $v \in V$ we have $0 \leq v$ if and only if $v \in C$, i.e., we can simply say that $C$ is the positive cone given by some order $\leq$. For $v, w \in V$ we then have $v \leq w$ if and only if $0 \leq w - v$, so in principle we can construct the order $\leq$ from the cone $C$. But the order $\leq$ satisfies \ref{item:cones-OVS-reflexive} - \ref{item:cones-OVS-multiplicative} only if the cone $C$ is convex and pointed.
\begin{proposition} \label{prop:cones-orderFromCone}
Let $C \subset V$ be a convex and pointed cone and let $\leq$ be given for $v, w \in V$ as
\begin{equation} \label{eq:cones-orderFromCone}
v \leq w \quad \Leftrightarrow \quad w-v \in C,
\end{equation}
then $(V, \leq)$ is an ordered vector space.
\end{proposition}
\begin{proof}
The proof is rather straightforward. Let $v,w,x \in V$ and $\lambda \in \Rp$, then we have $v - v = 0 \in C$, so $v \leq v$ and \ref{item:cones-OVS-reflexive} holds. Let $v - w \in C$ and $w - v \in C$, then since $w - v = -(v-w)$ and since $C$ is pointed, we have $v-w = 0$, so we get $v=w$ and \ref{item:cones-OVS-antiSymmetric} holds. Let $w-v \in C$ and $x-w \in C$, then we have $x - v = (x-w) + (w-v) \in C$ because $C$ is convex cone, so \ref{item:cones-OVS-transitive} holds. Let $w-v \in C$, then we have $(w+x) - (v+x) = w-v \in C$ and so \ref{item:cones-OVS-cancelative} holds. We also have $\lambda w - \lambda v = \lambda (w-v) \in C$ since $C$ is a cone, so \ref{item:cones-OVS-multiplicative} holds as well.
\end{proof}

\begin{proposition}
Let $C \subset V$ be a convex, pointed cone and let $\leq$ be the order constructed from $C$ as in \eqref{eq:cones-orderFromCone}. Let $C$ be a generating cone, then $V$ is directed set under $\leq$.
\end{proposition}
\begin{proof}
Let $v_1, v_2 \in V$, then since $C$ is generating, there are $x_1, x_2, y_1, y_2 \in C$ such that $v_1 = x_1 - y_1$ and $v_2 = x_2 - y_2$. Let $x = x_1 + x_2$, then we have $x - v_1 = x_2 + y_1 \in C$ and $x - v_2 = x_1 + y_2 \in C$, i.e., we have $v_1 \leq x$ and $v_2 \leq x$.
\end{proof}
Thus we have proved that convex, pointed cones are in one-to-one correspondence with ordered vector spaces. Moreover the cone is generating if and only if $V$ is a directed set.

We have not included the closeness of $C$ into the discussion, but one can easily show the following: let $\{ v_n \} \subset V$ be a Cauchy sequence and let $w \in V$, then $C$ is closed if and only if $w \leq v_n$ implies $w \leq \lim_{n \to \infty} v_n$.

\section{Functionals, duals and hyperplane separation theorems} \label{appendix:duals}

%
%
Let $V$ be a real finite-dimensional vector space. Functional $\psi: V \to \RR$ is a linear map from $V$ to $\RR$, i.e., for $v,w \in A(K)$ and $\alpha, \beta \in \RR$ we have $\psi ( \alpha v + \beta w ) = \alpha \psi(v) + \beta \psi(w)$. It is straightforward to define a linear combination of functionals, let $\psi, \varphi$ be linear functional on $V$, $\alpha, \beta \in \RR$ and $v \in V$, then we define $(\alpha \psi + \beta \varphi)(v) = \alpha \psi(v) + \beta \varphi(v)$. It follows that the set of all functionals $\psi: V \to \RR$ is a vector space.
\begin{definition}
Let $V$ be a real, finite-dimensional vectors space. The \emph{dual vector space} $V^*$ is the vector space of all functionals $\psi: V \to \RR$.
\end{definition}
Given a basis $\{v_1, \ldots, v_n\} \subset V$ we can define corresponding basis in $V^*$.
\begin{definition}
Let $\{v_1, \ldots, v_n\} \subset V$ be a basis of $V$, then the \emph{dual basis} is a basis $\{\psi_1, \ldots, \psi_n\} \subset V^*$ such that
\begin{equation}
\psi_i(v_j) = \delta_{ij},
\end{equation}
where $\delta_{ij}$ is the Kronecker delta,
\begin{equation} \label{eq:duals-KroneckerDelta}
\delta_{ij} =
\begin{cases}
1 & i=j \\
0 & i \neq j
\end{cases}
\end{equation}
\end{definition}
One can show that a dual basis always exists, see \cite{MacLaneBirkhoff-algebra}. The proof is rather simple, one can start with any basis of $V^*$ and solve a series of linear equations to construct the dual basis. Another option is to realize that given a real, finite-dimensional vector space $V$, we can always introduce the Euclidean inner product and use that inner product to construct the dual basis.

A natural question arises: what is the dual of the dual? For a general vector space, this is a non-trivial question, but since we are working with finite-dimensional vector spaces, the question considerably simplifies. Before we proceed, note that we can naturally identify vectors $v \in V$ with the functionals on functionals, $\xi \in V^{**}$, $\xi: V^* \to \RR$.
\begin{proposition}
Let $v \in V$, then we can identify $v$ with a functional on functionals $\xi_v \in V^{**}$.
\end{proposition}
\begin{proof}
The construction is rater simple, let $\psi \in V^*$ then we define $\xi_v(\psi) = \psi(v)$. In other words, for $\psi \in V^*$ the map $\xi_v: \psi \mapsto \psi(v)$ is a functional on $V^*$.
\end{proof}
We should, in principle, work with an isomorphism $v \mapsto \xi_v$ rather that simply putting $\xi_v = v$, but since this isomorphism is linear, we will omit it.

\begin{proposition} \label{prop:duals-doubleDual}
Let $V$ be a real, finite-dimensional vector space, then $V = V^{**}$.
\end{proposition}
\begin{proof}
Let $\{v_1, \ldots, v_n\} \subset V$ be a basis of $V$ and let $\{\psi_1, \ldots, \psi_n\} \subset V^*$ be the dual basis. Let $\xi \in V^{**}$, we will show that we have $\xi = \sum_{i=1}^n \xi(\psi_i) v_i$ and so $\xi \in V$. Let $j \in \{1, \ldots, n\}$, then we clearly have $\xi(\psi_j) = \sum_{i=1}^n \xi(\psi_i) \psi_j(v_i)$. It then follows that for every $\psi \in V^*$ we have $\xi(\psi) = \sum_{i=1}^n \xi(\psi_i) \psi(v_i) = (\sum_{i=1}^n \xi(\psi_i) v_i)(\psi)$ and the result follows.
\end{proof}

We will now look at the structure of the dual vector space $V^*$ given that the vector space $V$ is an ordered vector space, see Definition \ref{def:cones-OVS}. So let $V$ be a real, finite-dimensional vector space and let $C \subset V$ be convex cone. $C$ induces a cone $C^* \subset V^*$, $C^*$ is called the dual cone.
\begin{definition} \label{def:duals-dualCone}
Let $V$ be a real, finite-dimensional vector space and let $C \subset V$ be a cone. The \emph{dual cone} $C^* \subset V^*$ is defined as
\begin{equation}
C^* = \{ \psi \in V^* : \psi(x) \geq 0, \forall x \in C \}.
\end{equation}
\end{definition}
In other words, the dual cone is the cone of all functionals $\psi \in V^*$ that are positive on all the elements of $C$, i.e., $\psi$ is positive on all positive vectors. It is straightforward that $C^*$ is a cone.

We will show that the dual cone $C^*$ is always convex and closed, and that $C^*$ is generating if $C$ is pointed and vice-versa.
\begin{proposition} \label{prop:duals-dualConvexClosed}
Let $V$ be a real, finite-dimensional vector space and let $C \subset V$ be a cone. The dual cone $C^*$ is convex and closed in the standard topology given by the Euclidean inner product.
\end{proposition}
\begin{proof}
Let $\psi_1, \psi_2 \in C^*$ and $\lambda \in [0,1]$ and let $x \in C$. we have
\begin{equation}
(\lambda \psi_1 + (1-\lambda) \psi_2)(x) = \lambda \psi_1(x) + (1-\lambda) \psi_2(x) \geq 0
\end{equation}
and so $\lambda \psi_1 + (1-\lambda) \psi_2 \in C^*$. Now let $\{ \psi_n \}_{n=1}^\infty \subset C^*$ be a Cauchy sequence, i.e., there is $\psi \in V^*$ such that $\psi = \lim_{n \to \infty} \psi_n$. For every $x \in C$ we have $\psi(x) = \lim_{n \to \infty} \psi_n(x)$, but since $\psi_n(x) \geq 0$ we must also have $\psi(x) \geq 0$ and so $\psi \in C^*$.
\end{proof}

\begin{proposition} \label{prop:duals-generatingToPointed}
Let $V$ be a real, finite-dimensional vector space and let $C \subset V$ be a generating cone. Then the dual cone $C^*$ is pointed.
\end{proposition}
\begin{proof}
Remember that $C$ is generating if $C - C = V$ and $C^*$ is pointed if $C^* \cap (- C^*) = \{ 0 \}$. Let $\psi \in C^* \cap (- C^*)$, then it follows that for any $x \in C$ we must have $\psi(x) \geq 0$ and $\psi(x) \leq 0$ and so $\psi(x) = 0$ follows. Since $C$ is generating, for every $v \in C$ there are $y,y' \in C$ such that $v = y - y'$. We then have $\psi(v) = \psi(y) - \psi(y') = 0$ and so $\psi = 0$.
\end{proof}

\begin{proposition} \label{prop:duals-pointedToGenerating}
Let $V$ be a real, finite-dimensional vector space and let $C \subset V$ be a pointed cone. Then the dual cone $C^*$ is generating.
\end{proposition}
\begin{proof}
Remember that $C$ is pointed if $C \cap (-C) = \{ 0 \}$ and $C^*$ is generating if $C^* - C^* = V^*$. Assume that $C^*$ is not generating, then there is $\psi \in V^*$, $\psi \neq 0$ such that $\psi \notin C^* - C^*$. It follows that there is $v \in V$ such that $\psi(v) \neq 0$ but for all $\varphi \in C^* - C^*$ we have $\varphi(v) = 0$. One can construct such $v$ using the dual basis of $V^*$. It follows that $0 \neq v \in C \cap (-C)$, which is a contradiction with $C$ being pointed.
\end{proof}

We will now present two variants of an important theorem known as the Hanh-Banach hyperplane separation theorem. Note that affine function is very similar concept to linear functional, but for $\psi \in V^*$ we must have $\psi(0) = 0$, while for an affine function $f: V \to \RR$ we can have $f(0) \neq 0$. If $f: V \to \RR$ is an affine function such that $f(0) = 0$, then $f \in V^*$.
\begin{theorem}[Hyperplane separation theorem] \label{thm:duals-hyperplaneSeparation}
Let $V$ be a real, finite-dimensional vector space and let $X,Y \subset V$ be disjoint convex sets. Then there exists an affine function $f: V \to \RR$, that is a function such that for $v,w \in V$ and $\alpha \in \RR$ we have
\begin{equation}
f(\alpha v + (1-\alpha) w) = \alpha f(v) + (1-\alpha) f(w),
\end{equation}
such that
\begin{equation}
\max_{v \in X} f(x) \leq 0 \leq \max_{w \in Y} f(y).
\end{equation}
In other words, $f$ is non-positive on $X$ and non-negative on $Y$.
\end{theorem}
\begin{proof}
See \cite[Section 2.5.1]{BoydVandenberghe-convex}.
\end{proof}

\begin{theorem}[Strict hyperplane separation theorem] \label{thm:duals-strictHyperplaneSeparation}
Let $V$ be a real, finite-dimensional vector space equipped with Euclidean topology, let $X \subset V$ be a convex, closed sets and let $y \in V$ such that $y \notin X$. Then there exists an affine function $f: V \to \RR$ such that
\begin{equation}
\max_{v \in X} f(x) < 0 < f(y).
\end{equation}
\end{theorem}
\begin{proof}
See \cite[Section 2.5.1]{BoydVandenberghe-convex}.
\end{proof}

We have now all the tools we need to characterize the dual cone of the dual cone, i.e., the cone $C^{**}$ given as
\begin{equation}
C^{**} = \{ \xi \in V^{**} : \xi(\psi) \geq 0, \forall \psi \in C^* \}.
\end{equation}
\begin{proposition} \label{prop:duals-doubleDualCone}
Let $V$ be a real, finite-dimensional vector space and let $C \subset V$ be a convex, closed cone. Then $C^{**} = C$.
\end{proposition}
\begin{proof}
Since $V = V^{**}$, see Proposition \ref{prop:duals-doubleDual}, we must have $C^{**} \subset V$. Let $v \in C^{**}$ be such that $v \notin C$. According to Theorem \ref{thm:duals-strictHyperplaneSeparation} there is an affine function $f: V \to \RR$ such that
\begin{equation} \label{eq:duals-doubleDualCone-separation}
f(v) < 0 < \min_{x \in C} f(x).
\end{equation}
Now let $\psi \in V^*$ be given for $w \in V$ as $\psi(w) = f(w) - f(0)$.
It is easy to check that $\psi$ is linear: for $w', w'' \in V$ and $\alpha, \beta \in \RR$ we have $\alpha w' + \beta w'' = \alpha w' + \beta w'' + (1-\alpha - \beta) 0$ and so
\begin{align}
\psi(\alpha w' + \beta w'') &= f(\alpha w' + \beta w'' + (1-\alpha - \beta) 0) - f(0) \\
&= \alpha f(w') + \beta f(w'') + (1-\alpha - \beta) f(0) - f(0) \\
&= \alpha f(w') + \beta f(w'')  - (\alpha + \beta) f(0) \\
&= \alpha \psi(w') + \beta \psi(w'').
\end{align}
Using \eqref{eq:duals-doubleDualCone-separation} we get
\begin{equation} \label{eq:duals-doubleDualCone-min}
\psi(v) < \min_{x \in C} \psi(x)
\end{equation}
but since $0 \in C$ and $\psi(0) = 0$, we must have $\min_{x \in C} \psi(x) \leq 0$ and $\psi(v) < 0$. If $\min_{x \in C} \psi(x) = 0$, then $\psi \in C^*$ and $\psi(v) < 0$ is a contradiction with $v \in C^{**}$. So assume that $\min_{x \in C} \psi(x) < 0$, then there is $y \in C$ such that $\psi(y) < 0$. Let
\begin{equation}
\alpha = \dfrac{2 \psi(v)}{\psi(y)}
\end{equation}
then we have $\alpha y \in C$ and
\begin{equation}
\psi(\alpha y) = \dfrac{2 \psi(v)}{\psi(y)} \psi(y) = 2 \psi(v) < \psi(v)
\end{equation}
which is a contradiction with \eqref{eq:duals-doubleDualCone-min}.
\end{proof}

\section{Bilinear forms, linear maps and tensor products} \label{appendix:bilinear}

%
%
We are going to review several basic results and constructions on tensor products of vector spaces. We will be using the same approach as presented in \cite{Ryan-tensorProducts}. We will show how one can relate bilinear functionals, linear maps and tensor products of vector spaces. We will start from bilinear forms:
\begin{definition}
Let $V_A$, $V_B$ be real, finite-dimensional vector spaces. A \emph{bilinear form} is a map $B: V_A \times V_B \to \RR$, where $V_A \times V_B$ denotes the Cartesian product of $V_A$ and $V_B$, such that $B$ is linear in $V_A$ and $V_B$, i.e., such that for $v_A, w_A \in V_A$, $v_B, w_B \in V_B$ and $\alpha, \beta \in \RR$ we have
\begin{align}
B(\alpha v_A + \beta w_B, v_B) &= \alpha B(v_A, v_B) + \beta B(w_A, v_B), \\
B(v_A, \alpha v_B + \beta w_B) &= \alpha B(v_A, v_B) + \beta B(v_A, w_B).
\end{align}
\end{definition}

As first, we will show that every bilinear forms are in one-to-one correspondence with linear maps.
\begin{proposition} \label{prop:bilinear-LBcorrespondence}
Let $V_A$, $V_B$ be real, finite-dimensional vector spaces. Linear maps $L: V_A \to V_B$ are in one-to-one correspondence with bilinear forms $B: V_A \times V^*_B \to \RR$ via
\begin{equation} \label{eq:bilinear-LBcorrespondence}
\psi_B (L(v_A)) = B(v_A, \psi_B),
\end{equation}
where $v_A \in V_A$ and $\psi_B \in V^*_B$.
\end{proposition}
\begin{proof}
It is straightforward to check that given a linear map $L:V_A \to V_B$, we can define a bilinear form $B: V_A \times V^*_B \to \RR$ via \eqref{eq:bilinear-LBcorrespondence}. Now given a bilinear form $B: V_A \times V^*_B \to \RR$, fix $v_A \in V_A$ and define $\xi_B \in V^{**}_B$ as $\xi_B(\psi_B) = B(v_A, \psi_B)$, where $\psi_B \in V^*_B$. It is straightforward to check that $\xi_B$ is linear, so $\xi_B \in V_B^{**} = V_B$, where we used the result of Proposition \ref{prop:duals-doubleDual}. Now we can define a map $L: V_A \to V_B$ as $L(v_A) = \xi_B$. It is again straightforward to check that $L$ is linear. Let $\psi_B \in V^*_B$, then note that we have $\xi_B(\psi_B) = \psi_B(\xi_B)$ as a result of the isomorphism between $V_B$ and $V^{**}_B$. We have $\psi_B (L(v_A)) = \psi_B(\xi_B) = B(v_A, \psi_B)$ and so \eqref{eq:bilinear-LBcorrespondence} holds.
\end{proof}

Similar to linear functionals, also the set of all bilinear forms is a vector space. This is easy to see, let $B_1, B_2: V_A \times V_B \to \RR$ be bilinear forms and let $\alpha, \beta \in \RR$, then we define
\begin{equation}
(\alpha B_1 + \beta B_2)(v_A, v_B) = \alpha B_1(v_A, v_B) + \beta B_2(v_A, v_B),
\end{equation}
where $v_A \in V_A$ and $v_B \in V_B$. We will now look at the functionals on the vector space of bilinear forms, we will see that this leads to tensor products.

Let $v_A \in V_A$, $v_B \in V_B$ and let $B: V_A \times V_B \to \RR$ be a bilinear form, then the map $B \mapsto B(v_A, v_B)$ is a linear functional on the vector space of bilinear forms. We will use $v_A \otimes v_B$ to denote this functional, i.e., we have $(v_A \otimes v_B)(B) = B(v_A, v_B)$. It is straightforward to check that for $v_A, w_A \in V_A$, $v_B, w_B \in V_B$ and $\alpha, \beta \in \RR$ we have
\begin{align}
(\alpha v_A + \beta w_A) \otimes v_B &= \alpha v_A \otimes v_B + \beta w_A \otimes v_B, \\
v_A \otimes (\alpha v_B + \beta w_B) &= \alpha v_A \otimes v_B + \beta v_A \otimes w_B.
\end{align}
But note that in general
\begin{equation}
(v_A + w_A) \otimes (v_B + w_B) \neq v_A \otimes v_B + w_A \otimes w_B
\end{equation}
simply because
\begin{equation}
B(v_A + w_A, v_B + w_B) \neq B(v_A, v_B) + B(w_A, w_B).
\end{equation}

\begin{definition} \label{def:bilinear-tensorProduct}
Let $V_A$, $V_B$ be real, finite-dimensional vector spaces, then their \emph{tensor product} is the vector space
\begin{equation}
V_A \otimes V_B = \linspan(\{ v_A \otimes v_B : v_A \in V_A, v_B \in V_B \}).
\end{equation}
\end{definition}
Before we proceed, we will prove two simple results about the structure of $V_A \otimes V_B$.
\begin{lemma} \label{lemma:bilinear-productBasis}
Let $V_A$, $V_B$ be real, finite-dimensional vector spaces and let $\{ v_{1,A}, \ldots, v_{n_A, A}\} \subset V_A$ and $\{ v_{1,B}, \ldots, v_{n_B, B}\} \subset V_B$ be bases of $V_A$ and $V_B$ respectively. Then $\{ v_{i,A} \otimes v_{j,B} \}_{i,j = 1}^{n_A, n_B} \subset V_A \otimes V_B$ is a basis of $V_A \otimes V_B$.
\end{lemma}
\begin{proof}
The result follows easily from Definition \ref{def:bilinear-tensorProduct}. We only need to show that every vector of the form $w_A \otimes w_B$, where $w_A \in V_A$ and $w_B \in V_B$, can be written as linear combination of $\{ v_{i,A} \otimes v_{j,B} \}_{i,j = 1}^{n_A, n_B}$, but this is obvious.
\end{proof}

\begin{lemma} \label{lemma:bilinear-productSum}
Let $u_{AB} \in V_A \otimes V_B$, then there are $\{v_{1,A}, \ldots, v_{n,A}\} \subset V_A$ and $\{w_{1,A}, \ldots, w_{n,B}\} \subset V_B$ such that $u_{AB} = \sum_{i=1}^{n_A} v_{i,A} \otimes w_{i,B}$, where $\{ v_{1,A}, \ldots, v_{n,A} \}$ is a basis of $V_A$.
\end{lemma}
\begin{proof}
The result follows from Lemma \ref{lemma:bilinear-productBasis}. Let $\{ v_{1,A}, \ldots, v_{n_A, A}\} \subset V_A$ and $\{ v_{1,B}, \ldots, v_{n_B, B}\} \subset V_B$ be bases of $V_A$ and $V_B$ respectively, then there are numbers $\alpha_{ij} \in \RR$, $i \in \{1, \ldots, n_A\}$ and $j \in \{1, \ldots, n_B\}$ such that
\begin{equation}
u_{AB} = \sum_{i=1}^{n_A} \sum_{j=1}^{n_B} \alpha_{ij} v_{i,A} \otimes v_{j,B}.
\end{equation}
We then have
\begin{equation}
u_{AB} = \sum_{i=1}^{n_A} v_{i,A} \otimes \left( \sum_{j=1}^{n_B} \alpha_{ij} v_{j,B} \right).
\end{equation}
which is the form of $u_{AB}$ we wanted to obtain.
\end{proof}

One can actually show that $V_A \otimes V_B$ is the dual of the vector space of bilinear forms $B: V_A \times V_B \to \RR$ by constructing  a basis of the vector space of bilinear forms and showing that the dual basis is included in $V_A \otimes V_B$.

\begin{proposition} \label{prop:bilinear-bilinearDual}
$V_A \otimes V_B$ is the dual vector space to the vector space of bilinear forms $B: V_A \times V_B \to \RR$.
\end{proposition}
\begin{proof}
Let $\{ v_{1,A}, \ldots, v_{n_A,A}\} \subset V_A$ and $\{ v_{1,B}, \ldots, v_{n_B,B}\} \subset V_B$ be basis of $V_A$ and $V_B$ respectively. Let $B_{ij}: V_A \times V_B \to \RR$, where $i \in \{1, \ldots, n_A\}$ and $j \in \{1, \ldots, n_B\}$ be bilinear forms given as $B_{ij} (v_{k,A}, v_{\ell, B}) = \delta_{ik} \delta_{j \ell}$, where $\delta_{ik}, \delta_{j \ell}$ are the Kronecker deltas. The set $\{ B_{ij} \}_{i,j=1}^{n_A, n_B}$ is the basis of the vector space of bilinear forms: let $B:V_A \times V_B \to \RR$, then for any $w_A \in V_A$, $w_B \in V_B$ we have
\begin{equation}
B(w_A, w_B) = \sum_{i=1}^{n_a} \sum_{j=1}^{n_B} B(v_{i,A}, v_{j,B}) B_{ij}(w_A, w_B),
\end{equation}
which one can easily verify by writing $w_A$ and $w_B$ as linear combinations of the bases. It is now straightforward to verify that $\{ v_{i,A} \otimes v_{j,B} \}_{i,j=1}^{n_A, n_B}$ is the dual basis to $B_{ij}$, from which the result follows.
\end{proof}
It is a simple corollary of Proposition \ref{prop:bilinear-bilinearDual} that the dual of $V_A \otimes V_B$ is the vector space of the bilinear forms $B: V_A \times V_B \to \RR$. But we can also construct $V^*_A \otimes V^*_B$ and for $\psi_A \in V^*_A$, $\psi_B \in V^*_B$ we can define
\begin{equation}
(\psi_A \otimes \psi_B)(v_A \otimes v_B) = \psi_A(v_A) \psi_B(v_B).
\end{equation}
It follows that (up to an isomorphism that we will omit) that $\psi_A \otimes \psi_B \in (V_A \otimes V_B)^*$ and $V^*_A \otimes V^*_B \subset (V_A \otimes V_B)^*$. We will prove the other inclusion as well.
\begin{proposition} \label{prop:bilinear-dualOfTensor}
$V^*_A \otimes V^*_B = (V_A \otimes V_B)^*$.
\end{proposition}
\begin{proof}
Let $\{ v_{1,A}, \ldots, v_{n_A, A}\} \subset V_A$, $\{ v_{1,B}, \ldots, v_{n_B, B}\} \subset V_B$ be bases of $V_A$ and $V_B$ respectively and let $\{ \psi_{1,A}, \ldots, \psi_{n_A,A}\} \subset V^*_A$, $\{ \psi_{1,B}, \ldots, \psi_{n_B,B}\} \subset V^*_B$ be the dual bases. The according to Lemma \ref{lemma:bilinear-productBasis} we have that $\{ v_{i,A} \otimes v_{j,B} \}_{i,j=1}^{n_A, n_B}$ and $\{ \psi_{i,A} \otimes \psi_{j,B} \}_{i,j=1}^{n_A, n_B}$ are the bases of $V_A \otimes V_B$ and $V^*_A \otimes V^*_B$. We have
\begin{equation}
(\psi_{i,A} \otimes \psi_{j,B})(v_{k,A} \otimes v_{l,B}) = \psi_{i,A}(v_{k,A}) \psi_{j,B}(v_{l,B}) = \delta_{ik} \delta_{jl}
\end{equation}
and so $\{ \psi_{i,A} \otimes \psi_{j,B} \}_{i,j=1}^{n_A, n_B}$ is the dual basis to $\{ v_{i,A} \otimes v_{j,B} \}_{i,j=1}^{n_A, n_B}$. It follows that the linear hull of $\{ \psi_{i,A} \otimes \psi_{j,B} \}_{i,j=1}^{n_A, n_B}$ must be $(V_A \otimes V_B)^*$ and thus we get $V^*_A \otimes V^*_B = (V_A \otimes V_B)^*$.
\end{proof}

Finally, we can prove the following isomorphisms between the tensor product of vector spaces, the vector space of bilinear forms and the vectors space of linear maps.
\begin{proposition} \label{prop:bilinear-isomorphisms}
Let $V_A$, $V_B$ be real, finite-dimensional vector spaces. Then the following vector spaces are isomorphic:
\begin{itemize}
\item $V_A \otimes V_B$,
\item vector space of bilinear forms $B: V^*_A \times V^*_B \to \RR$,
\item vector space of linear maps $L: V^*_A \to V_B$,
\end{itemize}
\end{proposition}
\begin{proof}
We already know that the vector space of bilinear forms $B: V^*_A \times V^*_B \to \RR$ and the vector space of linear maps $L: V^*_A \to V_B$ are isomorphic as a result of Proposition \ref{prop:bilinear-LBcorrespondence}.

The vector space of bilinear forms $B: V^*_A \times V^*_B \to \RR$ is the dual of $V^*_A \otimes V^*_B$ as a result of Proposition \ref{prop:bilinear-bilinearDual} and $V^*_A \otimes V^*_B = (V_A \otimes V_B)^*$ as a result of Proposition \ref{prop:bilinear-dualOfTensor}. But then using the result of Proposition \ref{prop:duals-doubleDual} that $V^{**}$ is isomorphic to $V$ it follows that the vector space of bilinear forms $B: V^*_A \times V^*_B \to \RR$ is isomorphic to $(V^*_A \otimes V^*_B)^*$. The result follows from Proposition \ref{prop:bilinear-dualOfTensor} as we have $(V^*_A \otimes V^*_B)^* = V_A \otimes V_B$.
\end{proof}

\begin{corollary}
For $v_{AB} \in V_A \otimes V_B$ there is a bilinear form $B_v: V^*_A \times V^*_B \to \RR$ and linear map $L_v: V_A^* \to V_B$ such that for every $\psi_A \in V^*_A$ and $\psi_B \in V^*_B$ we have
\begin{equation}
(\psi_A \otimes \psi_B)(v_{AB}) = B_v(\psi_A, \psi_B) = \psi_B(L(\psi_A)).
\end{equation}
\end{corollary}
\begin{proof}
The result follows from Proposition \ref{prop:bilinear-isomorphisms}. $B_v$ and $L_v$ can be obtained from $v_{AB}$ using the corresponding isomorphisms given by Propositions \ref{prop:bilinear-LBcorrespondence} and \ref{prop:bilinear-bilinearDual}.
\end{proof}

\end{document}